\documentclass[11pt]{article}

\usepackage[margin=1in]{geometry}
\usepackage[T1]{fontenc}
\usepackage{lmodern}
\usepackage{amsmath, amssymb, amsthm, amsfonts}
\usepackage{mathtools}
\usepackage{bm}
\usepackage{graphicx}
\usepackage{caption}
\usepackage{subcaption}
\usepackage{booktabs}
\usepackage{multirow}
\usepackage{array}
\usepackage{enumerate}
\usepackage{algorithm}
\usepackage{algpseudocode}
\usepackage{xcolor}
\usepackage{setspace}
\usepackage[authoryear,round]{natbib}
\usepackage[colorlinks=true,allcolors=blue]{hyperref}
\usepackage{authblk}

\bibpunct{(}{)}{;}{a}{}{,}

\newcommand{\symbf}[1]{\boldsymbol{#1}}

\theoremstyle{plain}
\newtheorem{theorem}{Theorem}
\newtheorem{proposition}{Proposition}
\newtheorem{lemma}{Lemma}

\newtheorem{assumption}{Assumption}
\theoremstyle{definition}

\theoremstyle{remark}

\title{Sensitivity Analysis for Trials with Irregular and Informative Assessment Times:
Marginal Mean Models and Sensitivity Parameter Selection}

\author[1]{Yujing Gao}
\author[1]{Shu Yang}
\author[2,3]{Ravi Varadhan}
\author[4]{Andrew Redd}
\author[5]{Andrea J. Apter}
\author[6, *]{Daniel O. Scharfstein}

\affil[1]{Department of Statistics, North Carolina State University}
\affil[2]{Department of Oncology, Johns Hopkins School of Medicine}
\affil[3]{Department of Biostatistics, Bloomberg School of Public Health, Johns Hopkins University}
\affil[4]{Department of Internal Medicine, University of Utah School of Medicine}
\affil[5]{Department of Medicine, University of Pennsylvania Perelman School of Medicine}
\affil[6]{Department of Population Health Sciences, University of Utah School of Medicine}
\affil[*]{Daniel.Scharfstein@hsc.utah.edu}

\date{\today}

\begin{document}

\maketitle

\begin{abstract}
In randomized trials, outcome assessment times often vary across participants and may depend on the underlying outcome process, so outcomes observed near a target time need not represent the target population. Whether assessment is informative cannot be determined from the observed data alone. We extend the sensitivity analysis framework of Smith et al.\ (2024), which is anchored in an explainable assessment assumption and was restricted to continuous outcomes through an identity-link marginal mean model, to a broader class of link functions that accommodates binary and count outcomes. Within this framework, we develop an influence function-based augmented inverse intensity-weighted estimator that is root-$n$ consistent under standard regularity conditions. We also propose a data-driven procedure for selecting the range of the sensitivity parameter from plausible bounds on the target mean outcome and for constructing a confidence set for this range, which accounts for uncertainty in sensitivity parameter selection. We illustrate the method with data from the ARC study and evaluate its finite-sample performance in simulation studies.
\end{abstract}


\section{Introduction}
\label{sec:intro}

Randomized clinical trials are often designed to collect outcomes at pre-specified times following randomization. In practice, however, actual assessment times may vary substantially across participants. Such irregular assessment times complicate inference on treatment effects because outcomes are not observed for all participants at the target assessment times. Consequently, the outcomes observed around a given time may not be representative of the outcome distribution in the target population.

Assessment times are said to be informative when the distribution of outcomes among participants assessed at a given time differs from the corresponding population distribution at that time. For example, participants may postpone or miss assessments when their health status is worse, resulting in observed outcomes that tend to be better than those in the target population. Conversely, in studies that allow ``as-needed'' visits, participants may be more likely to undergo assessment when their health status deteriorates. Importantly, whether the assessment process is informative cannot generally be determined from the observed data alone \citep{Smith2024}.

Irregular assessment times can also complicate comparisons between treatment groups. Analyses based directly on outcomes at the observed assessment times may capture differences in assessment timing rather than differences in the underlying outcome trajectories. This issue can arise, for example, when assessment patterns differ between treatment groups or when treatment effects vary over time. Therefore, rather than focusing on outcomes at the random times at which participants happen to be assessed, it is often more scientifically relevant to estimate population mean outcomes and treatment effects at clinically meaningful target times.

Several approaches have been developed to account for informative assessment times, including inverse intensity weighting \citetext{\citealp*
{scharfstein2004dependent}; \citealp*{buzkova2007dependent}; \citealp*{buzkova2009repeated}; \citealp*{pullenayegum2013dr}; \citealp{sun2016quantile}}, semiparametric joint models \citetext{\citealp*{lin2001regression}; \citealp*{lin2003semiparametric}; \citealp*{sun2007informative}; \citealp*{liang2009jointmodeling}; \citealp{sun2011semiparametric}; \citealp*{sun2011regression}; \citealp{zhu2011regression}; \citealp{sun2012joint}}, pairwise likelihood approaches \citetext{\citealp*{chen2015regression}; \citealp{shen2019regression}}, and fully parametric models \citep{lipsitz2002dependent,ryu2007longitudinal,gasparini2020mixed}. However, analogous to analyses with missing data, inference about population outcomes at target assessment times generally requires assumptions about the relationship between the assessment process and the underlying outcome process that cannot be verified from the observed data.


In this paper, we extend the framework of \cite{Smith2024} by allowing for a broader class of link functions, thereby accommodating a wider range of outcome types, including binary and count outcomes. Within this generalized framework, we develop an influence function-based augmented inverse intensity-weighted estimator, enabling flexible semiparametric modeling while achieving root-$n$ consistency under standard regularity conditions. To address uncertainty in the selection of the sensitivity parameter, we provide a strategy for selecting the sensitivity parameter and constructing a confidence set for its plausible range. Simulation studies demonstrate that the proposed approach achieves satisfactory coverage performance.

The remainder of the paper is organized as follows. Section \ref{chp4:sec:setting} introduces the setting and notation. Section \ref{chp4:sec:framework} presents the proposed sensitivity analysis framework for general outcome types. Section \ref{chp4:sec:estimation} describes the estimation procedure, and Section \ref{sec:selection} develops the proposed method for sensitivity parameter selection. Section \ref{sec:analysis} illustrates the method using data from the ARC study. Section \ref{chp4:sec:simu} reports simulation results.

\section{Setting and Notation} 
\label{chp4:sec:setting}

We consider a randomized trial in which each participant has an outcome measured at baseline and at a sequence of post-baseline times. The number and timing of these post-baseline assessments may vary across participants. The primary objective is to estimate the treatment arm-specific mean outcome at target follow-up times. 

Let $[0,\tau]$ denote the study follow-up period, where $\tau$ is the end of follow-up. We focus on a time interval $[t_1,t_2]$, with $0<t_1<t_2\leq \tau$, that contains the target follow-up times of interest and over which assessments occur in both treatment arms. This interval is specified by the analyst and defines the time window for estimation and inference. We assume that there is no dropout during the study.

For each participant, let $A$ denote the treatment assignment, with $A=1$ for the treatment arm and $A=0$ for the control arm. For each time $t\in [0,\tau]$, let $Y(t)$ denote the (possibly unobserved) outcome at time $t$, and let $\symbf{X}(t)$ denote baseline and time-varying covariates collected up to time $t$. Let $N(t)$ denote the cumulative number of assessments up to time $t$, and define $\Delta N(t) = N(t) - N(t-)$ as the indicator of an assessment occurring at time $t$. Let $T_{k}$ denote the time of the participant's $k$th post-baseline assessment. We refer to $\{Y(t): t\in[0,\tau]\}$ as the outcome process and $\{N(t):t\in[0,\tau]\}$ as the assessment process. For each $t\in[0,\tau]$, let $\overline{\symbf{O}}(t)$ denote the observed history prior to time $t$, including baseline data, treatment assignment, assessment times before $t$, and all outcomes and covariates observed at those assessments.  Let $\symbf{O} = \overline{\symbf{O}}(\tau)$ denote the participant's observed data over the entire study.

For each $t\in[0,\tau]$, let $Y^a(t)$ denote the potential outcome that the participant would have at time $t$ under treatment assignment $A=a$ and $\mu^a(t) = \mathbb{E}\{Y^a(t)\}$ denote the corresponding population mean outcome at time $t$ were all participants assigned to treatment arm $A=a$ for $a=0,1$. Under randomization, $\mu^a(t) = \mathbb{E}\{Y(t)\mid A=a\}$, the mean outcome at time $t$ among participants assigned to treatment arm $a$, for $a\in\{0,1\}$. The treatment effect at time $t$ is then defined as $\delta(t) = \mu^1(t) - \mu^0(t)$.

To formalize the assumptions used in subsequent sections, we define the intensity function for the assessment process given the observed past as 
\begin{equation} 
  \label{eqn:lambda}
   \lambda\big\{t \mid  \overline{\symbf{O}}(t)\big\} 
   = 
   \lim_{\epsilon \to 0^+ } \big[\mathbb{P} \big\{N(t+\epsilon)-N(t-)  = 1 \mid  \overline{\symbf{O}}(t) \big\}  \big/ \epsilon \big] ,
\end{equation}
where $N(t+\epsilon)-N(t-)$ indicates whether the participant has an assessment
during the time interval $[t,t+\epsilon]$. For any time $t$ such that $\lambda\big\{t \mid \overline{\symbf{O}}(t)\big\}>0 $, we define 
\begin{align*}
     dF\big\{y(t) \mid \Delta N(t)=1, \overline{\symbf{O}}(t)\big\} 
     & = \lim_{\epsilon \to 0^+} dF \big\{y(t) \mid N(t+\epsilon)-N(t-) = 1, \overline{\symbf{O}}(t) \big\},\\
    dF\big\{y(t) \mid \Delta N(t)=0, \overline{\symbf{O}}(t)\big\} 
    & = \lim_{\epsilon \to 0^+} dF \big\{y(t) \mid  N(t+\epsilon)-N(t-) =0, \overline{\symbf{O}}(t) \big\},
\end{align*}
as the distributions of $Y(t)$ among individuals who were, and who were not, assessed at time $t$, given the observed history $\overline{\symbf{O}}(t)$.

\section{Sensitivity Analysis Framework for General Outcome Types}
\label{chp4:sec:framework}

When the outcome is continuous, \cite{Smith2024} developed an inferential framework for estimating the treatment effect $\delta(t)$ based on a sensitivity analysis under a range of plausible assumptions. Their approach proceeds by conducting inference for the treatment arm-specific mean outcomes $\mu^a(t)$ separately for $a \in \{0,1\}$.

We begin by reviewing their sensitivity analysis framework, in which all assumptions, distributions, and estimands are defined within each treatment arm. For notational simplicity, we suppress explicit dependence on treatment assignment throughout this section and reintroduce it when defining the treatment effect $\delta(t)$. We then extend this framework to accommodate more general outcome types by relaxing the modeling assumptions adopted in \cite{Smith2024} and introducing a more flexible class of marginal mean models.

\begin{assumption}[Tilting assumption] 
\label{chp4:assum:tilting}
For each time $t$ such that $\lambda\big\{t \mid \overline{\symbf{O}}(t)\big\} > 0$,    
\begin{equation}
    dF\big\{y(t) \mid \Delta N(t) = 0, \overline{\symbf{O}}(t)\big\} =
    dF\big\{y(t) \mid \Delta N(t) = 1, \overline{\symbf{O}}(t)\big\} \exp\{\alpha y(t)\} / c\left\{\overline{\symbf{O}}(t); \alpha\right\}
\end{equation}
where $c\left\{\overline{\symbf{O}}(t); \alpha\right\} = \mathbb{E}\left[\exp\{\alpha Y(t)\} \mid \Delta N(t) = 1, \overline{\symbf{O}}(t)  \right]$ and we assume $c\left\{\overline{\symbf{O}}(t);\alpha \right\} < \infty$.
\end{assumption}

Assumption \ref{chp4:assum:tilting} is a generalized version of the explainable assessment \citep{Smith2024} or the sequential ignorability \citep{scharfstein2004dependent}, meaning that outcomes among participants who are not assessed at a given time $t$, $dF\big\{y(t) \mid \Delta N(t) = 0, \overline{\symbf{O}}(t)\big\}$, tend to be larger or smaller than outcomes among similar participants who are assessed at time $t$, $dF\big\{y(t) \mid \Delta N(t) = 1, \overline{\symbf{O}}(t)\big\}$. Specifically, for a negative (positive) value of $\alpha$, the distribution of unobserved outcomes in the given arm is tilted to the left (right) relative to the distribution of observed outcomes in that arm, with smaller (larger) values of $Y(t)$ receiving greater weight. When $\alpha = 0$ (no tilt), Assumption \ref{chp4:assum:tilting} reduces to the explainable assessment assumption. As pointed out by \cite{Smith2024}, the tilting assumption can be adjusted by choosing a different form of $q\{t, Y(t);\alpha\}$ in place of $\alpha Y(t)$. The construction using $\alpha Y(t)$ is known as exponential tilting, which is widely used in sensitivity analysis for trials with missing or censored data \citep{barndorffnielsen1994inference, rotnitzky2001competing, birmingham2003patternmixture, vansteelandt2007nonmonotone, scharfstein2019globalsa}.

Because the assessment at time $t$ can depend on the current outcome $Y(t)$ under Assumption \ref{chp4:assum:tilting}, we define,
\begin{equation} 
  \label{eqn:rho}
   \rho\big\{t \mid \overline{\symbf{O}}(t), Y(t)\big\} 
   = 
   \lim_{\epsilon \to 0^+ } \big[\mathbb{P} \big\{N(t+\epsilon)-N(t-)  = 1 \mid \overline{\symbf{O}}(t), Y(t) \big\}  \big/ \epsilon \big].
\end{equation}

\begin{assumption}[Positivity assumption]
\label{chp4:assum:positivity}
There is some $c>0$ such that, for all $t$ in $[t_1,t_2]$, $\rho\big\{t\mid \overline{\symbf{O}}(t), Y(t)\big\} > c$ for all values of $Y(t)$ and $\overline{\symbf{O}}(t)$.
\end{assumption}

The intensity function $\rho\big\{t \mid \overline{\symbf{O}}(t), Y(t)\big\}$ is related to the intensity function $\lambda\big\{t \mid  \overline{\symbf{O}}(t)\big\}$ by the following proposition:

\begin{proposition}
\label{chp4:prop:lambda_rho}
Under Assumptions \ref{chp4:assum:tilting} and \ref{chp4:assum:positivity}, for each $t \in [t_1,t_2]$, 
\begin{align*}
    \rho\big\{t \mid \overline{\symbf{O}}(t), Y(t)\big\} = \lambda\big\{t \mid  \overline{\symbf{O}}(t)\big\}  \mathbb{E}\left[\exp\{\alpha Y(t)\} \mid \Delta N(t) = 1, \overline{\symbf{O}}(t)  \right] \big / \exp\{\alpha Y(t)\}.
\end{align*}
\end{proposition}
The proof is provided in Appendix \ref{apxB:appendix:condition}. 

To enable estimation of $\mu(t)$ from finite samples, we impose a smoothing assumption that allows information to be shared across different times.

\begin{assumption}[Generalized marginal mean assumption]
\label{chp4:assum:marginal} 
For all $t\in[t_1,t_2]$,
$\mu(t) = s\left\{\symbf{B}(t)^\prime \symbf{\beta}\right\}$, where $\symbf{B}(t) = (B_1(t), \cdots, B_p(t))^\prime$  is a specified vector of basis functions such that 
$\int_{t= t_1}^{t_2} \symbf{B}(t) \symbf{B}(t)^\prime dt$ is positive definite, $s\in \mathcal{C}^2$ is a known inverse link function of the link $g$ with strictly positive derivative $\frac{\partial s(z)}{\partial z} > 0$, and $\symbf{\beta} \in \mathbb{R}^p$ is a parameter vector.
\end{assumption}

Assumption \ref{chp4:assum:marginal} generalizes the marginal mean model in \cite{Smith2024} to accommodate a broader class of outcome types. In particular, the identity link used in \cite{Smith2024} is a special case of this formulation. More generally, different choices of link functions allow the framework to handle various outcome types. For example, for binary outcomes with support $\{0,1\}$, one may use the logit link $g(\mu) = \log\left\{\mu/ (1-\mu)\right\}$ with the inverse link $s(z) = 1/\{1+\exp(-z)\}$. For count outcomes with support $\{0,1,2,\cdots\}$, the log link $g(\mu) = \log(\mu)$ with inverse $s(z) = \exp(z)$ is appropriate.

\begin{proposition}
\label{chp4:prop:identify}
    The parameter $\symbf{\beta}$ is identified under Assumptions 
~\ref{chp4:assum:tilting}, ~\ref{chp4:assum:positivity} and ~\ref{chp4:assum:marginal}.
\end{proposition}

The proof is provided in Appendix \ref{apxB:appendix:condition}.

Finally, we assume that the assessment process depends on future values of the outcome and covariates, the current value of covariates, and past unobserved values of the outcome and covariates only through past observed data and the current value of the outcome, which is helpful in deriving the influence function for $\symbf{\beta}$ \citep{Smith2024}. 

\begin{assumption}[Non-future dependence assumption]
 \label{chp4:assum:nofuture}
Let $\symbf{L}=\{ Y(t), \symbf{X}(t) : 0 \leq t \leq \tau \}$.  Then 
\begin{equation*}
     \lim_{\epsilon \to 0^+ } \left[ \mathbb{P} \big\{ N(t+\epsilon)-N(t-) = 1 \mid \overline{\symbf{O}}(t), \symbf{L} \big\} / \epsilon \right] = \lim_{\epsilon \to 0^+ } \left[ \mathbb{P} \big\{ N(t+\epsilon)-N(t-) = 1 \mid \overline{\symbf{O}}(t), Y(t) \big\} /\epsilon \right].
\end{equation*}
\end{assumption}

\section{Estimation}
\label{chp4:sec:estimation}

\subsection{Influence Functions}
\label{subsection:influence}

Before deriving the observed-data influence function, we first obtain a full-data influence function for $\symbf{\beta}$ following the approach provided by \cite{diaz2011superlearner}. We specify a function $\symbf{B}(t)$ with $\int_{t=t_1}^{t_2}\symbf{B}(t) \symbf{B}(t)^\prime dt$ positive definite, 
and a strictly increasing inverse link function $s\in \mathcal{C}^2$, which we will use when we make Assumption~\ref{chp4:assum:marginal}.  However, we initially do not assume that the equality $ \mu(t) = s\left\{\symbf{B}(t)^\prime \symbf{\beta}\right\}$ in Assumption~\ref{chp4:assum:marginal} holds, and instead work in the nonparametric model on the full data $\symbf{L}$. 
We use a smooth function of the form $s\left\{\symbf{B}(t)^\prime \symbf{\gamma}\right\}$ to approximate the curve $\mu(t)$ with the following two types of loss functions:

\begin{align*}
    \mbox{Squared error loss in the transformed space: }&\mathcal{L}_1(\symbf{\gamma}) = \int_{t=t_1}^{t_2} \left[ s^{-1}\{\mu(t)\} - \symbf{B}(t)^\prime \symbf{\gamma}\right]^2 dt 
    \\
    \mbox{Quasi log-likelihood loss: }& \mathcal{L}_2(\symbf{\gamma}) = \int_{t=t_1}^{t_2} \left[ b\{\symbf{B}(t)^\prime \symbf{\gamma}\} - \mu(t)\symbf{B}(t)^\prime \symbf{\gamma}\right] dt
\end{align*}
where $b(z)$ is defined by 
$b(z) = \int s(z) dz$, which means $s(z) = \frac{\partial b(z)}{\partial z }$.

Corresponding to each loss function, we derive the observed-data influence functions for $\symbf{\beta}$.

\begin{theorem}[Influence functions of $\symbf{\beta}$]
\label{chp4:theorem:IF}
Under Assumptions \ref{chp4:assum:tilting},
\ref{chp4:assum:positivity},  \ref{chp4:assum:marginal},  \ref{chp4:assum:nofuture}, a class of influence functions for $\symbf{\beta}$ is given by 
\begin{align*}
\symbf{\varphi}_j(\symbf{O}) 
    = & \int_{t=t_1}^{t_2} \frac{\symbf{W}_j(t;\symbf{\beta})\left[ Y(t) - \mathbb{E}\left\{Y(t)\mid \overline{\symbf{O}}(t)\right\}\right]}{\rho\left\{t \mid \overline{\symbf{O}}(t), Y(t)\right\}} dN(t) \\
    \quad & + \int_{t=t_1}^{t_2} \symbf{W}_j(t;\symbf{\beta})
    \left[ \mathbb{E}\left\{Y(t)\mid \overline{\symbf{O}}(t)\right\} -  s\left\{ \symbf{B}(t)^\prime \symbf{\beta}\right\}\right]dt 
\end{align*}
where $\symbf{W}_j(t;\symbf{\beta})$ has the following form:
\begin{align*}
    \symbf{W}_1(t;\symbf{\beta}) & = \symbf{V}_1^{-1}\symbf{B}(t)\frac{\partial g(\mu)}{\partial \mu} \Big|_{\mu = s\left\{\symbf{B}(t)^\prime \symbf{\beta}\right\}}, \quad \symbf{V}_1 = \int_{u=t_1}^{t_2} \symbf{B}(u) \symbf{B}(u)^\prime du;
    \\
    \symbf{W}_2(t;\symbf{\beta}) & = \symbf{V}_2\left(\symbf{\beta}\right)^{-1}\symbf{B}(t), \quad \symbf{V}_2(\symbf{\beta}) = \int_{u=t_1}^{t_2} \symbf{B}(u)\symbf{B}(u)^\prime \frac{\partial s(z)}{\partial z} \Bigg|_{z = \symbf{B}(u)^\prime \symbf{\beta}}du.
\end{align*}
\end{theorem}

In Theorem \ref{chp4:theorem:IF}, $\symbf{W}_1(t;\symbf{\beta})$ corresponds to the squared error loss in the transformed space; 
$\symbf{W}_2(t;\symbf{\beta})$ corresponds to the quasi log-likelihood loss. The proof of Theorem \ref{chp4:theorem:IF} is provided in Appendix \ref{apxB:app:IF_deriv}.

\subsection{Estimation of the Mean Outcome}
\label{chp4:subsection:estimation}

Before implementing the influence functions in Theorem \ref{chp4:theorem:IF} to estimate $\symbf{\beta}$, as in \cite{Smith2024}, we need to fit two observed-data models: one for the intensity function in each arm $\lambda\{t \mid \overline{\symbf{O}}(t)\}$; and another for the conditional distribution of observed outcomes in each arm given the observed past $dF\{Y(t) = y(t) \mid \Delta N(t) = 1, \overline{\symbf{O}}(t)\}$. 

First, we model $\lambda\{t \mid\overline{\symbf{O}}(t)\}$ using an Andersen-Gill model \citep{andersen1982cox} stratified by assessment number $k$, $\lambda\{t \mid \overline{\symbf{O}}(t)\} = \lambda_{0,k}(t) \exp\{\symbf{\gamma}^{\prime}\symbf{Z}(t)\}D_k(t)$. Here $\symbf{Z}(t)$ is a specified function of the participants' observed past $\overline{\symbf{O}}(t)$, such as the outcomes at previous assessments. The function $\lambda_{0,k}(t)$ is an unspecified baseline intensity function for stratum $k$ estimated by kernel smoothing the Breslow estimator \citep{breslow1972discussion} of the cumulative baseline intensity functions, $\symbf{\gamma}$ is a parameter vector estimated by the partial likelihood, and $D_k(t)$ is an indicator that the participant is at risk for having the $k$-th assessment at time $t$.

Second, we model $dF\{Y(t) = y(t) \mid \Delta N(t) = 1, \overline{\symbf{O}}(t)\}$ using a single-index model \citep{chiang2012new}, where the conditional cumulative distribution function of $Y(t)$ is assumed to depend on predictors $\symbf{W}(t)$ only through a scalar $\symbf{\theta}^\prime \symbf{W}(t)$ and kernel smoothed with respect to $\symbf{\theta}^\prime \symbf{W}(t)$ via a bandwidth parameter $h$.  In the estimation process, $\symbf{\theta}$ and $h$ are jointly estimated by minimizing a pseudo sum of integrated squares. The estimator for the conditional cumulative distribution function of $Y(t)$, $\widehat{F}\{Y(t) = y(t) \mid \Delta N(t) = 1, \overline{\symbf{O}}(t) = \overline{\symbf{o}}(t)\}$, is a step function with jumps at all outcome values observed in the data.

Given these two observed-data models, suppose that we have data for $n$ independent individuals. 
We construct estimators $\widehat{\symbf{\beta}}$ and $\widehat{\mu}(t)$ using the following steps, where we use a subscript $i$ to denote data for individual $i$. 
For each individual $i$:

\begin{enumerate}
\item For each assessment $k$ with $T_{ik} \in [t_1,t_2]$, compute
\begin{align*}
& \widehat{\mathbb{E}} \big[\exp \{ \alpha Y(T_{ik}) \} \mid \Delta N(T_{ik})=1, \overline{\symbf{O}}_i(T_{ik})\big] \\
&\quad = \int_{y \in \mathcal{Y}} \exp( \alpha y )\, d\widehat{F} \big\{ Y(T_{ik})=y \mid \Delta N(T_{ik})=1, \overline{\symbf{O}}_i(T_{ik}) \big\},
\end{align*}
where $\mathcal{Y}$ is the set of all outcome values observed in the data. Then compute the estimated conditional mean
\begin{align*}
\widehat{\mathbb{E}}\big\{Y(T_{ik}) \mid \overline{\symbf{O}}_i(T_{ik})\big\} = \frac{\widehat{\mathbb{E}}\big[ Y(T_{ik}) \exp\{ \alpha Y(T_{ik})\} \mid \Delta N(T_{ik})=1, \overline{\symbf{O}}_i(T_{ik})\big] }{\widehat{\mathbb{E}}\big[ \exp\{ \alpha Y(T_{ik})\} \mid \Delta N(T_{ik})=1, \overline{\symbf{O}}_i(T_{ik})\big] },
\end{align*}
and, by Proposition \ref{chp4:prop:lambda_rho}, the estimated inverse intensity weight
\begin{align*}
& \widehat{\rho} \big\{ T_{ik} \mid \overline{\symbf{O}}_i(T_{ik}), Y_i(T_{ik})\big\} \\
&\quad = \widehat{\lambda}_{0,k}(T_{ik})\exp\{ \widehat{\symbf{\gamma}}^\prime \symbf{Z}_i(T_{ik})\} \exp\{- \alpha Y_i(T_{ik}) \} \\
&\qquad \times \widehat{\mathbb{E}}\big[ \exp\{ \alpha Y(T_{ik}) \} \mid \Delta N(T_{ik})=1,\overline{\symbf{O}}_i(T_{ik})\big].
\end{align*}
\item For each time $t \in [t_1,t_2]$, compute the predicted mean outcome given the observed past before time $t$:
\begin{align*}
\widehat{\mathbb{E}}\big\{Y(t) \mid \overline{\symbf{O}}_i(t)\big\} = \frac{\widehat{\mathbb{E}}\big[ Y(t) \exp\{ \alpha Y(t)\} \mid \Delta N(t)=1, \overline{\symbf{O}}_i(t)\big] }{\widehat{\mathbb{E}}\big[ \exp\{ \alpha Y(t)\} \mid \Delta N(t)=1, \overline{\symbf{O}}_i(t)\big] }.
\end{align*}
\item Compute
\begin{align*}
\widehat{\symbf{\Psi}}_j(\symbf{O}_i;\symbf{\beta})
={} & \sum_{k=1}^{K_i} \symbf{W}_j(T_{ik};\symbf{\beta})\,\frac{Y_i(T_{ik}) - \widehat{\mathbb{E}} \big\{Y(T_{ik}) \mid \overline{\symbf{O}}_i(T_{ik}) \big\} }{\widehat{\rho} \big\{T_{ik} \mid \overline{\symbf{O}}_i(T_{ik}), Y_i(T_{ik}) \big\}} \\
& + \int_{t=t_1}^{t_2} \symbf{W}_j(t;\symbf{\beta})\big[ \widehat{\mathbb{E}}\big\{Y(t)\mid \overline{\symbf{O}}_i(t)\big\} - s\big\{ \symbf{B}(t)^\prime \symbf{\beta}\big\} \big]\,dt.
\end{align*}
\end{enumerate}

The augmented inverse intensity-weighted estimator $\widehat{\symbf{\beta}}_j$ is solved from $\frac{1}{n}\sum_{i=1}^n \widehat{\symbf{\Psi}}_j(\symbf{O}_i;\symbf{\beta}) = 0$ and the corresponding estimated mean outcome is $\widehat{\mu}_j(t) = s\big\{\symbf{B}(t)^\prime \widehat{\symbf{\beta}}_j\big\}$, $j\in\{1,2\}$.

\subsection{Large-sample Distribution}
\label{subsection:large_sample}

If the models for $\lambda\{t\mid \overline{\symbf{O}}(t)\}$ and $dF\{Y(t) = y(t) \mid \Delta N(t) = 1, \overline{\symbf{O}}(t)\}$ are both correctly specified, and if Assumptions \ref{chp4:assum:tilting}--\ref{chp4:assum:nofuture} and additional regularity assumptions hold, then $\widehat{\symbf{\beta}}_j$ is a consistent estimator of the true value $\symbf{\beta}_0$ and $\sqrt{n}\big(\widehat{\symbf{\beta}}_j - \symbf{\beta}_0\big)  \overset{d}{\rightarrow}
\mathcal{N}(\symbf{0},  \Sigma_j^*)$, where $\Sigma_j^*$ is the influence function-based variance, $j\in\{1,2\}$. The proof is provided in Appendix \ref{apxB:app:asy_proof}. Here we derive the second-order remainder term in an expansion of $\sqrt{n}\big(\widehat{\symbf{\beta}}_j - \symbf{\beta}_0\big)$ following \cite{kennedy2016empiricalprocesses} and we show conditions under which the remainder term is asymptotically negligible.

Based on this result, we can derive the influence function-based variance estimators for $\widehat{\symbf{\beta}}_j$ as $\frac{1}{n}\widehat{\Sigma}_j^*$. The mean at any fixed time $t\in[t_1,t_2]$, $\mu(t) = \mathbb{E}\{Y(t)\} = s\{\symbf{B}(t)^\prime \symbf{\beta}\}$, can be estimated as $\widehat{\mu}_j(t) = s\{\symbf{B}(t)^\prime \widehat{\symbf{\beta}}_j\}$ with influence function-based variance estimator $\widehat{Var}\{\widehat{\mu}_j(t)\} = \frac{1}{n} \left[\frac{\partial s(z)}{\partial z} \big|_{z= \symbf{B}(t)^\prime \widehat{\symbf{\beta}}_j}\right]^2 \symbf{B}(t)^\prime\widehat{\Sigma}_j^* \symbf{B}(t)$, $j=1,2$.

A Wald confidence interval for $\mu(t)$ can be constructed using the influence function-based variance estimator or the variance estimator derived from the parametric bootstrap. The nonparametric bootstrap is not a feasible way of constructing confidence intervals, since ties in bootstrapped datasets cause estimates of conditional cumulative distribution functions based on the single-index model to be undefined \citep{Smith2024}.

Let $\alpha_a$ be the sensitivity parameter for treatment arm $a$ ($a=0,1$). For inference on $\delta(t)$, the estimation procedure above is repeated in the treatment arm to estimate $\mu^1(t)$ under a range of $\alpha_1$ values, and separately in the control arm to estimate $\mu^0(t)$ under a range of $\alpha_0$ values. These results are then combined to estimate the treatment effect as $\widehat{\delta}_j(t) = \widehat{\mu}_j^1(t) - \widehat{\mu}_j^0(t)$ over a grid of sensitivity parameters $(\alpha_1,\alpha_0)$, $j=1,2$.

\section{Selection and Inference for Sensitivity Parameter Values}
\label{sec:selection}

\subsection{True Range of the Sensitivity Parameter}
\label{subsection:true_range}

The estimation and inference of the mean outcome trajectory $\mu^a(t)$ in treatment arm $a$ are conducted over a range of values of a sensitivity parameter $\alpha_a$, $a=0,1$. Suppose clinical experts provide a plausible range $(\mu_a^{l}, \mu_a^{u})$ for the outcome mean $\mu^a(t)$ in each treatment arm, where $\mu_a^{l}< \mu_a^{u}$. For a given treatment arm $a$, we want to find the set of sensitivity parameter values $\alpha_a$ such that 
\begin{equation}
    \mu^l_{a} < \mathbb{E}\{Y^a(t);\alpha_a\} < \mu^u_{a}, \quad t \in [t_1, t_2], a\in\{0,1\}.
    \label{equ:require}
\end{equation}
In the simulation studies and the data analysis of  \cite{Smith2024}, we observed that the estimator $\widehat{\mathbb{E}}\{Y^a(t);\alpha_a\}$ increases with $\alpha_a$ for all $t$. The following result shows that this monotonicity also holds at the population level given Assumption \ref{chp4:assum:tilting} and bounded outcome values. 

\begin{theorem}
[Monotonicity]
\label{chp4:theorem:mono} Under Assumption \ref{chp4:assum:tilting} and assuming that the outcome is bounded for $t\in [t_1,t_2]$,  for each treatment arm $a$ and each time $t \in [t_1,t_2]$, $\mathbb{E}\{Y^a(t);\alpha_a \}$ is non-decreasing in $\alpha_a$.
\end{theorem}

The proof is provided in Appendix \ref{apxB:app:mono_proof}. Based on Theorem \ref{chp4:theorem:mono}, the true range of $\alpha_a$ that meets the requirement \eqref{equ:require} is an interval $(\alpha_a^l,\alpha_a^u)$: 
\begin{align*}
    \mathcal{W}_a = 
    \big\{\alpha \in [\alpha_a^{1}, \alpha_a^{2}]: \min\limits_{t\in[t_1,t_2]}\mathbb{E} \{Y^a(t);\alpha\} > \mu_a^{l}, \max\limits_{t\in[t_1,t_2]}\mathbb{E}\{Y^a(t);\alpha\} < \mu_a^{u}\big\} = (\alpha_a^l,\alpha_a^u).
\end{align*}
where $[\alpha_a^{1}, \alpha_a^{2}]$ is the initial search range for $\alpha_a$. Then the true values of $\alpha_a^l$ and $\alpha_a^u$ can be defined by the following optimization problems:
\begin{align*}
    \alpha_a^l & = \underset{[\alpha_a^{1}, \alpha_a^{2}]}{\mbox{inf }}\alpha \ \text{ s.t. } \min\limits_{t\in [t_1,t_2]}\mathbb{E}\{Y^a(t);\alpha\} > \mu_a^{l}, \\
    \alpha_a^u & = \underset{[\alpha_a^{1}, \alpha_a^{2}]}{\mbox{sup }} \alpha\  \text{ s.t. } \max\limits_{t\in [t_1,t_2]} \mathbb{E}\{Y^a(t);\alpha\} < \mu_a^{u}.
\end{align*}

Algorithm~\ref{alg:point_est} describes the procedure for estimating the range of the sensitivity parameter. We first discretize the initial search space into a finite grid of candidate values $\mathcal{A}_M$ and the study time interval into a finite grid of time points $\mathcal{T}_K$. The lower and upper bounds, $\alpha_a^l$ and $\alpha_a^u$, are then obtained by solving the corresponding optimization problems over $\mathcal{A}_M$, with the constraints evaluated at each time point in $\mathcal{T}_K$. Algorithm~\ref{alg:bias_corr} then uses a multiplier bootstrap to obtain bias-corrected estimates of these bounds, applying the same adaptive search at each bootstrap iteration.


Several practical issues arise in implementing these algorithms. First, for some datasets, no value of $\alpha \in \mathcal{A}_M$ satisfies condition \eqref{equ:require} in Algorithm \ref{alg:point_est}. This situation typically reflects poor finite-sample performance of the augmented inverse intensity-weighted estimator, leading to the failure of the initial search procedure. In such cases, the bias-correction procedure in Algorithm \ref{alg:bias_corr} is not performed because the original point estimate is unavailable. Second, even when Algorithm \ref{alg:point_est} successfully identifies $\widehat{\alpha}_a^l$ and $\widehat{\alpha}_a^u$, some multiplier bootstrap replicates may fail to produce valid estimates. Specifically, extreme values of the weighted influence function may cause the bootstrap estimator $\widehat{\symbf{\beta}}_j^{a,b}(\alpha)$ to deviate substantially from the corresponding original estimator $\widehat{\symbf{\beta}}_j^a(\alpha)$. Consequently, the optimization problem used to determine $\widehat{\alpha}_a^{l,b}$ or $\widehat{\alpha}_a^{u,b}$ may have no feasible solution. For example, if the bootstrap estimator $\widehat{\mathbb{E}}\{Y^a(t);\alpha\}$ exceeds the upper threshold $\mu_a^u$ for every $t \in \mathcal{T}_K$ and every $\alpha \in \mathcal{A}_M$, then the constraint $\max_{t \in \mathcal{T}_K} \widehat{\mathbb{E}}\{Y^a(t); \alpha\} < \mu_a^u$ cannot be satisfied, and $\widehat{\alpha}_{a}^{u,b}$ cannot be identified. Similarly, if $\widehat{\mathbb{E}}\{Y^a(t);\alpha\}$  is below the lower threshold $\mu_a^l$ for every  $t \in \mathcal{T}_K$ and every $\alpha \in \mathcal{A}_M$, then $\min_{t \in \mathcal{T}_K} \widehat{\mathbb{E}}\{Y^a(t); \alpha\} > \mu_a^l$  cannot be satisfied, and $\widehat{\alpha}_{a}^{l,b}$  cannot be identified. When either situation occurs, the corresponding bootstrap estimate is recorded as missing (NA). These invalid bootstrap replicates are excluded when computing the bias-corrected estimators of $\alpha_a^l$ and $\alpha_a^u$. The bias correction is therefore based only on bootstrap replicates for which the optimization problem admits a feasible solution.

The default implementation of Steps 2--3 in Algorithm \ref{alg:bias_corr} uses a grid search, which requires evaluating the multiplier bootstrap estimator over all candidate values in $\mathcal{A}_M$. When the search grid is sufficiently fine, this procedure can be computationally expensive. To improve efficiency, these steps may instead be implemented using a binary search, exploiting the monotonicity of $E\{Y^a(t);\alpha_a\}$ with respect to $\alpha_a$ established in Theorem \ref{chp4:theorem:mono}. In practice, we recommend an adaptive search strategy. Specifically, Steps 2--3 are first implemented using the binary search algorithm. If the binary search fails to identify a feasible solution (e.g., due to NA values arising during the bootstrap procedure), the algorithm automatically switches to the grid search. This hybrid approach substantially reduces computational cost while retaining the numerical robustness of the grid search in challenging bootstrap replicates.

\begin{algorithm}[!htbp]
\setlength{\baselineskip}{0.9\baselineskip}
\caption{Original point estimation for $\alpha_a^l$ and $\alpha_a^u$.}
\begin{itemize}
    \item Step 1: For each $\alpha \in \mathcal{A}_M$, compute the estimated influence function $\widehat{\symbf{\Psi}}_j(\symbf{O}_i; \symbf{\beta}^a, \alpha)$ for $i = 1, \ldots, n$.
     
    \item Step 2:  For each $\alpha \in \mathcal{A}_M$, obtain the augmented inverse intensity-weighted estimator $\widehat{\symbf{\beta}}_j^a(\alpha)$ by solving the estimating equation $\frac{1}{n} \sum_{i=1}^n \widehat{\symbf{\Psi}}_j(\symbf{O}_i; \symbf{\beta}^a, \alpha) = 0.$ 
    \item Step 3: For each $\alpha \in \mathcal{A}_M$ and each time point $t \in \mathcal{T}_K$, compute the marginal mean $\widehat{\mathbb{E}}\{Y^a(t); \alpha\} = s\big\{\symbf{B}(t)^\prime \widehat{\symbf{\beta}}_j^a(\alpha)\big\}$.
    \item Step 4: For each $\alpha \in \mathcal{A}_M$, compute the minimum and maximum estimated marginal means over the time grid $\mathcal{T}_K$:
    \[
    \widehat{{\cal M}}_K^{\min}(\alpha)
    =
    \min_{t \in \mathcal{T}_K}
    \widehat{\mathbb{E}}\{Y^a(t); \alpha\},
    \qquad
    \widehat{{\cal M}}_K^{\max}(\alpha)
    =
    \max_{t \in \mathcal{T}_K}
    \widehat{\mathbb{E}}\{Y^a(t); \alpha\}.
    \]
    \item Step 5 (Identify the smallest $\alpha$): Find the smallest $\alpha \in \mathcal{A}_M$ such that:
    \begin{itemize}
        \item  $\widehat{{\cal M}}_{K}^{\min}(\alpha) > \mu_a^l$, and 
        \item $\widehat{{\cal M}}_{K}^{\min}(\alpha)$ is locally monotone increasing in $\alpha$ at this value.
    \end{itemize}
    Denote this value by $\widehat{\alpha}_a^l$ and use it as the point estimate of $\alpha_a^l$.
    \item Step 6 (Identify the largest $\alpha$): Find the largest $\alpha \in \mathcal{A}_M$ such that:
    \begin{itemize}
        \item $\widehat{\mathcal{M}}_{K}^{\max}(\alpha) < \mu_a^u$, and 
        \item $\widehat{\mathcal{M}}_{K}^{\max}(\alpha)$ is locally monotone increasing in $\alpha$ at this value.
    \end{itemize}
    Denote this value by $\widehat{\alpha}_a^u$ and use it as the point estimate of $\alpha_a^u$.
\end{itemize}
\label{alg:point_est}
\end{algorithm}

\begin{algorithm}[!htbp]
\setlength{\baselineskip}{0.9\baselineskip}
\caption{Multiplier bootstrap bias-correction point estimation for $\alpha_a^l$ and $\alpha_a^u$.}
\begin{itemize}
    \item Follow Algorithm \ref{alg:point_est} to find $\widehat{\alpha}_a^u$ and $\widehat{\alpha}_a^l$, and record the augmented inverse intensity-weighted estimator $\widehat{\symbf{\beta}}_j^a(\alpha)$, the estimated influence function $\widehat{\symbf{\Psi}}_j(\symbf{O}_i; \symbf{\beta}^a, \alpha)$ and their derivatives $\widehat{\symbf{\Psi}}_j^1(\symbf{O}_i; \symbf{\beta}^a, \alpha)$, for $i = 1, \ldots, n$ and $\alpha \in \mathcal{A}_M$.
    \item For each bootstrap replication, $b=1,\cdots, B$:
    \begin{itemize}
    \item Step 1: Generate independent multiplier weights $w_1^b,\cdots, w_n^b$ from a mean-one, variance-one distribution (e.g., standard normal or Rademacher distribution). 
    \item Step 2: For each $\alpha \in \mathcal{A}_M$, find the multiplier bootstrap estimator $\widehat{\symbf{\beta}}_j^{a,b}(\alpha)$ by solving $\left[
     \frac{1}{n} \sum_{i=1}^nw_i^b \widehat{\symbf{\Psi}}_j\{\symbf{O}_i;\symbf{\beta}^a, \alpha\}
     \right] = 0$; or compute it as
    $$ \widehat{\symbf{\beta}}_j^{a,b}(\alpha) 
    = \widehat{\symbf{\beta}}_j^a(\alpha) - \left[
     \frac{1}{n}\sum_{i=1}^n \widehat{\symbf{\Psi}}^1_j\{\symbf{O}_i;\widehat{\symbf{\beta}}_j^a(\alpha), \alpha\}
     \right]^{-1} 
     \left[
     \frac{1}{n} \sum_{i=1}^n (w_i^b - 1) \widehat{\symbf{\Psi}}_j\{\symbf{O}_i;\widehat{\symbf{\beta}}_j^a(\alpha), \alpha\}
     \right],$$
     where we refer to the second form as the one-step multiplier bootstrap estimator and use it as the default.
    \item Step 3: Repeat Steps 3--6 of Algorithm \ref{alg:point_est} to obtain $\widehat{\alpha}_{a}^{l,b}$ and $\widehat{\alpha}_{a}^{u,b}$. 
\end{itemize}
\item Bias-corrected estimators for $\alpha_a^l$ and $\alpha_a^u$ are 
$$\widehat{\alpha}_a^{l,corr} = \widehat{\alpha}_a^l - \left(\frac{1}{B}\sum_{b=1}^B \widehat{\alpha}_a^{l,b} -\widehat{\alpha}_a^{l}\right),
\ \widehat{\alpha}_a^{u,corr} = \widehat{\alpha}_a^u - \left(\frac{1}{B}\sum_{b=1}^B \widehat{\alpha}_a^{u,b} -\widehat{\alpha}_a^u\right).$$
\end{itemize}
\label{alg:bias_corr}
\end{algorithm}


\subsection{Confidence Set for the Range of Sensitivity Parameter}
\label{subsection:true_set}

Having provided estimators for $\alpha_a^l$ and $\alpha_a^u$, we now construct a confidence set $(\widehat{\alpha}_a^l - \delta_a^l, \widehat{\alpha}_a^u + \delta_a^u)$ for $(\alpha_a^l, \alpha_a^u)$ such that the probability $\mathbb{P} \big\{(\alpha_a^l, \alpha_a^u) \subset (\widehat{\alpha}_a^l - \delta_a^l, \widehat{\alpha}_a^u + \delta_a^u)\big\}$ is at least the nominal confidence level. Suppose that we want to construct a $95\%$ confidence set $(\widehat{\alpha}_a^l - \delta_a^l, \widehat{\alpha}_a^u + \delta_a^u)$. We can find the $97.5\%$ lower confidence bound $\widehat{\alpha}_a^l - \delta_a^l$ for $\alpha_a^l$ and $97.5\%$ upper confidence bound $\widehat{\alpha}_a^u + \delta_a^u$ for $\alpha_a^u$. 
Using the fact that $\alpha_a^l < \alpha_a^u$, we have
\begin{align*}
    \mathbb{P} \big\{(\alpha_a^l, \alpha_a^u) \subset (\widehat{\alpha}_a^l - \delta_a^l, \widehat{\alpha}_a^u + \delta_a^u)\big\} & = \mathbb{P} \big(\widehat{\alpha}_a^l - \delta_a^l \leq \alpha_a^l, \alpha_a^u \leq \widehat{\alpha}_a^u + \delta_a^u \big)\\
    & = 1 - \mathbb{P}\big(\alpha_a^l < \widehat{\alpha}_a^l - \delta_a^l \mbox{ or } \alpha_a^u > \widehat{\alpha}_a^u + \delta_a^u \big) \\
    & \geq 1 - \mathbb{P}\big(\alpha_a^l < \widehat{\alpha}_a^l - \delta_a^l\big) -
    \mathbb{P}\big(\alpha_a^u > \widehat{\alpha}_a^u + \delta_a^u\big)\\
    & \geq 0.95.
\end{align*}

Thus, $(\widehat{\alpha}_a^l - \delta_a^l, \widehat{\alpha}_a^u + \delta_a^u)$ is a $95\%$ confidence set for $(\alpha_a^l, \alpha_a^u)$.

Algorithm \ref{alg:double_boot} presents a double bootstrap procedure for constructing this 95\% confidence set based on the one-step multiplier bias-correction procedure. In practice, the one-step multiplier bias-correction procedure may be replaced by the multiplier bootstrap procedure described in Algorithm \ref{alg:bias_corr}. In addition, the default implementation of Step 2 and Steps 2.2--2.3 in Algorithm \ref{alg:double_boot} uses a grid search. As discussed previously, these steps can instead be implemented using the adaptive search strategy to improve computational efficiency: the binary search is applied first, and the procedure reverts to the grid search only if NA issues are encountered. In the simulation studies presented in Section \ref{chp4:sec:simu}, all algorithms were implemented using the one-step multiplier bootstrap procedure together with the adaptive search strategy.

The implementation issues discussed for Algorithms \ref{alg:point_est} and \ref{alg:bias_corr} also apply to Algorithm \ref{alg:double_boot}. Specifically, if no value of $\alpha \in \mathcal{A}_M$ satisfies condition \eqref{equ:require} in Algorithm \ref{alg:point_est}, the original point estimate cannot be obtained, and the double bootstrap procedure is therefore not performed. Moreover, even when Algorithm \ref{alg:point_est} successfully identifies the original point estimates $\widehat{\alpha}_a^l$ and $\widehat{\alpha}_a^u$, some bootstrap or double-bootstrap replicates may fail to identify feasible values of $\alpha$ because the optimization problem has no solution. In these cases, the corresponding bootstrap estimates are recorded as missing (NA) and excluded from the bias-correction and confidence-set construction procedures.

\begin{algorithm}[!htbp]
\footnotesize
\setlength{\baselineskip}{0.9\baselineskip}
\setlength{\abovedisplayskip}{3pt}
\setlength{\belowdisplayskip}{3pt}
\setlength{\abovedisplayshortskip}{2pt}
\setlength{\belowdisplayshortskip}{2pt}
\caption{Multiplier double bootstrap bias-correction confidence set for $\alpha_a^l$ and $\alpha_a^u$.}
\begin{itemize}
    \item Step 1: Follow Algorithm \ref{alg:point_est} to find $\widehat{\alpha}_a^u$ and $\widehat{\alpha}_a^l$, and record the estimated coefficient $\widehat{\symbf{\beta}}_j^a(\alpha)$ and the estimated influence function $\widehat{\symbf{\Psi}}_j(\symbf{O}_i; \symbf{\beta}^a, \alpha)$ and their derivatives $\widehat{\symbf{\Psi}}_j^1(\symbf{O}_i; \symbf{\beta}^a, \alpha)$, for $i = 1, \ldots, n$ and $\alpha \in \mathcal{A}_M$.
    \item Step 2: For each outer layer bootstrap replication, $b_1=1,\cdots, B_1$:
    \begin{itemize}
    \item Generate independent multiplier weights $w_1^{b_1},\cdots, w_n^{b_1}$ from a mean-one, variance-one distribution (e.g., standard normal or Rademacher distribution). 
    \item For each $\alpha \in \mathcal{A}_M$, compute the one-step multiplier bootstrap estimator of $\widehat{\symbf{\beta}}_j^{a,b_1}(\alpha)$ as 
    \begin{align*}
        \widehat{\symbf{\beta}}_j^{a,b_1}(\alpha) 
        & = \  \widehat{\symbf{\beta}}_j^a(\alpha) - 
       \left[
     \frac{1}{n}\sum_{i=1}^n \widehat{\symbf{\Psi}}^1_j\{\symbf{O}_i; \widehat{\symbf{\beta}}_j^a(\alpha), \alpha\}
     \right]^{-1} 
     \left[
     \frac{1}{n} \sum_{i=1}^n (w_i^{b_1} -1)\widehat{\symbf{\Psi}}_j\big\{\symbf{O}_i;\widehat{\symbf{\beta}}_j^a(\alpha), \alpha\big\}
     \right].
    \end{align*} 
    \item Repeat Steps 3--6 of Algorithm \ref{alg:point_est} to obtain $\widehat{\alpha}_{a}^{l,b_1}$ and $\widehat{\alpha}_{a}^{u,b_1}$. 
    \item For each inner layer bootstrap replication, $b_2 =1,\cdots, B_2$:
    \begin{itemize}
        \item Step 2.1: Generate independent multiplier weights $w_1^{b_1,b_2},\cdots, w_n^{b_1,b_2}$ from a mean-one, variance-one distribution (e.g., standard normal or Rademacher distribution). 
       \item Step 2.2: For each $\alpha \in \mathcal{A}_M$, 
       compute the inner bootstrap estimator $\widehat{\symbf{\beta}}_j^{a,b_1,b_2}(\alpha)$ as
        \begin{align*}
        \widehat{\symbf{\beta}}_j^{a, b_1,b_2}(\alpha) 
        & = \  \widehat{\symbf{\beta}}_j^{a,b_1}(\alpha) - 
       \left[
     \frac{1}{n}\sum_{i=1}^n \widehat{\symbf{\Psi}}^1_j\{\symbf{O}_i; \widehat{\symbf{\beta}}_j^{a,b_1}(\alpha), \alpha \}
     \right]^{-1} 
     \left[
     \frac{1}{n} \sum_{i=1}^n (w_i^{b_1, b_2} - 1 ) \widehat{\symbf{\Psi}}_j\big\{\symbf{O}_i;\widehat{\symbf{\beta}}_j^{a,b_1}(\alpha), \alpha\big\}
     \right].
    \end{align*} 
       \item For each $t \in \mathcal{T}_K$, compute $\widehat{\mathbb{E}}^{b_1,b_2}\{Y^a(t);\alpha\} = s\big\{\symbf{B}(t)^\prime \widehat{\symbf{\beta}}_j^{a,b_1,b_2}(\alpha)\big\}$,
      $\widehat{{\cal M}}_K^{min, b_1,b_2}(\alpha) = \min\limits_{t\in \mathcal{T}_K} \widehat{\mathbb{E}}^{b_1,b_2}\{Y^a(t);\alpha\}$, $\widehat{{\cal M}}_K^{max,b_1,b_2}(\alpha) = \max\limits_{t\in \mathcal{T}_K} \widehat{\mathbb{E}}^{b_1,b_2}\{Y^a(t);\alpha\}.$
    \item Step 2.3: Repeat Steps 5--6 of Algorithm \ref{alg:point_est} to obtain $\widehat{\alpha}_{a}^{l,b_1,b_2}$ and $\widehat{\alpha}_{a}^{u,b_1,b_2}$. 
    \end{itemize}
    \item Bias-corrected outer bootstrap estimators are $$\widehat{\alpha}_a^{l,b_1,corr} = \widehat{\alpha}_a^{l,b_1} - \left(\frac{1}{B_2}\sum_{b_2=1}^{B_2} \widehat{\alpha}_a^{l,b_1,b_2} -\widehat{\alpha}_a^{l,b_1}\right), \ \widehat{\alpha}_a^{u,b_1,corr} = \widehat{\alpha}_a^{u,b_1} - \left(\frac{1}{B_2}\sum_{b_2=1}^{B_2} \widehat{\alpha}_a^{u,b_1,b_2} -\widehat{\alpha}_a^{u,b_1}\right).$$
\end{itemize}
\item Step 3:  Bias-corrected estimators for $\alpha_a^l$ and $\alpha_a^u$ are 
$$\widehat{\alpha}_a^{l,corr} = \widehat{\alpha}_a^l - \left(\frac{1}{B_1}\sum_{b_1=1}^{B_1} \widehat{\alpha}_a^{l,b_1} -\widehat{\alpha}_a^{l}\right),
\ \widehat{\alpha}_a^{u,corr} = \widehat{\alpha}_a^u - \left(\frac{1}{B_1}\sum_{b_1=1}^{B_1} \widehat{\alpha}_a^{u,b_1} -\widehat{\alpha}_a^u\right).$$
\item Step 4: We can construct two kinds of confidence sets: 
\begin{itemize}
    \item 1. Basic Bootstrap CI: find the 97.5\% quantile of $\{\widehat{\alpha}_a^{l,b_1,corr}\}_{b_1=1}^{B_1}$ as $C^{l,perc}_{0.975}$ and the 2.5\% quantile of $\{\widehat{\alpha}_a^{u,b_1,corr}\}_{b_1=1}^{B_1}$ as $C^{u,perc}_{0.025}$, then the 95\% confidence set is $(2\widehat{\alpha}_a^{l,corr} - C^{l,perc}_{0.975}, 2\widehat{\alpha}_a^{u,corr}  - C^{u,perc}_{0.025})$.
    \item 2. Wald CI: calculate the bias-corrected outer bootstrap standard error as $sd\{\widehat{\alpha}_a^{u,b_1,corr}\}_{b_1=1}^{B_1}$ and $sd\{\widehat{\alpha}_a^{l,b_1,corr}\}_{b_1=1}^{B_1}$, then the 95\% confidence set is $( \widehat{\alpha}_a^{l,corr} - q_{0.975} sd\{\widehat{\alpha}_a^{l,b_1,corr}\}_{b_1=1}^{B_1}, \widehat{\alpha}_a^{u,corr} + q_{0.975} sd\{\widehat{\alpha}_a^{u,b_1,corr}\}_{b_1=1}^{B_1})$, where $q_{0.975}$ is the standard normal 97.5\% quantile. 
\end{itemize}
\end{itemize}
\label{alg:double_boot}
\end{algorithm}


\section{Data Analysis: ARC Trial}
\label{sec:analysis}

We apply our sensitivity analysis method to the Asthma Research for the Community (ARC) study \citep{apter2019home}, a pragmatic randomized trial of 301 low-income participants with uncontrolled asthma. Participants in the active control group received usual
care plus access to and training in a web-based portal designed to improve communication between participants and their health-care providers. Participants in the intervention group received home visits by community health workers to promote care coordination and help with the use of the patient portal, in addition to usual care and portal training. The primary outcome was the score on the Asthma Control Questionnaire (ACQ) and the study protocol called for outcome data to be collected at 3, 6, 9 and 12 months after randomization. However, research coordinators were often unable to schedule data collection appointments until substantially later than these targeted times, and it is possible that assessment at time $t$ may be associated with the outcome at time $t$, even after adjusting for observed variables \citep{Smith2024}.

The ACQ score ranges from 0 (completely controlled asthma) to 6 (extremely uncontrolled asthma), taking values in $\{0, 1/6, 2/6, \cdots, 6\}$. Based on clinical input, a mean ACQ score of 3.0 or higher, or 1.2 or lower, at any time would be extreme in either treatment arm \citep{Smith2024}. In addition to the original ACQ score, we considered a binary outcome indicating asthma control status. Specifically, an asthma control indicator (AC indicator) was defined as 1 if the ACQ score was less than 1.5, corresponding to adequate asthma control, and 0 otherwise \citep{apter2019home}. For this binary outcome, the clinical collaborator considered a mean proportion of 0.5 or higher, or 0.2 or lower, at any follow-up time to represent an extreme outcome in either treatment arm.

Based on this information, we first obtain point estimates and confidence sets for the range of the sensitivity parameter in each arm of the ARC data and for both the original ACQ score and the AC indicator. To select the sensitivity parameter, we set the initial search grid of $\alpha_a$ as $\mathcal{A}_M = \{-3, -2.99, \cdots, 2.99, 3\}$ and the time set $\mathcal{T}_K = \{60,61,\cdots, 460\}$, which contains the target assessment times. For each $a=0,1$, we assume that $\mu^a(t) = s\{\symbf{B}(t)^\prime \symbf{\beta}\}$ for $t\in[60,460]$, with $\symbf{B}(t)$ a cubic spline basis with one interior knot at $t = 260$. We used the log link for the original ACQ score and the logit link for the AC indicator. As in \cite{Smith2024}, we modeled the intensity function separately for each arm using a stratified Andersen-Gill model, and modeled the conditional outcome distribution of observed outcomes separately for each treatment arm using a single-index model, then constructed the augmented inverse intensity-weighted estimators as in Section \ref{chp4:sec:estimation}.


We used the Rademacher weight distribution, the one-step multiplier bootstrap formula, and the adaptive search method in the bias correction process. We conducted 100 multiplier bootstrap replicates to obtain the bias-correction point estimates of $\alpha_a^l$ and $\alpha_a^u$. Then we used $400$ outer and $50$ inner bootstrap replicates to construct the bias-corrected confidence set for $(\alpha_a^l, \alpha_a^u)$.

\begin{table}[H]
\small
\centering
\caption[Bias-corrected estimates of the lower and upper sensitivity parameter bounds ($\alpha_l$ and $\alpha_u$), together with the corresponding bias-corrected 95\% confidence sets, for the ACQ score and the AC indicator in the ARC trial analysis.]{Bias-corrected estimates of the lower and upper sensitivity parameter bounds ($\alpha_l$ and $\alpha_u$), together with the corresponding bias-corrected 95\% confidence sets, for the ACQ score and the AC indicator in the ARC trial analysis. The confidence sets were constructed using the bias-corrected point estimates and bias-corrected bootstrap estimates according to Algorithm \ref{alg:double_boot}.}
\label{table:analysis_alpha_selection}
\begin{tabular}{c c c c c c}
\hline
Trt Group & Link Type & $\alpha_l$ & $\alpha_u$ & 95\% Basic Bootstrap CI &  95\% Wald CI \\
\hline

\multicolumn{6}{c}{\textbf{ACQ Score}}\\
\hline
\multirow{2}{*}{Treatment}
& Log ($\symbf{W}_1(t;\symbf{\beta})$)
& -0.95 & 0.55 & $(-1.21,\,0.77)$ & $(-1.23,\,0.81)$ \\
\cline{2-6}
& Log ($\symbf{W}_2(t;\symbf{\beta})$)
& -0.93 & 0.55 & $(-1.20,\,0.81)$ & $(-1.23,\,0.82)$ \\
\cline{1-6}
\multirow{2}{*}{Control}
& Log ($\symbf{W}_1(t;\symbf{\beta})$)
& -0.61 & 0.38 & $(-0.88,\,0.86)$ & $(-0.98,\,0.68)$ \\
\cline{2-6}
& Log ($\symbf{W}_2(t;\symbf{\beta})$)
& -0.62 & 0.37 & $(-0.91,\,0.81)$ & $(-1.01,\,0.66)$ \\
\hline

\multicolumn{6}{c}{\textbf{AC Indicator}}\\
\hline
\multirow{2}{*}{Treatment}
& Logit ($\symbf{W}_1(t;\symbf{\beta})$)
& -0.10 & 0.53 & $(-1.25,\,1.76)$ & $(-1.20,\,1.46)$ \\
\cline{2-6}
& Logit ($\symbf{W}_2(t;\symbf{\beta})$)
& -0.13 & 0.58 & $(-1.37,\,1.79)$ & $(-1.23,\,1.46)$ \\
\cline{1-6}
\multirow{2}{*}{Control}
& Logit ($\symbf{W}_1(t;\symbf{\beta})$)
& -0.27 & 1.16 & $(-0.96,\,2.15)$ & $(-0.81,\,1.74)$ \\
\cline{2-6}
& Logit ($\symbf{W}_2(t;\symbf{\beta})$)
& -0.25 & 1.16 &  $(-1.00,\,2.14)$ & $(-0.83,\,1.73)$ \\
\hline
\end{tabular}
\end{table}

Based on results in Table \ref{table:analysis_alpha_selection}, we choose $\alpha_0, \alpha_1 \in \{-1.2, -0.8, -0.4, 0, 0.4, 0.8\}$ for the ACQ score and $\{-0.8, -0.4, 0, 0.4, 0.8, 1.2\}$ for the AC indicator. Here we only present the analysis results with the first weight function $\symbf{W}_1(t;\symbf{\beta})$; results with the second weight function $\symbf{W}_2(t;\symbf{\beta})$
 are provided in Appendix \ref{apxB:appendix:data_analysis}.
Regarding the ACQ score, Figure \ref{fig:mean_curves_log} shows estimates of the mean curve $\mu^a(t)$ for each arm under these specific sensitivity parameter values, $t\in[60,460]$, and estimates and confidence intervals for $\mu^a(t)$ at the target times are also shown. Figure \ref{fig:contour_log} shows estimates and confidence intervals for $\delta(t)$ at 6 and 12 months under specific values of $\alpha_0$ and $\alpha_1$. Regarding the AC indicator, Figure \ref{fig:mean_curves_logit} shows estimates of the curve $\mu^a(t)$ for each arm under these specific sensitivity parameter values, $t\in[60,460]$. Estimates and confidence intervals for  $\mu^a(t)$ at the target times are also shown. Figure \ref{fig:contour_logit} shows estimates and confidence intervals for $\delta(t)$ at 6 and 12 months under specific values of $\alpha_0$ and $\alpha_1$.

\begin{figure}[H]
\centering
\begin{subfigure}[b]{0.43\textwidth}
    \centering
    \includegraphics[width=\textwidth]{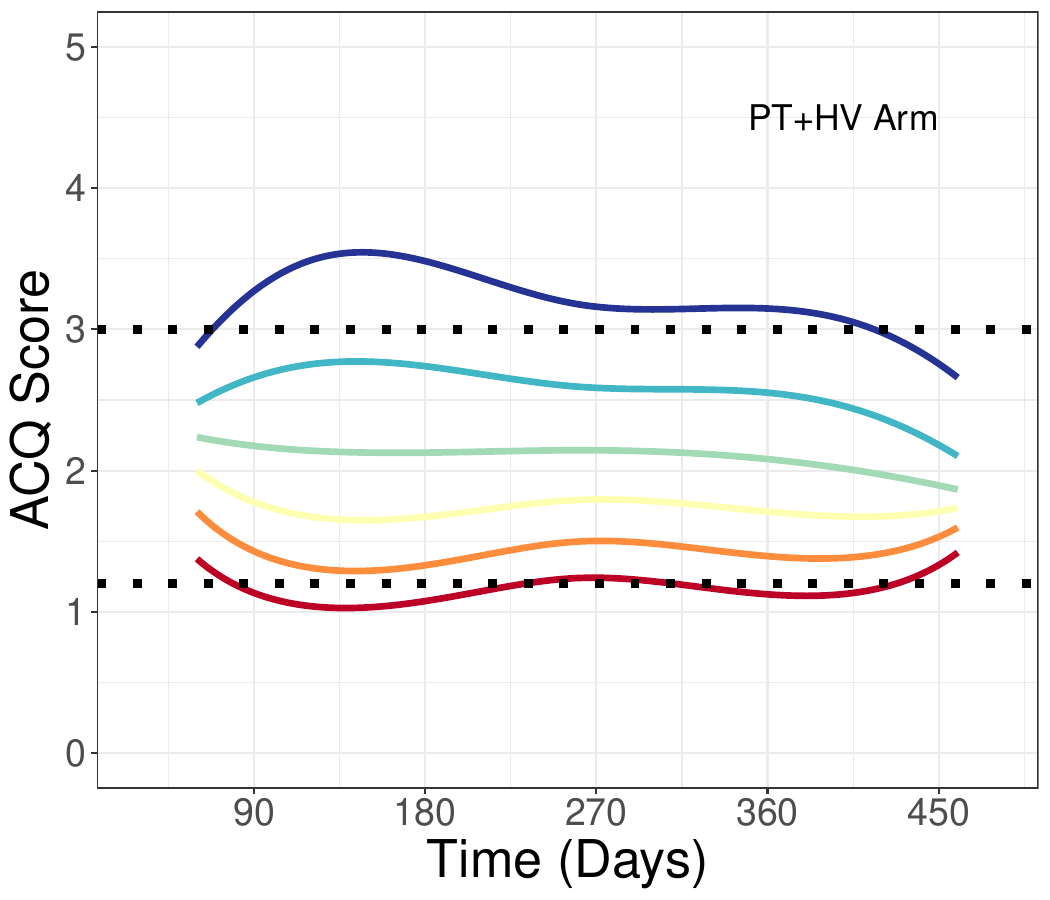}
\end{subfigure}
\hfill
\begin{subfigure}[b]{0.49\textwidth}  
    \centering 
    \includegraphics[width=\textwidth]{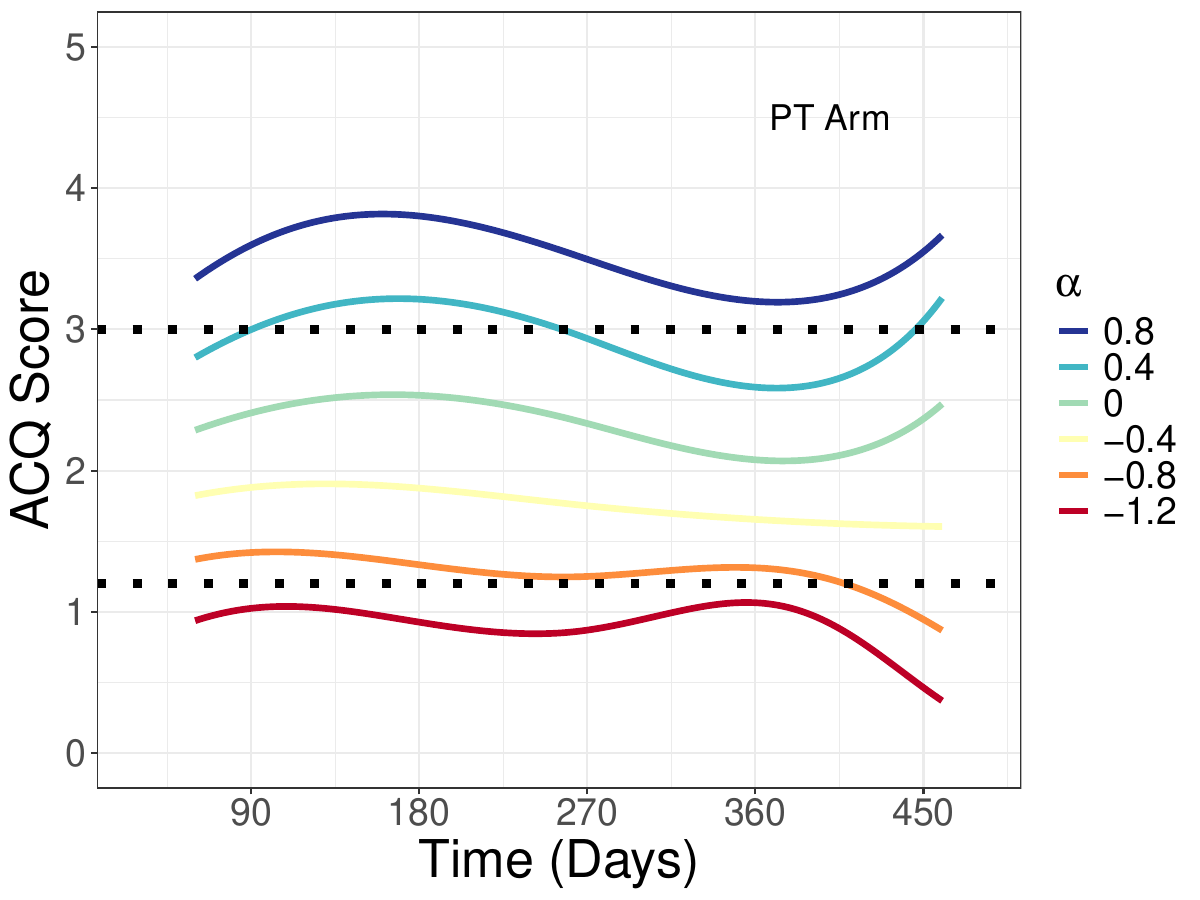}
\end{subfigure}

\medskip

\begin{subfigure}[b]{0.43\textwidth}   
    \centering 
    \includegraphics[width=\textwidth]{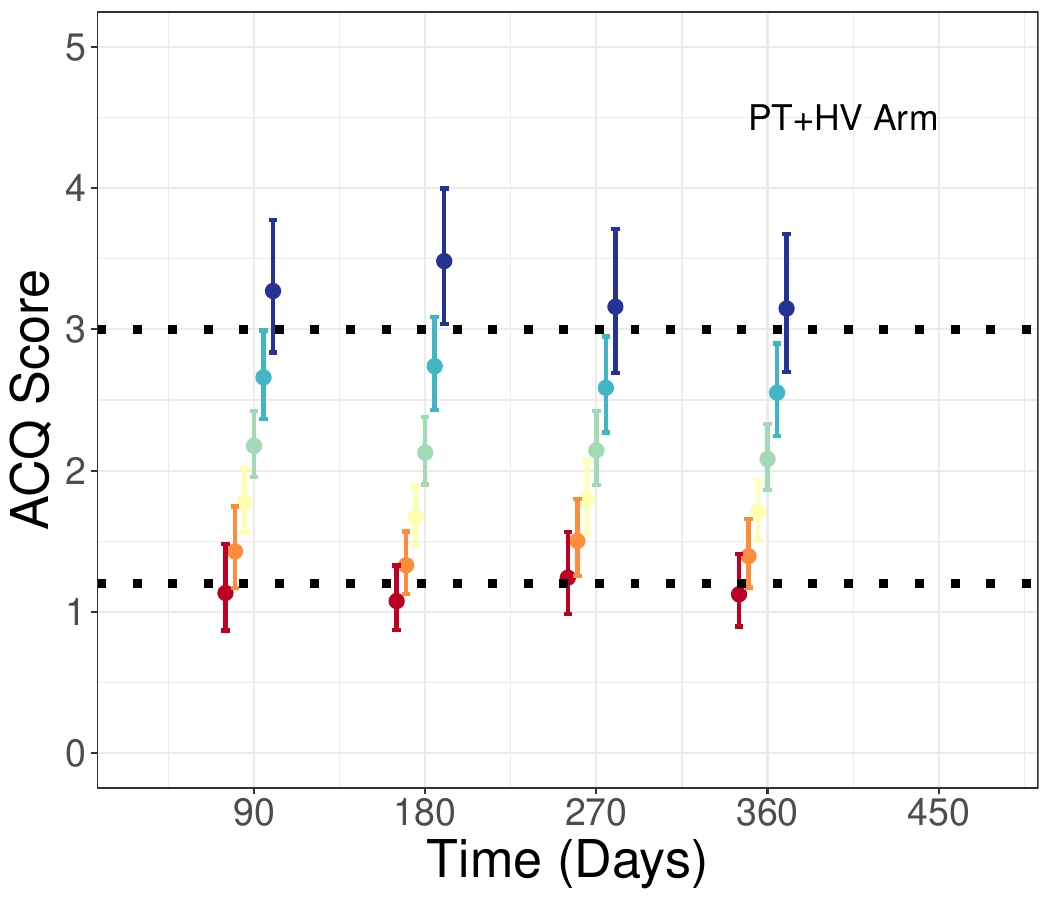}
   \end{subfigure}
\hfill
\begin{subfigure}[b]{0.49\textwidth}   
    \centering 
    \includegraphics[width=\textwidth]{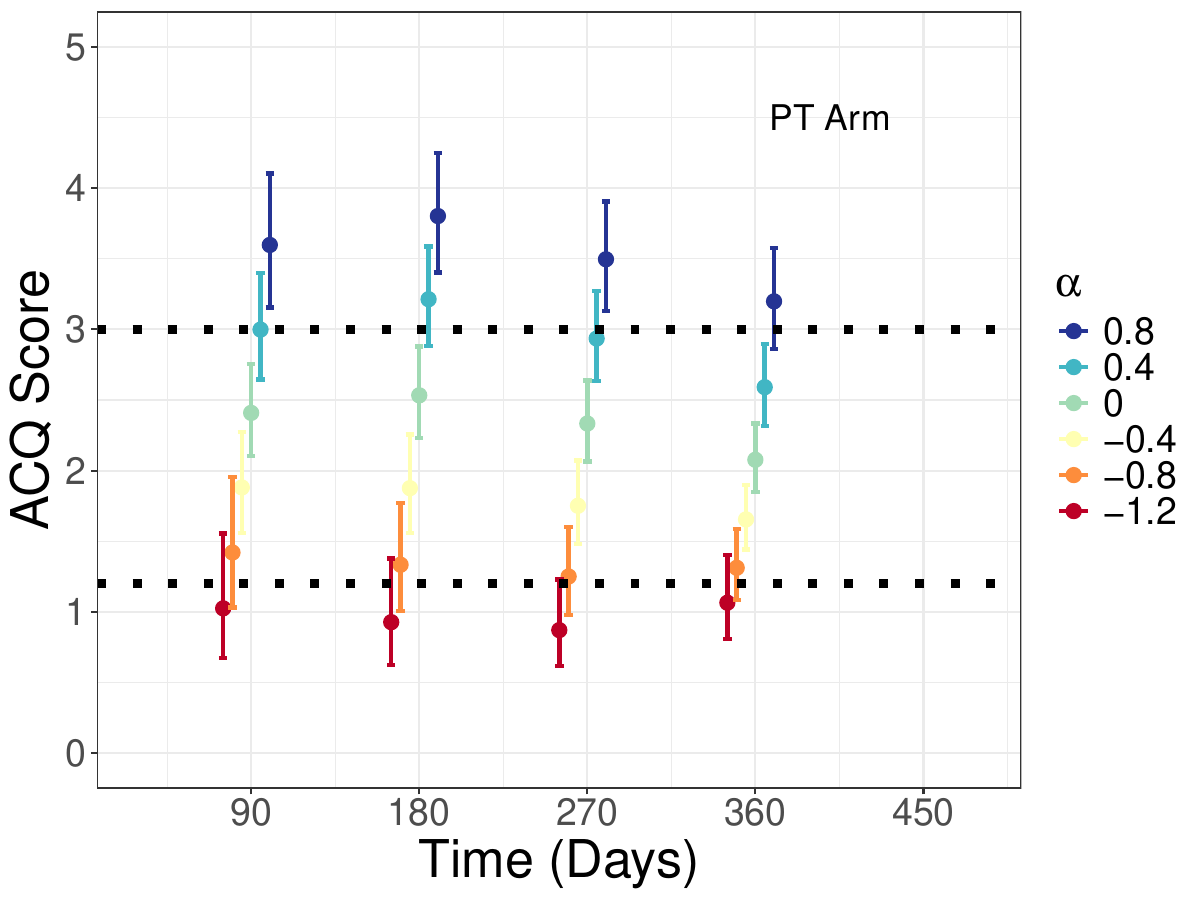}
\end{subfigure}

\caption[Population means of the Asthma Control Questionnaire (ACQ) score in the ARC trial under a range of sensitivity parameter values (log link with weight $\symbf{W}_1(t;\symbf{\beta})$).]{Population means of scores on the Asthma Control Questionnaire (ACQ) in the ARC trial under a range of sensitivity parameter values (log link with weight $\symbf{W}_1(t;\symbf{\beta})$). Here we show inference for the population mean ACQ score in the intervention (PT+HV) arm and the control (PT) arm.  Estimation is made under sensitivity parameter values of $\alpha = -1.2, -0.8, -0.4, 0, 0.4, 0.8$, where a positive (negative) value of $\alpha$ posits that unobserved outcome values at time $t$ tend to be higher (lower) than observed values, after controlling for variables observed before time $t$. Upper panels: point estimates on the interval 60 to 460 days after randomization.  Lower panels: point estimates and 95\% Wald confidence intervals using the parametric bootstrap variance estimate at each target time of 90, 180, 270, and 360 days. For each arm, only those values of $\alpha$ under which the estimated curves lie completely between the dotted lines at the lower bound $1.2$ and upper bound $3$ are considered plausible based on subject-matter expertise.} 
\label{fig:mean_curves_log}
\end{figure}

\begin{figure}[H]
\centering
\begin{subfigure}[b]{0.42\textwidth}
    \centering
    \includegraphics[width=\textwidth]{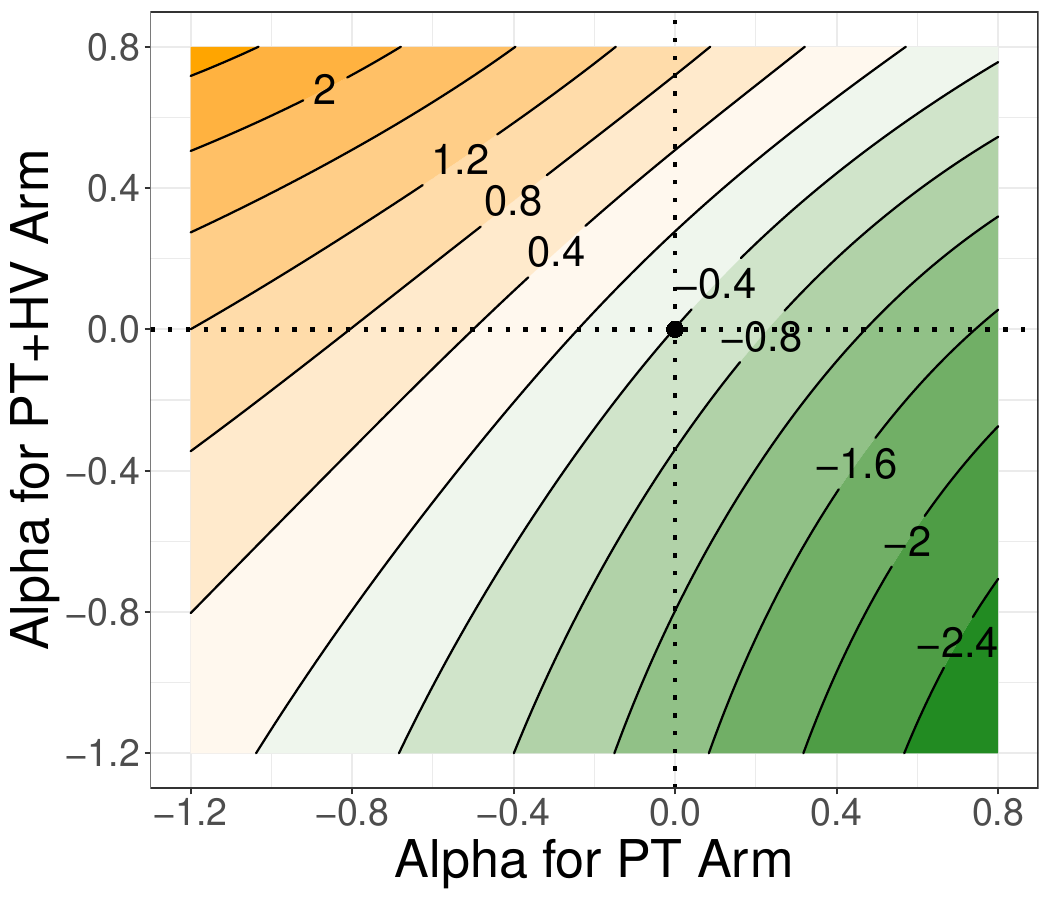}
\end{subfigure}
\hfill
\begin{subfigure}[b]{0.48\textwidth}  
    \centering 
    \includegraphics[width=\textwidth]{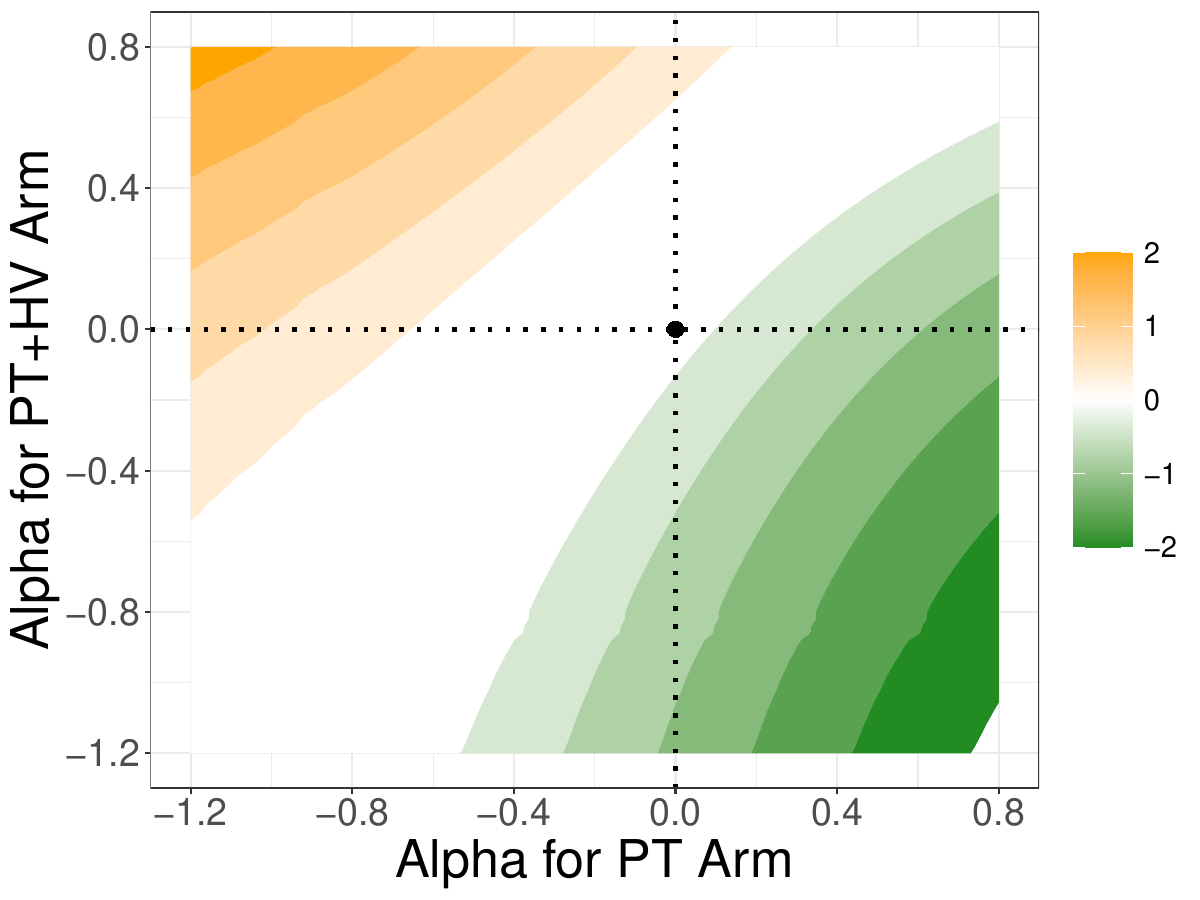}
\end{subfigure}

\medskip

\begin{subfigure}[b]{0.42\textwidth}   
    \centering 
    \includegraphics[width=\textwidth]{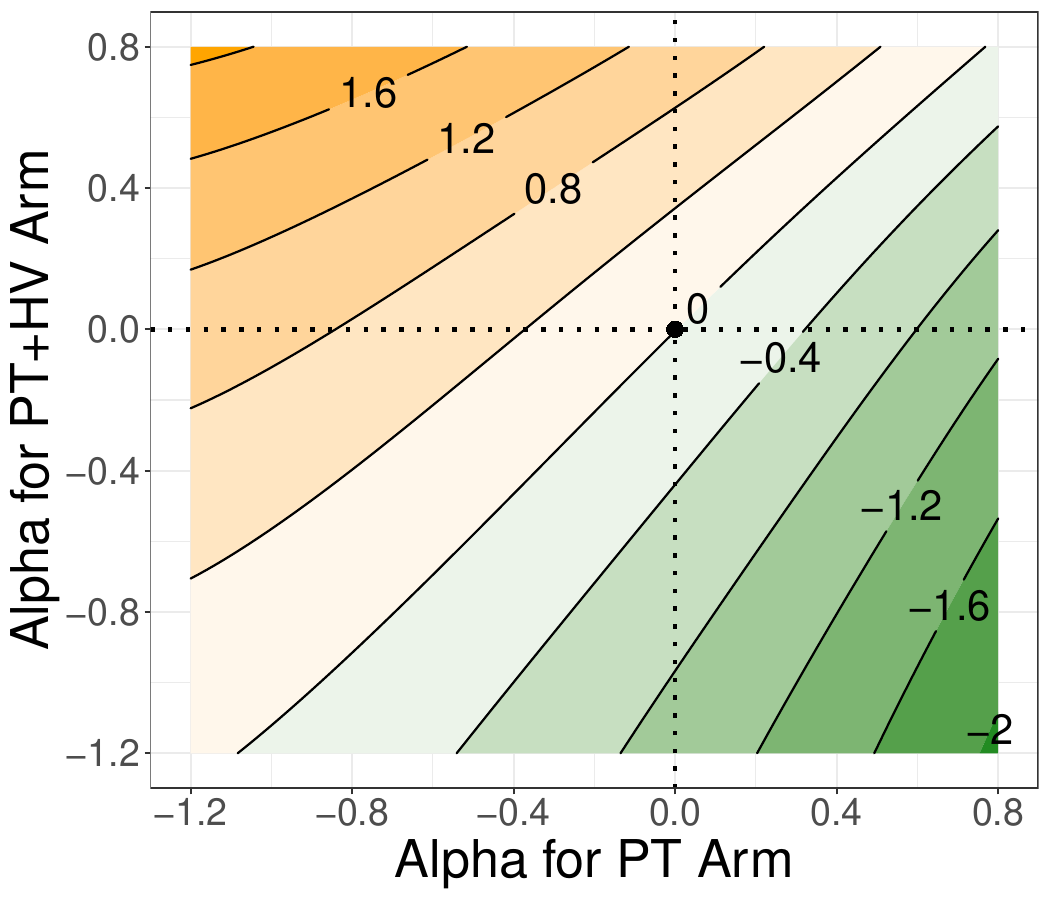}
   \end{subfigure}
\hfill
\begin{subfigure}[b]{0.48\textwidth}   
    \centering 
    \includegraphics[width=\textwidth]{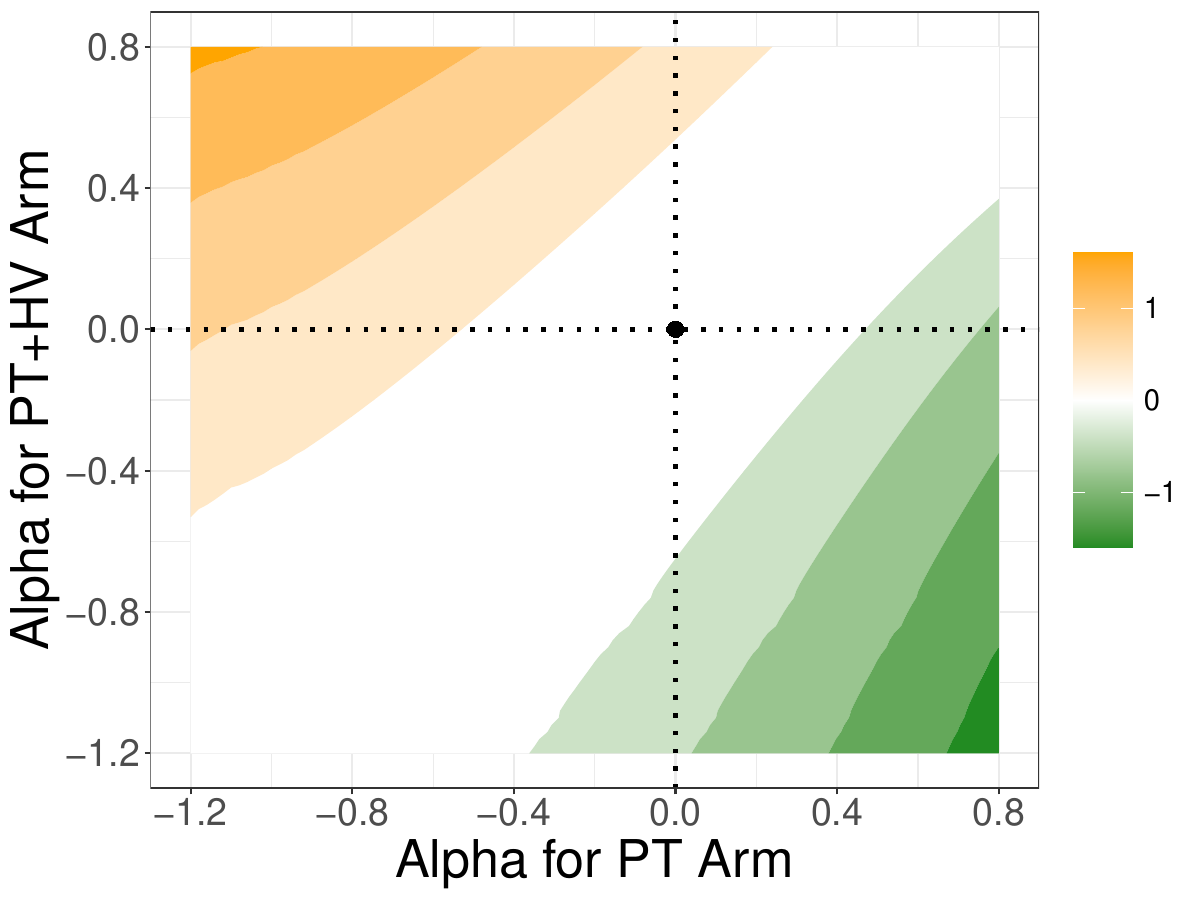}
\end{subfigure}
\caption[Sensitivity analysis for the Asthma Control Questionnaire (ACQ) score in the ARC trial (log link with weight $\symbf{W}_1(t;\symbf{\beta})$).]{Sensitivity analysis for the Asthma Control Questionnaire (ACQ) score in the ARC trial (log link with weight $\symbf{W}_1(t;\symbf{\beta})$). The left panels show the point estimate of the treatment effect $\delta(t)$ at 6 months (upper left) and 12 months (lower left) for each pair of sensitivity parameter values. The right panels summarize the 95\% confidence intervals for $\delta(t)$ at 6 months (upper right) and 12 months (lower right). The white region corresponds to sensitivity parameter values for which the confidence interval contains 0; the shaded region in the lower right (upper left) corresponds to values for which the confidence interval lies entirely below (above) 0. Contour heights give the value in the confidence interval closest to 0. Confidence intervals are Wald intervals based on the parametric bootstrap variance estimate.}
\label{fig:contour_log}
\end{figure}

\begin{figure}[H]
\centering
\begin{subfigure}[b]{0.43\textwidth}
    \centering
    \includegraphics[width=\textwidth]{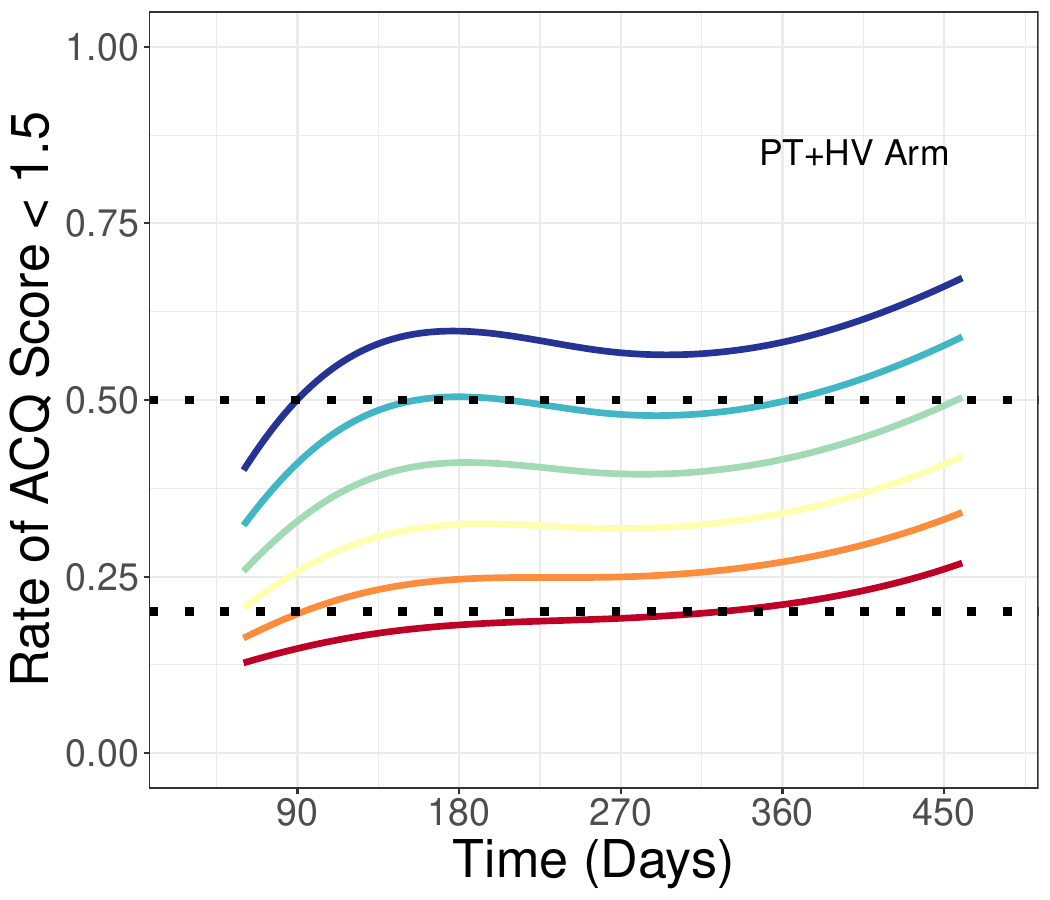}
\end{subfigure}
\hfill
\begin{subfigure}[b]{0.49\textwidth}  
    \centering 
    \includegraphics[width=\textwidth]{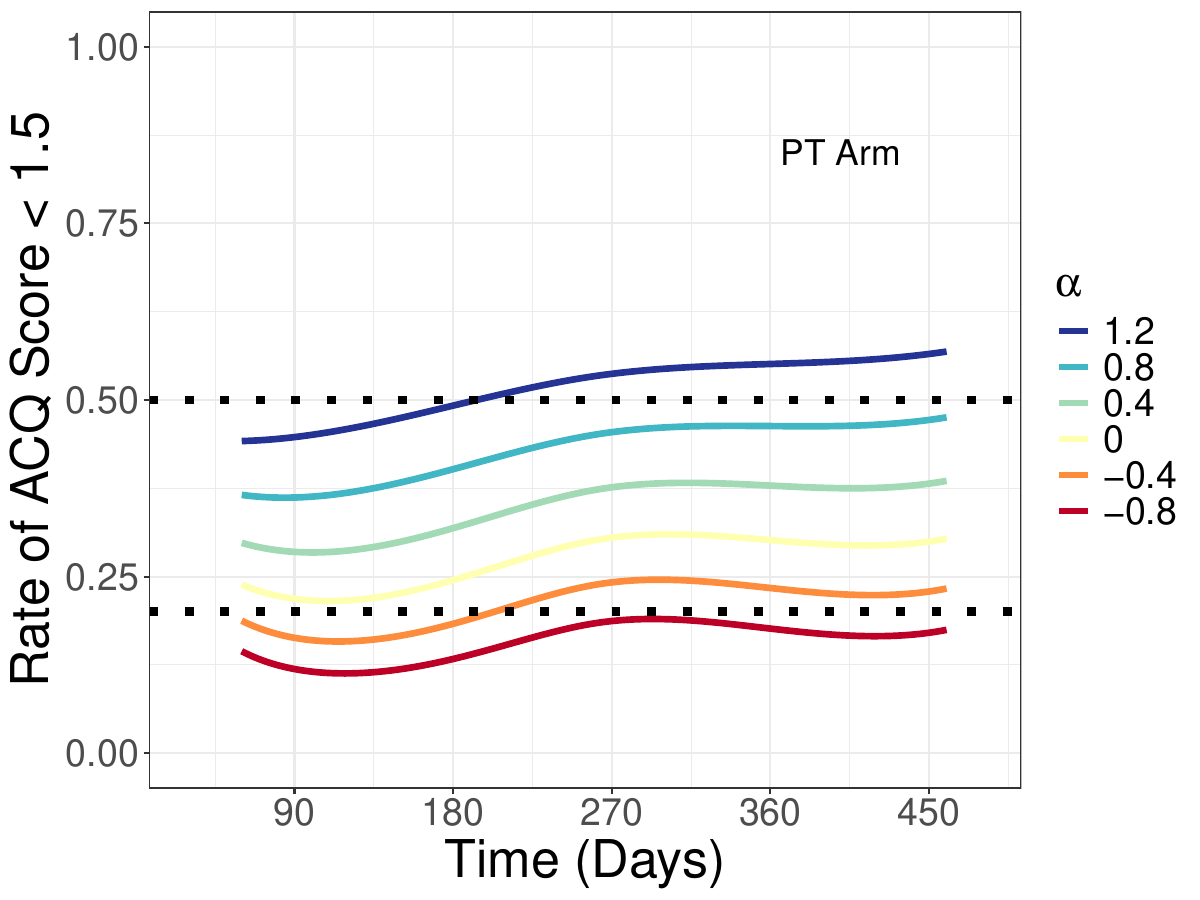}
\end{subfigure}

\medskip

\begin{subfigure}[b]{0.43\textwidth}   
    \centering 
    \includegraphics[width=\textwidth]{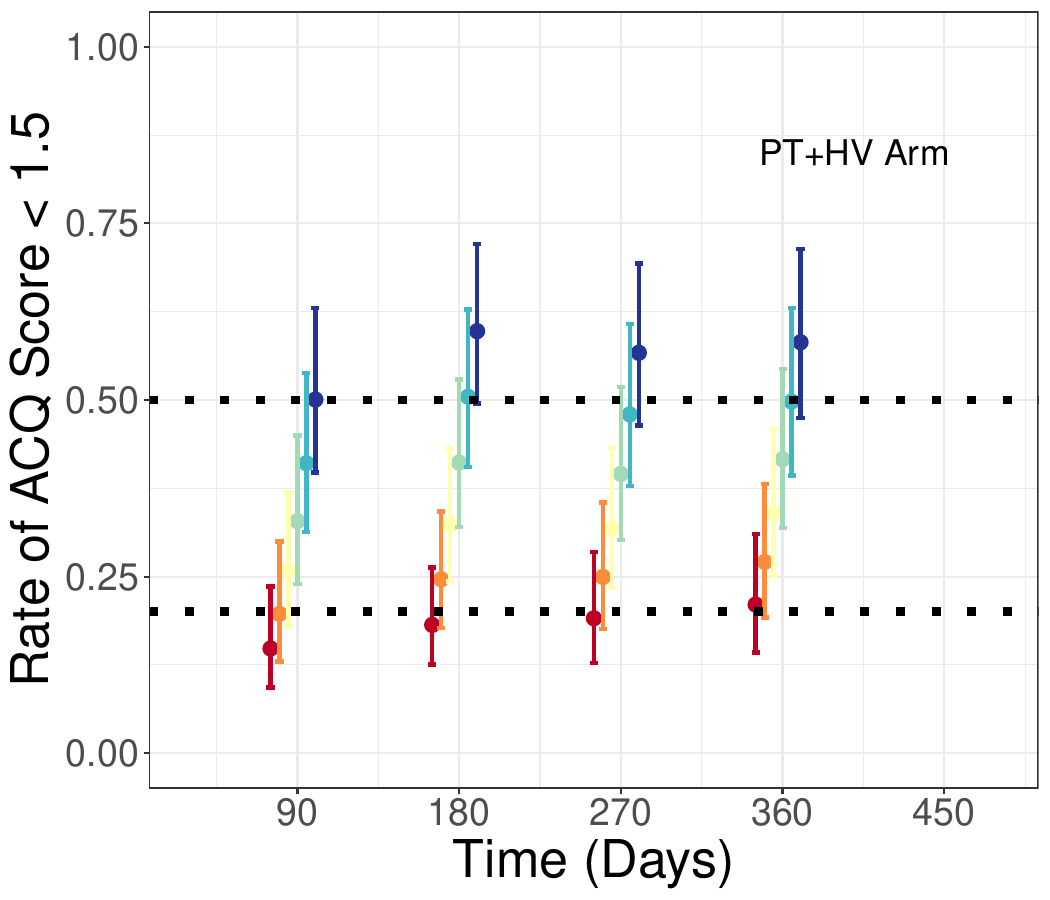}
   \end{subfigure}
\hfill
\begin{subfigure}[b]{0.49\textwidth}   
    \centering 
    \includegraphics[width=\textwidth]{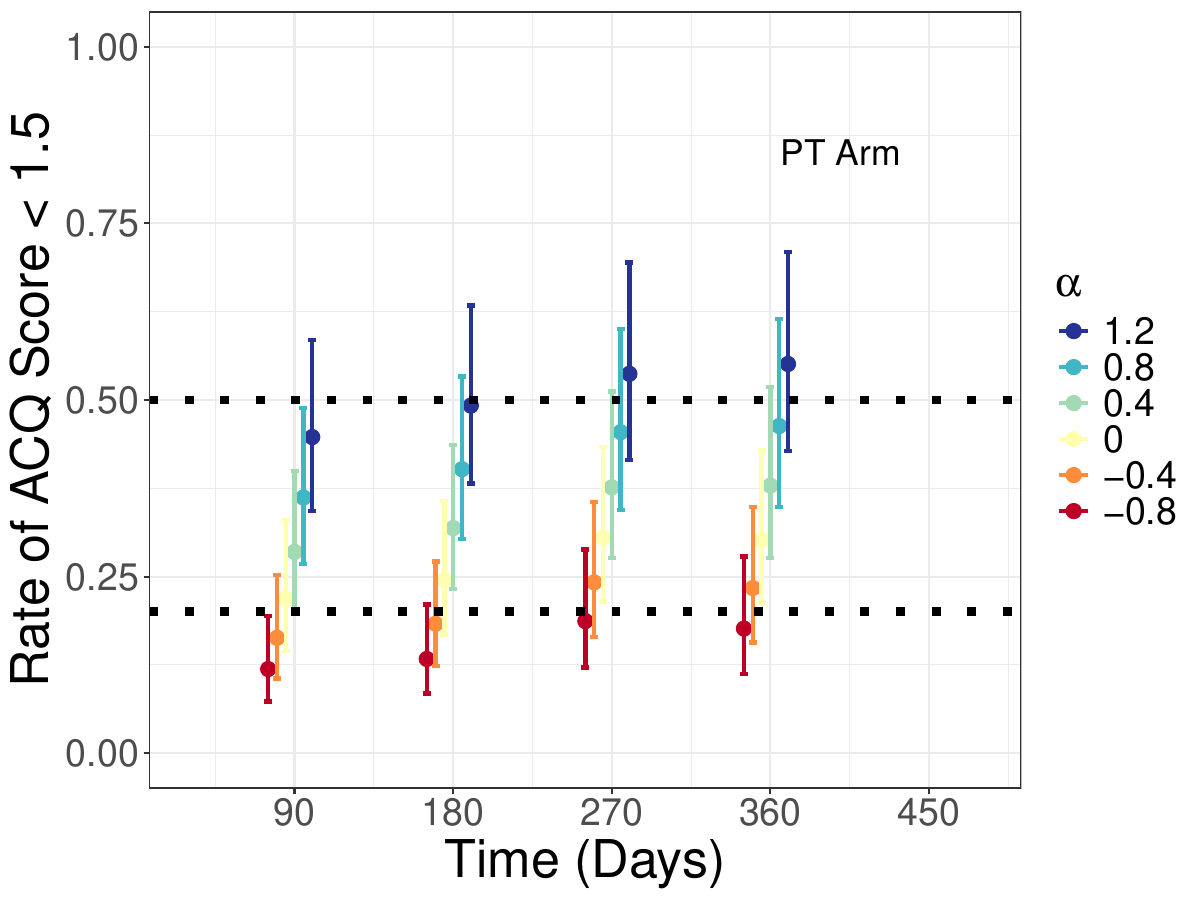}
\end{subfigure}
\caption[Population mean estimates of the proportion of patients with Asthma Control Questionnaire (ACQ) score less than 1.5 in the ARC trial under a range of sensitivity parameter values (logit link with weight $\symbf{W}_1(t;\symbf{\beta})$).]{Population mean estimates of the proportion of patients with Asthma Control Questionnaire (ACQ) score less than 1.5 in the ARC trial under a range of sensitivity parameter values (logit link with weight $\symbf{W}_1(t;\symbf{\beta})$). Here we show inference for the proportion of patients with ACQ score less than 1.5 in the intervention (PT+HV) arm and the control (PT) arm.  Estimation is made under sensitivity parameter values of $\alpha = -0.8, -0.4, 0, 0.4, 0.8, 1.2$, where a positive (negative) value of $\alpha$ posits that unobserved outcome values at time $t$ tend to be higher (lower) than observed values, after controlling for variables observed before time $t$. Upper panels: point estimates on the interval 60 to 460 days after randomization.  Lower panels: point estimates and 95\% Wald confidence intervals using the parametric bootstrap variance estimate at each target time of 90, 180, 270, and 360 days. For each arm, only those values of $\alpha$ under which the estimated curves lie completely between the dotted lines at the lower bound $0.2$ and upper bound $0.5$ are considered plausible based on subject-matter expertise.} 
\label{fig:mean_curves_logit}
\end{figure}

\begin{figure}[H]
\centering
\begin{subfigure}[b]{0.42\textwidth}
    \centering
    \includegraphics[width=\textwidth]{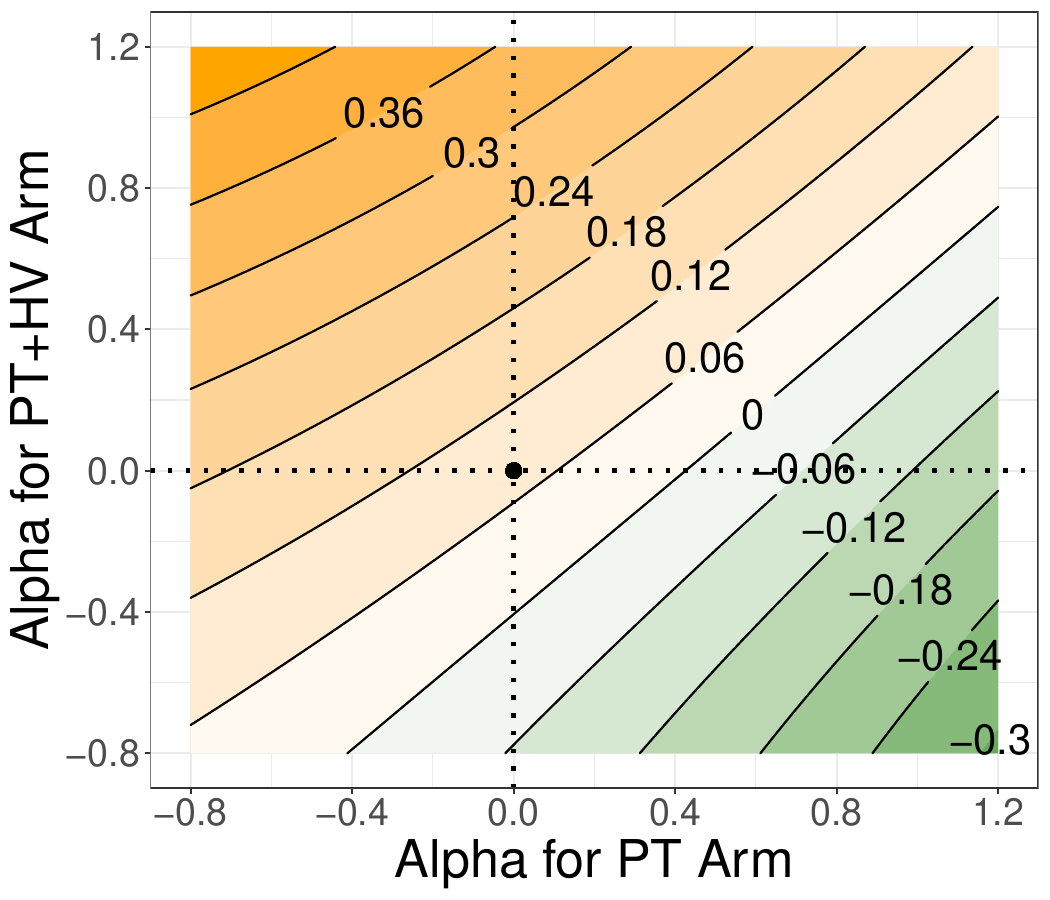}
\end{subfigure}
\hfill
\begin{subfigure}[b]{0.48\textwidth}  
    \centering 
    \includegraphics[width=\textwidth]{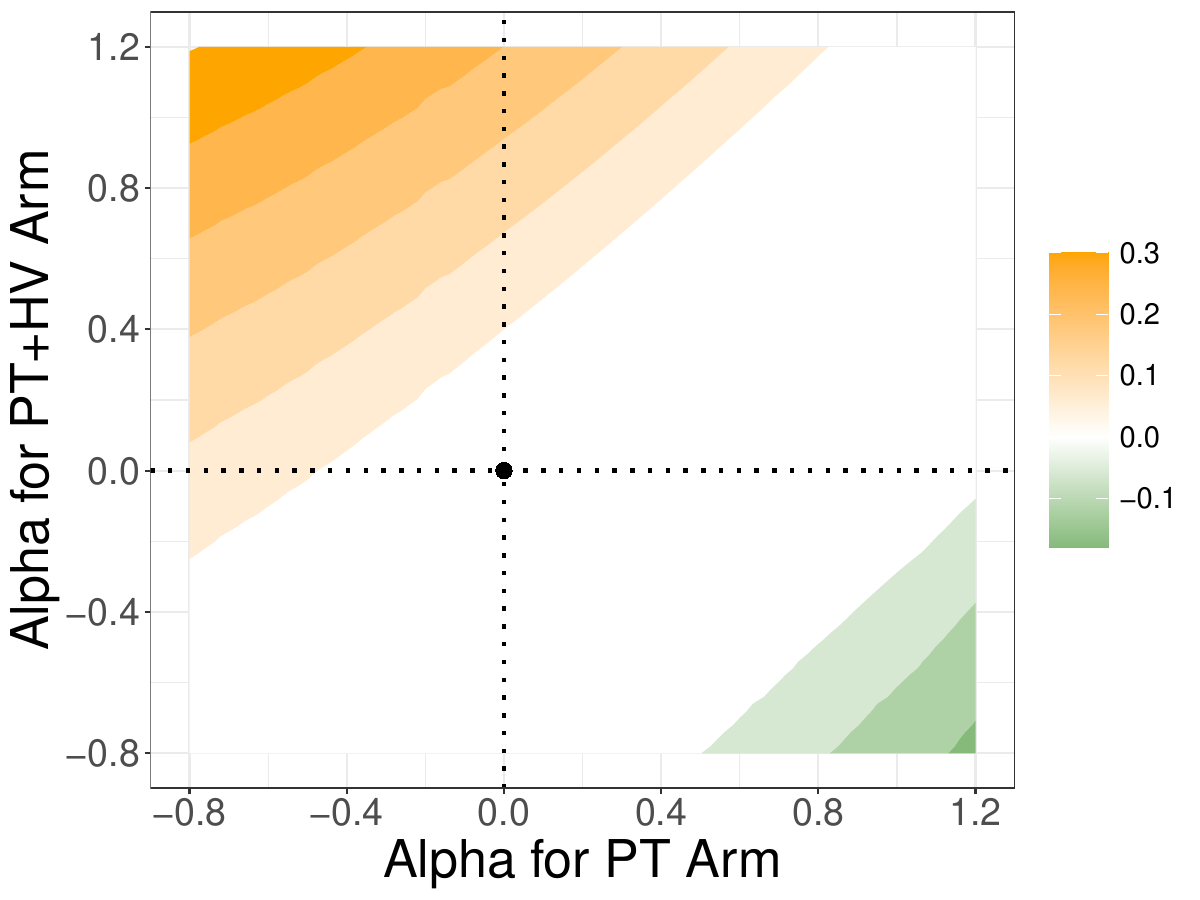}
\end{subfigure}

\medskip

\begin{subfigure}[b]{0.42\textwidth}   
    \centering 
    \includegraphics[width=\textwidth]{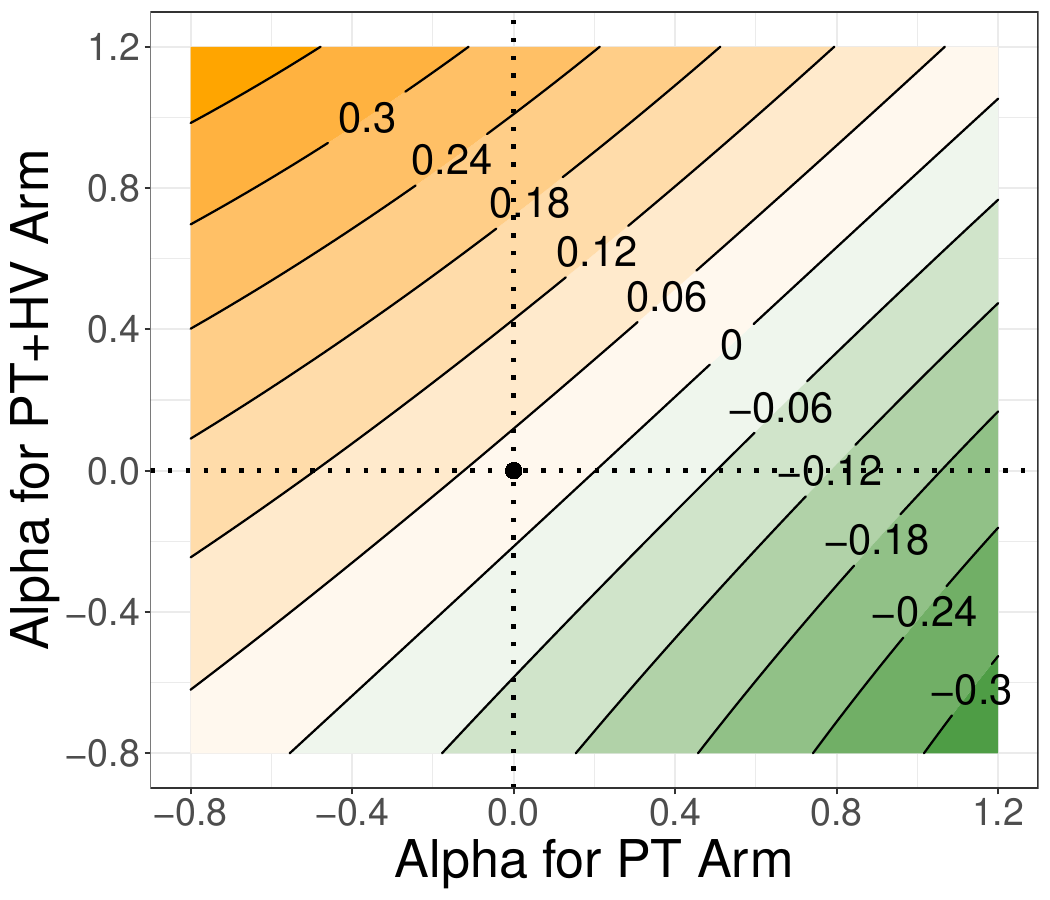}
   \end{subfigure}
\hfill
\begin{subfigure}[b]{0.48\textwidth}   
    \centering 
    \includegraphics[width=\textwidth]{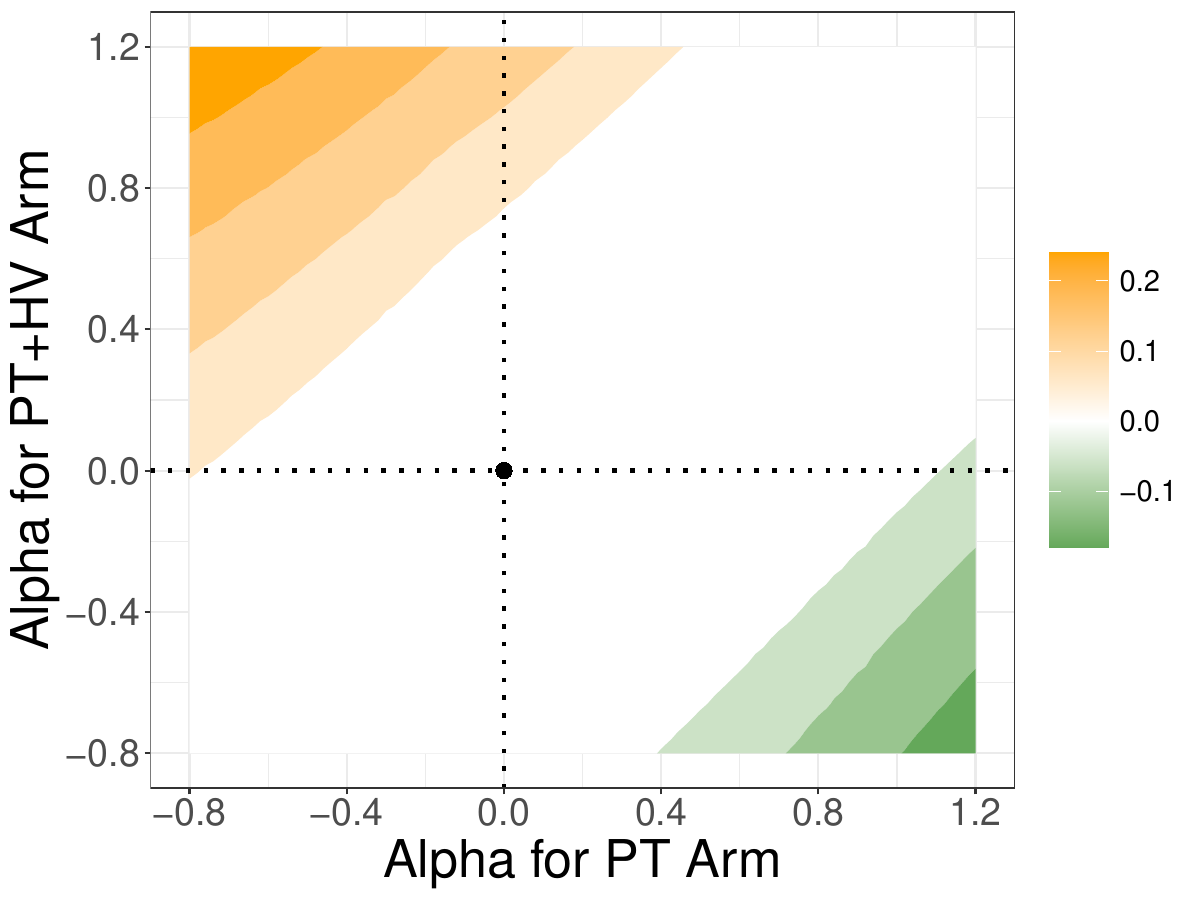}
\end{subfigure}

\caption[Sensitivity analysis for the proportion of patients with Asthma Control Questionnaire (ACQ) score less than 1.5 in the ARC trial (logit link with weight $\symbf{W}_1(t;\symbf{\beta})$).]{Sensitivity analysis for the proportion of patients with Asthma Control Questionnaire (ACQ) score less than 1.5 in the ARC trial (logit link with weight $\symbf{W}_1(t;\symbf{\beta})$). The left panels show the point estimate of the treatment effect $\delta(t)$ at 6 months (upper left) and 12 months (lower left) for each pair of sensitivity parameter values. The right panels summarize the 95\% confidence intervals for $\delta(t)$ at 6 months (upper right) and 12 months (lower right). The white region corresponds to sensitivity parameter values for which the confidence interval contains 0; the shaded region in the lower right (upper left) corresponds to values for which the confidence interval lies entirely below (above) 0. Contour heights give the value in the confidence interval closest to 0. Confidence intervals are Wald intervals based on the parametric bootstrap variance estimate.}
\label{fig:contour_logit}
\end{figure}

If we assume $\alpha_0 = \alpha_1$,
which includes the explainable assessment assumption within each treatment arm, there is insufficient evidence to conclude a treatment effect at either 6 or 12 months for either the ACQ score or the AC indicator. However, when informative assessment mechanisms are allowed to differ between treatment arms, statistically significant treatment effects emerge under some plausible sensitivity parameter values. For example, for the ACQ score, if we assume $\alpha_0=0$ and $\alpha_1=-0.8$, implying that unobserved ACQ scores tend to be lower than observed scores in the intervention arm, there is evidence that the home visits intervention improves the population mean ACQ score at both 6 and 12 months compared with portal training alone. For the AC indicator, if we assume $\alpha_0 = 0$ and $\alpha_1 = 0.8$, implying that unobserved AC indicators tend to be higher than observed values in the intervention arm, there is evidence that 
the home visits intervention results in better asthma control than the portal training alone at both 6 and 12 months.

Overall, the results for both outcomes indicate that when $\alpha_0$ and $\alpha_1$ are similar, differing by no more than some value, there is insufficient evidence to conclude a treatment effect at either follow-up time. In contrast, when $\alpha_0 > \alpha_1 + 0.8$, the analyses using the ACQ scores provide evidence that the home visits intervention improves the population mean asthma control outcome relative to portal training alone at both 6 and 12 months. Similarly, when $\alpha_1 > \alpha_0 + 0.8$, the analyses using the AC indicator provide evidence of the same improvement at both follow-up times.

\section{Simulation Study}
\label{chp4:sec:simu}

We generated realistic simulated data based on the ARC data, with a sample size of $N=200$ per treatment arm. Details of the data-generating process are provided in Appendix \ref{apxB:appendix:simu_mechanism}. The simulated data were generated under the proposed sensitivity analysis framework, with prespecified values of $\alpha_1$ and $\alpha_0$ corresponding to the link function used in the analysis. Specifically, the ACQ score was generated using a negative binomial regression model, whereas the AC indicator was generated using a logistic regression model. The resulting datasets were analyzed using the proposed augmented inverse intensity-weighted estimators.

We first evaluated the finite-sample performance of the proposed estimators by analyzing the simulated data using the true values of $\alpha_1$ and $\alpha_0$. We then analyzed the same datasets under the  explainable assessment assumption, setting $\alpha_0 = \alpha_1 =0$, to investigate the performance of methods that rely on this assumption when it is violated. Results based on the first weight function $\symbf{W}_1(t;\symbf{\beta})$ for the treatment effects at 6 and 12 months are presented in Tables \ref{table:coverage_log} and \ref{table:coverage_logit}. Corresponding results using the second weight function $\symbf{W}_2(t;\symbf{\beta})$ are provided in Appendix \ref{apxB:appendix:simu_result}. When the analysis was conducted using the true values of $\alpha_1$ and $\alpha_0$, the proposed estimators exhibited good finite-sample performance. Across 500 simulation replicates, the empirical bias was small, with absolute bias below 0.01 for the AC indicator and below 0.06 for the ACQ score. The empirical coverage probabilities of the nominal 0.95 confidence intervals were close to the target level, ranging from 0.938 to 0.970 for the AC indicator and from 0.926 to 0.972 for the ACQ score. In contrast, when the analyses assumed $\alpha_1 = \alpha_0 = 0$, the resulting estimators often exhibited substantial bias and the confidence interval coverage deviated markedly from the nominal 0.95 level.
This highlights the importance of considering a range of different assumptions through sensitivity analysis.

\begin{table}[htbp]
\centering
\caption[Simulation results for the ACQ score with the log link ($\symbf{W}_1(t;\symbf{\beta})$).]{Simulation results for the ACQ score with the log link ($\symbf{W}_1(t;\symbf{\beta})$). Data were generated under our sensitivity analysis assumption using values of $\alpha_0,\alpha_1=-1.2, -0.8, -0.4, 0, 0.4, 0.8$. The treatment effects at 6 and 12 months were then estimated using our augmented inverse intensity-weighted estimator with the log link ($\symbf{W}_1(t;\symbf{\beta})$): (a) using the true values of $\alpha_0,\alpha_1$ (rows denoted ``S.A.''), and (b) under the explainable assessment assumption that $\alpha_0=\alpha_1=0$ (rows denoted ``Expl.''). Shown are the absolute value of the empirical bias and the confidence interval coverage across 500 simulations. Confidence intervals are Wald confidence intervals using the parametric bootstrap variance estimate.}
\label{table:coverage_log}
\footnotesize
\setlength{\tabcolsep}{3pt}
\begin{tabular}{ccc*{12}{c}}
\hline
 & & & \multicolumn{12}{c}{True $\alpha_0$} \\
\cline{4-15}
 & & & \multicolumn{2}{c}{$-1.2$} & \multicolumn{2}{c}{$-0.8$} & \multicolumn{2}{c}{$-0.4$}
     & \multicolumn{2}{c}{0} & \multicolumn{2}{c}{0.4} & \multicolumn{2}{c}{0.8} \\
\cline{4-15}
Month & True $\alpha_1$ & Method & Bias & Cov. & Bias & Cov. & Bias & Cov. & Bias & Cov. & Bias & Cov. & Bias & Cov. \\
\hline
\multirow{12}{*}{6}
 & \multirow{2}{*}{0.8}    & S.A.  & 0.046 & 0.952 & 0.041 & 0.954 & 0.038 & 0.964 & 0.031 & 0.962 & 0.019 & 0.958 & 0.003 & 0.954 \\
 &                         & Expl. & 2.155 & 0.000 & 1.916 & 0.000 & 1.572 & 0.000 & 1.073 & 0.000 & 0.408 & 0.306 & 0.356 & 0.412 \\ \cline{2-15}
 & \multirow{2}{*}{0.4}    & S.A.  & 0.025 & 0.962 & 0.020 & 0.968 & 0.017 & 0.960 & 0.010 & 0.964 & 0.002 & 0.960 & 0.024 & 0.958 \\
 &                         & Expl. & 1.571 & 0.000 & 1.333 & 0.000 & 0.989 & 0.000 & 0.490 & 0.158 & 0.176 & 0.832 & 0.940 & 0.000 \\ \cline{2-15}
 & \multirow{2}{*}{0}      & S.A.  & 0.013 & 0.960 & 0.009 & 0.950 & 0.006 & 0.946 & 0.002 & 0.954 & 0.013 & 0.958 & 0.036 & 0.950 \\
 &                         & Expl. & 1.080 & 0.000 & 0.841 & 0.002 & 0.498 & 0.148 & 0.002 & 0.954 & 0.667 & 0.032 & 1.431 & 0.000 \\ \cline{2-15}
 & \multirow{2}{*}{$-0.4$} & S.A.  & 0.006 & 0.954 & 0.001 & 0.952 & 0.002 & 0.942 & 0.009 & 0.942 & 0.020 & 0.952 & 0.043 & 0.940 \\
 &                         & Expl. & 0.697 & 0.010 & 0.459 & 0.214 & 0.115 & 0.886 & 0.384 & 0.348 & 1.050 & 0.000 & 1.814 & 0.000 \\ \cline{2-15}
 & \multirow{2}{*}{$-0.8$} & S.A.  & 0.002 & 0.960 & 0.003 & 0.958 & 0.006 & 0.948 & 0.013 & 0.940 & 0.025 & 0.944 & 0.047 & 0.938 \\
 &                         & Expl. & 0.412 & 0.298 & 0.173 & 0.830 & 0.170 & 0.840 & 0.670 & 0.030 & 1.335 & 0.000 & 2.099 & 0.000 \\ \cline{2-15}
 & \multirow{2}{*}{$-1.2$} & S.A.  & 0.003 & 0.958 & 0.008 & 0.958 & 0.011 & 0.946 & 0.018 & 0.940 & 0.029 & 0.948 & 0.052 & 0.942 \\
 &                         & Expl. & 0.197 & 0.804 & 0.042 & 0.940 & 0.385 & 0.346 & 0.885 & 0.000 & 1.550 & 0.000 & 2.314 & 0.000 \\
\hline
\multirow{12}{*}{12}
 & \multirow{2}{*}{0.8}    & S.A.  & 0.035 & 0.954 & 0.029 & 0.958 & 0.023 & 0.962 & 0.013 & 0.962 & 0.004 & 0.962 & 0.010 & 0.938 \\
 &                         & Expl. & 2.028 & 0.000 & 1.805 & 0.000 & 1.485 & 0.000 & 1.019 & 0.000 & 0.387 & 0.304 & 0.363 & 0.366 \\ \cline{2-15}
 & \multirow{2}{*}{0.4}    & S.A.  & 0.020 & 0.948 & 0.014 & 0.944 & 0.008 & 0.956 & 0.002 & 0.960 & 0.011 & 0.966 & 0.025 & 0.944 \\
 &                         & Expl. & 1.466 & 0.000 & 1.243 & 0.000 & 0.922 & 0.000 & 0.456 & 0.192 & 0.175 & 0.784 & 0.925 & 0.000 \\ \cline{2-15}
 & \multirow{2}{*}{0}      & S.A.  & 0.012 & 0.946 & 0.007 & 0.944 & 0.001 & 0.948 & 0.009 & 0.966 & 0.018 & 0.972 & 0.033 & 0.954 \\
 &                         & Expl. & 1.000 & 0.000 & 0.777 & 0.000 & 0.457 & 0.192 & 0.009 & 0.966 & 0.641 & 0.008 & 1.391 & 0.000 \\ \cline{2-15}
 & \multirow{2}{*}{$-0.4$} & S.A.  & 0.007 & 0.932 & 0.002 & 0.952 & 0.004 & 0.954 & 0.015 & 0.966 & 0.024 & 0.958 & 0.038 & 0.940 \\
 &                         & Expl. & 0.641 & 0.008 & 0.418 & 0.248 & 0.098 & 0.922 & 0.368 & 0.350 & 1.000 & 0.000 & 1.750 & 0.000 \\ \cline{2-15}
 & \multirow{2}{*}{$-0.8$} & S.A.  & 0.000 & 0.934 & 0.006 & 0.940 & 0.012 & 0.950 & 0.022 & 0.952 & 0.031 & 0.946 & 0.045 & 0.934 \\
 &                         & Expl. & 0.371 & 0.344 & 0.147 & 0.860 & 0.173 & 0.790 & 0.639 & 0.008 & 1.271 & 0.000 & 2.021 & 0.000 \\ \cline{2-15}
 & \multirow{2}{*}{$-1.2$} & S.A.  & 0.009 & 0.934 & 0.014 & 0.946 & 0.020 & 0.948 & 0.030 & 0.946 & 0.039 & 0.942 & 0.054 & 0.926 \\
 &                         & Expl. & 0.167 & 0.818 & 0.056 & 0.944 & 0.376 & 0.334 & 0.842 & 0.000 & 1.474 & 0.000 & 2.224 & 0.000 \\
\hline
\end{tabular}
\end{table}

\begin{table}[htbp]
\centering
\caption[Simulation results for the AC indicator with the logit link ($\symbf{W}_1(t;\symbf{\beta})$).]{Simulation results for the AC indicator with the logit link ($\symbf{W}_1(t;\symbf{\beta})$). Data were generated under our sensitivity analysis assumption using values of $\alpha_0,\alpha_1 = -0.8, -0.4, 0, 0.4, 0.8, 1.2$. The treatment effects at 6 and 12 months were then estimated using our augmented inverse intensity-weighted estimator with the logit link ($\symbf{W}_1(t;\symbf{\beta})$): (a) using the true values of $\alpha_0,\alpha_1$ (rows denoted ``S.A.''), and (b) under the explainable assessment assumption that $\alpha_0=\alpha_1=0$ (rows denoted ``Expl.''). Shown are the absolute value of the empirical bias and the confidence interval coverage across 500 simulations. Confidence intervals are Wald confidence intervals using the parametric bootstrap variance estimate.}
\label{table:coverage_logit}
\footnotesize
\setlength{\tabcolsep}{3pt}
\begin{tabular}{ccc*{12}{c}}
\hline
 & & & \multicolumn{12}{c}{True $\alpha_0$} \\
\cline{4-15}
 & & & \multicolumn{2}{c}{$-0.8$} & \multicolumn{2}{c}{$-0.4$} & \multicolumn{2}{c}{0}
     & \multicolumn{2}{c}{0.4} & \multicolumn{2}{c}{0.8} & \multicolumn{2}{c}{1.2} \\
\cline{4-15}
Month & True $\alpha_1$ & Method & Bias & Cov. & Bias & Cov. & Bias & Cov. & Bias & Cov. & Bias & Cov. & Bias & Cov. \\
\hline
\multirow{12}{*}{6}
 & \multirow{2}{*}{1.2}    & S.A.  & 0.004 & 0.964 & 0.004 & 0.970 & 0.004 & 0.970 & 0.005 & 0.964 & 0.006 & 0.960 & 0.006 & 0.962 \\
 &                         & Expl. & 0.365 & 0.000 & 0.315 & 0.000 & 0.254 & 0.014 & 0.184 & 0.164 & 0.104 & 0.642 & 0.017 & 0.944 \\ \cline{2-15}
 & \multirow{2}{*}{0.8}    & S.A.  & 0.004 & 0.970 & 0.004 & 0.970 & 0.004 & 0.966 & 0.005 & 0.962 & 0.006 & 0.958 & 0.006 & 0.962 \\
 &                         & Expl. & 0.278 & 0.002 & 0.228 & 0.048 & 0.167 & 0.222 & 0.096 & 0.698 & 0.017 & 0.944 & 0.070 & 0.810 \\ \cline{2-15}
 & \multirow{2}{*}{0.4}    & S.A.  & 0.003 & 0.958 & 0.003 & 0.962 & 0.003 & 0.960 & 0.004 & 0.958 & 0.005 & 0.960 & 0.005 & 0.958 \\
 &                         & Expl. & 0.192 & 0.138 & 0.142 & 0.380 & 0.081 & 0.792 & 0.011 & 0.948 & 0.069 & 0.826 & 0.156 & 0.278 \\ \cline{2-15}
 & \multirow{2}{*}{0}      & S.A.  & 0.003 & 0.950 & 0.003 & 0.944 & 0.003 & 0.956 & 0.004 & 0.958 & 0.005 & 0.954 & 0.004 & 0.954 \\
 &                         & Expl. & 0.114 & 0.582 & 0.064 & 0.858 & 0.003 & 0.956 & 0.068 & 0.826 & 0.147 & 0.334 & 0.234 & 0.046 \\ \cline{2-15}
 & \multirow{2}{*}{$-0.4$} & S.A.  & 0.004 & 0.944 & 0.003 & 0.946 & 0.003 & 0.950 & 0.004 & 0.956 & 0.005 & 0.952 & 0.005 & 0.952 \\
 &                         & Expl. & 0.045 & 0.912 & 0.005 & 0.956 & 0.066 & 0.832 & 0.137 & 0.402 & 0.216 & 0.070 & 0.303 & 0.006 \\ \cline{2-15}
 & \multirow{2}{*}{$-0.8$} & S.A.  & 0.002 & 0.938 & 0.002 & 0.942 & 0.002 & 0.946 & 0.003 & 0.948 & 0.004 & 0.952 & 0.004 & 0.946 \\
 &                         & Expl. & 0.015 & 0.948 & 0.065 & 0.832 & 0.126 & 0.478 & 0.197 & 0.124 & 0.276 & 0.018 & 0.363 & 0.000 \\
\hline
\multirow{12}{*}{12}
 & \multirow{2}{*}{1.2}    & S.A.  & 0.001 & 0.950 & 0.001 & 0.952 & 0.002 & 0.950 & 0.002 & 0.954 & 0.003 & 0.950 & 0.003 & 0.952 \\
 &                         & Expl. & 0.366 & 0.000 & 0.312 & 0.000 & 0.246 & 0.012 & 0.172 & 0.206 & 0.089 & 0.694 & 0.003 & 0.958 \\ \cline{2-15}
 & \multirow{2}{*}{0.8}    & S.A.  & 0.001 & 0.954 & 0.001 & 0.950 & 0.002 & 0.956 & 0.002 & 0.954 & 0.003 & 0.954 & 0.003 & 0.956 \\
 &                         & Expl. & 0.281 & 0.000 & 0.226 & 0.028 & 0.161 & 0.268 & 0.086 & 0.708 & 0.003 & 0.958 & 0.083 & 0.700 \\ \cline{2-15}
 & \multirow{2}{*}{0.4}    & S.A.  & 0.002 & 0.958 & 0.002 & 0.952 & 0.003 & 0.958 & 0.003 & 0.958 & 0.004 & 0.954 & 0.004 & 0.952 \\
 &                         & Expl. & 0.196 & 0.096 & 0.141 & 0.378 & 0.076 & 0.752 & 0.001 & 0.954 & 0.081 & 0.716 & 0.168 & 0.234 \\ \cline{2-15}
 & \multirow{2}{*}{0}      & S.A.  & 0.002 & 0.950 & 0.003 & 0.954 & 0.003 & 0.954 & 0.003 & 0.958 & 0.004 & 0.954 & 0.004 & 0.950 \\
 &                         & Expl. & 0.117 & 0.526 & 0.062 & 0.806 & 0.003 & 0.954 & 0.078 & 0.734 & 0.160 & 0.262 & 0.247 & 0.024 \\ \cline{2-15}
 & \multirow{2}{*}{$-0.4$} & S.A.  & 0.002 & 0.956 & 0.003 & 0.958 & 0.004 & 0.958 & 0.004 & 0.948 & 0.005 & 0.952 & 0.004 & 0.950 \\
 &                         & Expl. & 0.046 & 0.890 & 0.009 & 0.954 & 0.075 & 0.748 & 0.149 & 0.322 & 0.232 & 0.036 & 0.318 & 0.000 \\ \cline{2-15}
 & \multirow{2}{*}{$-0.8$} & S.A.  & 0.003 & 0.952 & 0.004 & 0.950 & 0.005 & 0.952 & 0.004 & 0.948 & 0.005 & 0.948 & 0.005 & 0.948 \\
 &                         & Expl. & 0.017 & 0.944 & 0.071 & 0.772 & 0.137 & 0.384 & 0.211 & 0.060 & 0.294 & 0.000 & 0.380 & 0.000 \\
\hline
\end{tabular}
\end{table}

We conducted an additional simulation study to evaluate the finite-sample performance of the proposed bias-corrected point estimators and confidence sets for the sensitivity parameter range. Sample sizes of 200, 400, and 600 were considered, with 500 simulation replicates generated for each scenario. To obtain the bias-corrected point estimates of $\alpha_a^l$ and $\alpha_a^u$, we applied Algorithm \ref{alg:bias_corr} using 100 multiplier bootstrap replicates. To construct the corresponding bias-corrected confidence sets, we implemented the double-bootstrap procedure in Algorithm \ref{alg:double_boot} with 400 outer bootstrap replicates and 50 inner bootstrap replicates.

\begin{figure}[htbp]
\centering
\begin{subfigure}[b]{1\textwidth}
    \centering
    \includegraphics[width=\textwidth]{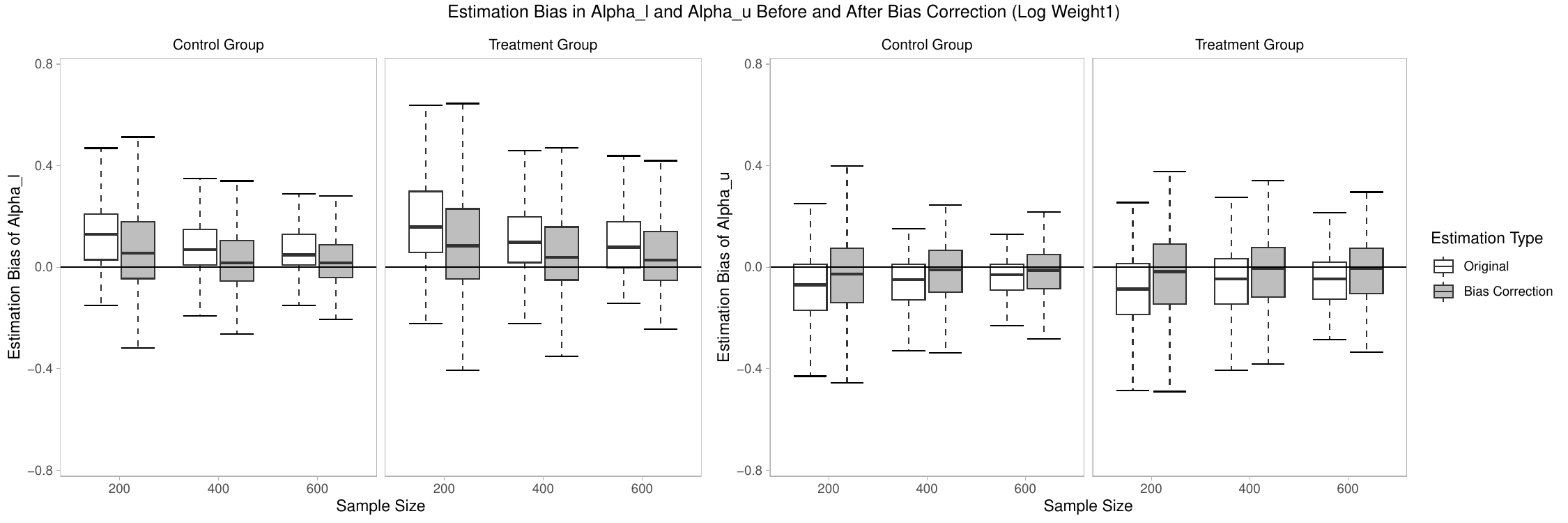}
\end{subfigure}

\medskip

\begin{subfigure}[b]{1\textwidth}   
    \centering 
    \includegraphics[width=\textwidth]{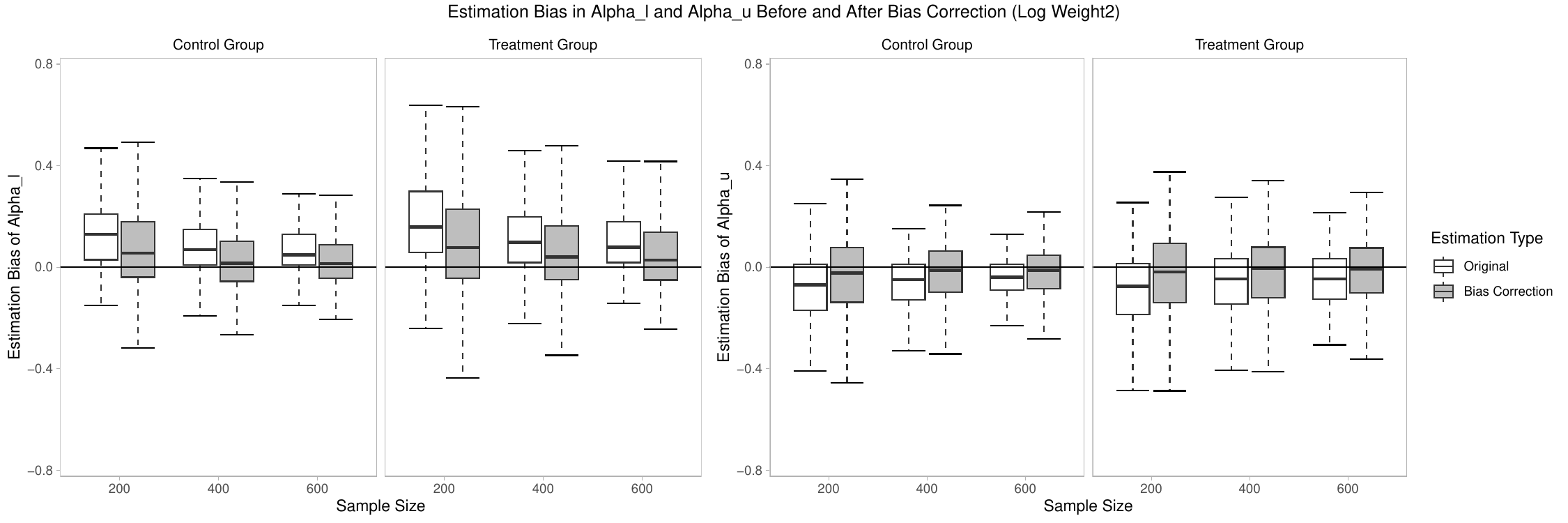}
   \end{subfigure}

   \medskip

\begin{subfigure}[b]{1\textwidth}   
    \centering 
    \includegraphics[width=\textwidth]{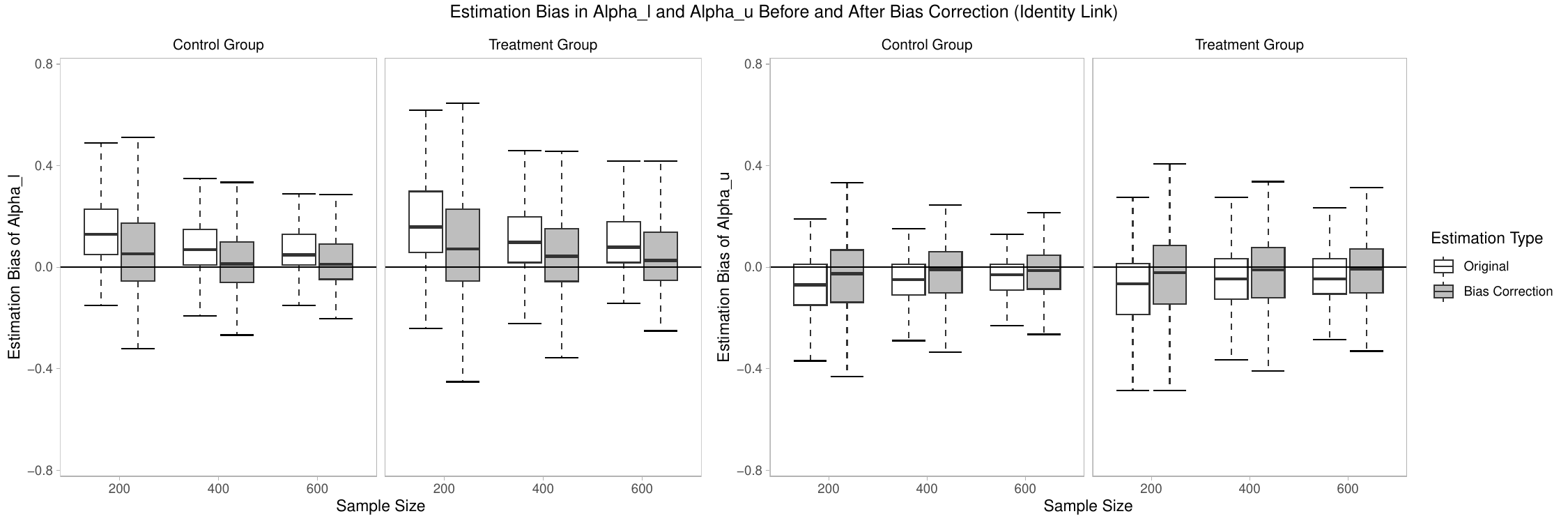}
   \end{subfigure}
   
\caption[Boxplots of the estimation bias of the lower and upper sensitivity parameter bounds, $\alpha_a^l$ and $\alpha_a^u$, before and after bias correction for the ACQ score under sample sizes of 200, 400, and 600.]{Boxplots of the estimation bias of the lower and upper sensitivity parameter bounds, $\alpha_a^l$ and $\alpha_a^u$, before and after bias correction for the ACQ score under sample sizes of 200, 400, and 600. The first row of panels corresponds to the log link with weight function $\symbf{W}_1(t;\symbf{\beta})$, the second row corresponds to the log link with weight function $\symbf{W}_2(t;\symbf{\beta})$, and the third row corresponds to the identity link. The ``Original'' estimates are obtained using Algorithm \ref{alg:point_est}, whereas the ``Bias Correction'' estimates are obtained using Algorithm \ref{alg:bias_corr} with 100 multiplier bootstrap replicates. Within each row, the left panel displays the estimation bias of $\alpha_a^l$, and the right panel displays the estimation bias of $\alpha_a^u$.} 
\label{fig:boxplot_log}
\end{figure}

\begin{figure}[htbp]
\centering
\begin{subfigure}[b]{1\textwidth}
    \centering
    \includegraphics[width=\textwidth]{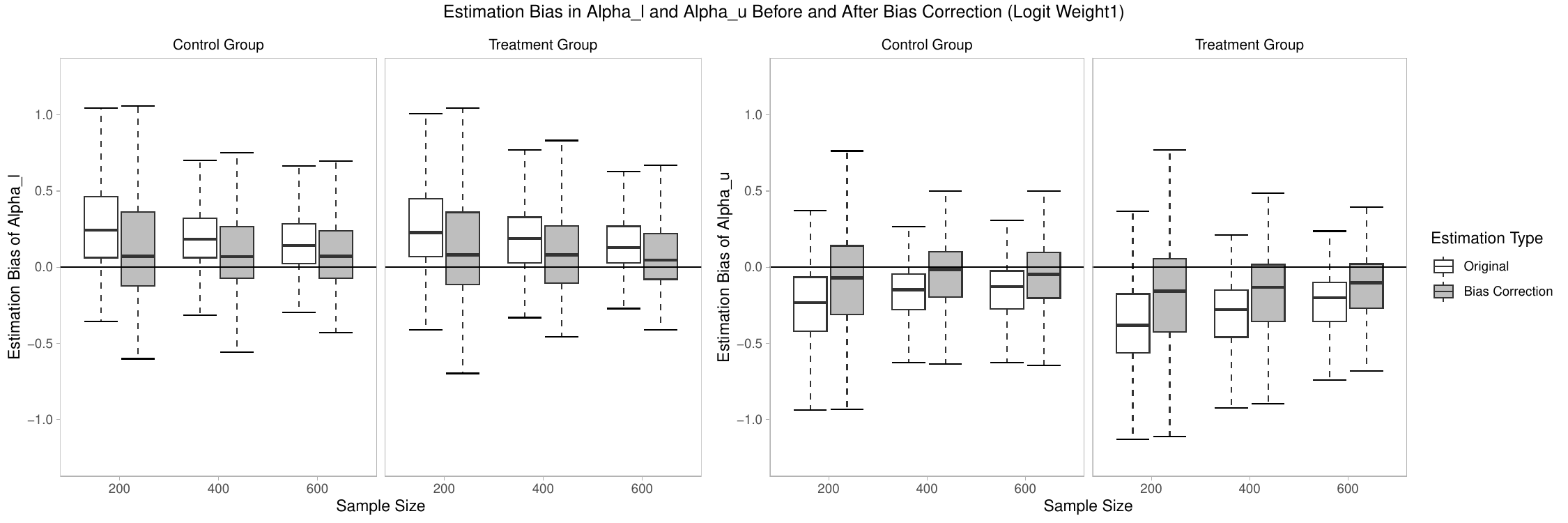}
\end{subfigure}

\medskip

\begin{subfigure}[b]{1\textwidth}   
    \centering 
    \includegraphics[width=\textwidth]{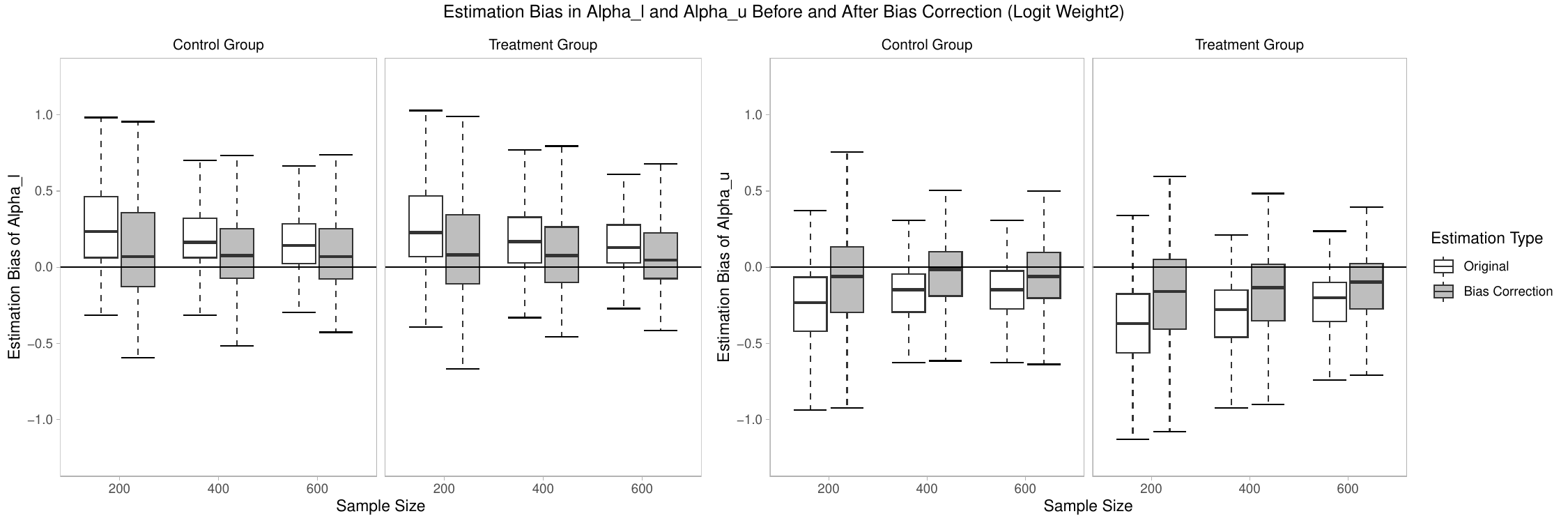}
   \end{subfigure}
\caption[Boxplots of the estimation bias of the lower and upper sensitivity parameter bounds, $\alpha_a^l$ and $\alpha_a^u$, before and after bias correction for the AC indicator under sample sizes of 200, 400, and 600.]{Boxplots of the estimation bias of the lower and upper sensitivity parameter bounds, $\alpha_a^l$ and $\alpha_a^u$, before and after bias correction for the AC indicator under sample sizes of 200, 400, and 600. The upper panels correspond to the logit link with weight function $\symbf{W}_1(t;\symbf{\beta})$, and the lower panels correspond to the logit link with weight function $\symbf{W}_2(t;\symbf{\beta})$. The ``Original'' estimates are obtained using Algorithm \ref{alg:point_est}, whereas the ``Bias Correction'' estimates are obtained using Algorithm \ref{alg:bias_corr} with 100 multiplier bootstrap replicates. Within each panel, the left plot displays the estimation bias of $\alpha_a^l$, and the right plot displays the estimation bias of $\alpha_a^u$.} 
\label{fig:boxplot_logit}
\end{figure}

Figures \ref{fig:boxplot_log} and \ref{fig:boxplot_logit} compare the estimation bias of the original point estimators obtained from Algorithm \ref{alg:point_est} with that of the bias-corrected estimators obtained from Algorithm \ref{alg:bias_corr}. As expected, the absolute bias of both the original and bias-corrected estimators decreases as the sample size increases. Moreover, across all sample sizes, treatment groups, and link functions, the bias-corrected estimators consistently exhibit smaller bias than the corresponding original estimators, demonstrating the effectiveness of the proposed bias-correction procedure.

A notable difference is observed between the continuous and binary outcome settings. For the ACQ score (identity and log links), the bias decreases rapidly with increasing sample size. When the sample size reaches 600, the boxplots are generally centered around zero, indicating that the remaining bias is negligible in most scenarios. In contrast, for the asthma control indicator (logit link), the bias remains noticeably larger. Although the bias is substantially reduced after bias correction and continues to decrease as the sample size increases, several boxplots at a sample size of 600 remain shifted below zero, indicating the presence of residual negative bias. This observation is consistent with the more challenging finite-sample behavior of the logit-link estimator and helps explain the relatively lower empirical coverage observed for the corresponding confidence sets.

\begin{table}[htbp]
\centering
\caption[Empirical coverage probabilities of the 95\% confidence sets for the sensitivity parameter range under sample sizes of 200, 400, and 600.]{Empirical coverage probabilities of the 95\% confidence sets for the sensitivity parameter range under sample sizes of 200, 400, and 600. The confidence sets were constructed using the basic bias-corrected bootstrap confidence interval and the Wald confidence interval based on the bias-corrected point estimates and bias-corrected bootstrap variance estimates.}
\label{table:simu_alpha_selection}
\small
\setlength{\tabcolsep}{4pt}
\begin{tabular}{lllcccccc}
\hline
 & & & \multicolumn{3}{c}{Basic} & \multicolumn{3}{c}{Wald} \\
\cline{4-9}
Outcome & Trt Group & Link & 200 & 400 & 600 & 200 & 400 & 600 \\
\hline
\multirow{6}{*}{ACQ Score}
 & \multirow{3}{*}{Treatment} & Identity                                 & 0.778 & 0.830 & 0.862 & 0.904 & 0.948 & 0.952 \\ \cline{3-9}
 &                            & Log ($\symbf{W}_1(t;\symbf{\beta})$)     & 0.778 & 0.830 & 0.868 & 0.908 & 0.944 & 0.952 \\ \cline{3-9}
 &                            & Log ($\symbf{W}_2(t;\symbf{\beta})$)     & 0.776 & 0.828 & 0.872 & 0.912 & 0.942 & 0.958 \\ \cline{2-9}
 & \multirow{3}{*}{Control}   & Identity                                 & 0.808 & 0.848 & 0.858 & 0.938 & 0.930 & 0.944 \\ \cline{3-9}
 &                            & Log ($\symbf{W}_1(t;\symbf{\beta})$)     & 0.810 & 0.848 & 0.872 & 0.936 & 0.932 & 0.944 \\ \cline{3-9}
 &                            & Log ($\symbf{W}_2(t;\symbf{\beta})$)     & 0.818 & 0.852 & 0.870 & 0.938 & 0.938 & 0.944 \\
\hline
\multirow{4}{*}{AC Indicator}
 & \multirow{2}{*}{Treatment} & Logit ($\symbf{W}_1(t;\symbf{\beta})$)   & 0.858 & 0.886 & 0.872 & 0.890 & 0.898 & 0.890 \\ \cline{3-9}
 &                            & Logit ($\symbf{W}_2(t;\symbf{\beta})$)   & 0.862 & 0.898 & 0.868 & 0.882 & 0.902 & 0.882 \\ \cline{2-9}
 & \multirow{2}{*}{Control}   & Logit ($\symbf{W}_1(t;\symbf{\beta})$)   & 0.856 & 0.880 & 0.896 & 0.886 & 0.906 & 0.916 \\ \cline{3-9}
 &                            & Logit ($\symbf{W}_2(t;\symbf{\beta})$)   & 0.848 & 0.882 & 0.884 & 0.886 & 0.902 & 0.912 \\
\hline
\end{tabular}
\end{table}

Table \ref{table:simu_alpha_selection} summarizes the empirical coverage probabilities of the confidence sets across treatment groups, link functions, weight functions, and sample sizes. During the simulation studies, we observed that, for the asthma control indicator (logit link), the point estimation procedure in Algorithm \ref{alg:point_est} occasionally failed to identify any value of $\alpha$ within the initial search space that satisfied the outcome mean constraint. For these simulated datasets, the corresponding confidence sets were treated as non-covering. The frequency of such failures decreased substantially as the sample size increased from 200 to 400 and 600, indicating improved finite-sample performance of the proposed estimator. Similarly, the number of NA values arising during the bootstrap bias-correction procedure also decreased with increasing sample size. Overall, the Wald confidence interval consistently achieved better empirical coverage than the basic bootstrap confidence interval. For most scenarios, the empirical coverage approached the nominal 95\% level as the sample size increased. This trend was particularly evident for the ACQ score, where the Wald confidence intervals attained coverage probabilities close to 0.95 across both treatment groups and all link and weight function combinations when the sample size was 600. In contrast, for the asthma control indicator, the empirical coverage remained below the nominal level even at the largest sample size. Moreover, the coverage probabilities at sample size 600 were similar to those at sample size 400, despite the reduction in finite-sample bias. This finding is consistent with the persistent non-negligible bias observed for the logit-link estimator and suggests that the remaining bias continues to affect the finite-sample accuracy of the confidence sets. Consequently, increasing the sample size beyond 400 provided only limited improvement in the empirical coverage for the binary outcome.

\section{Discussion}
\label{sec:discussion}

In modern clinical trials, the timing of outcome assessments often varies across participants, and assessment times may be associated with the underlying outcome process. If this dependence is not appropriately accounted for, treatment effect estimates at the target assessment times may be biased and may reflect differences in assessment timing rather than differences in participants' underlying outcome process. Most existing analysis methods for this setting assume a certain type of informativeness. However, multiple assumptions may be consistent with the observed study data, and as demonstrated in our simulation studies, different assumptions can lead to substantially different treatment effect estimates. The proposed sensitivity analysis framework enables inference for treatment effects under a range of plausible assumptions, providing a more comprehensive representation of the overall uncertainty in the study conclusions.

Compared with the previous sensitivity analysis framework of \cite{Smith2024}, the proposed methodology substantially broadens its scope. Our approach accommodates a broad class of outcome types, including continuous, binary, and count outcomes, through appropriate link functions. Furthermore, rather than requiring investigators to specify a plausible range of sensitivity parameters solely based on subject-matter knowledge, we develop a data-driven procedure for identifying plausible ranges of the sensitivity parameters and constructing confidence sets based on the plausible bounds for the target mean outcome. This additional layer of inference provides a principled approach to accounting for uncertainty in the specification of the sensitivity parameters and leads to more interpretable sensitivity analyses.

Several limitations need further investigation. The estimation procedure continues to rely on the non-future dependence assumption, and as in our previous work \citep{Smith2024}, assumes that no participants are censored throughout the study. Extending the proposed framework to relax these assumptions and accommodate informative censoring is an important direction for future research. In addition, although we propose a data-driven procedure for identifying plausible ranges of the sensitivity parameters and constructing confidence sets, the empirical coverage probabilities in the simulation studies were slightly below the nominal level in some settings, suggesting that larger sample sizes or further methodological refinements may be needed to improve finite-sample performance. We also observed occasional NA values during the sensitivity parameter selection procedure, primarily due to extreme values of the weighted influence function. These numerical issues occurred more frequently in the inner bootstrap procedure and under the logit link, and may affect the accuracy of the estimated sensitivity parameter range. Developing more numerically stable procedures that better handle extreme values of the weighted influence function is another important direction for future research.

\section*{Funding Statement}

Research reported in this article was funded through a Patient-Centered Outcomes Research Institute (PCORI) award ME-2021C3-24972. Shu Yang, Yujing Gao, Andrew Redd, Andrea Apter and Daniel Scharfstein received funding from this award. The views and statements presented in this article are solely the responsibility of the authors and do not necessarily represent the views of PCORI. The ARC study was funded through PCORI award AS-1307-05218.

\section*{Data Availability Statement}

The data from the Asthma Research for the Community study that support the findings in this paper will be shared on reasonable request to the corresponding author. All code required for the replication of the simulations and the case study is publicly available at \url{https://github.com/yujing24/sensitivity_general.git}.


\bibliographystyle{apalike}
\bibliography{bibliography}

\clearpage
\appendix
\renewcommand{\thesection}{\Alph{section}}
\setcounter{equation}{0}
\renewcommand{\theequation}{\thesection.\arabic{equation}}

\section{Theory}

\subsection{Proof of Propositions \ref{chp4:prop:lambda_rho} and \ref{chp4:prop:identify}}
\label{apxB:appendix:condition}



\subsubsection{Proof of Proposition \ref{chp4:prop:lambda_rho}}
\label{app:prop1}

\begin{proof}[Proof of Proposition~\ref{chp4:prop:lambda_rho} ]

Fix a time $t \in [t_1,t_2]$ and a value $y(t)$.  For each $\epsilon > 0$ let: 
\begin{align*}
A_\epsilon := & \ dF\big\{y(t) \mid N(t+\epsilon)-N(t-)=0, \overline{\symbf{O}}(t) \big\} \\
B_\epsilon := & \ dF\big\{y(t)\mid N(t+\epsilon)-N(t-) =1, \overline{\symbf{O}}(t) \big\} 
\frac{ \exp \{ \alpha Y(t) \} }{ \mathbb{E} \big[ \exp\{ \alpha Y(t) \} \mid  N(t+\epsilon)-N(t-)=1, \overline{\symbf{O}}(t) \big] }.
\end{align*}

\noindent For each $\epsilon >0$, by Bayes' theorem we have:
\begin{align}
 & \mathbb{P}\big\{ N(t+\epsilon)-N(t-) = 1 \mid \overline{\symbf{O}}(t), Y(t)=y(t)\big\} \nonumber \\
={} & \frac{ \mathbb{P} \big\{ N(t+\epsilon)-N(t-)=1 \mid \overline{\symbf{O}}(t) \big\}}{
  \mathbb{P} \big\{ N(t+\epsilon)-N(t-)=1 \mid \overline{\symbf{O}}(t) \big\} + \mathbb{P} \big\{ N(t+\epsilon)-N(t-)=0 \mid \overline{\symbf{O}}(t) \big\} D_\epsilon } ,
\label{bayes}
\end{align}

\noindent for $D_\epsilon :=  A_\epsilon / dF\big\{y(t) \mid N(t+\epsilon)-N(t-)=1, \overline{\symbf{O}}(t) \big\}$.  

Dividing each side of \eqref{bayes} by $\epsilon$ and taking the limit as $\epsilon \to 0^+$, the left side of \eqref{bayes} is
\begin{align*}
\lim_{\epsilon \to 0^+} & \frac{ 
\mathbb{P} \big\{ N(t+\epsilon)-N(t-) =1 \mid \overline{\symbf{O}}(t), Y(t)=y(t) \big\} }{\epsilon} = \rho\big\{t\mid \overline{\symbf{O}}(t), Y(t)\big\},
\end{align*}
and the numerator of the right side of \eqref{bayes} is
\begin{align*}
\lim_{\epsilon \to 0^+} & \frac{ 
\mathbb{P} \big\{ N(t+\epsilon)-N(t-) =1 \mid \overline{\symbf{O}}(t) \big\} }{\epsilon} = \lambda\big\{t\mid \overline{\symbf{O}}(t)\big\}.
\end{align*}

For the denominator of the right side of \eqref{bayes}, write:

\begin{align*}
D_\epsilon = & \  
\frac{ B_\epsilon}{ dF\big\{y(t)\mid N(t+\epsilon)-N(t-)=1 , \overline{\symbf{O}}(t) \big\} }
+
\frac{ A_\epsilon- B_\epsilon }{ dF\big\{y(t)\mid N(t+\epsilon)-N(t-)=1, \overline{\symbf{O}}(t) \big\} }\\
= & \ \frac{ \exp\{ \alpha Y(t) \} }{ \mathbb{E} \big[ \exp \{ \alpha Y(t) \} \mid N(t+\epsilon)-N(t-) =1, \overline{\symbf{O}}(t) \big]  }  
+ 
\frac{ A_\epsilon- B_\epsilon}{ dF\big\{ y(t) | N(t+\epsilon)-N(t-)=1, \overline{\symbf{O}}(t) \big\} } ,
\end{align*}

\noindent and note that:
\begin{align*}
\lim_{\epsilon \to 0^+} (A_\epsilon  - B_\epsilon) 
= & \  dF \big\{y(t) \mid \Delta N(t)=0, \overline{\symbf{O}}(t) \big\} \\
& \quad -  
dF \big\{y(t) \mid \Delta N(t)=1, \overline{\symbf{O}}(t) \big\}
\frac{ \exp \{ \alpha Y(t) \} }{ \mathbb{E} \big[ \exp \{ \alpha Y(t) \} \mid \Delta N(t)=1, \overline{\symbf{O}}(t) \big] } \\
= & \  0 
\end{align*}
\noindent by Assumption \ref{chp4:assum:tilting}.  

Therefore, 
\begin{align*}
\lim_{\epsilon \to 0^+} D_{\epsilon} = \frac{ \exp \{ \alpha Y(t) \} }{ \mathbb{E} \big[ \exp \{ \alpha Y(t) \} \mid \Delta N(t)=1, \overline{\symbf{O}}(t) \big] }
\end{align*}
completes the proof.

\end{proof}

\subsubsection{Proof of Proposition \ref{chp4:prop:identify}}
\label{app:prop2}

Before we prove Proposition \ref{chp4:prop:identify}, we provide the following proposition to show that the mean outcome $\mu(t)$ in each arm is identified from the observed data.

\begin{proposition}
\label{prop:mu_identity}
For each time $t$ with $\lambda\{t\mid \overline{\symbf{O}}(t)\} > 0$, the mean $\mu(t)$ in each arm is identified under Assumption \ref{chp4:assum:tilting}, as
\begin{align*}
\mu(t) = \mathbb{E} \left[  \frac{\mathbb{E} \big[ Y(t) \exp\{ \alpha Y(t) \}\mid \Delta N(t)=1,\overline{\symbf{O}}(t) \big]}{\mathbb{E} \big[  \exp\{ \alpha Y(t) \}\mid \Delta N(t)=1, \overline{\symbf{O}}(t) \big] } \right].
\end{align*}
\end{proposition}

\begin{proof}[Proof of Proposition~{\rm\ref{prop:mu_identity}}]

Let $t$ be a time with $\lambda\{t \mid \overline{\symbf{O}}(t)\} >0$.  For each $\epsilon > 0$,

\begin{align}
\mu(t) = & \int_{\overline{\symbf{o}}(t)} \int_{y(t)} y(t)\,
  dF \big\{ y(t) \mid \overline{\symbf{O}}(t)=\overline{\symbf{o}}(t) \big\}\, dF \big\{\overline{\symbf{o}}(t) \big\} \nonumber \\
={} & \int_{\overline{\symbf{o}}(t)} \int_{y(t)} y(t)
  \Big[ dF \big\{ y(t) \mid N(t+\epsilon)-N(t-)=0, \overline{\symbf{O}}(t)=\overline{\symbf{o}}(t) \big\} \nonumber \\
 & \hspace{1.2in} \times \mathbb{P} \big\{ N(t+\epsilon)-N(t-)=0 \mid \overline{\symbf{O}}(t)=\overline{\symbf{o}}(t) \big\} \nonumber \\
 & \hspace{0.9in} + dF \big\{ y(t) \mid N(t+\epsilon)-N(t-)=1, \overline{\symbf{O}}(t)=\overline{\symbf{o}}(t) \big\} \nonumber \\
 & \hspace{1.2in} \times \mathbb{P} \big\{ N(t+\epsilon)-N(t-)=1 \mid \overline{\symbf{O}}(t)=\overline{\symbf{o}}(t) \big\} \Big]\, dF\big\{\overline{\symbf{o}}(t) \big\}. \nonumber
\end{align}

Taking the limit as $\epsilon \to 0^+$, we have:
\begin{align}
\mu(t) ={} & \int_{\overline{\symbf{o}}(t)} \int_{y(t)} y(t) \Big[ dF \big\{ y(t) \mid \Delta N(t)=0, \overline{\symbf{O}}(t)=\overline{\symbf{o}}(t) \big\} \times 1 \nonumber \\
 & \hspace{0.9in} + dF \big\{ y(t) \mid \Delta N(t)=1, \overline{\symbf{O}}(t)=\overline{\symbf{o}}(t) \big\} \times 0 \Big]\, dF \big\{\overline{\symbf{o}}(t) \big\} \nonumber \\
={} & \int_{\overline{\symbf{o}}(t)} \int_{y(t)} \frac{ y(t) \exp \big\{ \alpha y(t) \big\}\, dF \big\{ y(t) \mid \Delta N(t)=1, \overline{\symbf{O}}(t) =\overline{\symbf{o}}(t) \big\} }{ \mathbb{E} \big[ \exp \big\{\alpha Y(t) \big\} \mid \Delta N(t) =1, \overline{\symbf{O}}(t)=\overline{\symbf{o}}(t) \big] }\, dF\big\{\overline{\symbf{o}}(t) \big\} \nonumber \\
={} & \mathbb{E} \left[ \frac{\mathbb{E} \big[ Y(t) \exp\{ \alpha Y(t) \}\mid \Delta N(t)=1,\overline{\symbf{O}}(t) \big]}{\mathbb{E} \big[ \exp\{ \alpha Y(t) \}\mid \Delta N(t)=1, \overline{\symbf{O}}(t) \big] } \right]. \nonumber
\end{align}

\end{proof}

Based on Proposition \ref{prop:mu_identity}, we provide the proof for the Proposition \ref{chp4:prop:identify}.

\begin{proof}
[Proof of Proposition~{\rm\ref{chp4:prop:identify}}]
Under Assumptions \ref{chp4:assum:tilting} and \ref{chp4:assum:positivity}, by Proposition \ref{prop:mu_identity}, $\mu(t)$ is identified from the observed data for each $t\in [t_1,t_2]$, and under Assumption \ref{chp4:assum:marginal}, 
$\int_{t=t_1}^{t_2} \symbf{B}(t) \symbf{B}(t)^\prime dt$ is invertible, then 
$$\symbf{\beta} = \left\{\int_{t=t_1}^{t_2} \symbf{B}(t) \symbf{B}(t)^\prime dt\right\}^{-1} \int_{t=t_1}^{t_2} \symbf{B}(t) s^{-1}\left\{\mu(t)\right\}dt.$$
\end{proof}


\subsection{Influence Function Derivation}
\label{apxB:app:IF_deriv}

\begin{proof}[Proof of Theorem~\ref{chp4:theorem:IF}] 
In this appendix, we prove Theorem ~\ref{chp4:theorem:IF} using the semiparametric theory \citep{scharfstein1999adjusting, tsiatis2006semiparametric, Smith2024}. We first define the corresponding notation and review elements of this theory applied to our context.

Let $\{\symbf{\mathcal{F}}\}_{t=0}^{\tau}$ be the filtration defined by $\symbf{\mathcal{F}}_t = (\{ \symbf{X}(u), Y(u): \Delta N(u) = 1, 0\leq u \leq t\}, \{N(u): 0\leq u \leq t\})$. Here $\symbf{\mathcal{F}}_t$ is the observed past up through time $t$, and $\symbf{\mathcal{F}}_{t-} = \overline{\symbf{O}}(t)$ is the observed past prior to, but not including, time $t$. The intensity function $\lambda\{t \mid \overline{\symbf{O}}(t)\}$ is the intensity for the counting process $\{N(t):t\in[0,\tau]\}$ with respect to the filtration $\{\symbf{\mathcal{F}}_t\}_{t=0}^{\tau}$. Additionally, let $\symbf{L} = \{\symbf{X}(t), Y(t): 0\leq t\leq \tau\}$ and let $\{\symbf{\mathcal{G}}_t\}_{t=0}^{\tau}$ be the filtration defined by $\symbf{\mathcal{G}}_t = (\symbf{L}, \{N(u): 0\leq u\leq t\}) = (\symbf{L}, \symbf{\mathcal{F}}_t)$. Under Assumption \ref{chp4:assum:nofuture}, $\rho\big\{t \mid \overline{\symbf{O}}(t), Y(t)\big\}$ is the intensity function for $\{N(t)\}_{t=0}^{\tau}$ with respect to the filtration $\{\symbf{\mathcal{G}}_t\}_{t=0}^\tau$. 

We will refer to $\symbf{L} = \{\symbf{X}(t), Y(t): 0\leq t\leq \tau\}$ as the full data. We write $\symbf{R} = \{N(t):t_1\leq t\leq t_2\}$, and we refer to $\symbf{Z} = (\symbf{L}, \symbf{R})$ as the total data. Let $\mathcal{H}^{L}$, $\mathcal{H}^{O}$, $\mathcal{H}^{Z}$ be the Hilbert spaces of mean-zero, finite-variance, p-dimensional (where $p$ is the length of $\symbf{\beta}$) functions of the full, observed, and total data respectively, with the covariance inner product. Full-data influence functions are normalized elements of the orthogonal complement of the full-data nuisance tangent space $\Lambda( F_{L})$. Observed data influence functions are normalized elements of $\Lambda^{O,\perp} = \Lambda_1^{O,\perp} \cap \Lambda_2^{O,\perp}$, where $\Lambda_1^{O}$ is the image of $\Lambda(F_{L})$, and $\Lambda_2^{O}$ is the image of the coarsening nuisance tangent space $\Lambda(F_{R|L})$, under the map $A: \mathcal{H}^{Z} \rightarrow \mathcal{H}^{O}$ given by $A(\cdot) = \mathbb{E}\{\cdot \mid \symbf{O}\}$. Therefore, the orthogonal complement $\Lambda_1^{O,\perp}$ is equal to the null space of the adjoint of the map $A_1: \Lambda (F_L) \rightarrow \mathcal{H}^{O}$ given by $A_1(\cdot) = \mathbb{E}\{\cdot \mid \symbf{O}\}$. Similarly, $\Lambda_2^{O,\perp}$ is the null space of the adjoint of the map $A_2: \Lambda(F_{R|L})\rightarrow \mathcal{H}^{O}$ given by $A_2(\cdot) = \mathbb{E}\{\cdot \mid \symbf{O}\}$. The adjoint of $A_1$ is the map $A_1^*: \mathcal{H}^{O}\rightarrow \Lambda (F_{L})$ given by $A_1^*(\cdot) = \Pi[\mathbb{E}\{\cdot \mid \symbf{L}\}| \Lambda(F_L)]$. Therefore, $\Lambda_1^{O,\perp} = \{h(\symbf{O}) \in \mathcal{H}^{O}: \mathcal{K}(g(\symbf{O})) \in \Lambda (F_L)^{\perp}\}$, where $\mathcal{K}: \mathcal{H}^{O} \rightarrow \mathcal{H}^{L}$ is the map defined by $\mathcal{K}(g(\symbf{O})) = \mathbb{E}\{g(\symbf{O})\mid \symbf{L}\} $. Again using adjoints, $\Lambda_2^{O,\perp} = \{g(\symbf{O}) \in \mathcal{H}^{O}: g(\symbf{O}) \in \Lambda (F_{R|L})^{\perp}\} = \mathcal{H}^{O} \cap \Lambda(F_{R|L})^{\perp}$, and therefore, $\Lambda^{O,\perp} = \Lambda_1^{O,\perp} \cap \Lambda(F_{R|L}^{\perp})$.

In the following part, we first derive full-data influence functions $\symbf{\varphi}_j\left(\symbf{L}\right)$, obtain elements of $\Lambda_1^{O,\perp}$ by inverse-weighting $\symbf{\varphi}_j\left(\symbf{L}\right)$ and adding elements of the augmentation space. Then we identify one of these elements that is also orthogonal to $\Lambda(F_{R|L})$. By the above theory, this element is in the space $\Lambda^{O,\perp}$. We further show this element inherits the property of being correctly normalized from $\symbf{\varphi}_j\left(\symbf{L}\right)$, and is therefore an observed data influence function.

\textbf{Part 1: a full-data influence function}.

We obtain a full-data influence function for $\beta$ following the approach provided by  \cite{diaz2011superlearner}. We specify a function $\symbf{B}(t)$ with $\int_{t=t_1} ^{t_2} \symbf{B}(t) \symbf{B}(t)^\prime dt$ positive definite, and a strictly increasing inverse link function $s\in \mathbb{C}^2$, which we will use when we make Assumption~\ref{chp4:assum:marginal}.
However, we initially do not assume that the equality $ \mu(t) = s\left\{\symbf{B}(t)^\prime \symbf{\gamma}\right\}$ in Assumption \ref{chp4:assum:marginal}
holds, and instead work in the nonparametric model on $\symbf{L}$. We use a smooth function of the form $s\left\{\symbf{B}(t)^\prime \symbf{\gamma}\right\}$ to approximate the curve $\mu(t)$, where the inverse link function $s(\cdot)$ is strictly monotonic and invertible under Assumption \ref{chp4:assum:marginal}.

Consider the following two types of loss functions:
\begin{align}
    \mathcal{L}_1(\symbf{\gamma}) &= \int_{t=t_1}^{t_2} \left[ g\{\mu(t)\} - \symbf{B}(t)^\prime \symbf{\gamma}\right]^2 dt \label{equ:loss_square1}
    \\
    \mathcal{L}_2(\symbf{\gamma}) & = \int_{t=t_1}^{t_2} \left[ b\{\symbf{B}(t)^\prime \symbf{\gamma}\} - \mu(t)\symbf{B}(t)^\prime \symbf{\gamma}\right] dt
    \label{equ:loss_exp}
\end{align}
where $b(z)$ is defined by 
$b(z) = \int s(z) dz$, which means $s(z) = \frac{\partial b(z)}{\partial z }$. 

The loss function $\mathcal{L}_1(\symbf{\gamma})$ is the squared error loss function in the transformed space and the loss function $\mathcal{L}_2(\symbf{\gamma})$ is the quasi log-likelihood loss function inspired by the negative log-likelihood of the exponential family distribution.

Next, we derive the second derivative of each loss function:
\begin{align*}
    \frac{\partial ^2 \mathcal{L}_1 (\symbf{\gamma})}{\partial \symbf{\gamma}\partial  \symbf{\gamma}^\prime} & = \int_{t=t_1}^{t_2} 2\symbf{B}(t) \symbf{B}(t) ^\prime dt \\
    \frac{\partial ^2 \mathcal{L}_2 (\symbf{\gamma})}{\partial \symbf{\gamma}\partial  \symbf{\gamma}^\prime} & = \int_{t=t_1}^{t_2} \frac{\partial  s(z)}{\partial z} \Bigg|_{z = \symbf{B}(t)^\prime \symbf{\gamma}}\symbf{B}(t) \symbf{B}(t)^\prime dt 
\end{align*}

Since $\int_{t=t_1} ^{t_2} \symbf{B}(t) \symbf{B}(t)^\prime dt$ is positive definite, then $\mathcal{L}_1(\symbf{\gamma})$ in equation~\eqref{equ:loss_square1} is strictly convex.

Since $s(\cdot)$ is strictly increasing and $\int_{t=t_1} ^{t_2} \symbf{B}(t) \symbf{B}(t)^\prime dt$ is positive definite, $\frac{\partial ^2 \mathcal{L}_2 (\symbf{\gamma})}{\partial \symbf{\gamma}\partial \symbf{\gamma}^\prime}$ is positive definite, and $\mathcal{L}_2(\symbf{\gamma})$ in equation \eqref{equ:loss_exp} is also strictly convex.

Thus, loss functions $\mathcal{L}_j(\symbf{\gamma})$, $j=1,2$, have a unique minimizer, defined as a parameter $\symbf{\gamma}_0$, which is the unique solution of the equation by setting the first derivative of $\mathcal{L}_j (\symbf{\gamma})$ to be $\symbf{0}$:
\begin{align}
    \frac{\partial \mathcal{L}_1 (\symbf{\gamma})}{\partial \symbf{\gamma}}\Big |_{\symbf{\gamma} = \symbf{\gamma}_0} & = \int_{t=t_1}^{t_2} -2 \symbf{B}(t) \left[g\{\mu(t)\}  - \symbf{B}(t)^\prime \symbf{\gamma}_0\right] dt = \symbf{0}
    \label{equ:loss1_deriv} 
    \\
    \frac{\partial \mathcal{L}_2 (\symbf{\gamma})}{\partial \symbf{\gamma}} \Big |_{\symbf{\gamma} = \symbf{\gamma}_0}& = \int_{t=t_1}^{t_2} \symbf{B}(t) \left[s\left\{\symbf{B}(t)^\prime \symbf{\gamma}_0\right\} - \mu(t)\right] dt = \symbf{0}
    \label{equ:loss3_deriv}
\end{align}

Assuming the equality $\mu(t) = s\left\{\symbf{B}(t)^\prime \symbf{\beta}\right\}$ in Assumption~\ref{chp4:assum:marginal}, if we substitute $\mu(t)$ in the first derivative of each loss function with $s\left\{\symbf{B}(t)^\prime \symbf{\beta}\right\}$ and set $\symbf{\gamma}_0 = \symbf{\beta}$, then we have $\frac{\partial \mathcal{L}_j (\symbf{\gamma})}{\partial \symbf{\gamma}}\big|_{\symbf{\gamma} = \symbf{\beta}} = 0$ for $j=1,2$. Thus, under Assumption \ref{chp4:assum:marginal}, 
$\symbf{\gamma}_0 = \symbf{\beta}$ is the unique minimizer of each loss function $\mathcal{L}_j (\symbf{\gamma})$ for $j=1,2$.

Below we compute the influence function for $\symbf{\gamma}_0$ in the nonparametric model, which we denote by $\symbf{\varphi}^{NP}\left(\symbf{L};\symbf{\gamma}_0\right)$. Then $\symbf{\varphi}^{NP}\left(\symbf{L};\symbf{\gamma}_0\right)$ will also be a valid full-data influence function for $\symbf{\beta}$ in our model.

Let $\mathcal{P}$ be any parametric submodel (parametrized by $\epsilon$ such that $\epsilon = 0$ corresponds to the true distribution of $\symbf{L}$) and let $a(\symbf{L})$ denote the score vector for $\mathcal{P}$. The influence function for $\symbf{\gamma}_0$ is the mean-zero, p-dimensional function of $\symbf{L}$ satisfying $\frac{\partial \symbf{\gamma}_0}{\partial \epsilon} \big |_{\epsilon = 0} = \mathbb{E}\left\{\symbf{\varphi}^{NP}\left(\symbf{L};\symbf{\gamma}_0\right) a(\symbf{L})\right\}$. For a fixed $a(\symbf{L})$, we consider the parametric submodel $\mathcal{P}_a = \left\{P_a(\epsilon): \epsilon \in D \subset \mathbb{R}\right\}$, where $P_a(\epsilon) = dF_0(\symbf{L})\left\{1 + \epsilon a(\symbf{L})\right\}$, and where the score of $P_a(\epsilon)$ evaluated at $\epsilon = 0$ is $a(\symbf{L})$.

In $\mathcal{P}_a$, equations ~\eqref{equ:loss1_deriv}, ~\eqref{equ:loss3_deriv} become:
\begin{align}
    \mathcal{L}_1 (\symbf{\gamma}): & \int_{t=t_1}^{t_2} \symbf{B}(t) \left[g\left[ \int_{l} y (t) dF_0(\symbf{l}) \left\{1+ \epsilon a(\symbf{l})\right\}\right]  - \symbf{B}(t)^\prime \symbf{\gamma}(\epsilon)\right] dt = \symbf{0} 
    \label{equ:loss1_deriv1}\\
  \mathcal{L}_2 (\symbf{\gamma}): & \int_{t=t_1}^{t_2} \symbf{B}(t)\left[ s\left\{\symbf{B}(t)^\prime \symbf{\gamma}(\epsilon)\right\}- \int_{\symbf{l}}y(t) dF_0(\symbf{l})\left\{1 + \epsilon a(\symbf{l})\right\}\right] dt =  \symbf{0} 
    \label{equ:loss3_deriv1}
\end{align}


Taking the derivative with respect to $\epsilon$ on both sides of equations~\eqref{equ:loss1_deriv1} and ~\eqref{equ:loss3_deriv1} gives:
\begin{align}
    \mathcal{L}_1 (\symbf{\gamma}): & \int_{t=t_1}^{t_2} \symbf{B}(t) \left[\frac{\partial g(z)}{\partial z} \Big|_{z = \int_{l} y (t) dF_0(\symbf{l}) \left\{1+ \epsilon a(\symbf{l})\right\}} \left[\int_{l} y (t) dF_0(\symbf{l}) a(\symbf{l})\right]  - \symbf{B}(t)^\prime \frac{\partial \symbf{\gamma}}{\partial \epsilon} \right] dt = \symbf{0}
    \label{equ:loss1_deriv2} \\
    \mathcal{L}_2 (\symbf{\gamma}): & \int_{t=t_1}^{t_2} \symbf{B}(t) \left[\frac{\partial s(z)}{\partial z} \Bigg|_{z = \symbf{B}(t)^\prime \symbf{\gamma}(\epsilon)} \symbf{B}(t)^\prime \frac{\partial \symbf{\gamma}}{\partial \epsilon} - \int_{\symbf{l}}y(t) dF_0(\symbf{l})a(\symbf{l})\right]dt =  \symbf{0}.
    \label{equ:loss3_deriv2}
\end{align}

Setting $\epsilon = 0$ on both sides of equation~\eqref{equ:loss1_deriv2} and solving for $\frac{\partial \symbf{\gamma}}{\partial \epsilon}\big|_{\epsilon = 0}$, we have:
\begin{align}
     \mathcal{L}_1 (\symbf{\gamma}):  \frac{\partial \symbf{\gamma}}{\partial \epsilon}\Bigg|_{\epsilon = 0}
    & = \symbf{V}_1^{-1} \int_{t=t_1}^{t_2} \symbf{B}(t) \frac{\partial g(z)}{\partial z} \Big|_{z = \mathbb{E}\left\{Y(t)\right\}} \mathbb{E}\left\{Y(t)a(\symbf{L})\right\} dt \nonumber \\
    & = \mathbb{E}\left[\left\{\symbf{V}_1^{-1} \int_{t=t_1}^{t_2} \symbf{B}(t) \frac{\partial g(z)}{\partial z} \Big|_{z = \mathbb{E}\left\{Y(t)\right\}}  Y(t)dt \right\}a(\symbf{L}) \right] \nonumber \\
    & = \mathbb{E}\left[\left[\symbf{V}_1^{-1}\int_{t=t_1}^{t_2} \symbf{B}(t) \frac{\partial g(z)}{\partial z} \Big|_{z = \mathbb{E}\left\{Y(t)\right\}} \left[ Y(t) - \mathbb{E}\{Y(t)\} \right]dt \right]a(\symbf{L}) \right],
    \label{equ:loss1_deriv3}
\end{align}
where $\symbf{V}_1 = \int_{t=t_1}^{t_2} \symbf{B}(t) \symbf{B}(t)^\prime dt$.

Setting $\epsilon = 0$ on both sides of equation~\eqref{equ:loss3_deriv2} and solving for $\frac{\partial \symbf{\gamma}}{\partial \epsilon}\big|_{\epsilon = 0}$, we have:

\begin{align}
    \mathcal{L}_2(\symbf{\gamma}): \frac{\partial \symbf{\gamma}}{\partial \epsilon}\Bigg|_{\epsilon = 0}
    & = \symbf{V}_2(\symbf{\gamma}_0)^{-1}\int_{t=t_1}^{t_2} \symbf{B}(t) \mathbb{E}\left\{Y(t)a(\symbf{L})\right\} dt \nonumber \\
    & = \mathbb{E}\left[\left\{\symbf{V}_2(\symbf{\gamma}_0)^{-1} \int_{t=t_1}^{t_2} \symbf{B}(t) Y(t)dt \right\}a(\symbf{L}) \right] \nonumber \\
    & = \mathbb{E}\left[\left[\symbf{V}_2(\symbf{\gamma}_0)^{-1} \int_{t=t_1}^{t_2} \symbf{B}(t) \left[ Y(t) - \mathbb{E}\{Y(t)\} \right]dt \right]a(\symbf{L}) \right],
    \label{equ:loss3_deriv3}
\end{align}
where $\symbf{V}_2(\symbf{\gamma}_0) = \int_{t=t_1}^{t_2} 
\symbf{B}(t)\symbf{B}(t)^\prime \frac{\partial s(z)}{\partial z} \Bigg|_{z = \symbf{B}(t)^\prime \symbf{\gamma}_0}dt$.

Therefore, under Assumption~\ref{chp4:assum:marginal},
the influence function for $\symbf{\gamma}_0$ of each loss function is:
\begin{align}
  \mathcal{L}_j(\symbf{\gamma}): \quad & \symbf{\varphi}_j\left(\symbf{L}\right) = \int_{t=t_1}^{t_2} \symbf{W}_j(t;\symbf{\beta})\left[ Y(t) - s\left\{ \symbf{B}(t)^\prime \symbf{\beta}\right\}\right]dt,
  \label{equ:full_influ1}
\end{align}

where 
\begin{align*}
    \symbf{W}_1(t;\symbf{\beta}) & = \symbf{V}_1^{-1}\symbf{B}(t)\frac{\partial g(z)}{\partial z} \Big|_{z = s\left\{\symbf{B}(t)^\prime \symbf{\beta}\right\}}, \quad \symbf{V}_1 = \int_{t=t_1}^{t_2} \symbf{B}(t) \symbf{B}(t)^\prime dt;
    \\
    \symbf{W}_2(t;\symbf{\beta}) & = \symbf{V}_2\left(\symbf{\beta}\right)^{-1}\symbf{B}(t), \quad \symbf{V}_2(\symbf{\beta}) = \int_{t=t_1}^{t_2} \symbf{B}(t)\symbf{B}(t)^\prime \frac{\partial s(z)}{\partial z} \Bigg|_{z = \symbf{B}(t)^\prime \symbf{\beta}}dt.
\end{align*}


\textbf{Part 2: an inverse-weighted element and the augmentation space}. 

Next, for $\symbf{\varphi}_j\left(\symbf{L}\right)$ in equations~\eqref{equ:full_influ1}, we find an element $\symbf{g}_j^*\left(\symbf{O}\right) \in \mathcal{H}^O$ which maps to $\symbf{\varphi}_j\left(\symbf{L}\right)$ under the map $\mathcal{K}(\cdot)$. One such element is the following inverse-weighted element:
\begin{align}
\symbf{g}_j^*\left(\symbf{O}\right) \in \mathcal{H}^O 
& = \int_{t=t_1}^{t_2} \frac{
\symbf{W}_j(t;\symbf{\beta}) \left[ Y(t) - s\left\{ \symbf{B}(t)^\prime \symbf{\beta}\right\}\right]}{\rho\left\{t \mid \overline{\symbf{O}}(t), Y(t)\right\}} dN(t), 
\label{equ:gstar}
\end{align}
as we verify below, $j=1,2$.

To show this, we will make use of a counting process martingale whose properties we will exploit throughout this appendix. Define $M(t):= N(t) - \int_{u=0}^t \rho\left\{u \mid \overline{\symbf{O}}(u), Y(u)\right\} du$. Under Assumption~\ref{chp4:assum:nofuture},
$\left\{M(t)\right\}_{t=0}^{\tau}$ is a martingale adapted to the filtration $\{\symbf{\mathcal{G}}_t\}_{t=0}^\tau$. Let $H(t)$ be any predictable process. A sufficient condition for this is that $H(t)$ is left continuous in $t$, and that for each $t$, $H(t)$ is $\symbf{\mathcal{G}}_t$-measurable. In particular, with an assumption that $Y(t)$ is left continuous, $H(t)$ can be any function of $\overline{\symbf{O}}(t)$ and $Y(t)$. 

Then $U(t) := \int_{u=0}^{t}H(u)dM(u)$ is also a martingale adapted to $\{\symbf{\mathcal{G}}_t\}_{t=0}^\tau$, so $\mathbb{E}\left\{ U(t)\mid \symbf{\mathcal{G}}_s\right\} = U(s)$ for all $s<t$.  In particular, $U(t)$ has mean zero. Also, since $\symbf{\mathcal{G}}_0 = \symbf{L}$, 
\begin{align}
\label{equ:martingale1}
\mathbb{E}\left\{\int_{u=0}^t H(u)dM(u) \big | \symbf{L} \right\} = \mathbb{E}\left\{U(t)\mid \symbf{\mathcal{G}}_{0}\right\} = U(0) = 0    
\end{align}

Thus, by equation~\eqref{equ:martingale1}, we have 
\begin{align*}
    &\  \mathbb{E}\left\{\symbf{g}_j^*\left(\symbf{O}\right) \mid  \symbf{L}\right\} \nonumber \\
  = &\ \mathbb{E}\left[ \int_{t=t_1}^{t_2} \frac{\symbf{W}_j(t;\symbf{\beta})  \left[ Y(t) - s\left\{ \symbf{B}(t)^\prime \symbf{\beta}\right\}\right]}{\rho\left\{t \mid \overline{\symbf{O}}(t), Y(t)\right\}} \left[dN(t) - \rho\left\{t \mid \overline{\symbf{O}}(t), Y(t)
  \right\} dt\right] \bigg | \symbf{L}\right]\nonumber \\
  & \  + \mathbb{E}\left[\int_{t=t_1}^{t_2} \symbf{W}_j(t;\symbf{\beta}) \left[Y(t) - s\left\{\symbf{B}(t)^\prime \symbf{\beta}\right\}\right]dt \bigg | \symbf{L} \right] \nonumber \\
 = &  \ \symbf{0} +  \symbf{\varphi}_j(\symbf{L})
\end{align*}
as claimed. 

The inverse-weighted element $\symbf{g}_j^*(\symbf{O})$ and the augmentation space
\begin{equation*}
\mathrm{Aug} := \big\{ \symbf{g}(\symbf{O}) \in \mathcal{H}^{O} : \mathbb{E}\{\symbf{g}(\symbf{O}) \mid \symbf{L}\} = 0 \big\}
\end{equation*}
determine the space $\Lambda_1^{O,\perp}$ as
\begin{equation*}
\Lambda_1^{O,\perp} = \big\{ \symbf{g}_j^*(\symbf{O}) + \symbf{g}(\symbf{O}) : \symbf{g}(\symbf{O}) \in \mathrm{Aug} \big\}, \quad j=1,2.
\end{equation*}

We know that for $\symbf{g}(\symbf{O}) \in \mathcal{H}^O$: 
\begin{equation*}
    \symbf{g}(\symbf{O}) = \int \symbf{h}\left\{t, \overline{\symbf{O}}(t), Y(t)\right\}dN(t) +  \int \symbf{I}\left\{t, \overline{\symbf{O}}(t)\right\}\left[dN(t) - \lambda\left\{t\mid \overline{\symbf{O}}(t)\right\}dt\right]
\end{equation*}
where $\symbf{h}\left\{t, \overline{\symbf{O}}(t), Y(t)\right\}$ satisfies $\mathbb{E}\left[\int \symbf{h}\left\{t, \overline{\symbf{O}}(t), Y(t)\right\} dN(t)\right] = 0$ and $\symbf{I}\left\{t, \overline{\symbf{O}}(t)\right\}$ can be function of $t$ and $\overline{\symbf{O}}(t)$.

Thus, we need to find $\symbf{h}\left\{t, \overline{\symbf{O}}(t), Y(t)\right\}$ and $\symbf{I}\left\{t, \overline{\symbf{O}}(t)\right\}$ to make $\mathbb{E}\left\{\symbf{g}(\symbf{O})\mid \symbf{L}\right\} = 0$, which becomes 
\begin{equation}
\label{equ:aug_deriv1}
    \int \symbf{h}\left\{t, \overline{\symbf{O}}(t), Y(t)\right\}\rho\left\{t \mid \overline{\symbf{O}}(t), Y(t)\right\} dt +
    \int \symbf{I}\left\{t, \overline{\symbf{O}}(t)\right\} \left[\rho\left\{t \mid \overline{\symbf{O}}(t), Y(t)\right\}dt - \lambda\left\{t \mid \overline{\symbf{O}}(t)\right\}dt \right] = 0
\end{equation}


If we write $\symbf{I}\left\{t, \overline{\symbf{O}}(t)\right\}$ as 
\begin{equation}
\label{equ:aug_deriv2}
    \symbf{I}\left\{t, \overline{\symbf{O}}(t)\right\} = - \frac{\symbf{h}\left\{t, \overline{\symbf{O}}(t), Y(t)\right\}\rho\left\{t \mid \overline{\symbf{O}}(t), Y(t)\right\} }{\rho\left\{t \mid \overline{\symbf{O}}(t), Y(t)\right\} - \lambda\left\{t\mid \overline{\symbf{O}}(t)\right\}},
\end{equation}
then equation \eqref{equ:aug_deriv1} holds. Thus, we can rewrite equation \eqref{equ:aug_deriv2} as 
\begin{equation}
\label{equ:aug_deriv3}
    \symbf{h}\left\{t, \overline{\symbf{O}}(t), Y(t)\right\} =  -\symbf{I}\left\{t, \overline{\symbf{O}}(t)\right\} \left[1 - \frac{\lambda\left\{t \mid \overline{\symbf{O}}(t)\right\}}{\rho\left\{t \mid \overline{\symbf{O}}(t), Y(t)\right\}}\right].
\end{equation}


Thus, plugging equation \eqref{equ:aug_deriv3} into $\symbf{g}(\symbf{O})$, we have
\begin{align*}
  \symbf{g}(\symbf{O}) & = \int -\symbf{I}\left\{t, \overline{\symbf{O}}(t)\right\} \left[1 - \frac{\lambda\left\{t\mid \overline{\symbf{O}}(t)\right\}}{\rho\left\{t \mid \overline{\symbf{O}}(t), Y(t)\right\}}\right]dN(t) +  \int \symbf{I}\left\{t, \overline{\symbf{O}}(t)\right\}\left[dN(t) - \lambda\left\{t\mid \overline{\symbf{O}}
  (t)\right\}dt\right] \\   
  & =  \int  \frac{\symbf{I}\left\{t, \overline{\symbf{O}}(t)\right\}\lambda\left\{t\mid \overline{\symbf{O}}(t)\right\}}{\rho\left\{t \mid \overline{\symbf{O}}(t), Y(t)\right\}}dN(t) - \int \symbf{I}\left\{t, \overline{\symbf{O}}(t)\right\} \lambda\left\{t\mid \overline{\symbf{O}}
  (t)\right\}dt\\
  & = \int \frac{\symbf{I}\left\{t, \overline{\symbf{O}}(t)\right\}\lambda\left\{t\mid \overline{\symbf{O}}(t)\right\}}{ \rho\left\{t \mid \overline{\symbf{O}}(t), Y(t)\right\}} \left[ dN(t)- \rho\left\{t \mid \overline{\symbf{O}}(t), Y(t)\right\}dt\right]
\end{align*}

Then we can write one subset $Aug^*$ of the augmentation space $Aug$ as:
\begin{align}
\label{equ:aug_space1}
Aug^* = \left\{
\int_{t=t_1}^{t_2} \frac{\symbf{h}\left\{t, \overline{\symbf{O}}(t)\right\}}{ \rho\left\{t \mid \overline{\symbf{O}}(t), Y(t)\right\}} \left[ dN(t)- \rho\left\{t \mid \overline{\symbf{O}}(t), Y(t)\right\}dt\right] :\mbox{ any }  \underbrace{ \symbf{h} \left\{ t, \overline{\symbf{O}}(t)\right\} }_{ p \times 1} \right\} \subseteq Aug
\end{align}


\textbf{Part 3: the coarsening tangent space.}

The coarsening tangent space $\Lambda (F_{R\mid L})$ is in the set of all elements of $\mathcal{H}^Z$ that are linear combinations of scores of $\symbf{R}\mid \symbf{L}$ in a parametric submodel of our model. Following \cite{cooklawless2007book}, we can compute the likelihood for $\symbf{R} \mid \symbf{L}$ by first partitioning the interval $[t_1,t_2]$ into a finite number of subintervals $[a_k, a_k+\epsilon)$, $k=0,\cdots, K$, each of length $\epsilon$, and then taking the limit as $K \rightarrow \infty$. For fixed $\epsilon > 0 $, let $\Delta(a_k)$ denote the number of visits during the interval $[a_k, a_k+\epsilon)$; for small $\epsilon$, this number of visits is assumed to be 0 or 1. For fixed $K$, the conditional (on $\symbf{L}$) likelihood of $\symbf{R} = (N(a_0), \cdots, N(a_K))$ is the same as the conditional (on $\symbf{L}$) likelihood of $(\Delta(a_0), \cdots, \Delta(a_K))$, which can be written as:
\begin{equation*}
    \mathbb{P}\left\{ \Delta(a_0) \mid \symbf{L}\right\} \prod\limits_{k=1}^K \mathbb{P}\left\{\Delta(a_k) \mid \Delta(a_0), \cdots, \Delta(a_{k-1}), \symbf{L} \right\}.
\end{equation*}

For each $k=1,\cdots,K$, 
\begin{align*}
     \mathbb{P}\left\{\Delta(a_k) \mid \Delta(a_0), \cdots, \Delta(a_{k-1}), \symbf{L} \right\} 
     & = \rho\left\{a_k \mid \symbf{L}, \overline{\symbf{O}}(a_k)\right\} \epsilon + o(\epsilon) \\
    & = \rho\left\{a_k \mid \overline{\symbf{O}}(a_k), Y(a_k)\right\} \epsilon + o(\epsilon)
\end{align*}
where the last equality holds by Assumption~\ref{chp4:assum:nofuture}.

Similarly, $ \mathbb{P}\left\{\Delta(a_0) \mid \symbf{L} \right\} = \rho\left\{a_0 \mid \symbf{L}, \overline{\symbf{O}}(a_0)\right\} \epsilon + o(\epsilon) = \rho\left\{a_0 \mid \overline{\symbf{O}}(a_0), Y(a_0)\right\} \epsilon + o(\epsilon)$. Therefore, for fixed $K$, the conditional (on $\symbf{L}$) likelihood of $(\Delta(a_0), \cdots, \Delta(a_K))$ is 
\begin{equation*}
    \prod_{k=0}^K   \left[  \rho\left\{a_k \mid \overline{\symbf{O}}(a_k),  Y(a_k)\right\} \epsilon + o(\epsilon) \right]^{\Delta(a_k) } \left[ 1- \rho\left\{a_k\mid \overline{\symbf{O}}(a_k), Y(a_k)\right\} \epsilon + o(\epsilon)\right]^{1-\Delta(a_k) } .
\end{equation*}

After dividing by the constant $\prod\limits_{k=0}^K \epsilon ^{\Delta(a_k)}$, the conditional likelihood is proportional to:
\begin{equation}
\label{equ:like_h} \prod_{k=0}^K   \left[  \rho \left\{a_k \mid \overline{\symbf{O}}(a_k) , Y(a_k)\right\} + o(\epsilon)/\epsilon  \right]^{\Delta(a_k) } \left[ 1- \rho \left\{ a_k \mid \overline{\symbf{O}}(a_k) , Y(a_k)\right\} \epsilon + o(\epsilon) \right] ^{1-\Delta(a_k) }. 
\end{equation}

The continuous-time conditional (on $\symbf{L}$) likelihood is the limit of equation \eqref{equ:like_h} as $K\rightarrow \infty$. We compute this by viewing \eqref{equ:like_h} as a product of three factors and note that:
\begin{equation*}
     \lim_{K \to \infty} \prod_{k=0}^K   \left[  \rho \left\{ a_k \mid \overline{\symbf{O}}(a_k), Y(a_k) \right\}  + o(\epsilon)/\epsilon \right]^{\Delta(a_k) } = \prod_{t=t_1}^{t_2}  \left[ \rho \left\{ t \mid \overline{\symbf{O}}(t), Y(t) \right\}  \right]^{\Delta N(t)}.
\end{equation*} 
By a result on product integrals \citep{cooklawless2007book},

\begin{equation*}
     \lim_{K \to \infty} \prod_{k=0}^K  \left[ 1- \rho \left\{a_k \mid \overline{\symbf{O}}(a_k) , Y(a_k) \right\} \epsilon + o(\epsilon)\right] =  \exp \left[ - \int_{t=t_1}^{t_2} \rho \left\{ t \mid \overline{\symbf{O}}(t), Y(t)\right\} dt \right].
\end{equation*}

Finally, since $\Delta(a_k)=1$ at only a fixed number of values $a_k$ as $K$ varies, and since
\begin{equation*}
\lim_{\epsilon \to 0^+} \big[ 1- \rho \big\{ a_k \mid \overline{\symbf{O}}(a_k), Y(a_k) \big\} \epsilon + o(\epsilon) \big] = 1
\end{equation*}
for each $a_k$, we have
\begin{equation*}
\lim_{K \to \infty} \prod_{k=0}^K \big[ 1- \rho \big\{ a_k \mid \overline{\symbf{O}}(a_k), Y(t_k) \big\} \epsilon + o(\epsilon) \big]^{\Delta(a_k)} = 1.
\end{equation*}

Therefore, the continuous-time conditional (on $\symbf{L}$) likelihood $\mathcal{L}$ is:
\begin{equation*}
    \mathcal{L} = \left[ \prod_{t=t_1}^{t_2} \rho \left\{t \mid  \overline{\symbf{O}}(t), Y(t) \right\}^{\Delta N(t)} \right] \exp \left[  - \int_{t=t_1}^{t_2} \rho \left\{ t \mid \overline{\symbf{O}}(t), Y(t)\right\} dt \right].
\end{equation*}

To compute $\Lambda(F_{R \mid L})$, we now consider the conditional likelihood in a parametric submodel of our model for $dF(\symbf{r}\mid \symbf{L})$. Any such  parametric submodel is induced by a model for the intensity function $\rho\left\{t\mid  \overline{\symbf{O}}(t), Y(t)\right\}$, say $\left\{\rho\left\{t\mid \overline{\symbf{O}}(t), Y(t) ; \symbf{\gamma}\right\}: \symbf{\gamma} \in \Gamma \right\}$ with $\symbf{\gamma}$ a finite-dimensional parameter and $\symbf{\gamma} = \symbf{0}$ corresponding to the truth. The conditional likelihood in this parametric submodel is:
\begin{equation*}
    \mathcal{L}(\symbf{\gamma}) =  \left[ \prod_{t=t_1}^{t_2} \rho \left\{ t \mid \overline{\symbf{O}}(t), Y(t); \symbf{\gamma} \right\}^{\Delta N(t)} \right] \exp \left[  - \int_{t=t_1}^{t_2} \rho \left\{ t \mid \overline{\symbf{O}}(t), Y(t); \symbf{\gamma} \right\} dt \right].
\end{equation*}
The log-conditional likelihood is:
\begin{equation*}
    \ell( \symbf{\gamma}) = \int_{t=t_1}^{t_2}   \log \left[ \rho \left\{ t \mid \overline{\symbf{O}}(t), Y(t) ; \symbf{\gamma} \right \} \right] dN(t)  \ - \ \int_{t=t_1}^{t_2} \rho \left\{ t \mid \overline{\symbf{O}}(t), Y(t) ; \symbf{\gamma} \right \} dt,
\end{equation*}
and the score is:
\begin{align*}
 \symbf{S}_{\symbf{\gamma}}(\symbf{O}) =  & \ \int_{t=t_1}^{t_2}     \left[ \frac{   \frac{ \partial  }{ \partial \gamma} \rho \left\{ t \mid \overline{\symbf{O}}(t), Y(t); \symbf{\gamma} \right\} }{ \rho \left\{ t \mid \overline{\symbf{O}}(t), Y(t); \symbf{\gamma} \right\} } \right]dN(t)  \ - \ \int_{t=t_1}^{t_2}   \frac{ \partial }{ \partial \symbf{\gamma}} \rho \left\{ t \mid \overline{\symbf{O}}(t), Y(t) ; \symbf{\gamma} \right \} dt  \\
= & \  \int_{t=t_1}^{t_2} \left[  \frac{   \frac{ \partial }{ \partial \gamma}  \rho \left\{ t \mid \overline{\symbf{O}}(t), Y(t); \symbf{\gamma} \right \}  }{ \rho \left\{ t \mid \overline{\symbf{O}}(t), Y(t); \symbf{\gamma} \right \} } \right] \left[ dN(t) -\rho \left\{ t \mid \overline{\symbf{O}}(t), Y(t) ; \symbf{\gamma} \right\}dt \right].
\end{align*}

In order to satisfy the relation in Proposition \ref{chp4:prop:lambda_rho}, we must have
\begin{equation*}
\rho \big\{ t \mid \overline{\symbf{O}}(t), Y(t) ; \symbf{\gamma} \big\} = r \big\{ t, \overline{\symbf{O}}(t) ; \symbf{\gamma}\big\} \exp\{-\alpha Y(t) \}
\end{equation*}
for some function $r \{ t, \overline{\symbf{O}}(t) ; \symbf{\gamma} \}$. Since
\begin{equation*}
\frac{ \frac{ \partial }{ \partial \gamma} \rho \big\{ t \mid \overline{\symbf{O}}(t), Y(t) ; \symbf{\gamma} \big\} }{ \rho \big\{ t \mid \overline{\symbf{O}}(t), Y(t) ; \symbf{\gamma} \big\} } = \frac{ \frac{ \partial }{ \partial \gamma} r \big\{ t, \overline{\symbf{O}}(t) ; \symbf{\gamma} \big\} }{ r \big\{ t, \overline{\symbf{O}}(t) ; \symbf{\gamma} \big\} }
\end{equation*}
is a function of $t$ and $\overline{\symbf{O}}(t)$ only, evaluating the score at the truth $\symbf{\gamma}=\symbf{0}$ shows that any score of a parametric submodel of $\Lambda(F_{R|L})$ must be of the form
\begin{equation*}
\int_{t=t_1}^{t_2} \symbf{h} \big\{ t, \overline{\symbf{O}}(t) \big\} \big[ dN(t) - \rho \big\{ t \mid \overline{\symbf{O}}(t),Y(t) \big\} dt \big]
\end{equation*}
for some function $\symbf{h} \{ t, \overline{\symbf{O}}(t) \}$. That is, we have shown the following containment:
\begin{equation} \label{equ:coarsening_spc}
\begin{split}
\Lambda(F_{R|L}) \subseteq \bigg\{ & \int_{t=t_1}^{t_2} \symbf{H} \big\{ t, \overline{\symbf{O}}(t) \big\} \big[dN(t) - \rho \big\{ t \mid \overline{\symbf{O}}(t), Y(t) \big\}dt \big] \in \mathcal{H}^Z : \\
 & \mbox{any } \underbrace{ \symbf{H} \big\{ t, \overline{\symbf{O}}(t)\big\} }_{ p \times 1} \bigg\}.
\end{split}
\end{equation}

\textbf{Part 4: an observed-data influence function.}

We have defined the inverse-weighted element in equation \eqref{equ:gstar},

\begin{equation*}
\symbf{g}_j^*\left(\symbf{O}\right) = \int_{t=t_1}^{t_2} \frac{ \symbf{W}_j(t;\symbf{\beta}) \left[ Y(t) - s\left\{ \symbf{B}(t)^\prime \symbf{\beta}\right\}\right]}{\rho\left\{t \mid \overline{\symbf{O}}(t), Y(t)\right\}} dN(t),
\end{equation*}

Now consider the function: 
\begin{equation*}
\symbf{h}_j^*\left(\symbf{O}\right) = \int_{t=t_1}^{t_2} \frac{\symbf{W}_j(t;\symbf{\beta})  \left[\mathbb{E}\left\{Y(t)\mid \overline{\symbf{O}}(t)\right\} - s\left\{\symbf{B}(t)^\prime \symbf{\beta}\right\}\right]}{\rho\left\{t \mid \overline{\symbf{O}}(t), Y(t)\right\}} \left[dN(t) - \rho\left\{ t \mid \overline{\symbf{O}}(t), Y(t)\right\} dt\right].
\end{equation*}
Here $\symbf{h}_j^*\left(\symbf{O}\right) \in Aug^* \subseteq Aug$, and therefore $\symbf{g}_j^*\left(\symbf{O}\right) - \symbf{h}_j^*\left(\symbf{O}\right) \in \Lambda_1^{O,\perp}$. In order to make $\symbf{\varphi}_j(\symbf{O}) = \symbf{g}_j^*\left(\symbf{O}\right) - \symbf{h}_j^*\left(\symbf{O}\right)$ be an influence function of $\symbf{\beta}$, we need to show that $\symbf{g}_j^*\left(\symbf{O}\right) - \symbf{h}_j^*\left(\symbf{O}\right)$ is orthogonal to $\Lambda(F_{R\mid L})$, which means $\symbf{g}_j^*\left(\symbf{O}\right) - \symbf{h}_j^*\left(\symbf{O}\right) \in \Lambda^{O,\perp} = \Lambda_1^{O,\perp} \cap (\Lambda(F_{R\mid L}))^{\perp}$ and that it is properly normalized. To show $\symbf{g}_j^*\left(\symbf{O}\right) - \symbf{h}_j^*\left(\symbf{O}\right) \in (\Lambda(F_{R\mid L}))^{\perp}$, we need to show that $\mathbb{E}\left[\left\{ \symbf{g}_j^*(\symbf{O}) - \symbf{h}_j^*(\symbf{O})\right\} \symbf{w}(\symbf{Z})\right] = 0$
 for any $\symbf{w}(\symbf{Z}) \in \Lambda(F_{R\mid L})$. 

Let $\symbf{w}(\symbf{Z})= \int_{t=t_1}^{t_2} \symbf{H} \left\{ t, \overline{\symbf{O}}(t) \right\}  \left[dN(t) - \rho \left\{ t \mid \overline{\symbf{O}}(t), Y(t) \right\}dt \right]  =  \big( w_1(\symbf{Z}), \ldots, w_p(\symbf{Z}) \big) $ be any element of $\Lambda(F_{R|L})$, and for each $l$, write $w_l(\symbf{Z}) = \int_{t=t_1}^{t_2} 
H_l\left\{ t, \overline{\symbf{O}}(t)\right\}  \left[dN(t) - \rho \left\{ t \mid \overline{\symbf{O}}(t),Y(t) \right\} dt \right] $.  Denote the $l$-th components of $\symbf{g}_j^*(\symbf{O})$ and $\symbf{h}_j^*(\symbf{O})$ as 
\[ g_{j,l}^*(\symbf{O}) = \int_{t=t_1}^{t_2} \frac{ W_{j,l}(t;\symbf{\beta}) \left[ Y(t) -   s\left\{\symbf{B}(t)^\prime \symbf{\beta}  \right\}\right] }{ \rho \left\{ t \mid \overline{\symbf{O}}(t),Y(t) \right\} }   dN(t) \]
\[ h_{j,l}^*(\symbf{O})= \int_{t=t_1}^{t_2} \frac{ W_{j,l}(t;\symbf{\beta}) \left[\mathbb{E}\left\{Y(t)\mid \overline{\symbf{O}}(t)\right\} - s\left\{\symbf{B}(t)^\prime \symbf{\beta}\right\}\right] }{\rho \left\{ t \mid \overline{\symbf{O}}(t),Y(t) \right\}} \left[ dN(t) - \rho \left\{ t \mid \overline{\symbf{O}}(t),Y(t) \right\} dt \right] \]
where $W_{j,l}(t;\symbf{\beta})$ is the $l$-th component of $\symbf{W}_j(t;\symbf{\beta})$, $j=1,2$, $l=1,\cdots, p$.

We show orthogonality using the martingale covariance property \citep{andersen1993bookABGK}, which says that, for predictable processes $H_1(t)$ and $H_2(t)$, 
\begin{equation*}
     \mathbb{E} \left\{ \left( \int_{t=t_1}^{t_2} H_1(t) \left[ dN(t) - \rho \left\{ t \mid \overline{\symbf{O}}(t),Y(t) \right\} dt \Big] \right)  \left( \int_{t=t_1}^{t_2} H_2(t) \Big[ dN(t) - \rho \left\{ t \mid \overline{\symbf{O}}(t),Y(t) \right\} dt \right] \right) \right\}  
\end{equation*}
\begin{equation}
 \label{equ:mart_covar}
 =  \mathbb{E} \left[  \int_{t=t_1}^{t_2} H_1(t) H_2(t) \rho \left\{t \mid \overline{\symbf{O}}(t), Y(t) \right\} dt  \right]. 
\end{equation}
By \eqref{equ:mart_covar},
$\mathbb{E} \left\{ h_{j,l}^*(\symbf{O}) w_l(\symbf{Z}) \right\} 
=  \mathbb{E} \left[ \int_{t=t_1}^{t_2} W_{j,l}(t;\symbf{\beta}) H_l \left\{ t, \overline{\symbf{O}}(t) \right\}   \left[ \mathbb{E} \left\{ Y(t) \mid  \overline{\symbf{O}}(t) \right\} - s\left\{\symbf{B}(t)^\prime \symbf{\beta} \right\} \right] dt \right]$.
On the other hand, by adding and subtracting we can write:
\begin{align}
& \mathbb{E} \left\{ g_{j,l}^*(\symbf{O})w_{l}(\symbf{Z})\right\}  \nonumber \\
= & \ \mathbb{E} \left\{ \left(  \int_{t=t_1}^{t_2} \frac{ W_{j,l}(t;\symbf{\beta}) \left[Y(t) - s\left\{\symbf{B}(t)^\prime \symbf{\beta}\right\}\right]  }{ \rho \left\{ t \mid \overline{\symbf{O}}(t),Y(t) \right\} } \left[ dN(t) - \rho \left\{ t \mid \overline{\symbf{O}}(t),Y(t) \right\} dt \right] \right) w_l(\symbf{Z}) \right\} \label{equ:mean_gw} \\
&  \hspace{0.1in} + \mathbb{E} \left( \left[ \int_{t=t_1}^{t_2} W_{j,l}(t;\symbf{\beta}) \left[ Y(t)  - s\left\{\symbf{B}(t)^\prime \symbf{\beta} \right\} \right] dt \right] w_l(\symbf{Z}) \right). \nonumber 
\end{align}

By \eqref{equ:mart_covar}, the first term in equation \eqref{equ:mean_gw} equals 
\begin{align*}
    & \mathbb{E}\left[\int_{t=t_1}^{t_2}
    W_{j,l}(t;\symbf{\beta}) H_l\left\{t, \overline{\symbf{O}}(t)\right\}\left[Y(t) - s\left\{\symbf{B}(t)^\prime \symbf{\beta}\right\}\right]dt\right]\\
    & = \mathbb{E}\left[\int_{t=t_1}^{t_2}W_{j,l}(t;\symbf{\beta}) H_l\left\{t, \overline{\symbf{O}}(t)\right\}\left[\mathbb{E}\left\{Y(t) \mid \overline{\symbf{O}}(t)\right\} - s\left\{\symbf{B}(t)^\prime \symbf{\beta}\right\}\right]dt\right]\\
    & = \mathbb{E} \left\{ h_{j,l}^*(\symbf{O}) w_l(\symbf{Z}) \right\}
\end{align*}
where the first equation is by iterated expectation conditioning on $\overline{\symbf{O}}(t)$.

Using iterated expectation conditioning on $\symbf{L}$, the second term in equation \eqref{equ:mean_gw} equals:
\begin{align*}
& \mathbb{E}\left[   \left( \int_{t=t_1}^{t_2}  W_{j,l}(t;\symbf{\beta}) \left[ Y(t) - s\left\{\symbf{B}(t)^\prime \symbf{\beta}\right\} \right] dt \right) \times  \right. \\
& \hspace{0.1in} \underbrace{ \mathbb{E} \left\{  \left( \int_{t=t_1}^{t_2} H_l \left\{ t, \overline{\symbf{O}}(t) \right\}  \Big[ dN(t) - \rho \left\{ t \mid \overline{\symbf{O}}(t),Y(t) \right\}dt \Big] \right) \Big | \symbf{L}   \right\} }_{ = \ 0  \mbox{ by equation \eqref{equ:martingale1} }} \bigg] =0.
\end{align*}

Thus, $\mathbb{E}\left[\left\{ \symbf{g}_j^*(\symbf{O}) - \symbf{h}_j^*(\symbf{O})\right\} \symbf{w}(\symbf{Z})\right] = 0$. To make $\symbf{g}_j^*(\symbf{O}) - \symbf{h}_j^*(\symbf{O})$ be an influence function, we also need to show $\mathbb{E}\left[\left\{\symbf{g}_j^*(\symbf{O}) - \symbf{h}_j^*(\symbf{O})\right\} \symbf{S}_{\symbf{\beta}} (\symbf{O})^\prime \right] = \symbf{I}_{p}$, where $\symbf{S}_{\symbf{\beta}} (\symbf{O})$ is the observed-data score for $\symbf{\beta}$ and $\symbf{I}_{p}$ is the $p\times p$ identity matrix. Following Theorem 8.3 in \cite{tsiatis2006semiparametric}, we rewrite the left side using iterated expectations and the fact that $\symbf{S}_{\symbf{\beta}}(\symbf{O}) = \mathbb{E}\left\{\symbf{S}_{\symbf{\beta}}(\symbf{L}) \mid \symbf{O}\right\}$, where $\symbf{S}_{\symbf{\beta}}(\symbf{L})$ is the full-data score for $\symbf{\beta}$:
\begin{align*}
    \mathbb{E}\left[\left\{\symbf{g}_j^*(\symbf{O}) - \symbf{h}_j^*(\symbf{O})\right\} \symbf{S}_{\symbf{\beta}} (\symbf{O})^\prime \right] 
    & = \mathbb{E}\left[\left\{\symbf{g}_j^*(\symbf{O}) - \symbf{h}_j^*(\symbf{O})\right\} \mathbb{E}\left\{ \symbf{S}_{\symbf{\beta}}(\symbf{L})^\prime\mid \symbf{O}\right\}\right]\\
    & = \mathbb{E}\left[\left\{\symbf{g}_j^*(\symbf{O}) - \symbf{h}_j^*(\symbf{O})\right\} \symbf{S}_{\symbf{\beta}}(\symbf{L})^\prime\right]\\
     = & \ \mathbb{E} \big[ \underbrace{ \mathbb{E} \left\{ \symbf{g}_j^*(\symbf{O}) \mid \symbf{L} \right\} }_{\symbf{\varphi}(\symbf{L})} \symbf{S}_{\symbf{\beta}}(\symbf{L})^\prime  \big] - \mathbb{E} \big[ \underbrace{ \mathbb{E} \left\{ \symbf{h}_j^*(\symbf{O}) \mid \symbf{L} \right\} }_{= \symbf{0}} \symbf{S}_{\symbf{\beta}}(\symbf{L})^\prime  \big] \\
     = & \  \mathbb{E} \left[ \symbf{\varphi}(\symbf{L}) \symbf{S}_{\symbf{\beta}}(\symbf{L})^\prime  \right]  \\
     & = \symbf{I}_p
\end{align*}
where the last equality holds since $\symbf{\varphi}(\symbf{L})$ is a full-data influence function. 

Denoting $\symbf{g}_j^*(\symbf{O}) - \symbf{h}_j^*(\symbf{O})$ now by $\symbf{\varphi}_j(\symbf{O})$ as 
\begin{align*}
    \symbf{\varphi}_j(\symbf{O}) 
    = & \int_{t=t_1}^{t_2} \frac{
    \symbf{W}_j(t;\symbf{\beta}) \left[ Y(t) - \mathbb{E}\left\{Y(t)\mid \overline{\symbf{O}}(t)\right\}\right]}{\rho\left\{t \mid \overline{\symbf{O}}(t), Y(t)\right\}} dN(t) \\
    & \quad + \int_{t=t_1}^{t_2} \symbf{W}_j(t;\symbf{\beta}) \left[ \mathbb{E}\left\{Y(t)\mid \overline{\symbf{O}}(t)\right\} - s\left\{ \symbf{B}(t)^\prime \symbf{\beta}\right\}\right]dt 
\end{align*}
where 
\begin{align*}
    \symbf{W}_1(t;\symbf{\beta}) & = \symbf{V}_1^{-1}\symbf{B}(t)\frac{\partial g(z)}{\partial z} \Big|_{z = s\left\{\symbf{B}(t)^\prime \symbf{\beta}\right\}}, \quad \symbf{V}_1 = \int_{t=t_1}^{t_2} \symbf{B}(t) \symbf{B}(t)^\prime dt;
    \\
    \symbf{W}_2(t;\symbf{\beta}) & = \symbf{V}_2\left(\symbf{\beta}\right)^{-1}\symbf{B}(t), \quad \symbf{V}_2(\symbf{\beta}) = \int_{t=t_1}^{t_2} \symbf{B}(t)\symbf{B}(t)^\prime \frac{\partial s(z)}{\partial z} \Bigg|_{z = \symbf{B}(t)^\prime \symbf{\beta}}dt.
\end{align*}
\end{proof}

\subsection{Proof of the Asymptotic Results}
\label{apxB:app:asy_proof}

In the following part, we use $\mathbb{P}_{n}f$ to denote $n^{-1}\sum_{i=1}^n f(X_i)$, $\mathbb{P}f$ to denote $\mathbb{E}\{f(X)\} = \int f dP$,  and $\mathbb{G}_n f$ to denote $ n^{-1/2}(\mathbb{P}_n - \mathbb{P})f =  n^{-1/2} \sum_{i=1}^n \{f(X_i) - \mathbb{P}f\}$. In addition, we use $\symbf{\eta}$ to denote all parameters of the models for $\lambda\{t \mid \overline{\symbf{O}}(t)\}$ and $dF\{ y(t) \mid \Delta N(t) = 1, \overline{\symbf{O}}(t)\}$ in the influence function. Then we write the influence function for $\symbf{\beta}$ as $\symbf{\varphi}_j \left(\symbf{O}; \symbf{\eta}, \symbf{\beta}\right)$. We will assume the parameter space $\mathcal{B}^* \subset \mathbb{R}^p$ is an open set. For p-dimensional vectors $\symbf{x}, \symbf{y} \in \mathbb{R}^p$, we define the L2 norm as $||\symbf{x}||_2 = \sqrt{\symbf{x}^\prime \symbf{x}}$, and the distance as $d(\symbf{x},\symbf{y}) = ||\symbf{x}-\symbf{y}||_2$. For a matrix $A = (a_{ij}) \in \mathbb{R}^{p_1\times p_2}$, we define the matrix norm as $||A||_2 = ||vec(A)||_2 = \sqrt{\sum_{i=1}^{p_1}\sum_{j=1}^{p_2} a_{ij}^2}$, where $a_{ij}$ denotes the $(i,j)$-th element of $A$ and $vec(\cdot)$ means stacking all elements of the matrix $A$ column-wise into a single vector.

We use $\symbf{\eta}_0$ and $\symbf{\beta}_0 \in \mathcal{B}^*$ to denote the true values of these parameters, and use $\widehat{\symbf{\eta}}$ and $\widehat{\symbf{\beta}}\in \mathcal{B}^*$ to denote the estimated values of these parameters, where
 \begin{align*}
     & \mathbb{P}_n \symbf{\varphi}_j \big( \symbf{O} ; \widehat{\symbf{\eta}}, \widehat{\symbf{\beta}} \big) = 0, \\
     & \mathbb{P} \symbf{\varphi}_j \left( \symbf{O} ; \symbf{\eta}_0, \symbf{\beta}_0\right) = 0.
 \end{align*}

Suppose the intensity function $\lambda \left\{ t \mid \overline{\symbf{O}}(t) \right\}$ is modeled using a stratified Andersen-Gill model as described in Section \ref{chp4:subsection:estimation}. The partial likelihood estimator $\widehat{\symbf{\gamma}}$ and the Breslow estimator $\widehat{\Lambda}_{0,k}(t)$ are used to estimate $\symbf{\gamma}$ and the cumulative baseline intensity functions $\Lambda_{0,k}(t) = \int_{s=0}^t \lambda_{0,k}(s) ds$ for each assessment $k$. Then kernel smoothing of $\widehat{\Lambda}_{0,k}(t)$ is used to estimate the baseline intensity $\lambda_{0,k}(t)$ as $\widehat{\lambda}_{0,k}(t) = \frac{1}{h} \sum_{k} K\left(\frac{t-T_j}{h}\right)d\widehat{\Lambda}_{0,k} (T_j)$, where $h$ is the bandwidth, $K(\cdot)$ is a choice of kernel, $T_j$ is a time at which one or more participants had their $k$th assessment, and $d\widehat{\Lambda}_{0,k}(T_j)$ is the size of the jump in $\widehat{\Lambda}_{0,k}(t)$ at $T_j$. For the model of the conditional outcome distribution $dF\{ y(t) | \Delta N(t) = 1, \overline{\symbf{O}}(t)\}$, we impose a certain rate condition below, and we will show that the single-index model satisfies this rate. 

In the following part, we denote these estimated nuisance functions with the estimated nuisance parameter $\widehat{\symbf{\eta}}$ as $\widehat{\lambda} \left\{ t \mid \overline{\symbf{O}}(t) \right\}$, 
$\widehat{\rho} \left\{t \mid \overline{\symbf{O}}(t), Y(t)\right\}$,
$\widehat{\mathbb{E}}\{Y(t)\mid \overline{\symbf{O}}(t)\}$,  
$\widehat{\mathbb{E}} \big[ Y(t) \exp \{ \alpha Y(t) \}  \mid \Delta N(t) = 1, \overline{\symbf{O}}(t) \big]$, 
$\widehat{\mathbb{E}} \big[ \exp \{ \alpha Y(t) \}  \mid \Delta N(t) = 1, \overline{\symbf{O}}(t) \big]$; and with the true nuisance parameter $\symbf{\eta}_0$ as $\lambda \left\{ t \mid \overline{\symbf{O}}(t) \right\}$, $\rho\left\{t \mid \overline{\symbf{O}}(t), Y(t)\right\}$,
$\mathbb{E} \{Y(t)\mid \overline{\symbf{O}}(t)\}$,  
$\mathbb{E} \big[ Y(t) \exp \{ \alpha Y(t) \}  \mid \Delta N(t) = 1, \overline{\symbf{O}}(t) \big]$, $\mathbb{E} \big[ \exp \{ \alpha Y(t) \}  \mid \Delta N(t) = 1, \overline{\symbf{O}}(t) \big]$.

\subsubsection{Additional Assumptions and Lemmas}
\label{app:assum_lemma}

Here we introduce some additional assumptions for the influence function $\symbf{\varphi}_j \left(\symbf{O}; \symbf{\eta}, \symbf{\beta}\right)$, $j=1,2$.

\begin{assumption}
    \label{assum:bound_Y}
$Y(t)$ is bounded for $t\in[t_1,t_2]$.
\end{assumption}


\begin{assumption}
\label{assum:donsker1}
   The class of functions $\left\{ \symbf{\varphi}_j \left(\symbf{O} ;  
   \symbf{\eta}, \symbf{\beta}\right):
  \symbf{\beta} \in \mathcal{B}^* \right\}$ 
  is P-Donsker, j=1,2.
\end{assumption}

\begin{assumption}
\label{assum:donsker2} 
   The class of functions $\left\{
   \frac{\partial \varphi_j(\symbf{O}; \symbf{\eta}, \symbf{\beta})}{\partial \symbf{\beta}}:
  \symbf{\beta} \in \mathcal{B}^* \right\}$ 
  is P-Donsker, j=1,2.
\end{assumption}


Here we introduce some additional assumptions for estimation of nuisance functions in the influence function $\symbf{\varphi}_j \left(\symbf{O}; \symbf{\eta}, \symbf{\beta}\right)$.

\begin{assumption}
    \label{assum: corr_spec}
    The models for $\lambda\{t \mid \overline{\symbf{O}}(t)\}$ and $dF\{ y(t) \mid \Delta N(t) = 1, \overline{\symbf{O}}(t)\}$ are correctly specified.
\end{assumption}

\begin{assumption}
    \label{assum: lambda_rate}
     The model for $\lambda\{t \mid \overline{\symbf{O}}(t) \}$ is a (possibly stratified) Andersen-Gill model, $\gamma$ and $\lambda_{0,k}(t)$, for each visit number $k$, are estimated as described in Section \ref{chp4:subsection:estimation}, and the bandwidth $h=h_n$ used for kernel smoothing is chosen using a procedure such that $\lim_{n \to \infty} nh_n^5=d < \infty$.
\end{assumption}

\begin{assumption}
    \label{assum:Y_rate}
     The model for $dF\{ y(t) \mid \Delta N(t) = 1, \overline{\symbf{O}}(t)\} $ and estimators of $\mathbb{E}[\exp \{ \alpha Y(t) \} \mid \Delta N(t) = 1, \overline{\symbf{O}}(t) ]$ and $\mathbb{E}[Y(t) \exp \{ \alpha Y(t) \} \mid \Delta N(t) = 1, \overline{\symbf{O}}(t) ]$ are such that: 
     \[ \mathbb{E} \left[ \int_{t=t_1}^{t_2} \Big\{ \widehat{\mathbb{E}}[ \exp\{\alpha Y(t)\} \mid \Delta N(t) = 1, \overline{\symbf{O}}(t)]  
     - 
     \mathbb{E}[\exp \{ \alpha Y(t) \} \mid \Delta N(t) = 1, \overline{\symbf{O}}(t) ] \Big\}^2 dt \right] \]
     and 
     \[ \mathbb{E} \left[ \int_{t=t_1}^{t_2} \Big\{ \widehat{\mathbb{E}}[ Y(t) \exp\{\alpha Y(t)\} \mid \Delta N(t) = 1, \overline{\symbf{O}}(t)] 
     - 
     \mathbb{E}[ Y(t)\exp \{ \alpha Y(t) \} \mid \Delta N(t) = 1, \overline{\symbf{O}}(t) ] \Big\}^2 dt \right]\]
are each $o_P(n^{-1/2})$.
\end{assumption}

If the models for $\lambda\{t\mid \overline{\symbf{O}}(t)\}$ and $dF\{Y(t) = y(t) \mid \Delta N(t) = 1, \overline{\symbf{O}}(t)\}$ are both correctly specified, and if Assumptions \ref{chp4:assum:tilting}--\ref{chp4:assum:nofuture} and additional regularity assumptions hold, then $\widehat{\symbf{\beta}}$ derived in Section \ref{chp4:subsection:estimation} is a consistent estimator of the true value $\symbf{\beta}$ and is asymptotically normal.

\begin{theorem}[Consistency]
\label{theorem:consist}
Under Assumptions \ref{chp4:assum:tilting}--\ref{assum:donsker1}, \ref{assum: corr_spec}--\ref{assum:Y_rate}, $\widehat{\symbf{\beta}}_j \overset{p}{\rightarrow} \symbf{\beta}$, $j=1,2$.
\end{theorem}

\begin{theorem}[Asymptotic normality]
Under Assumptions \ref{chp4:assum:tilting}--\ref{assum:Y_rate},
\label{theorem:asy_normal}
    \begin{align*} \sqrt{n}\big(\widehat{\symbf{\beta}}_j - \symbf{\beta}_0\big)  
    \overset{d}{\rightarrow}
    \mathcal{N}(\symbf{0},  \Sigma_j^*),
\end{align*}
where 
\begin{align*}
     \Sigma_j^* = 
     \left[\mathbb{P} \frac{\partial \symbf{\varphi}_j \big( \symbf{O} ; \symbf{\eta}_0, \symbf{\beta}_0 \big)}{\partial \symbf{\beta}}\right]^{-1}
     \mathbb{E} \left\{\symbf{\varphi}_j \left( \symbf{O}; \symbf{\eta}_0, \symbf{\beta}_0\right)  \symbf{\varphi}_j \left( \symbf{O}; \symbf{\eta}_0, \symbf{\beta}_0\right) ^\prime \right\} \left[\mathbb{P} \frac{\partial \symbf{\varphi}_j \big( \symbf{O} ; \symbf{\eta}_0, \symbf{\beta}_0 \big)}{\partial \symbf{\beta}}\right]^{-1, \prime},
\end{align*}
and 
\begin{align*}
    \frac{\partial \symbf{\varphi}_j \big( \symbf{O} ; \symbf{\eta}_0, \symbf{\beta}_0 \big)}{\partial \symbf{\beta}} = &  \ 
    \int_{t=t_1}^{t_2} \frac{\partial  \symbf{W}_j(t;\symbf{\beta}_0)}{\partial \symbf{\beta}}\frac{
    \left[ Y(t) - \mathbb{E}\left\{Y(t)\mid \overline{\symbf{O}}(t)\right\}\right]}{\rho\left\{t \mid \overline{\symbf{O}}(t), Y(t)\right\}} dN(t) \\
    & \quad + \int_{t=t_1}^{t_2} \frac{\partial  \symbf{W}_j(t;\symbf{\beta}_0)}{\partial \symbf{\beta}}\left[ \mathbb{E}\left\{Y(t)\mid \overline{\symbf{O}}(t)\right\} -  s\left\{ \symbf{B}(t)^\prime \symbf{\beta}_0\right\}\right]dt \\
    & \quad -
    \int_{t=t_1}^{t_2} \symbf{W}_j(t;\symbf{\beta}_0)\frac{\partial s\left\{ \symbf{B}(t)^\prime \symbf{\beta}_0\right\}}{\partial \symbf{\beta}}dt,\ j=1,2.
\end{align*}
\end{theorem}

Before we provide the proof for these two theorems, we first provide some lemmas. The proofs of Theorems \ref{theorem:consist} and \ref{theorem:asy_normal} are provided in Appendices \ref{app:proof_consist} and \ref{app:proof_asy}.


\begin{lemma}
    \label{lemma:unique}
    Suppose that Assumptions \ref{chp4:assum:tilting}, \ref{chp4:assum:positivity}, \ref{chp4:assum:marginal}, \ref{chp4:assum:nofuture} hold,
    $\symbf{\beta}_0 \in \mathcal{B}^*$ is a locally unique solution of the equation $\mathbb{P}\varphi_j(\symbf{O}; \symbf{\eta}_0, \symbf{\beta}) = 0$ for $j=1,2$.
\end{lemma}

\begin{proof}
\begin{align*}
    \mathbb{P}\varphi_j(\symbf{O}; \symbf{\eta}_0, \symbf{\beta})
    = & \ \mathbb{E}\left[\int_{t=t_1}^{t_2} \frac{
    \symbf{W}_j(t;\symbf{\beta}) \left[ Y(t) - \mathbb{E}\left\{Y(t)\mid \overline{\symbf{O}}(t)\right\}\right]}{\rho\left\{t \mid \overline{\symbf{O}}(t), Y(t)\right\}} dN(t)\right.  \\
    & \left. \quad + \int_{t=t_1}^{t_2} \symbf{W}_j(t;\symbf{\beta}) \left[ \mathbb{E}\left\{Y(t)\mid \overline{\symbf{O}}(t)\right\} -  s\left\{ \symbf{B}(t)^\prime \symbf{\beta}\right\}\right]dt \right] \\
    = & \  \mathbb{E}\left[\int_{t=t_1}^{t_2} 
    \symbf{W}_j(t;\symbf{\beta}) \left[ Y(t) -  s\left\{ \symbf{B}(t)^\prime \symbf{\beta}\right\}\right] dt\right] \\
    = & \  \int_{t=t_1}^{t_2} \symbf{W}_j(t;\symbf{\beta}) \left[ s\left\{ \symbf{B}(t)^\prime \symbf{\beta}_0\right\} -  s\left\{ \symbf{B}(t)^\prime \symbf{\beta}\right\}\right] dt. 
\end{align*}

We know that $\mathbb{P}\varphi_j(\symbf{O}; \symbf{\eta}_0, \symbf{\beta}_0) = 0$, which means $\symbf{\beta}_0$ is a solution for $\mathbb{P}\varphi_j(\symbf{O}; \symbf{\eta}_0, \symbf{\beta}) = 0$, $j=1,2$.  Define $g_j(\symbf{\beta}) = \mathbb{P}\varphi_j(\symbf{O}; \symbf{\eta}_0, \symbf{\beta})^\prime \mathbb{P}\varphi_j(\symbf{O}; \symbf{\eta}_0, \symbf{\beta})$, then
$g_j(\symbf{\beta}) \geq 0$, $j=1,2$. Thus, $g_j(\symbf{\beta}_0) = 0$ means $\symbf{\beta}_0$ is a minimizer of $g_j(\symbf{\beta})$.

If we can show that $\symbf{\beta}_0$ is a locally unique minimizer of $g_j(\symbf{\beta})$, we can show that $\symbf{\beta}_0$ is also a locally unique solution of $\mathbb{P}\varphi_j(\symbf{O}; \symbf{\eta}_0, \symbf{\beta}) = 0$. In the following, we verify that the Hessian matrix for $g_j(\symbf{\beta})$ at $\symbf{\beta}_0$ is positive definite, $j=1,2$.

Define $\mathbb{P}^l\varphi_j(\symbf{O}; \symbf{\eta}_0, \symbf{\beta})$ as the $l$-th component of $\mathbb{P}\varphi_j(\symbf{O}; \symbf{\eta}_0, \symbf{\beta})$ and $\symbf{\beta}^l$ as the $l$-th element of $\symbf{\beta}$, for $l=1,\cdots,p$. Define $J_j(\symbf{\beta}) = \frac{\partial \mathbb{P}\varphi_j(\symbf{O}; \symbf{\eta}_0, \symbf{\beta})}{\partial \symbf{\beta}^\prime}$ as
\begin{align*}
    J_j(\symbf{\beta}) 
    & =\ 
    \frac{\partial \mathbb{P}\varphi_j(\symbf{O}; \symbf{\eta}_0, \symbf{\beta})}{\partial \symbf{\beta}^\prime} \\
    & = \
    \frac{\partial \int_{t=t_1}^{t_2} 
    \symbf{W}_j(t;\symbf{\beta}) \left[ s\left\{ \symbf{B}(t)^\prime \symbf{\beta}_0\right\} -  s\left\{ \symbf{B}(t)^\prime \symbf{\beta}\right\}\right] dt }{\partial \symbf{\beta}^\prime} \\
    & = \
    \int_{t=t_1}^{t_2} \frac{\partial \symbf{W}_j(t;\symbf{\beta}) }{\partial \symbf{\beta}^\prime}\left[ s\left\{ \symbf{B}(t)^\prime \symbf{\beta}_0\right\} -  s\left\{ \symbf{B}(t)^\prime \symbf{\beta}\right\}\right] dt 
    - \int_{t=t_1}^{t_2}  \symbf{W}_j(t;\symbf{\beta}) 
    \frac{\partial s\left\{ \symbf{B}(t)^\prime \symbf{\beta}\right\}}{\partial \symbf{\beta}^\prime} dt, 
\end{align*}
where the third equality is by the continuity of the $s(\cdot)$ function, and the $(l_1,l_2)$-th element of $J_j(\symbf{\beta})$ is $\frac{\partial \mathbb{P}^{l_1}\varphi_j(\symbf{O}; \symbf{\eta}_0, \symbf{\beta})}{\partial \symbf{\beta}^{l_2}}$, for $l_1,l_2=1,\cdots, p$ and $j=1,2$. Additionally, define the Hessian matrix for the $l$-th component of $\mathbb{P}\varphi_j(\symbf{O}; \symbf{\eta}_0, \symbf{\beta})$ as $H_j^l(\symbf{\beta}) = \frac{\partial^2 \mathbb{P}^l\varphi_j(\symbf{O}; \symbf{\eta}_0, \symbf{\beta})}{\partial \symbf{\beta} \partial \symbf{\beta}^\prime }$, for $l=1,\cdots, p$ and $j=1,2$.

The first and second derivatives of $g(\symbf{\beta})$ with respect to $\symbf{\beta}$ are 
\begin{align*}
    \frac{\partial g_j(\symbf{\beta})}{\partial \symbf{\beta}} 
    = & \ 
    2 J_j(\symbf{\beta})^\prime \mathbb{P}\varphi_j(\symbf{O}; \symbf{\eta}_0, \symbf{\beta}),
    \\
    \frac{\partial^2 g_j(\symbf{\beta})}{\partial \symbf{\beta} \partial \symbf{\beta}^\prime }
    = & \ 
    2 J_j(\symbf{\beta})^\prime J_j(\symbf{\beta}) + 2 \sum_{l=1}^p \mathbb{P}^l\varphi_j(\symbf{O}; \symbf{\eta}_0, \symbf{\beta}) H_j^l(\symbf{\beta}).
\end{align*}

Set $\symbf{\beta} = \symbf{\beta}_0$ and using $\mathbb{P}^l\varphi_j(\symbf{O}; \symbf{\eta}_0, \symbf{\beta}_0) = 0$ for $l=1,\cdots, p$, we have 
\begin{align*}
   \frac{\partial^2 g_j(\symbf{\beta})}{\partial \symbf{\beta} \partial \symbf{\beta}^\prime }
    = \ 
    2 J_j(\symbf{\beta})^\prime J_j(\symbf{\beta}), \quad \quad J_j(\symbf{\beta}_0) = - \int_{t=t_1}^{t_2}  \symbf{W}_j(t;\symbf{\beta}_0)
   \frac{\partial s\left\{ \symbf{B}(t)^\prime \symbf{\beta}\right\}}{\partial \symbf{\beta}^\prime}\Big|_{\symbf{\beta} = \symbf{\beta}_0} 
    dt.
\end{align*}
where 
\begin{align*}
    \frac{\partial s\left\{ \symbf{B}(t)^\prime \symbf{\beta}\right\}}{\partial \symbf{\beta}^\prime}\Big|_{\symbf{\beta} = \symbf{\beta}_0} =  \frac{\partial s(z)}{\partial z} \Big|_{z = \symbf{B}(t)^\prime \symbf{\beta}_0} \symbf{B}(t)^\prime.
\end{align*}

\begin{align*}
    \frac{\partial^2 g(\symbf{\beta})}{\partial \symbf{\beta} \partial \symbf{\beta}^\prime }\Big|_{\symbf{\beta} = \symbf{\beta}_0}
    = & \ 
    2 J_j(\symbf{\beta}_0)^\prime J_j(\symbf{\beta}_0) \\
    = & \ 
    2 \left[\int_{t=t_1}^{t_2}  \symbf{W}_j(t;\symbf{\beta}_0) 
    \frac{\partial s\left\{ \symbf{B}(t)^\prime \symbf{\beta}\right\}}{\partial \symbf{\beta}^\prime}\Big|_{\symbf{\beta} = \symbf{\beta}_0} dt\right]^\prime \int_{t=t_1}^{t_2} \symbf{W}_j(t;\symbf{\beta}_0) 
    \frac{\partial s\left\{ \symbf{B}(t)^\prime \symbf{\beta}\right\}}{\partial \symbf{\beta}^\prime}\Big|_{\symbf{\beta} = \symbf{\beta}_0} dt.
\end{align*}

For $\symbf{W}_1(t;\symbf{\beta}_0)$, since $s^{-1}$ is the inverse function of $s$, we have 
$$\frac{\partial s^{-1}(z)}{\partial z} \Big|_{z = s\left\{\symbf{B}(t)^\prime \symbf{\beta}_0\right\}} = \left\{\frac{\partial s(z)}{\partial z} \Big|_{z = \symbf{B}(t)^\prime \symbf{\beta}_0}
\right\}^{-1}, $$
and therefore
\begin{align*}
    \symbf{W}_1(t;\symbf{\beta}_0) 
    \frac{\partial s\left\{ \symbf{B}(t)^\prime \symbf{\beta}\right\}}{\partial \symbf{\beta}^\prime}\Big|_{\symbf{\beta} = \symbf{\beta}_0}
    & = \symbf{V}_1^{-1}\symbf{B}(t) \symbf{B}(t)^\prime.
\end{align*}
Integrating over $[t_1,t_2]$ gives $J_1(\symbf{\beta}_0) = -\symbf{V}_1^{-1}\symbf{V}_1 = - \symbf{I}_p$.

For $\symbf{W}_2(t;\symbf{\beta}_0)$, we have
\begin{align*}
    & \symbf{W}_2(t;\symbf{\beta}_0) 
    \frac{\partial s\left\{ \symbf{B}(t)^\prime \symbf{\beta}\right\}}{\partial \symbf{\beta}^\prime}\Big|_{\symbf{\beta} = \symbf{\beta}_0}
    = \symbf{V}_2^{-1}(\symbf{\beta}_0)\symbf{B}(t) \symbf{B}(t)^\prime \frac{\partial s(z)}{\partial z} \Big|_{z = \symbf{B}(t)^\prime \symbf{\beta}_0},
\end{align*}
and integrating over $[t_1,t_2]$ gives $J_2(\symbf{\beta}_0) = -\symbf{V}_2(\symbf{\beta}_0)^{-1}\symbf{V}_2(\symbf{\beta}_0) = -\symbf{I}_p$. Here $\symbf{V}_2(\symbf{\beta}_0)$ is invertible because $s$ is strictly increasing and $\int_{t_1}^{t_2} \symbf{B}(t) \symbf{B}(t)^\prime dt$ is positive definite.

Hence, for $j=1,2$, $\frac{\partial^2 g_j(\symbf{\beta})}{\partial \symbf{\beta} \partial \symbf{\beta}^\prime }\big|_{\symbf{\beta} = \symbf{\beta}_0} = 2\symbf{I}_p > 0$, which is positive definite. Therefore, $\symbf{\beta}_0$ is a locally unique minimizer for $g_j(\symbf{\beta})$ and a locally unique solution for $\mathbb{P}\varphi_j(\symbf{O}; \symbf{\eta}_0, \symbf{\beta}) = 0$.
\end{proof}


\begin{lemma}
  \label{lemma:convg1}
 Suppose that Assumptions \ref{chp4:assum:tilting}, \ref{chp4:assum:positivity}, \ref{chp4:assum:marginal}, \ref{chp4:assum:nofuture}, and \ref{assum: corr_spec} - \ref{assum:Y_rate} hold, then each element of $\mathbb{P}\symbf{\varphi}_j \big(\symbf{O} ; \widehat{\symbf{\eta}}, \symbf{\beta}\big) - \mathbb{P}\symbf{\varphi}_j \big(\symbf{O} ; \symbf{\eta}_0, \symbf{\beta}\big)$ is $o_P(n^{-1/2})$, 
 for $\symbf{\beta} \in \mathcal{B}^*$,
 $j=1,2$. 
\end{lemma}

\begin{proof}
\label{proof:lemma_convg1}
\begin{align*}
  & \ \mathbb{P}\symbf{\varphi}_j \big( \symbf{O} ; \widehat{\symbf{\eta}}, \symbf{\beta} \big) 
  - \mathbb{P}\symbf{\varphi}_j \big( \symbf{O} ; \symbf{\eta}_0, \symbf{\beta} \big) \nonumber \\
   = & \ \mathbb{E}\left[\int_{t=t_1}^{t_2} \frac{\symbf{W}_j(t;\symbf{\beta})\left[ Y(t) - \widehat{\mathbb{E}}\left\{Y(t)\mid\overline{\symbf{O}}(t)\right\}\right]}{\widehat{\rho}\left\{t \mid \overline{\symbf{O}}(t), Y(t)\right\}} dN(t) + 
  \int_{t=t_1}^{t_2}  \symbf{W}_j(t;\symbf{\beta})
    \widehat{\mathbb{E}}\left\{Y(t)\mid\overline{\symbf{O}}(t)\right\}dt \right] \nonumber \\
    & \quad  - \mathbb{E} \left[\int_{t=t_1}^{t_2} \frac{\symbf{W}_j(t;\symbf{\beta})\left[ Y(t) - \mathbb{E}\left\{Y(t)\mid \overline{\symbf{O}}(t)\right\}\right]}{\rho\left\{t \mid \overline{\symbf{O}}(t), Y(t)\right\}} dN(t) + 
    \int_{t=t_1}^{t_2} \symbf{W}_j(t;\symbf{\beta}) \mathbb{E}\left\{Y(t)\mid\overline{\symbf{O}}(t)\right\}dt\right] \\
   = & \ \underbrace{\mathbb{E}\left[\int_{t=t_1}^{t_2} \frac{\symbf{W}_j(t;\symbf{\beta})\left[ Y(t) - \widehat{\mathbb{E}}\left\{Y(t)\mid \overline{\symbf{O}}(t)\right\}\right]}{\widehat{\rho}\left\{t \mid \overline{\symbf{O}}(t), Y(t)\right\}} \big[dN(t) - \rho\left\{t \mid \overline{\symbf{O}}(t), Y(t)\right\} dt \big]\right]}_{=0} \nonumber \\
    & \quad + \mathbb{E}\left[\int_{t=t_1}^{t_2} \frac{\symbf{W}_j(t;\symbf{\beta})\left[ Y(t) - \widehat{\mathbb{E}}\left\{Y(t)\mid\overline{\symbf{O}}(t)\right\}\right]}{\widehat{\rho}\left\{t \mid \overline{\symbf{O}}(t), Y(t)\right\}}\rho\left\{t \mid \overline{\symbf{O}}(t), Y(t)\right\} dt\right] \\
    & \quad 
    + \mathbb{E}\left[ \int_{t=t_1}^{t_2} \symbf{W}_j(t;\symbf{\beta})
    \widehat{\mathbb{E}}\left\{Y(t)\mid\overline{\symbf{O}}(t)\right\}dt \right] \nonumber \\
     & \quad 
     - \underbrace{\mathbb{E} \left[\int_{t=t_1}^{t_2} \frac{\symbf{W}_j(t;\symbf{\beta})\left[ Y(t) - \mathbb{E}\left\{Y(t)\mid \overline{\symbf{O}}(t)\right\}\right]}{\rho\left\{t \mid \overline{\symbf{O}}(t), Y(t)\right\}} \big[dN(t) - \rho\left\{t \mid \overline{\symbf{O}}(t), Y(t)\right\} dt \big]\right]}_{=0} \nonumber \\
     & \quad - \mathbb{E} \left[\int_{t=t_1}^{t_2} \symbf{W}_j(t;\symbf{\beta})\left[ Y(t) - \mathbb{E}\left\{Y(t)\mid \overline{\symbf{O}}(t)\right\}\right] dt\right] - \mathbb{E}\left[\int_{t=t_1}^{t_2}\symbf{W}_j(t;\symbf{\beta}) \mathbb{E}\left\{Y(t)\mid \overline{\symbf{O}}(t)\right\}dt\right]. 
\end{align*}

Then 
\begin{align}
  & \ \mathbb{P}\symbf{\varphi}_j \big( \symbf{O} ; \widehat{\symbf{\eta}}, \symbf{\beta} \big) 
  - \mathbb{P}\symbf{\varphi}_j \big( \symbf{O} ; \symbf{\eta}_0, \symbf{\beta} \big) \nonumber \\
    = & \ \mathbb{E}\left[\int_{t=t_1}^{t_2} \frac{\symbf{W}_j(t;\symbf{\beta})\left[ Y(t) - \widehat{\mathbb{E}}\left\{Y(t) \mid \overline{\symbf{O}}(t)\right\}\right]}{\widehat{\rho}\left\{t \mid \overline{\symbf{O}}(t), Y(t)\right\}}\rho\left\{t \mid \overline{\symbf{O}}(t), Y(t)\right\} dt\right] \\
    & \ + 
    \mathbb{E}\left[ \int_{t=t_1}^{t_2} \symbf{W}_j(t;\symbf{\beta})
    \widehat{\mathbb{E}}\left\{Y(t)\mid \overline{\symbf{O}}(t)\right\}dt \right] \nonumber \\
     & \quad 
     - \mathbb{E} \left[\int_{t=t_1}^{t_2} \symbf{W}_j(t;\symbf{\beta})\left[ Y(t) - \mathbb{E}\left\{Y(t)\mid \overline{\symbf{O}}(t)\right\}\right] dt\right] - \mathbb{E}\left[\int_{t=t_1}^{t_2}\symbf{W}_j(t;\symbf{\beta}) \mathbb{E}\left\{Y(t)\mid \overline{\symbf{O}}(t)\right\}dt\right] \nonumber \\
    = & \ \mathbb{E}\left[\int_{t=t_1}^{t_2} \frac{\symbf{W}_j(t;\symbf{\beta})\left[ Y(t) - \widehat{\mathbb{E}}\left\{Y(t) \mid \overline{\symbf{O}}(t)\right\}\right]}{\widehat{\rho}\left\{t \mid \overline{\symbf{O}}(t), Y(t)\right\}}\rho\left\{t \mid \overline{\symbf{O}}(t), Y(t)\right\} dt\right]
    \nonumber \\
    & \quad 
    + \mathbb{E}\left[ \int_{t=t_1}^{t_2} \symbf{W}_j(t;\symbf{\beta})
    \widehat{\mathbb{E}}\left\{Y(t)\mid \overline{\symbf{O}}(t)\right\}dt \right] - \mathbb{E}\left\{\int_{t=t_1}^{t_2}\symbf{W}_j(t;\symbf{\beta}) Y(t)dt\right\}
    \nonumber \\
    = & \ \mathbb{E}\left[\int_{t=t_1}^{t_2} \frac{\symbf{W}_j(t;\symbf{\beta})\rho\left\{t \mid \overline{\symbf{O}}(t), Y(t)\right\}}{\widehat{\rho}\left\{t \mid \overline{\symbf{O}}(t), Y(t)\right\}}
    \left[ \mathbb{E}\{Y(t)\mid \overline{\symbf{O}}(t)\} - \widehat{\mathbb{E}}\left\{Y(t)\mid \overline{\symbf{O}}(t)\right\}\right] 
     dt\right] \nonumber \\
     & \quad - \mathbb{E}\left[ \int_{t=t_1}^{t_2} \symbf{W}_j(t;\symbf{\beta})
     \left[\mathbb{E}\left\{Y(t)\mid\overline{\symbf{O}}(t)\right\} - \widehat{\mathbb{E}}\left\{Y(t)\mid \overline{\symbf{O}}(t)\right\}\right]dt \right],
     \label{equ:P_varphi}
\end{align}

\noindent where in the last equality we used the fact that, by Proposition \ref{chp4:prop:lambda_rho}:
\[ \frac{ \rho \left\{ t \mid \overline{\symbf{O}}(t),Y(t) \right\} }{ \widehat{\rho} \{ t \mid \overline{\symbf{O}}(t), Y(t) \}} 
= 
\frac{ \lambda \left\{ t \mid \overline{\symbf{O}}(t) \right\} \mathbb{E}[\exp\{ \alpha Y(t) \mid \Delta N(t)=1, \overline{\symbf{O}}(t) \} ] }{\widehat{\lambda} \left\{ t \mid \overline{\symbf{O}}(t) \right\} \widehat{\mathbb{E}}[ \exp\{ \alpha Y(t) \} \mid  \Delta N(t)=1, \overline{\symbf{O}}(t)  ] }, \]
which means $\rho \left\{ t \mid \overline{\symbf{O}}(t),Y(t) \right\}/ \widehat{\rho} \left\{ t \mid \overline{\symbf{O}}(t),Y(t) \right\}$ is a function of $t$ and $\overline{\symbf{O}}(t)$ only.

Thus, 
\begin{align*}
& \mathbb{P}\symbf{\varphi}_j \big( \symbf{O} ; \widehat{\symbf{\eta}}, \symbf{\beta} \big) 
- \mathbb{P}\symbf{\varphi}_j \big( \symbf{O} ; \symbf{\eta}_0, \symbf{\beta} \big)\\
= & 
\mathbb{E}\left[\int_{t=t_1}^{t_2}\symbf{W}_j(t;\symbf{\beta})\left[\frac{\rho\left\{t \mid \overline{\symbf{O}}(t), Y(t)\right\}}{\widehat{\rho}\left\{t \mid \overline{\symbf{O}}(t), Y(t)\right\}} - 1 \right]
  \left[ \mathbb{E}\{Y(t) \mid \overline{\symbf{O}}(t)\} - \widehat{\mathbb{E}}\left\{Y(t)\mid\overline{\symbf{O}}(t)\right\}\right] 
     dt\right].
\end{align*}

Write $\mathbb{P}\symbf{\varphi}_j \big( \symbf{O} ; \widehat{\symbf{\eta}}, \symbf{\beta} \big) 
- \mathbb{P}\symbf{\varphi}_j \big( \symbf{O} ; \symbf{\eta}_0, \symbf{\beta} \big)$  as $ (R_j^1(\symbf{\beta}),\cdots, R_j^p(\symbf{\beta}))^\prime$, we show that for $j=1,2$, $l=1,\cdots, p$, $R_j^l(\symbf{\beta}) = o_P(n^{-1/2})$; then the result for $\mathbb{P}\symbf{\varphi}_j \big( \symbf{O} ; \widehat{\symbf{\eta}}, \symbf{\beta} \big) 
- \mathbb{P}\symbf{\varphi}_j \big( \symbf{O} ; \symbf{\eta}_0, \symbf{\beta} \big)$ follows.

Write $\symbf{W}_j(t;\symbf{\beta}) = (\symbf{W}_j^1(t;\symbf{\beta}), \cdots, \symbf{W}_j^p(t;\symbf{\beta}))$, then for $j=1,2$, $l=1,\cdots, p$, we have 
\begin{align*}
R_j^l(\symbf{\beta}) = 
\mathbb{E}\left[\int_{t=t_1}^{t_2}
\symbf{W}_j^l(t;\symbf{\beta})\left[\frac{\rho\left\{t \mid \overline{\symbf{O}}(t), Y(t)\right\}}{\widehat{\rho}\left\{t \mid \overline{\symbf{O}}(t), Y(t)\right\}} - 1 \right]
  \left[ \mathbb{E}\{Y(t) \mid \overline{\symbf{O}}(t)\} - \widehat{\mathbb{E}}\left\{Y(t)\mid\overline{\symbf{O}}(t)\right\}\right] 
     dt\right].
\end{align*}

By the Cauchy-Schwarz inequality,
\begin{align}
    \{R_j^l (\symbf{\beta})\}^2
     \leq  & \ \mathbb{E}\left[\int_{t=t_1}^{t_2} \left[\frac{
     \symbf{W}_j^l(t;\symbf{\beta})}{\widehat{\rho}\left\{t \mid \overline{\symbf{O}}(t), Y(t)\right\}}\right]^2\left[\rho\left\{t \mid \overline{\symbf{O}}(t), Y(t)\right\} - \widehat{\rho}\left\{t \mid \overline{\symbf{O}}(t), Y(t)\right\} \right]^2 dt\right]
     \label{equ:rho_rate}\\
     & \  \times \mathbb{E}\left[\int_{t=t_1}^{t_2} \left[ \mathbb{E}\{Y(t) \mid \overline{\symbf{O}}(t)\} - \widehat{\mathbb{E}}\left\{Y(t) \mid \overline{\symbf{O}}(t)\right\}\right]^2 dt \right].
     \label{equ:E_rate}
\end{align}

\noindent Now we expand:
 \begin{align*}
\lefteqn{\widehat{\rho}\{t \mid \overline{\symbf{O}}(t), Y(t) \} - \rho\{t \mid \overline{\symbf{O}}(t), Y(t)\}} \\
&= \widehat{\lambda}_{0,k}(t) \exp\{ \widehat{\symbf{\gamma}} \symbf{Z}(t) \} \exp\{ -\alpha Y(t) \} \widehat{\mathbb{E}} \big[ \exp\{ \alpha Y(t) \} \mid \Delta N(t) = 1, \overline{\symbf{O}}(t) \big] \\
&\quad - \lambda_{0,k}(t) \exp\{ \symbf{\gamma} \symbf{Z}(t) \} \exp\{ -\alpha Y(t) \} \mathbb{E} \big[ \exp\{ \alpha Y(t) \} \mid \Delta N(t) = 1, \overline{\symbf{O}}(t) \big] \\
&= \Big\{ \widehat{\lambda}_{0,k}(t) - \lambda_{0,k}(t) \Big\} \exp\{ \widehat{\symbf{\gamma}} \symbf{Z}(t) \} \exp\{ -\alpha Y(t) \} \widehat{\mathbb{E}} \big[ \exp\{ \alpha Y(t) \} \mid \Delta N(t) = 1, \overline{\symbf{O}}(t) \big] \\
&\quad + \lambda_{0,k}(t) \Big[ \exp\{ \widehat{\symbf{\gamma}}^\prime \symbf{Z}(t) \} - \exp\{ \symbf{\gamma}^\prime \symbf{Z}(t) \} \Big] \exp\{ -\alpha Y(t) \} \widehat{\mathbb{E}} \big[ \exp\{ \alpha Y(t) \} \mid \Delta N(t) = 1, \overline{\symbf{O}}(t) \big] \\
&\quad + \lambda_{0,k}(t) \exp\{ \symbf{\gamma}^\prime \symbf{Z}(t) \} \exp\{ -\alpha Y(t) \} \Big\{ \widehat{\mathbb{E}} \big[ \exp\{ \alpha Y(t) \} \mid \Delta N(t) = 1, \overline{\symbf{O}}(t) \big] \\
&\quad\qquad\qquad\qquad\qquad\qquad\qquad - \mathbb{E} \big[ \exp\{ \alpha Y(t) \} \mid \Delta N(t) = 1, \overline{\symbf{O}}(t) \big] \Big\}.
\end{align*}

 Using this expansion and the inequality $rs \leq \frac{1}{2}(r^2+s^2)$, the rate for \eqref{equ:rho_rate} will be determined by the rates of three terms, namely the squares of each of the terms in the expansion above.

 For the factor \eqref{equ:E_rate}, by Proposition \ref{prop:mu_identity}, we write 
\begin{align*}
 & \ \mathbb{E} \{ Y(t) \mid \overline{\symbf{O}}(t) \} 
 -
 \widehat{\mathbb{E}}\{ Y(t) \mid \overline{\symbf{O}}(t)\} \\
 = & \ \frac{ \mathbb{E} \big[ Y(t) \exp \{ \alpha Y(t) \}  \mid \Delta N(t) = 1, \overline{\symbf{O}}(t) \big] }{ \mathbb{E} \big[ \exp \{ \alpha Y(t) \}  \mid \Delta N(t) = 1, \overline{\symbf{O}}(t) \big]  }
 - 
 \frac{ \widehat{\mathbb{E}} \big[ Y(t) \exp \{ \alpha Y(t) \}  \mid \Delta N(t) = 1, \overline{\symbf{O}}(t)  \big] }{ \widehat{\mathbb{E}} \big[ \exp \{ \alpha Y(t) \}  \mid \Delta N(t) = 1, \overline{\symbf{O}}(t)  \big]  } \\
 = & \ \frac{\mathbb{E} \big[ Y(t)\exp \{ \alpha Y(t) \}  \mid \Delta N(t) = 1, \overline{\symbf{O}}(t) \big]}{ 
 \mathbb{E} \big[ \exp \{ \alpha Y(t) \}  \mid \Delta N(t) = 1, \overline{\symbf{O}}(t)  \big] \widehat{\mathbb{E}} \big[ \exp \{ \alpha Y(t) \} \mid \Delta N(t) = 1, \overline{\symbf{O}}(t) \big]} \times \\
 & \quad \Big[ \widehat{\mathbb{E}} \big[ \exp \{ \alpha Y(t) \}  \mid \Delta N(t) = 1, \overline{\symbf{O}}(t) \big]- \mathbb{E} \big[  \exp \{ \alpha Y(t) \}  \mid \Delta N(t) = 1, \overline{\symbf{O}}(t) \big]\Big] \\
 & \ - \frac{\mathbb{E} \big[ \exp \{ \alpha Y(t) \}  \mid \Delta N(t) = 1, \overline{\symbf{O}}(t)  \big]}{
 \mathbb{E} \big[ \exp \{ \alpha Y(t) \} \mid \Delta N(t) = 1, \overline{\symbf{O}}(t)  \big] \widehat{\mathbb{E}} \big[ \exp \{ \alpha Y(t) \} \mid \Delta N(t) = 1, \overline{\symbf{O}}(t) \big]} \times \\
 & \quad  \Big[ \widehat{\mathbb{E}} \big[ Y(t) \exp \{ \alpha Y(t) \} \mid \Delta N(t) = 1, \overline{\symbf{O}}(t) \big]- \mathbb{E} \big[ Y(t) \exp \{ \alpha Y(t) \} \mid \Delta N(t) = 1, \overline{\symbf{O}}(t) \big]\Big].
\end{align*}
The rates for \eqref{equ:E_rate} will similarly be determined by the rates for the squares of the two terms in this expansion. Combining these results, we thus have
 \begin{align}
  & \{R_j^l(\symbf{\beta})\}^2 \ \leq \nonumber \\
  & \quad \Bigg[ 
  C_j^1(\symbf{\beta}, \widehat{\symbf{\eta}}) \ \mathbb{E} \left[ \int_{t=t_1}^{t_2} \Big\{ \widehat{ \lambda}_{0,k}(t)  - \lambda_{0,k}(t) \Big\}^2 dt \right] + \label{equ:term1} \\
  &  \quad\quad  
  C_j^2(\symbf{\beta}, \widehat{\symbf{\eta}}) \ \mathbb{E} \left[ \int_{t=t_1}^{t_2} \Big[  \exp\{ \widehat{\symbf{\gamma}}^\prime \symbf{Z}(t) \} - \exp\{ \symbf{\gamma} \symbf{Z}(t) \} \Big]^2 dt \right] + \label{equ:term2}  \\
  & \quad \quad  
  C_j^3(\symbf{\beta}, \widehat{\symbf{\eta}}) \ \mathbb{E} \bigg[ \int_{t=t_1}^{t_2} \Big\{ \widehat{\mathbb{E}} \big[ \exp\{ \alpha Y(t) \} \mid \Delta N(t) =1 , \overline{\symbf{O}}(t)  \big] - \label{equ:term3} \\
  & \hspace{1.6in} \mathbb{E} \big[ \exp\{ \alpha Y(t) \} \mid \Delta N(t) =1 , \overline{\symbf{O}}(t) \big] \Big\}^2 dt \bigg] \Bigg] \times \nonumber \\
  & \quad \Bigg[
  C_j^4(\symbf{\beta}, \widehat{\symbf{\eta}}) \ \mathbb{E} \bigg[ \int_{t=t_1}^{t_2} \bigg[  \widehat{\mathbb{E}} \big[ Y(t) \exp\{ \alpha Y(t) \} \mid \Delta N(t) =1 , \overline{\symbf{O}}(t) \big] -  \label{equ:term4} \\
  & \hspace{1.6in}
  \mathbb{E} \big[ Y(t) \exp\{ \alpha Y(t) \} \mid \Delta N(t) =1 , \overline{\symbf{O}}(t)  \big] \bigg]^2 dt \Bigg] + \nonumber \\
  & \quad \  
  C_j^5(\symbf{\beta}, \widehat{\symbf{\eta}}) \ \mathbb{E} \bigg[ \int_{t=t_1}^{t_2} \Big[ \widehat{\mathbb{E}} \big[ \exp\{ \alpha Y(t) \} \mid \Delta N(t) =1 , \overline{\symbf{O}}(t)  \big] - \label{equ:term5} \\
  & \hspace{1.6in} \mathbb{E} \big[ \exp\{ \alpha Y(t) \} \mid \Delta N(t) =1 , \overline{\symbf{O}}(t) \big] \Big]^2 dt \bigg] \Bigg],
  \nonumber 
\end{align}

\noindent where $C_j^1(\symbf{\beta}, \widehat{\symbf{\eta}}), \ldots,C_j^5(\symbf{\beta}, \widehat{\symbf{\eta}})$ are terms that are each $O_P(1)$ under Assumptions \ref{chp4:assum:positivity}  and \ref{assum:Y_rate}, $j=1,2$.

We first consider term \eqref{equ:term1}. Let $X_n(t)= (n h_n)^{1/2} \Big\{ \widehat{\lambda}_{0,k}(t)  - \lambda_{0,k}(t) \Big\}$. \cite{wells1994kernelestimation} has shown that, under Assumption \ref{assum: lambda_rate}, the stochastic process $\{X_n(t) : t_1 \leq t \leq t_2 \}_{n=1}^\infty$ converges weakly to a process $\{ X(t) : t_1 \leq t \leq t_2 \}$, where $\{ X(t) : t_1 \leq t \leq t_2 \}$ is a Wiener process plus a term depending on $\lambda_{0,k}(t)$.  Therefore, it follows by the Continuous Mapping Theorem that $\{ X_n^2(t) : t_1 \leq t \leq t_2 \}_{n=1}^\infty$ converges weakly to $\{ X^2(t) : t_1 \leq t \leq t_2 \}$.  Again by the Continuous Mapping Theorem, $\int_{t=t_1}^{t_2} X_n^2(t) dt = n h_n \left[ \int_{t=t_1}^{t_2}  \big\{ \widehat{\lambda}_{0,k}(t)  - \lambda_{0,k}(t) \big\}^2 dt \right]$ is $O_P(1)$.  Therefore, the term in \eqref{equ:term1} is $O_P((nh_n)^{-1})$. By Assumption \ref{assum: lambda_rate}, for large $n$, $n h_n$ is approximately $n^{4/5}$, so in particular the convergence rate for \eqref{equ:term1} is faster than $\sqrt{n}$. 

Since $\widehat{\gamma}$ converges to $\gamma$ at root-$n$ rates, the term in \eqref{equ:term2} is $O_P(n^{-1})$.  

Finally, terms \eqref{equ:term3}, \eqref{equ:term4}, and \eqref{equ:term5} are $o_P(n^{-1/2})$ by Assumption \ref{assum:Y_rate}. Hence the product that bounds $\{R_j^l(\symbf{\beta})\}^2$ is $o_P(n^{-1})$, so that $R_j^l(\symbf{\beta})$ is $o_P(n^{-1/2})$ as claimed.
\end{proof}


\begin{lemma}
  \label{lemma:convg2}
 Suppose that Assumptions \ref{chp4:assum:tilting}, \ref{chp4:assum:positivity}, \ref{chp4:assum:marginal}, \ref{chp4:assum:nofuture}, \ref{assum: corr_spec} - \ref{assum:Y_rate} hold, then each element of $\mathbb{P} \frac{\partial \symbf{\varphi}_j ( \symbf{O} ; \widehat{\symbf{\eta}}, \symbf{\beta} )}{\partial \symbf{\beta}}
 -
 \mathbb{P} \frac{\partial \symbf{\varphi}_j ( \symbf{O} ; \symbf{\eta}_0, \symbf{\beta} )}{\partial \symbf{\beta}}$ is $ o_P(n^{-1/2})$, for $\symbf{\beta} \in \mathcal{B}^*$, 
$j=1,2$.
\end{lemma}

\begin{proof}
\label{proof:lemma_convg2} 

\begin{align*}
    & \mathbb{P} \left[ \frac{\partial \symbf{\varphi}_j \big( \symbf{O} ; \widehat{\symbf{\eta}}, \symbf{\beta} \big)}{\partial \symbf{\beta}} \right] - 
    \mathbb{P} \left[ \frac{\partial \symbf{\varphi}_j \big( \symbf{O} ; \symbf{\eta}_0, \symbf{\beta} \big)}{\partial \symbf{\beta}} \right] \\
    = &\ \mathbb{E}\left[ \int_{t=t_1}^{t_2} \frac{\partial \symbf{W}_j(t;\symbf{\beta})}{\partial \symbf{\beta}}\frac{\left[Y(t) - \widehat{\mathbb{E}}\left\{Y(t)\mid \overline{\symbf{O}}(t)\right\}\right]}{\widehat{\rho}\left\{t \mid \overline{\symbf{O}}(t), Y(t)\right\}} dN(t) \right] \\
    & \ + \mathbb{E}\left[ \int_{t=t_1}^{t_2} \frac{\partial \symbf{W}_j(t;\symbf{\beta})}{\partial \symbf{\beta}}\left[ \widehat{\mathbb{E}}\left\{Y(t)\mid \overline{\symbf{O}}(t)\right\} - s\left\{ \symbf{B}(t)^\prime \symbf{\beta}\right\}\right]dt \right] 
    - 
    \mathbb{E}\left[ \int_{t=t_1}^{t_2} \symbf{W}_j(t;\symbf{\beta}) \frac{\partial}{\partial \symbf{\beta}} s\left\{ \symbf{B}(t)^\prime \symbf{\beta}\right\} dt \right]\\
    & \ - \mathbb{E}\left[ \int_{t=t_1}^{t_2} \frac{\partial \symbf{W}_j(t;\symbf{\beta})}{\partial \symbf{\beta}}\frac{\left[ Y(t) - \mathbb{E}\left\{Y(t)\mid\overline{\symbf{O}}(t)\right\}\right]}{\rho\left\{t \mid \overline{\symbf{O}}(t), Y(t)\right\}} dN(t) \right]\\
    & \ - \mathbb{E}\left[ \int_{t=t_1}^{t_2} \frac{\partial \symbf{W}_j(t;\symbf{\beta})}{\partial \symbf{\beta}}\left[ \mathbb{E}\left\{Y(t)\mid\overline{\symbf{O}}(t)\right\} - s\left\{ \symbf{B}(t)^\prime \symbf{\beta}\right\}\right]dt \right]  
    + 
    \mathbb{E}\left[ \int_{t=t_1}^{t_2} \symbf{W}_j(t;\symbf{\beta}) \frac{\partial}{\partial \symbf{\beta}} s\left\{ \symbf{B}(t)^\prime \symbf{\beta}\right\} dt \right]\\
    = & \ \mathbb{E}\left[ \int_{t=t_1}^{t_2} \frac{\partial \symbf{W}_j(t;\symbf{\beta})}{\partial \symbf{\beta}}\frac{\left[Y(t) - \widehat{\mathbb{E}}\left\{Y(t)\mid \overline{\symbf{O}}(t)\right\}\right]}{\widehat{\rho}\left\{t \mid \overline{\symbf{O}}(t), Y(t)\right\}} dN(t) \right] \\
    & \quad -
    \mathbb{E}\left[ \int_{t=t_1}^{t_2} \frac{\partial \symbf{W}_j(t;\symbf{\beta})}{\partial \symbf{\beta}}\frac{\left[ Y(t) - \mathbb{E}\left\{Y(t)\mid \overline{\symbf{O}}(t)\right\}\right]}{\rho\left\{t \mid \overline{\symbf{O}}(t), Y(t)\right\}} dN(t) \right]\\
    & \quad + \mathbb{E}\left[ \int_{t=t_1}^{t_2} \frac{\partial \symbf{W}_j(t;\symbf{\beta})}{\partial \symbf{\beta}}\widehat{\mathbb{E}}\left\{Y(t)\mid \overline{\symbf{O}}(t)\right\}dt \right]  - 
    \mathbb{E}\left[ \int_{t=t_1}^{t_2} \frac{\partial \symbf{W}_j(t;\symbf{\beta})}{\partial \symbf{\beta}}\mathbb{E}\left\{Y(t) \mid \overline{\symbf{O}}(t)\right\} dt \right] 
\end{align*}

Then we can write
$\mathbb{P} \frac{\partial \symbf{\varphi}_j \big( \symbf{O}; \widehat{\symbf{\eta}}, \symbf{\beta} \big)}{\partial \symbf{\beta}} - 
\mathbb{P} \frac{\partial \symbf{\varphi}_j \big( \symbf{O} ; \symbf{\eta}_0, \symbf{\beta} \big)}{\partial \symbf{\beta}}$ similar to the equation \eqref{equ:P_varphi} by changing $\symbf{W}_j(t;\symbf{\beta})$ to $\frac{\partial \symbf{W}_j(t;\symbf{\beta})}{\partial \symbf{\beta}}$. 
    
The proof for each element in $\mathbb{P} \frac{\partial \symbf{\varphi}_j \big( \symbf{O} ; \widehat{\symbf{\eta}}, \symbf{\beta} \big)}{\partial \symbf{\beta}} - 
\mathbb{P} \frac{\partial \symbf{\varphi}_j \big( \symbf{O} ; \symbf{\eta}_0, \symbf{\beta} \big)}{\partial \symbf{\beta}}$  is similar to the proof of Lemma \ref{lemma:convg1}.
\end{proof}


\begin{lemma}
    \label{lemma:L2P}
    Suppose that Assumptions \ref{chp4:assum:tilting}, \ref{chp4:assum:positivity}, \ref{chp4:assum:marginal}, \ref{chp4:assum:nofuture}, and \ref{assum:bound_Y} hold, 
    for each 
    $\symbf{\beta}\in \mathcal{B}^*$, $\varphi_j(\symbf{O}; \symbf{\eta}_0, \symbf{\beta}) \in L^2(P)$ and $\frac{\partial \varphi_j(\symbf{O}; \symbf{\eta}_0, \symbf{\beta})}{\partial \symbf{\beta}} \in L^2(P)$, which means
    \begin{align*}
        \mathbb{P} ||
        \varphi_j(\symbf{O}; \symbf{\eta}_0, \symbf{\beta})||_2^2 < +\infty, 
        \ \mathbb{P} 
        \Big|\Big|\frac{\partial \varphi_j(\symbf{O}; \symbf{\eta}_0, \symbf{\beta})}{\partial \symbf{\beta}}\Big|\Big|_2^2 < +\infty.
    \end{align*}
\end{lemma}

\begin{proof}

\textbf{Part 1: proof of }$\varphi_j(\symbf{O}; \symbf{\eta}_0, \symbf{\beta}) \in L^2(P)$.

For $\symbf{\beta}\in \mathcal{B}^*$, we show that each element of $\varphi_j(\symbf{O}; \symbf{\eta}_0, \symbf{\beta})$ is in $L^2(P)$, which implies that $\varphi_j(\symbf{O}; \symbf{\eta}_0, \symbf{\beta}) \in L^2(P)$. Write $\symbf{W}_j(t;\symbf{\beta}) = (\symbf{W}_j^1(t;\symbf{\beta}), \cdots, \symbf{W}_j^p(t;\symbf{\beta}))$ and define the $l$-th element of $\symbf{\varphi}_j(\symbf{O}; \symbf{\eta}_0, \symbf{\beta})$ as $\varphi_j^l(\symbf{O}; \symbf{\eta}_0, \symbf{\beta})$, $j=1,2$, $l=1,\cdots, p$, where 
\begin{align*}
    \varphi_j^l(\symbf{O}; \symbf{\eta}_0, \symbf{\beta})
    & = \
   \int_{t=t_1}^{t_2} \frac{
    \symbf{W}_j^l(t;\symbf{\beta}) \left[ Y(t) - \mathbb{E}\left\{Y(t) \mid \overline{\symbf{O}}(t)\right\}\right]}{\rho\left\{t \mid \overline{\symbf{O}}(t), Y(t)\right\}} dN(t) \\
    & \quad 
    + 
    \int_{t=t_1}^{t_2} \symbf{W}_j^l(t;\symbf{\beta}) \left[ \mathbb{E}\left\{Y(t)\mid \overline{\symbf{O}}(t)\right\} -  s\left\{ \symbf{B}(t)^\prime \symbf{\beta}\right\}\right]dt \\
    & = \ \int_{t=t_1}^{t_2} \frac{
    \symbf{W}_j^l(t;\symbf{\beta}) \left[ Y(t) - \mathbb{E}\left\{Y(t) \mid \overline{\symbf{O}}(t)\right\}\right]}{\rho\left\{t \mid \overline{\symbf{O}}(t), Y(t)\right\}} \left[dN(t) - \rho\left\{t \mid \overline{\symbf{O}}(t), Y(t)\right\}dt\right]\\
    & \quad + 
    \int_{t=t_1}^{t_2} \symbf{W}_j^l(t;\symbf{\beta}) \left[ Y(t) - \mathbb{E}\left\{Y(t)\mid \overline{\symbf{O}}(t)\right\}\right] dt \\
    & \quad 
    + \int_{t=t_1}^{t_2} \symbf{W}_j^l(t;\symbf{\beta}) \left[ \mathbb{E}\left\{Y(t) \mid \overline{\symbf{O}}(t)\right\} -  s\left\{ \symbf{B}(t)^\prime \symbf{\beta}\right\}\right]dt \\
     & = \ \int_{t=t_1}^{t_2} \frac{
    \symbf{W}_j^l(t;\symbf{\beta}) \left[ Y(t) - \mathbb{E}\left\{Y(t) \mid \overline{\symbf{O}}(t)\right\}\right]}{\rho\left\{t \mid \overline{\symbf{O}}(t), Y(t)\right\}} \left[dN(t) - \rho\left\{t \mid \overline{\symbf{O}}(t), Y(t)\right\} dt\right]\\
    & \quad + 
    \int_{t=t_1}^{t_2} \symbf{W}_j^l(t;\symbf{\beta}) \left[Y(t) -  s\left\{ \symbf{B}(t)^\prime \symbf{\beta}\right\}\right]dt.
\end{align*}

Then 
\begin{align*}
    \mathbb{E}\{\varphi_j^l(\symbf{O}; \symbf{\eta}_0, \symbf{\beta})\}^2 
    & = \ E\Bigg[\int_{t=t_1}^{t_2} \frac{
    \symbf{W}_j^l(t;\symbf{\beta}) \left[ Y(t) - \mathbb{E}\left\{Y(t) \mid \overline{\symbf{O}}(t)\right\}\right]}{\rho\left\{t \mid \overline{\symbf{O}}(t), Y(t)\right\}} \left[dN(t) - \rho\left\{t \mid \overline{\symbf{O}}(t), Y(t)\right\}dt\right] \\
    & \quad \quad \quad + \int_{t=t_1}^{t_2} \symbf{W}_j^l(t;\symbf{\beta}) \left[ Y(t) -  s\left\{ \symbf{B}(t)^\prime \symbf{\beta}\right\}\right]dt \Bigg]^2  \\
    & \leq \ 2 E\Bigg[\int_{t=t_1}^{t_2} \frac{\symbf{W}_j^l(t;\symbf{\beta}) \left[ Y(t) - \mathbb{E}\left\{Y(t) \mid \overline{\symbf{O}}(t)\right\}\right]}{\rho\left\{t \mid \overline{\symbf{O}}(t), Y(t)\right\}} \left[dN(t) - \rho\left\{t \mid \overline{\symbf{O}}(t), Y(t)\right\}dt\right]\Bigg]^2 \\
    &  \quad \quad \quad + 
    2 \mathbb{E}\left[\int_{t=t_1}^{t_2} \symbf{W}_j^l(t;\symbf{\beta}) \left[ Y(t)-  s\left\{ \symbf{B}(t)^\prime \symbf{\beta}\right\}\right]dt\right]^2,
\end{align*}
where
\begin{align*}
    & \ E\Bigg[
    \int_{t=t_1}^{t_2} \frac{
    \symbf{W}_j^l(t;\symbf{\beta}) \left[ Y(t) - \mathbb{E}\left\{Y(t) \mid \overline{\symbf{O}}(t)\right\}\right]}{\rho\left\{t \mid \overline{\symbf{O}}(t), Y(t)\right\}} \left[dN(t) - \rho\left\{t \mid \overline{\symbf{O}}(t), Y(t)\right\}dt\right]\Bigg]^2
    \\
    = & \ \mathbb{E}\left[\int_{t=t_1}^{t_2} \frac{
    \{\symbf{W}_j^l(t;\symbf{\beta}) \}^2\left[ Y(t) - \mathbb{E}\left\{Y(t) \mid \overline{\symbf{O}}(t)\right\}\right]^2 }{\rho\left\{t \mid \overline{\symbf{O}}(t), Y(t)\right\}} dt \right],
\end{align*}
and 
\begin{align*}
    \mathbb{E}\left[\int_{t=t_1}^{t_2} \symbf{W}_j^l(t;\symbf{\beta}) \left[ Y(t)-  s\left\{ \symbf{B}(t)^\prime \symbf{\beta}\right\}\right]dt\right]^2 
    \leq  
    \mathbb{E}\left[
    \int_{t=t_1}^{t_2} \{\symbf{W}_j^l(t;\symbf{\beta})\}^2 dt 
    \int_{t=t_1}^{t_2} \left[ Y(t)-  s\left\{ \symbf{B}(t)^\prime \symbf{\beta}\right\}\right]^2 dt\right].
\end{align*}

Since $\symbf{W}_j(t;\symbf{\beta})$ is continuous in $\symbf{\beta}$, then  $\int_{t=t_1}^{t_2} \{\symbf{W}_j^l(t;\symbf{\beta})\}^2 dt < +\infty$. Under Assumption \ref{assum:bound_Y}, we have
\begin{align*}
    \mathbb{E}\left[\int_{t=t_1}^{t_2}    \frac{
    \{\symbf{W}_j^l(t;\symbf{\beta}) \}^2\left[ Y(t) - \mathbb{E}\left\{Y(t) \mid \overline{\symbf{O}}(t)\right\}\right]^2 }{\rho\left\{t \mid \overline{\symbf{O}}(t), Y(t)\right\}} dt \right] < + \infty, \\
    \ 
    \mathbb{E}\left[
    \int_{t=t_1}^{t_2} \{\symbf{W}_j^l(t;\symbf{\beta})\}^2 dt 
    \int_{t=t_1}^{t_2} \left[ Y(t)-  s\left\{ \symbf{B}(t)^\prime \symbf{\beta}\right\}\right]^2 dt\right] < + \infty,
\end{align*}
which proves $\mathbb{E}\{\varphi_j^l(\symbf{O}; \symbf{\eta}_0, \symbf{\beta})\}^2  < +\infty$ for $j=1,2$, $l=1,2,\cdots,p$. 

Thus, each element of $\varphi_j(\symbf{O}; \symbf{\eta}_0, \symbf{\beta}) \in L^2(P)$ and further 
$\varphi_j(\symbf{O}; \symbf{\eta}_0, \symbf{\beta}) \in L^2(P)$.
\\

\textbf{Part 2: proof of }$\frac{\partial \varphi_j(\symbf{O}; \symbf{\eta}_0, \symbf{\beta})}{\partial \symbf{\beta}} \in L^2(P)$.

Next, we show that each element of $\frac{\partial \varphi_j(\symbf{O}; \symbf{\eta}_0, \symbf{\beta})}{\partial \symbf{\beta}} \in L^2(P)$, which implies that $\frac{\partial \varphi_j(\symbf{O}; \symbf{\eta}_0, \symbf{\beta})}{\partial \symbf{\beta}} \in L^2(P)$.

Write the $(l_1,l_2)$-th element of $\frac{\partial  \symbf{W}_j(t;\symbf{\beta})}{\partial \symbf{\beta}}$ as $\left[\frac{\partial  \symbf{W}_j(t;\symbf{\beta})}{\partial \symbf{\beta}} \right]_{l_1,l_2}$, and write the $(l_1,l_2)$-th element of $\frac{\partial \symbf{\varphi}_j ( \symbf{O}; \symbf{\eta}_0, \symbf{\beta} )}{\partial \symbf{\beta}}$ as $\left[\frac{\partial \symbf{\varphi}_j ( \symbf{O} ; \symbf{\eta}_0, \symbf{\beta} )}{\partial \symbf{\beta}} \right]_{l_1,l_2}$, for $j=1,2$, $l_1,l_2=1,\cdots, p$, where 
\begin{align*}
    \quad & \left[\frac{\partial \symbf{\varphi}_j \big( \symbf{O} ; \symbf{\eta}_0, \symbf{\beta} \big)}{\partial \symbf{\beta}} \right]_{l_1,l_2}\\
    = & \ 
    \int_{t=t_1}^{t_2} \left[\frac{\partial  \symbf{W}_j(t;\symbf{\beta})}{\partial \symbf{\beta}} \right]_{l_1,l_2}
    \frac{
    \left[ Y(t) - \mathbb{E}\left\{Y(t) \mid \overline{\symbf{O}}(t)\right\}\right]}{\rho\left\{t \mid \overline{\symbf{O}}(t), Y(t)\right\}} dN(t)\\
    & \quad + 
    \int_{t=t_1}^{t_2} 
    \left[\frac{\partial  \symbf{W}_j(t;\symbf{\beta})}{\partial \symbf{\beta}} \right]_{l_1,l_2}
    \left[ \mathbb{E}\left\{Y(t)\mid \overline{\symbf{O}}(t)\right\} -  s\left\{ \symbf{B}(t)^\prime \symbf{\beta}\right\}\right]dt\\
    & \quad -
    \int_{t=t_1}^{t_2} \symbf{W}_j^{l_1}(t;\symbf{\beta}) \left[\frac{\partial s\left\{ \symbf{B}(t)^\prime \symbf{\beta}\right\}}{\partial \symbf{\beta}}\right]_{l_2} dt \\
    = & \ 
    \underbrace{\int_{t=t_1}^{t_2} \left[\frac{\partial  \symbf{W}_j(t;\symbf{\beta})}{\partial \symbf{\beta}} \right]_{l_1,l_2}
    \frac{
    \left[ Y(t) - \mathbb{E}\left\{Y(t) \mid \overline{\symbf{O}}(t)\right\}\right]}{\rho\left\{t \mid \overline{\symbf{O}}(t), Y(t)\right\}} \left[dN(t) - \rho\left\{t \mid \overline{\symbf{O}}(t), Y(t)\right\}dt\right]}_{term1} \\
    & \quad + 
    \underbrace{\int_{t=t_1}^{t_2} 
    \left[\frac{\partial  \symbf{W}_j(t;\symbf{\beta})}{\partial \symbf{\beta}} \right]_{l_1,l_2}
    \left[ Y(t) -  s\left\{ \symbf{B}(t)^\prime \symbf{\beta}\right\}\right]dt}_{term2}
    - 
    \underbrace{\int_{t=t_1}^{t_2} \symbf{W}_j^{l_1}(t;\symbf{\beta}) \left[\frac{\partial s\left\{ \symbf{B}(t)^\prime \symbf{\beta}\right\}}{\partial \symbf{\beta}}\right]_{l_2} dt}_{term3}
\end{align*}

 As in the previous part, we need to show that each term satisfies:

\begin{align*}
    & \mathbb{E}\left[
    \int_{t=t_1}^{t_2} \left[\frac{\partial  \symbf{W}_j(t;\symbf{\beta})}{\partial \symbf{\beta}} \right]_{l_1,l_2}
    \frac{
    \left[ Y(t) - \mathbb{E}\left\{Y(t) \mid \overline{\symbf{O}}(t)\right\}\right]}{\rho\left\{t \mid \overline{\symbf{O}}(t), Y(t)\right\}} \left[dN(t) - \rho\left\{t \mid \overline{\symbf{O}}(t), Y(t)\right\}dt\right]
    \right]^2 
    < +\infty, \\
    & \mathbb{E}\left[\int_{t=t_1}^{t_2} 
    \left[\frac{\partial  \symbf{W}_j(t;\symbf{\beta})}{\partial \symbf{\beta}} \right]_{l_1,l_2}
    \left[Y(t) -  s\left\{ \symbf{B}(t)^\prime \symbf{\beta}\right\}\right]dt\right]^2 
    < +\infty, \\
    & \mathbb{E}\left[\int_{t=t_1}^{t_2} \symbf{W}_j^{l_1}(t;\symbf{\beta}) \left[\frac{\partial s\left\{ \symbf{B}(t)^\prime \symbf{\beta}\right\}}{\partial \symbf{\beta}}\right]_{l_2} dt\right]^2 
    < +\infty,
\end{align*}
where term1 and term2 can be proved similarly to the first part, and term3 can be proved by the continuity of the inverse link function, $s \in \mathbb{C}^2$.

Thus, we have proved that for $\symbf{\beta}\in \mathcal{B}^*$, each element of $\frac{\partial \symbf{\varphi}_j ( \symbf{O}; \symbf{\eta}_0, \symbf{\beta} )}{\partial \symbf{\beta}} \in L^2(P)$, then $\frac{\partial \varphi_j(\symbf{O}; \symbf{\eta}_0, \symbf{\beta})}{\partial \symbf{\beta}} \in L^2(P)$.
\end{proof}


\begin{lemma}
  \label{lemma:L2P_continuous}
      Suppose that Assumptions \ref{chp4:assum:tilting}, \ref{chp4:assum:positivity}, \ref{chp4:assum:marginal}, \ref{chp4:assum:nofuture}, and \ref{assum:bound_Y}, \ref{assum: corr_spec} - \ref{assum:Y_rate} hold, for $\symbf{\beta} \in \mathcal{B}^*$, $\varphi_j(\symbf{O}; \symbf{\eta}, \symbf{\beta})$ and $\frac{\partial \varphi_j(\symbf{O}; \symbf{\eta}, \symbf{\beta})}{\partial \symbf{\beta}}$ satisfy
  \begin{align*}
   \sup_{\symbf{\beta} \in \mathcal{B}^*}
   \mathbb{P} 
   ||\varphi_j(\symbf{O}; \widehat{\symbf{\eta}}, \symbf{\beta}) - \varphi_j(\symbf{O}; \symbf{\eta}_0, \symbf{\beta})||_2^2  = o_P(1), 
   \ 
   \sup_{\symbf{\beta} \in \mathcal{B}^*}\mathbb{P} 
   \Big|\Big|\frac{\partial \varphi_j(\symbf{O}; \widehat{\symbf{\eta}}, \symbf{\beta})}{\partial \symbf{\beta}} - \frac{\partial \varphi_j(\symbf{O}; \symbf{\eta}_0, \symbf{\beta})}{\partial \symbf{\beta}} \Big|\Big|_2^2  = o_P(1),
    \end{align*}
    as $\widehat{\symbf{\eta}} \rightarrow \symbf{\eta} _0$.
\end{lemma}

\begin{proof}
   \textbf{Part 1: proof of }$\sup_{\symbf{\beta} \in \mathcal{B}^*}
   \mathbb{P} 
   ||\varphi_j(\symbf{O}; \widehat{\symbf{\eta}}, \symbf{\beta}) - \varphi_j(\symbf{O}; \symbf{\eta}_0, \symbf{\beta})||_2^2  = o_P(1)$.
\begin{align*}
    & \varphi_j(\symbf{O}; \widehat{\symbf{\eta}}, \symbf{\beta}) - \varphi_j(\symbf{O}; \symbf{\eta}_0, \symbf{\beta}) \\
   =  & \ 
   \int_{t=t_1}^{t_2} \frac{
    \symbf{W}_j(t;\symbf{\beta}) \left[ Y(t) - \widehat{\mathbb{E}}\left\{Y(t)\mid \overline{\symbf{O}}(t)\right\}\right]}{\widehat{\rho}\left\{t \mid \overline{\symbf{O}}(t), Y(t)\right\}} dN(t) \\
    & \quad 
    + \int_{t=t_1}^{t_2} \symbf{W}_j(t;\symbf{\beta}) \left[ \widehat{\mathbb{E}}\left\{Y(t)\mid \overline{\symbf{O}}(t)\right\} -  s\left\{ \symbf{B}(t)^\prime \symbf{\beta}\right\}\right]dt  \\
    & \quad 
    - \int_{t=t_1}^{t_2} \frac{
    \symbf{W}_j(t;\symbf{\beta}) \left[ Y(t) - \mathbb{E}\left\{Y(t)\mid \overline{\symbf{O}}(t)\right\}\right]}{\rho\left\{t \mid \overline{\symbf{O}}(t), Y(t)\right\}} dN(t)\\
    & \quad 
    - \int_{t=t_1}^{t_2} \symbf{W}_j(t;\symbf{\beta}) \left[ \mathbb{E}\left\{Y(t)\mid \overline{\symbf{O}}(t)\right\} -  s\left\{ \symbf{B}(t)^\prime \symbf{\beta}\right\}\right]dt \\
  = & \ 
    \int_{t=t_1}^{t_2} 
    \symbf{W}_j(t;\symbf{\beta})
    \left[ \frac{
    \left[ Y(t) - \widehat{\mathbb{E}}\left\{Y(t)\mid \overline{\symbf{O}}(t)\right\}\right]}{\widehat{\rho}\left\{t \mid \overline{\symbf{O}}(t), Y(t)\right\}} 
    - 
    \frac{
    \left[ Y(t) - \mathbb{E}\left\{Y(t)\mid \overline{\symbf{O}}(t)\right\}\right]}{\rho\left\{t \mid \overline{\symbf{O}}(t), Y(t)\right\}} 
    \right] dN(t) \\
    & \quad 
    + \int_{t=t_1}^{t_2} \symbf{W}_j(t;\symbf{\beta}) \left[ \widehat{\mathbb{E}}\left\{Y(t)\mid \overline{\symbf{O}}(t)\right\} -  \mathbb{E}\left\{Y(t)\mid \overline{\symbf{O}}(t)\right\}\right]dt \\
   = & \ 
   \underbrace{
   \int_{t=t_1}^{t_2} \frac{\symbf{W}_j(t;\symbf{\beta})
    Y(t)}{\rho\left\{t \mid \overline{\symbf{O}}(t), Y(t)\right\}}\left[
    \frac{\rho\left\{t \mid \overline{\symbf{O}}(t), Y(t)\right\}}{\widehat{\rho}\left\{t \mid \overline{\symbf{O}}(t), Y(t)\right\}} - 1\right]dN(t)
   }_{\Delta_{j,1}}
   \\
    & \quad \underbrace{- \int_{t=t_1}^{t_2} \frac{\symbf{W}_j(t;\symbf{\beta})}{\widehat{\rho}\left\{t \mid \overline{\symbf{O}}(t), Y(t)\right\}}\left[\widehat{\mathbb{E}}\left\{Y(t)\mid \overline{\symbf{O}}(t)\right\} - \mathbb{E}\left\{Y(t)\mid \overline{\symbf{O}}(t)\right\}\right] dN(t)}_{\Delta_{j,2}} \\
    & \quad \underbrace{ - \int_{t=t_1}^{t_2}\frac{\symbf{W}_j(t;\symbf{\beta})  \mathbb{E}\left\{Y(t)\mid \overline{\symbf{O}}(t)\right\}}{\rho\left\{t \mid \overline{\symbf{O}}(t), Y(t)\right\}}\left[\frac{\rho\left\{t \mid \overline{\symbf{O}}(t), Y(t)\right\}}{\widehat{\rho}\left\{t \mid \overline{\symbf{O}}(t), Y(t)\right\}} - 1\right] dN(t)}_{\Delta_{j,3}} \\
    & \quad \underbrace{
    + \int_{t=t_1}^{t_2} \symbf{W}_j(t;\symbf{\beta}) \left[ \widehat{\mathbb{E}}\left\{Y(t)|\overline{\symbf{O}}(t)\right\} -  \mathbb{E}\left\{Y(t)\mid \overline{\symbf{O}}(t)\right\}\right]dt}_{\Delta_{j,4}}.
\end{align*}

By the Cauchy-Schwarz inequality, we have:
\begin{align*}
    \mathbb{P} ||\varphi_j(\symbf{O}; \widehat{\symbf{\eta}}, \symbf{\beta}) - \varphi_j(\symbf{O}; \symbf{\eta}_0, \symbf{\beta})||_2^2 
    & = \ 
    \mathbb{P} ||\Delta_{j,1} + \Delta_{j,2} + \Delta_{j,3} + \Delta_{j,4}||_2^2 \\
    & = \ 
    \mathbb{P}\left\{\left(\Delta_{j,1} + \Delta_{j,2} + \Delta_{j,3} + \Delta_{j,4}\right)^\prime \left(\Delta_{j,1} + \Delta_{j,2} + \Delta_{j,3} +\Delta_{j,4}\right)\right\}
    \\
    & \leq \ 
     \mathbb{P}
     \left(||\Delta_{j,1}||_2 +
     ||\Delta_{j,2}||_2 + ||\Delta_{j,3}||_2 + ||\Delta_{j,4}||_2
     \right)^2
    \\
    & \leq \ 
    4\left(
    \mathbb{P} ||\Delta_{j,1}||_2^2 + \mathbb{P} ||\Delta_{j,2}||_2^2 + 
    \mathbb{P} ||\Delta_{j,3}||_2^2 + 
    \mathbb{P} ||\Delta_{j,4}||_2^2\right).
\end{align*}
Thus, 
\begin{align*}
    & \sup_{\symbf{\beta} \in \mathcal{B}^*}
   \mathbb{P} 
   ||\varphi_j(\symbf{O}; \widehat{\symbf{\eta}}, \symbf{\beta}) - \varphi_j(\symbf{O}; \symbf{\eta}_0, \symbf{\beta})||_2^2 \\
   \leq & \ 
   \sup_{\symbf{\beta} \in \mathcal{B}^*} 4 
   \left(
    \mathbb{P} ||\Delta_{j,1}||_2^2 + \mathbb{P} ||\Delta_{j,2}||_2^2 + 
    \mathbb{P} ||\Delta_{j,3}||_2^2 + 
    \mathbb{P} ||\Delta_{j,4}||_2^2\right) \\
   \leq & \
    4\left(
    \sup_{\symbf{\beta} \in \mathcal{B}^*}\mathbb{P} ||\Delta_{j,1}||_2^2 + \sup_{\symbf{\beta} \in \mathcal{B}^*}\mathbb{P} ||\Delta_{j,2}||_2^2 + 
    \sup_{\symbf{\beta} \in \mathcal{B}^*}\mathbb{P} ||\Delta_{j,3}||_2^2 + 
    \sup_{\symbf{\beta} \in \mathcal{B}^*}\mathbb{P} ||\Delta_{j,4}||_2^2\right).
\end{align*}

Thus, if we can show that $\sup_{\symbf{\beta} \in \mathcal{B}^*}\mathbb{P} ||\Delta_{j,k}||_2^2 = o_P(1)$ for $j=1,2$, $k=1,2,3,4$, we can then prove $\sup_{\symbf{\beta} \in \mathcal{B}^*}
\mathbb{P} ||\varphi_j(\symbf{O}; \widehat{\symbf{\eta}}, \symbf{\beta}) - \varphi_j(\symbf{O}; \symbf{\eta}_0, \symbf{\beta})||_2^2  = o_P(1)$.

Define the $l$-th element of $\Delta_{j,k}$ as $\Delta_{j,k}^l$, we need to show that each element of $\Delta_{j,k}$ satisfies: $\sup_{\symbf{\beta} \in \mathcal{B}^*}\mathbb{P} (\Delta_{j,k}^l)^2 = o_P(1)$, for $j=1,2$, $k=1,2,3,4$, $l=1,\cdots, p$.

For the $l$-th element of $\Delta_{j,1}$, we have
\begin{align*}
   \Delta_{j,1}^l 
   = & \ \int_{t=t_1}^{t_2} \frac{\symbf{W}_j^l(t;\symbf{\beta})Y(t)}{\rho\left\{t \mid \overline{\symbf{O}}(t), Y(t)\right\}}
    \left[
    \frac{\rho\left\{t \mid \overline{\symbf{O}}(t), Y(t)\right\}}{\widehat{\rho}\left\{t \mid \overline{\symbf{O}}(t), Y(t)\right\}} - 1\right]
    dN(t) \\
   = & \ 
   \int_{t=t_1}^{t_2} \frac{\symbf{W}_j^l(t;\symbf{\beta})Y(t)}{\rho\left\{t \mid \overline{\symbf{O}}(t), Y(t)\right\}}
    \left[
    \frac{\rho\left\{t \mid \overline{\symbf{O}}(t), Y(t)\right\}}{\widehat{\rho}\left\{t \mid \overline{\symbf{O}}(t), Y(t)\right\}} - 1\right]
    \left[dN(t) - \rho\left\{t \mid \overline{\symbf{O}}(t), Y(t)\right\} dt \right] \\
    & \quad + 
    \int_{t=t_1}^{t_2} \symbf{W}_j^l(t;\symbf{\beta})Y(t)
    \left[
    \frac{\rho\left\{t \mid \overline{\symbf{O}}(t), Y(t)\right\}}{\widehat{\rho}\left\{t \mid \overline{\symbf{O}}(t), Y(t)\right\}} - 1\right]dt 
\end{align*}

Then we have 
\begin{align*}
    & \sup_{\symbf{\beta} \in \mathcal{B}^*}\mathbb{P} (\Delta_{j,1}^l)^2\\
    = & \  \sup_{\symbf{\beta} \in \mathcal{B}^*} 
   \mathbb{E}\left[\int_{t=t_1}^{t_2} \frac{\symbf{W}_j^l(t;\symbf{\beta})Y(t)}{\rho\left\{t \mid \overline{\symbf{O}}(t), Y(t)\right\}}
    \left[
    \frac{\rho\left\{t \mid \overline{\symbf{O}}(t), Y(t)\right\}}{\widehat{\rho}\left\{t \mid \overline{\symbf{O}}(t), Y(t)\right\}} - 1\right]
    dN(t)\right]^2 \nonumber \\
   \leq & \ 
   2 \sup_{\symbf{\beta} \in \mathcal{B}^*} 
   \mathbb{E} \left[
   \int_{t=t_1}^{t_2} \frac{\symbf{W}_j^l(t;\symbf{\beta})Y(t)}{\rho\left\{t \mid \overline{\symbf{O}}(t), Y(t)\right\}}
    \left[
    \frac{\rho\left\{t \mid \overline{\symbf{O}}(t), Y(t)\right\}}{\widehat{\rho}\left\{t \mid \overline{\symbf{O}}(t), Y(t)\right\}} - 1\right]
    \left[dN(t) - \rho\left\{t \mid \overline{\symbf{O}}(t), Y(t)\right\} dt \right]
   \right]^2 \\
   & \quad + 
   2 \sup_{\symbf{\beta} \in \mathcal{B}^*}  \mathbb{E} \left[
   \int_{t=t_1}^{t_2} \symbf{W}_j^l(t;\symbf{\beta})Y(t)\left[
    \frac{\rho\left\{t \mid \overline{\symbf{O}}(t), Y(t)\right\}}{\widehat{\rho}\left\{t \mid \overline{\symbf{O}}(t), Y(t)\right\}} - 1\right]dt \right]^2,
\end{align*}
where 
\begin{align}
    & \  2 \sup_{\symbf{\beta} \in \mathcal{B}^*} 
   \mathbb{E} \left[
   \int_{t=t_1}^{t_2} \frac{\symbf{W}_j^l(t;\symbf{\beta})Y(t)}{\rho\left\{t \mid \overline{\symbf{O}}(t), Y(t)\right\}}
    \left[
    \frac{\rho\left\{t \mid \overline{\symbf{O}}(t), Y(t)\right\}}{\widehat{\rho}\left\{t \mid \overline{\symbf{O}}(t), Y(t)\right\}} - 1\right]
    \left[dN(t) - \rho\left\{t \mid \overline{\symbf{O}}(t), Y(t)\right\} dt \right]
   \right]^2 \nonumber \\
   = & \  
   2\sup_{\symbf{\beta} \in \mathcal{B}^*}
   \mathbb{E} \left[
   \int_{t=t_1}^{t_2} \frac{\left\{\symbf{W}_j^l(t;\symbf{\beta})Y(t)\right\}^2}{ \rho\left\{t \mid \overline{\symbf{O}}(t), Y(t)\right\} \widehat{\rho}^2\left\{t \mid \overline{\symbf{O}}(t), Y(t)\right\}} \left[
   \rho\left\{t \mid \overline{\symbf{O}}(t), Y(t)\right\} - \widehat{\rho}\left\{t \mid \overline{\symbf{O}}(t), Y(t)\right\}
   \right]^2 dt 
   \right] \nonumber \\
   \leq & \  
    2 \mathbb{E} \left[
   \int_{t=t_1}^{t_2} 
   \sup_{\symbf{\beta} \in \mathcal{B}^*} \frac{\left\{\symbf{W}_j^l(t;\symbf{\beta})Y(t)\right\}^2}{ \rho\left\{t \mid \overline{\symbf{O}}(t), Y(t)\right\} \widehat{\rho}^2\left\{t \mid \overline{\symbf{O}}(t), Y(t)\right\}} \left[
   \rho\left\{t \mid \overline{\symbf{O}}(t), Y(t)\right\} - \widehat{\rho}\left\{t \mid \overline{\symbf{O}}(t), Y(t)\right\}
   \right]^2 dt 
   \right], 
   \label{equ:delta1_l_1}
\end{align}
and 
\begin{align}
    & \ 2 \sup_{\symbf{\beta} \in \mathcal{B}^*}  \mathbb{E} \left[
   \int_{t=t_1}^{t_2} \symbf{W}_j^l(t;\symbf{\beta})Y(t)\left[
    \frac{\rho\left\{t \mid \overline{\symbf{O}}(t), Y(t)\right\}}{\widehat{\rho}\left\{t \mid \overline{\symbf{O}}(t), Y(t)\right\}} - 1\right]dt \right]^2 \nonumber \\
    \leq & \
    2 \sup_{\symbf{\beta} \in \mathcal{B}^*} \mathbb{E} \left[
   \left[\int_{t=t_1}^{t_2} \frac{\left\{\symbf{W}_j^l(t;\symbf{\beta})Y(t)\right\}^2}{\widehat{\rho}^2 \left\{t \mid \overline{\symbf{O}}(t), Y(t)\right\}}
   dt \right]
   \int_{t=t_1}^{t_2} \left[
   \rho\left\{t \mid \overline{\symbf{O}}(t), Y(t)\right\} - \widehat{\rho}\left\{t \mid \overline{\symbf{O}}(t), Y(t)\right\}\right]^2 dt \right] \nonumber \\
   \leq & \
   2 \mathbb{E} \left[\left[\int_{t=t_1}^{t_2} \sup_{\symbf{\beta} \in \mathcal{B}^*} \frac{\left\{\symbf{W}_j^l(t;\symbf{\beta})Y(t)\right\}^2}{\widehat{\rho}^2 \left\{t \mid \overline{\symbf{O}}(t), Y(t)\right\}}
   dt \right]
   \int_{t=t_1}^{t_2} \left[
   \rho\left\{t \mid \overline{\symbf{O}}(t), Y(t)\right\} - \widehat{\rho}\left\{t \mid \overline{\symbf{O}}(t), Y(t)\right\}\right]^2 dt \right]
     \label{equ:delta1_l_2}
\end{align}
Both equations \eqref{equ:delta1_l_1} and \eqref{equ:delta1_l_2} can be proved to be $o_P(1)$ similarly to the proof for equation \eqref{equ:rho_rate} in Lemma \ref{lemma:convg1}.


For the $l$-th element of $\Delta_{j,2}$, we have
\begin{align*}
    \Delta_{j,2}^l 
    = & \ 
    \int_{t=t_1}^{t_2} \frac{\symbf{W}_j^l(t;\symbf{\beta})}{\widehat{\rho}\left\{t \mid \overline{\symbf{O}}(t), Y(t)\right\}}\left[\widehat{\mathbb{E}}\left\{Y(t)\mid \overline{\symbf{O}}(t)\right\} - \mathbb{E}\left\{Y(t)\mid \overline{\symbf{O}}(t)\right\}\right] dN(t)\\
    = & \ 
    \int_{t=t_1}^{t_2} \frac{\symbf{W}_j^l(t;\symbf{\beta})}{\widehat{\rho}\left\{t \mid \overline{\symbf{O}}(t), Y(t)\right\}}\left[\widehat{\mathbb{E}}\left\{Y(t)\mid \overline{\symbf{O}}(t)\right\} - \mathbb{E}\left\{Y(t)\mid \overline{\symbf{O}}(t)\right\}\right] \left[dN(t) - \rho\left\{t \mid \overline{\symbf{O}}(t), Y(t)\right\} dt \right] \\ 
    & \quad + 
     \int_{t=t_1}^{t_2} \frac{\symbf{W}_j^l(t;\symbf{\beta})\rho\left\{t \mid \overline{\symbf{O}}(t), Y(t)\right\}}{\widehat{\rho}\left\{t \mid \overline{\symbf{O}}(t), Y(t)\right\}}\left[\widehat{\mathbb{E}}\left\{Y(t)\mid \overline{\symbf{O}}(t)\right\} - \mathbb{E}\left\{Y(t)\mid \overline{\symbf{O}}(t)\right\}\right] dt.
\end{align*}
Then we have 
\begin{align*}
     & \sup_{\symbf{\beta} \in \mathcal{B}^*}\mathbb{P} (\Delta_{j,2}^l)^2 \\
     = & 
    \sup_{\symbf{\beta} \in \mathcal{B}^*} \mathbb{E}\left[
    \int_{t=t_1}^{t_2} \frac{\symbf{W}_j^l(t;\symbf{\beta})}{\widehat{\rho}\left\{t \mid \overline{\symbf{O}}(t), Y(t)\right\}}\left[\widehat{\mathbb{E}}\left\{Y(t)\mid \overline{\symbf{O}}(t)\right\} - \mathbb{E}\left\{Y(t)\mid \overline{\symbf{O}}(t)\right\}\right] dN(t)
    \right]^2\nonumber \\
    \leq &
    2 \sup_{\symbf{\beta} \in \mathcal{B}^*} \mathbb{E}\left[
    \int_{t=t_1}^{t_2} \frac{\symbf{W}_j^l(t;\symbf{\beta})}{\widehat{\rho}\left\{t \mid \overline{\symbf{O}}(t), Y(t)\right\}}\left[\widehat{\mathbb{E}}\left\{Y(t)\mid \overline{\symbf{O}}(t)\right\} - \mathbb{E}\left\{Y(t)\mid \overline{\symbf{O}}(t)\right\}\right] \left[dN(t) - \rho\left\{t \mid \overline{\symbf{O}}(t), Y(t)\right\} dt \right]
    \right]^2 \\
    & \quad + 
    2 \sup_{\symbf{\beta} \in \mathcal{B}^*}\mathbb{E}\left[
    \int_{t=t_1}^{t_2} \frac{\symbf{W}_j^l(t;\symbf{\beta})\rho\left\{t \mid \overline{\symbf{O}}(t), Y(t)\right\}}{\widehat{\rho}\left\{t \mid \overline{\symbf{O}}(t), Y(t)\right\}}\left[\widehat{\mathbb{E}}\left\{Y(t)\mid \overline{\symbf{O}}(t)\right\} - \mathbb{E}\left\{Y(t)\mid \overline{\symbf{O}}(t)\right\}\right] dt
    \right]^2,
\end{align*}
where 
\begin{align}
    & \ 2 \sup_{\symbf{\beta} \in \mathcal{B}^*} \mathbb{E}\left[
    \int_{t=t_1}^{t_2} \frac{\symbf{W}_j^l(t;\symbf{\beta})}{\widehat{\rho}\left\{t \mid \overline{\symbf{O}}(t), Y(t)\right\}}\left[\widehat{\mathbb{E}}\left\{Y(t)\mid \overline{\symbf{O}}(t)\right\} - \mathbb{E}\left\{Y(t)\mid \overline{\symbf{O}}(t)\right\}\right] \left[dN(t) - \rho\left\{t \mid \overline{\symbf{O}}(t), Y(t)\right\} dt \right]
    \right]^2 \nonumber \\
    = & \ 
    2 \sup_{\symbf{\beta} \in \mathcal{B}^*} \mathbb{E}\left[
     \int_{t=t_1}^{t_2} 
     \frac{\{\symbf{W}_j^l(t;\symbf{\beta})\}^2 \rho\left\{t \mid \overline{\symbf{O}}(t), Y(t)\right\}}{\widehat{\rho}^2\left\{t \mid \overline{\symbf{O}}(t), Y(t)\right\}}
     \left[\widehat{\mathbb{E}}\left\{Y(t)\mid \overline{\symbf{O}}(t)\right\} - \mathbb{E}\left\{Y(t)\mid \overline{\symbf{O}}(t)\right\}\right]^2 dt 
    \right]\nonumber \\
    \leq & \ 
    2\mathbb{E}\left[
    \int_{t=t_1}^{t_2}
    \sup_{\symbf{\beta} \in \mathcal{B}^*} \frac{\{\symbf{W}_j^l(t;\symbf{\beta})\}^2 \rho\left\{t \mid \overline{\symbf{O}}(t), Y(t)\right\}}{\widehat{\rho}^2\left\{t \mid \overline{\symbf{O}}(t), Y(t)\right\}}
    \left[\widehat{\mathbb{E}}\left\{Y(t)\mid \overline{\symbf{O}}(t)\right\} - \mathbb{E}\left\{Y(t)\mid \overline{\symbf{O}}(t)\right\}\right]^2 dt
    \right],
    \label{equ:delta2_l_1}
\end{align}
and 
\begin{align}
    & \ 
    2 \sup_{\symbf{\beta} \in \mathcal{B}^*} \mathbb{E}\left[
    \int_{t=t_1}^{t_2} \frac{\symbf{W}_j^l(t;\symbf{\beta})\rho\left\{t \mid \overline{\symbf{O}}(t), Y(t)\right\}}{\widehat{\rho}\left\{t \mid \overline{\symbf{O}}(t), Y(t)\right\}}\left[\widehat{\mathbb{E}}\left\{Y(t)\mid \overline{\symbf{O}}(t)\right\} - \mathbb{E}\left\{Y(t)\mid \overline{\symbf{O}}(t)\right\}\right] dt
    \right]^2 \nonumber \\
    \leq & \ 
    2 \sup_{\symbf{\beta} \in \mathcal{B}^*} \mathbb{E}\left[
    \left[ \int_{t=t_1}^{t_2} \frac{\left[\symbf{W}_j^l(t;\symbf{\beta})\rho\left\{t \mid \overline{\symbf{O}}(t), Y(t)\right\}\right]^2}{\widehat{\rho}^2\left\{t \mid \overline{\symbf{O}}(t), Y(t)\right\}} dt \right]
    \int_{t=t_1}^{t_2} \left[\widehat{\mathbb{E}}\left\{Y(t)\mid \overline{\symbf{O}}(t)\right\} - \mathbb{E}\left\{Y(t)\mid \overline{\symbf{O}}(t)\right\}\right]^2 dt
    \right] \nonumber \\
    \leq & \
    2 \mathbb{E}\left[
    \left[ \int_{t=t_1}^{t_2} 
    \sup_{\symbf{\beta} \in \mathcal{B}^*} \frac{\left[\symbf{W}_j^l(t;\symbf{\beta})\rho\left\{t \mid \overline{\symbf{O}}(t), Y(t)\right\}\right]^2}{\widehat{\rho}^2\left\{t \mid \overline{\symbf{O}}(t), Y(t)\right\}} dt \right]
    \int_{t=t_1}^{t_2} \left[\widehat{\mathbb{E}}\left\{Y(t)\mid \overline{\symbf{O}}(t)\right\} - \mathbb{E}\left\{Y(t)\mid \overline{\symbf{O}}(t)\right\}\right]^2 dt
    \right].
     \label{equ:delta2_l_2}
\end{align}

Both equations \eqref{equ:delta2_l_1} and \eqref{equ:delta2_l_2} can be proved to be $o_P(1)$ similarly to the proof of equation \eqref{equ:E_rate} in Lemma \ref{lemma:convg1}.

For the $l$-th element of $\Delta_{j,3}$, we have
\begin{align*}
    \Delta_{j,3}^l = & \ \int_{t=t_1}^{t_2}\frac{\symbf{W}_j^l(t;\symbf{\beta})  \mathbb{E}\left\{Y(t)\mid \overline{\symbf{O}}(t)\right\}}{\rho\left\{t \mid \overline{\symbf{O}}(t), Y(t)\right\}}\left[\frac{\rho\left\{t \mid \overline{\symbf{O}}(t), Y(t)\right\}}{\widehat{\rho}\left\{t \mid \overline{\symbf{O}}(t), Y(t)\right\}} - 1\right] dN(t)\\
    = & \ \int_{t=t_1}^{t_2}\frac{\symbf{W}_j^l(t;\symbf{\beta})  \mathbb{E}\left\{Y(t)\mid \overline{\symbf{O}}(t)\right\}}{\rho\left\{t \mid \overline{\symbf{O}}(t), Y(t)\right\}}\left[\frac{\rho\left\{t \mid \overline{\symbf{O}}(t), Y(t)\right\}}{\widehat{\rho}\left\{t \mid \overline{\symbf{O}}(t), Y(t)\right\}} - 1\right] \left[dN(t) - \rho\left\{t \mid \overline{\symbf{O}}(t), Y(t)\right\}dt\right] \\
    & \quad + \int_{t=t_1}^{t_2}
    \symbf{W}_j^l(t;\symbf{\beta})  \mathbb{E}\left\{Y(t)\mid \overline{\symbf{O}}(t)\right\}\left[\frac{\rho\left\{t \mid \overline{\symbf{O}}(t), Y(t)\right\}}{\widehat{\rho}\left\{t \mid \overline{\symbf{O}}(t), Y(t)\right\}} - 1\right] dt. 
\end{align*}

Then we have 
\begin{align*}
    & \sup_{\symbf{\beta} \in \mathcal{B}^*}\mathbb{P} (\Delta_{j,3}^l)^2 \\
    = & \ 
    \sup_{\symbf{\beta} \in \mathcal{B}^*} E
    \left[
    \int_{t=t_1}^{t_2} \frac{\symbf{W}_j^l(t;\symbf{\beta})  \mathbb{E}\left\{Y(t)\mid \overline{\symbf{O}}(t)\right\}}{\rho\left\{t \mid \overline{\symbf{O}}(t), Y(t)\right\}}
    \left[\frac{\rho\left\{t \mid \overline{\symbf{O}}(t), Y(t)\right\}}{\widehat{\rho}\left\{t \mid \overline{\symbf{O}}(t), Y(t)\right\}} - 1\right] dN(t)\right]^2 \\
    \leq & \ 
    2 \sup_{\symbf{\beta} \in \mathcal{B}^*} \mathbb{E} \left[\int_{t=t_1}^{t_2}\frac{\symbf{W}_j^l(t;\symbf{\beta})  \mathbb{E}\left\{Y(t)\mid \overline{\symbf{O}}(t)\right\}}{\rho\left\{t \mid \overline{\symbf{O}}(t), Y(t)\right\}}\left[\frac{\rho\left\{t \mid \overline{\symbf{O}}(t), Y(t)\right\}}{\widehat{\rho}\left\{t \mid \overline{\symbf{O}}(t), Y(t)\right\}} - 1\right] \left[dN(t) - \rho\left\{t \mid \overline{\symbf{O}}(t), Y(t)\right\}dt\right]
    \right]^2 \\
    & \quad + 
    2 \sup_{\symbf{\beta} \in \mathcal{B}^*} \mathbb{E} \left[
    \int_{t=t_1}^{t_2}
    \symbf{W}_j^l(t;\symbf{\beta})  \mathbb{E}\left\{Y(t)\mid \overline{\symbf{O}}(t)\right\}\left[\frac{\rho\left\{t \mid \overline{\symbf{O}}(t), Y(t)\right\}}{\widehat{\rho}\left\{t \mid \overline{\symbf{O}}(t), Y(t)\right\}} - 1\right] dt
    \right]^2,
\end{align*}
where 
\begin{align}
    & \ 
    2 \sup_{\symbf{\beta} \in \mathcal{B}^*} \mathbb{E} \left[
    \int_{t=t_1}^{t_2}\frac{\symbf{W}_j^l(t;\symbf{\beta})  \mathbb{E}\left\{Y(t)\mid \overline{\symbf{O}}(t)\right\}}{\rho\left\{t \mid \overline{\symbf{O}}(t), Y(t)\right\}}\left[\frac{\rho\left\{t \mid \overline{\symbf{O}}(t), Y(t)\right\}}{\widehat{\rho}\left\{t \mid \overline{\symbf{O}}(t), Y(t)\right\}} - 1\right] \left[dN(t) - \rho\left\{t \mid \overline{\symbf{O}}(t), Y(t)\right\}dt\right]
    \right]^2 \nonumber \\
    = & \ 
    2 \sup_{\symbf{\beta} \in \mathcal{B}^*} \mathbb{E}\left[\int_{t=t_1}^{t_2} \frac{\left[\symbf{W}_j^l(t;\symbf{\beta})  \mathbb{E}\left\{Y(t)\mid \overline{\symbf{O}}(t)\right\}\right]^2}{\rho\left\{t \mid \overline{\symbf{O}}(t), Y(t)\right\} \widehat{\rho}^2\left\{t \mid\overline{\symbf{O}}(t),Y(t)\right\}}\left[\rho\left\{t \mid \overline{\symbf{O}}(t), Y(t)\right\}  - \widehat{\rho}\left\{t \mid \overline{\symbf{O}}(t), Y(t)\right\} \right]^2 dt
    \right] 
    \nonumber\\
    \leq & \ 
    2 \mathbb{E}\left[\int_{t=t_1}^{t_2}\sup_{\symbf{\beta} \in \mathcal{B}^*} \frac{\left[\symbf{W}_j^l(t;\symbf{\beta})  \mathbb{E}\left\{Y(t)\mid \overline{\symbf{O}}(t)\right\}\right]^2}{\rho\left\{t \mid \overline{\symbf{O}}(t), Y(t)\right\} \widehat{\rho}^2\left\{t \mid\overline{\symbf{O}}(t),Y(t)\right\}}\left[\rho\left\{t \mid \overline{\symbf{O}}(t), Y(t)\right\}  - \widehat{\rho}\left\{t \mid \overline{\symbf{O}}(t), Y(t)\right\} \right]^2 dt
    \right],
    \label{equ:delta3_l_1}
\end{align}
and 
\begin{align}
    & \ 2 \sup_{\symbf{\beta} \in \mathcal{B}^*} \mathbb{E} \left[
    \int_{t=t_1}^{t_2}
    \symbf{W}_j^l(t;\symbf{\beta})  \mathbb{E}\left\{Y(t)\mid \overline{\symbf{O}}(t)\right\}\left[\frac{\rho\left\{t \mid \overline{\symbf{O}}(t), Y(t)\right\}}{\widehat{\rho}\left\{t \mid \overline{\symbf{O}}(t), Y(t)\right\}} - 1\right] dt
    \right]^2 \nonumber \\
    \leq & \ 
    2 \sup_{\symbf{\beta} \in \mathcal{B}^*} 
    \mathbb{E}\left[
    \left[\int_{t=t_1}^{t_2} \frac{\left[\symbf{W}_j^l(t;\symbf{\beta})  \mathbb{E}\left\{Y(t)\mid \overline{\symbf{O}}(t)\right\}\right]^2}{\widehat{\rho}^2\left\{t \mid \overline{\symbf{O}}(t), Y(t)\right\}}dt
    \right] \int_{t=t_1}^{t_2}\left[
    \rho\left\{t \mid \overline{\symbf{O}}(t), Y(t)\right\} - \widehat{\rho}\left\{t \mid \overline{\symbf{O}}(t), Y(t)\right\}\right]^2 dt
    \right] \nonumber \\
    \leq & \ 
    2\mathbb{E}\left[
    \left[\int_{t=t_1}^{t_2} \sup_{\symbf{\beta} \in \mathcal{B}^*} \frac{\left[\symbf{W}_j^l(t;\symbf{\beta})  \mathbb{E}\left\{Y(t)\mid \overline{\symbf{O}}(t)\right\}\right]^2}{\widehat{\rho}^2\left\{t \mid \overline{\symbf{O}}(t), Y(t)\right\}}dt
    \right] \int_{t=t_1}^{t_2}\left[
    \rho\left\{t \mid \overline{\symbf{O}}(t), Y(t)\right\} - \widehat{\rho}\left\{t \mid \overline{\symbf{O}}(t), Y(t)\right\}\right]^2 dt
    \right]
    \label{equ:delta3_l_2}
\end{align}

Both equations \eqref{equ:delta3_l_1} and \eqref{equ:delta3_l_2} can be proved to be $o_P(1)$ similarly to the proof for equation \eqref{equ:rho_rate} in Lemma \ref{lemma:convg1}.

For the $l$-th element of $\Delta_{j,4}$, we have
\begin{align}
     \sup_{\symbf{\beta} \in \mathcal{B}^*}\mathbb{P} (\Delta_{j,4}^l)^2  
     = & \ 
     \sup_{\symbf{\beta} \in \mathcal{B}^*} E
    \left[\int_{t=t_1}^{t_2} \symbf{W}_j^l(t;\symbf{\beta}) \left[ \widehat{\mathbb{E}}\left\{Y(t)\mid\overline{\symbf{O}}(t)\right\} -  \mathbb{E}\left\{Y(t)\mid\overline{\symbf{O}}(t)\right\}\right]dt\right]^2 \nonumber \\
    \leq & \ 
    \sup_{\symbf{\beta} \in \mathcal{B}^*} \mathbb{E}\left[
    \left[\int_{t=t_1}^{t_2}  
    \{\symbf{W}_j^l(t;\symbf{\beta})\}^2 dt \right] \int_{t=t_1}^{t_2}
    \left[ \widehat{\mathbb{E}}\left\{Y(t)\mid\overline{\symbf{O}}(t)\right\} -  \mathbb{E}\left\{Y(t)\mid\overline{\symbf{O}}(t)\right\}\right]^2dt
    \right]\\
    \leq & \ 
    \mathbb{E}\left[
    \left[\int_{t=t_1}^{t_2}
    \sup_{\symbf{\beta} \in \mathcal{B}^*}\{\symbf{W}_j^l(t;\symbf{\beta})\}^2 dt \right] \int_{t=t_1}^{t_2}
    \left[ \widehat{\mathbb{E}}\left\{Y(t)\mid\overline{\symbf{O}}(t)\right\} -  \mathbb{E}\left\{Y(t)\mid\overline{\symbf{O}}(t)\right\}\right]^2dt
    \right]
    \label{equ:delta4_l}
\end{align}
Equation \eqref{equ:delta4_l} can be proved to be $o_P(1)$ similarly to the proof of equation \eqref{equ:E_rate} in Lemma \ref{lemma:convg1}.

Thus, we have proved that $\sup_{\symbf{\beta} \in \mathcal{B}^*}\mathbb{P} ||\Delta_{j,k}||_2^2 = o_P(1)$, for $j=1,2$, $k=1,2,3,4$, which further proves that $\sup_{\symbf{\beta} \in \mathcal{B}^*}\mathbb{P} 
||\varphi_j(\symbf{O}; \widehat{\symbf{\eta}}, \symbf{\beta}) - \varphi_j(\symbf{O}; \symbf{\eta}_0, \symbf{\beta})||_2^2  = o_P(1)$.

\textbf{Part 2: proof of }$\sup_{\symbf{\beta} \in \mathcal{B}^*}
\mathbb{P} 
\Big|\Big|
\frac{\partial \varphi_j(\symbf{O}; \widehat{\symbf{\eta}}, \symbf{\beta})}{\partial \symbf{\beta}} - \frac{\partial \varphi_j(\symbf{O}; \symbf{\eta}_0, \symbf{\beta})}{\partial \symbf{\beta}} \Big|\Big|_2^2  = o_P(1)$.

\begin{align*}
    & \ \frac{\partial \varphi_j(\symbf{O}; \widehat{\symbf{\eta}}, \symbf{\beta})}{\partial \symbf{\beta}} - \frac{\partial \varphi_j(\symbf{O}; \symbf{\eta}_0, \symbf{\beta})}{\partial \symbf{\beta}} \\
    = & \ 
    \int_{t=t_1}^{t_2} \frac{\partial  \symbf{W}_j(t;\symbf{\beta})}{\partial \symbf{\beta}}\frac{
    \left[ Y(t) - \widehat{\mathbb{E}}\left\{Y(t) \mid \overline{\symbf{O}}(t)\right\}\right]}{\widehat{\rho}\left\{t \mid \overline{\symbf{O}}(t), Y(t)\right\}} dN(t) \\
    & \quad + 
    \int_{t=t_1}^{t_2} \frac{\partial  \symbf{W}_j(t;\symbf{\beta})}{\partial \symbf{\beta}}\left[ \widehat{\mathbb{E}}\left\{Y(t)\mid \overline{\symbf{O}}(t)\right\} -  s\left\{ \symbf{B}(t)^\prime \symbf{\beta}\right\}\right]dt
    -
    \int_{t=t_1}^{t_2} \symbf{W}_j(t;\symbf{\beta})\frac{\partial s\left\{ \symbf{B}(t)^\prime \symbf{\beta}\right\}}{\partial \symbf{\beta}}dt\\
    & \quad - 
    \int_{t=t_1}^{t_2} \frac{\partial  \symbf{W}_j(t;\symbf{\beta})}{\partial \symbf{\beta}}\frac{
    \left[ Y(t) - \mathbb{E}\left\{Y(t) \mid \overline{\symbf{O}}(t)\right\}\right]}{\rho\left\{t \mid \overline{\symbf{O}}(t), Y(t)\right\}} dN(t) \\
    & \quad - 
    \int_{t=t_1}^{t_2} \frac{\partial  \symbf{W}_j(t;\symbf{\beta})}{\partial \symbf{\beta}}\left[ \mathbb{E}\left\{Y(t)\mid \overline{\symbf{O}}(t)\right\} -  s\left\{ \symbf{B}(t)^\prime \symbf{\beta}\right\}\right]dt
    +
    \int_{t=t_1}^{t_2} \symbf{W}_j(t;\symbf{\beta})\frac{\partial s\left\{ \symbf{B}(t)^\prime \symbf{\beta}\right\}}{\partial \symbf{\beta}}dt\\
   = & \ 
   \int_{t=t_1}^{t_2} \frac{\partial  \symbf{W}_j(t;\symbf{\beta})}{\partial \symbf{\beta}}\frac{
    \left[ Y(t) - \widehat{\mathbb{E}}\left\{Y(t) \mid \overline{\symbf{O}}(t)\right\}\right]}{\widehat{\rho}\left\{t \mid \overline{\symbf{O}}(t), Y(t)\right\}} dN(t) \\
    & \quad - \int_{t=t_1}^{t_2} \frac{\partial  \symbf{W}_j(t;\symbf{\beta})}{\partial \symbf{\beta}}\frac{
    \left[ Y(t) - \mathbb{E}\left\{Y(t) \mid \overline{\symbf{O}}(t)\right\}\right]}{\rho\left\{t \mid \overline{\symbf{O}}(t), Y(t)\right\}} dN(t) \\
    & \quad + 
    \int_{t=t_1}^{t_2} 
    \frac{\partial  \symbf{W}_j(t;\symbf{\beta})}{\partial \symbf{\beta}}\left[\widehat{\mathbb{E}}\left\{Y(t) \mid \overline{\symbf{O}}(t)\right\} - \mathbb{E}\left\{Y(t) \mid \overline{\symbf{O}}(t)\right\}\right] dt\\
  = & \ 
  \underbrace{\int_{t=t_1}^{t_2} \frac{\partial  \symbf{W}_j(t;\symbf{\beta})}{\partial \symbf{\beta}}
  \frac{Y(t)}{\rho\left\{t \mid \overline{\symbf{O}}(t), Y(t)\right\}} \left[\frac{\rho\left\{t \mid \overline{\symbf{O}}(t), Y(t)\right\}}{\widehat{\rho}\left\{t \mid \overline{\symbf{O}}(t), Y(t)\right\}}  - 1 \right]dN(t)}_{D_{j,1}}
  \\
  & \quad 
  \underbrace{-  
  \int_{t=t_1}^{t_2} \frac{\partial   \symbf{W}_j(t;\symbf{\beta})}{\partial \symbf{\beta}} \frac{1}{\widehat{\rho}\left\{t \mid \overline{\symbf{O}}(t), Y(t)\right\}} \left[\widehat{\mathbb{E}}\left\{Y(t) \mid \overline{\symbf{O}}(t)\right\} - \mathbb{E}\left\{Y(t) \mid \overline{\symbf{O}}(t)\right\}\right]dN(t)}_{D_{j,2}}\\
  & \quad 
  \underbrace{-  
  \int_{t=t_1}^{t_2} \frac{\partial   \symbf{W}_j(t;\symbf{\beta})}{\partial \symbf{\beta}} 
  \frac{\mathbb{E}\left\{Y(t) \mid \overline{\symbf{O}}(t)\right\}}{\rho\left\{t \mid \overline{\symbf{O}}(t), Y(t)\right\}} \left[\frac{\rho\left\{t \mid \overline{\symbf{O}}(t), Y(t)\right\}}{\widehat{\rho}\left\{t \mid \overline{\symbf{O}}(t), Y(t)\right\}}  - 1 \right] dN(t)}_{D_{j,3}}\\
  & \quad \underbrace{
  + \int_{t=t_1}^{t_2} 
    \frac{\partial  \symbf{W}_j(t;\symbf{\beta})}{\partial \symbf{\beta}}\left[\widehat{\mathbb{E}}\left\{Y(t) \mid \overline{\symbf{O}}(t)\right\} - \mathbb{E}\left\{Y(t) \mid \overline{\symbf{O}}(t)\right\}\right] dt}_{D_{j,4}}\\
 \end{align*}

By Cauchy-Schwarz inequality, we have:

\begin{align*}
    \mathbb{P} 
   \Big|\Big|\frac{\partial \varphi_j(\symbf{O}; \widehat{\symbf{\eta}}, \symbf{\beta})}{\partial \symbf{\beta}} - \frac{\partial \varphi_j(\symbf{O}; \symbf{\eta}_0, \symbf{\beta})}{\partial \symbf{\beta}} \Big|\Big|_2^2
   & = \ 
    \mathbb{P} ||D_{j,1} + D_{j,2} + D_{j,3} + D_{j,4}||_2^2 \\
    & \leq \ 
    4\left(
    \mathbb{P} ||D_{j,1}||_2^2 + \mathbb{P} ||D_{j,2}||_2^2 + 
    \mathbb{P} ||D_{j,3}||_2^2 + 
    \mathbb{P} ||D_{j,4}||_2^2\right),
\end{align*}
and 
\begin{align*}
   & \sup_{\symbf{\beta} \in \mathcal{B}^*}\mathbb{P} \Big|\Big|\frac{\partial \varphi_j(\symbf{O}; \widehat{\symbf{\eta}}, \symbf{\beta})}{\partial \symbf{\beta}} - \frac{\partial \varphi_j(\symbf{O}; \symbf{\eta}_0, \symbf{\beta})}{\partial \symbf{\beta}} \Big|\Big|_2^2 \\ 
  \leq &  \ 
   \sup_{\symbf{\beta} \in \mathcal{B}^*} 4 
   \left(
    \mathbb{P} ||D_{j,1}||_2^2 + \mathbb{P} ||D_{j,2}||_2^2 + 
    \mathbb{P} ||D_{j,3}||_2^2 + 
    \mathbb{P} ||D_{j,4}||_2^2\right) \\
  \leq &  \
    4\left(
    \sup_{\symbf{\beta} \in \mathcal{B}^*}\mathbb{P} ||D_{j,1}||_2^2 + \sup_{\symbf{\beta} \in \mathcal{B}^*}\mathbb{P} ||D_{j,2}||_2^2 + 
    \sup_{\symbf{\beta} \in \mathcal{B}^*}\mathbb{P} ||D_{j,3}||_2^2 + 
    \sup_{\symbf{\beta} \in \mathcal{B}^*}\mathbb{P} ||D_{j,4}||_2^2\right).
\end{align*}

Similarly to Part 1, we can prove that $(l_1,l_2)$-th element of $D_{j,k}$, $[D_{j,k}]_{l_1,l_2}$, satisfies 
$$\sup_{\symbf{\beta} \in \mathcal{B}^*}\mathbb{P} ([D_{j,k}]_{l_1,l_2})^2 = o_P(1),$$ for $j=1,2$, $k=1,2,3,4$, $l=1,\cdots, p$. Thus, we can 
prove that $\sup_{\symbf{\beta} \in \mathcal{B}^*}\mathbb{P} ||D_{j,k}||_2^2 = o_P(1)$, for $j=1,2$, $k=1,2,3,4$, which further proves that $\sup_{\symbf{\beta} \in \mathcal{B}^*}\mathbb{P} 
\Big|\Big|
\frac{\partial \varphi_j(\symbf{O}; \widehat{\symbf{\eta}}, \symbf{\beta})}{\partial \symbf{\beta}} - \frac{\partial \varphi_j(\symbf{O}; \symbf{\eta}_0, \symbf{\beta})}{\partial \symbf{\beta}} \Big|\Big|_2^2  = o_P(1)$.
\end{proof}


\begin{lemma}
  \label{lemma:convg3}
  Suppose that Assumptions \ref{chp4:assum:tilting}--\ref{assum:donsker1}, 
  \ref{assum: corr_spec}--\ref{assum:Y_rate} hold, then $\mathbb{P}_n\symbf{\varphi}_j \big(\symbf{O} ; \widehat{\symbf{\eta}}, \symbf{\beta}\big) - \mathbb{P}\symbf{\varphi}_j \big(\symbf{O} ; \symbf{\eta}_0, \symbf{\beta}\big) = n^{-1/2} \mathcal{N}(\symbf{0}, \Sigma_j) + o_P(n^{-1/2})$, where 
  $\Sigma_j = E \big\{\symbf{\varphi}_j \left( \symbf{O}; \symbf{\eta}_0, \symbf{\beta}\right)  \symbf{\varphi}_j \left( \symbf{O}; \symbf{\eta}_0, \symbf{\beta}\right) ^\prime \big\}$, 
  for $\symbf{\beta} \in \mathcal{B}^*$,
  j=1,2.
\end{lemma}

\begin{proof}
For $\symbf{\beta} \in \mathcal{B}^*$, we have
\begin{align*}
    \mathbb{P}_n \symbf{\varphi}_j
    \big( \symbf{O} ; \widehat{\symbf{\eta}}, \symbf{\beta} \big) - \mathbb{P} \symbf{\varphi}_j \big( \symbf{O} ; \symbf{\eta}_0, \symbf{\beta} \big)
    = & \ \underbrace{(\mathbb{P}_n  - \mathbb{P})\symbf{\varphi}_j \big( \symbf{O} ; \symbf{\eta}_0, \symbf{\beta}\big)}_{term1} \ \\
    & \ +  \underbrace{(\mathbb{P}_n  - \mathbb{P}) \left\{\symbf{\varphi}_j \big( \symbf{O} ; \widehat{\symbf{\eta}}, \symbf{\beta} \big) - \symbf{\varphi}_j \big( \symbf{O} ; \symbf{\eta}_0, \symbf{\beta} \big)\right\}}_{term2}\\
    & \ + \underbrace{\mathbb{P}\symbf{\varphi}_j \big( \symbf{O} ; \widehat{\symbf{\eta}}, \symbf{\beta} \big) - \mathbb{P}\symbf{\varphi}_j \big( \symbf{O} ; \symbf{\eta}_0, \symbf{\beta} \big)}_{term3} 
\end{align*}

We next check each term:

(1) By Lemma \ref{lemma:L2P}, $\symbf{\varphi}_j \left( \symbf{O} ; \symbf{\eta}_0, \symbf{\beta} \right) \in L^2(\mathbb{P})$, since $\symbf{O}_1,\cdots, \symbf{O}_n$ are i.i.d., term1 is asymptotically normal by the central limit theorem:
\begin{align*}
    \sqrt{n}(\mathbb{P}_n - \mathbb{P}) \symbf{\varphi}_j \left( \symbf{O}; \symbf{\eta}_0, \symbf{\beta}\right) \overset{d}{\rightarrow} \mathcal{N}(\symbf{0}, \Sigma_j)
\end{align*}
where 
\begin{align*}
    \Sigma_j = \mathbb{E} \left\{\symbf{\varphi}_j \left( \symbf{O}; \symbf{\eta}_0, \symbf{\beta}\right)  \symbf{\varphi}_j \left( \symbf{O}; \symbf{\eta}_0, \symbf{\beta}\right) ^\prime \right\}
\end{align*}
or we can write it as  
\begin{align*}
(\mathbb{P}_n - \mathbb{P}) \symbf{\varphi}_j \left( \symbf{O}; \symbf{\eta}_0, \symbf{\beta} \right)  = n^{-1/2} \mathcal{N}(\symbf{0}, \Sigma_j) + o_P(n^{-1/2})
\end{align*}

(2)  By Lemma \ref{lemma:L2P_continuous}, $\mathbb{P}\ || \symbf{\varphi}_j \left( \symbf{O} ; \widehat{\symbf{\eta}}, \symbf{\beta}\right) - \symbf{\varphi}_j \left( \symbf{O} ; \symbf{\eta}_0, \symbf{\beta} \right)||_2^2 = o_P(1)$ as $\widehat{\symbf{\eta}} \rightarrow \symbf{\eta}_0$. Then under Assumption \ref{assum:donsker1}, by Lemma 19.24 of \cite{van2000asymptotic}, $||(\mathbb{P}_n  - \mathbb{P}) \left\{\symbf{\varphi}_j \big( \symbf{O} ; \widehat{\symbf{\eta}}, \symbf{\beta} \big) - \symbf{\varphi}_j \big( \symbf{O} ; \symbf{\eta}_0, \symbf{\beta} \big)\right\}||_2$ is the empirical process remainder which is ensured to be negligible ($o_P(n^{-1/2})$) \cite{wellner2013weak}. Then each element of  $(\mathbb{P}_n  - \mathbb{P}) \left\{\symbf{\varphi}_j \big( \symbf{O} ; \widehat{\symbf{\eta}}, \symbf{\beta} \big) - \symbf{\varphi}_j \big( \symbf{O} ; \symbf{\eta}_0, \symbf{\beta} \big)\right\}$ is also $o_P(n^{-1/2})$.

(3) By Lemma \ref{lemma:convg1}, we know that each element of $\mathbb{P}\symbf{\varphi}_j \big(\symbf{O} ; \widehat{\symbf{\eta}}, \symbf{\beta} \big) - \mathbb{P}\symbf{\varphi}_j \big(\symbf{O} ; \symbf{\eta}_0, \symbf{\beta} \big) = o_P(n^{-1/2})$.

Thus, we know that $\mathbb{P}_n \symbf{\varphi}_j
    \big( \symbf{O} ; \widehat{\symbf{\eta}}, \symbf{\beta} \big) - \mathbb{P} \symbf{\varphi}_j \big( \symbf{O} ; \symbf{\eta}_0, \symbf{\beta} \big) = n^{-1/2} \mathcal{N}(\symbf{0}, \Sigma_j) + o_P(n^{-1/2})$, where $\Sigma_j = E \big\{\symbf{\varphi}_j \left( \symbf{O}; \symbf{\eta}_0, \symbf{\beta}\right)  \symbf{\varphi}_j \left( \symbf{O}; \symbf{\eta}_0, \symbf{\beta}\right) ^\prime \big\}$, j=1,2.

\end{proof}

\begin{lemma}
  \label{lemma:convg4}
  Suppose that Assumptions \ref{chp4:assum:tilting}--\ref{assum:bound_Y}, 
  \ref{assum:donsker2}--\ref{assum:Y_rate} hold, then each element of $\mathbb{P}_n \frac{\partial \symbf{\varphi}_j \big(\symbf{O} ; \widehat{\symbf{\eta}}, \symbf{\beta}\big)}{\partial \symbf{\beta}} - 
  \mathbb{P} \frac{\partial \symbf{\varphi}_j \big(\symbf{O} ; \symbf{\eta}_0, \symbf{\beta}\big)}{\partial \symbf{\beta}}$ is $o_p(1)$. If we vectorize 
  the matrix $\frac{\partial \symbf{\varphi}_j \big(\symbf{O} ; \widehat{\symbf{\eta}}, \symbf{\beta}\big)}{\partial \symbf{\beta}} 
  \in \mathbb{R}^{p\times p}$ as $\symbf{\phi}_j^{1}\big(\symbf{O}; \widehat{\symbf{\eta}}, \symbf{\beta}\big) = vec\Big(\frac{\partial \symbf{\varphi}_j \big(\symbf{O} ; \widehat{\symbf{\eta}}, \symbf{\beta}\big)}{\partial \symbf{\beta}}\Big)$,
  and the matrix 
  $\frac{\partial \symbf{\varphi}_j \big(\symbf{O} ; \symbf{\eta}_0, \symbf{\beta}\big)}{\partial \symbf{\beta}} \in \mathbb{R}^{p\times p}$ as $\symbf{\phi}_j^{1}\big(\symbf{O} ; \symbf{\eta}_0, \symbf{\beta}\big) = vec\Big(\frac{\partial \symbf{\varphi}_j \big(\symbf{O} ; \symbf{\eta}_0, \symbf{\beta}\big)}{\partial \symbf{\beta}}\Big)$, which stacks all elements of the matrix column-wise into a single vector. Then $\mathbb{P}_n \symbf{\phi}_j^{1}\big(\symbf{O} ; \widehat{\symbf{\eta}}, \symbf{\beta}\big) - \mathbb{P}\symbf{\phi}_j^{1}\big(\symbf{O} ; \symbf{\eta}_0, \symbf{\beta}\big) = n^{-1/2} \mathcal{N}(\symbf{0}, \Sigma^1_j) + o_P(n^{-1/2})$, where 
  $\Sigma^1_j = E \big\{\symbf{\phi}_j^1 \left( \symbf{O}; \symbf{\eta}_0, \symbf{\beta}\right)  \symbf{\phi}_j^1 \left( \symbf{O}; \symbf{\eta}_0, \symbf{\beta}\right) ^\prime \big\}$, 
  for $\symbf{\beta} \in \mathcal{B}^*$, $j=1,2$.
\end{lemma}

\begin{proof}
    For $\symbf{\beta} \in \mathcal{B}^*$, we have
\begin{align*}
    \mathbb{P}_n \symbf{\phi}_j^{1}
    \big( \symbf{O} ; \widehat{\symbf{\eta}}, \symbf{\beta} \big) 
    - \mathbb{P} \symbf{\phi}_j^{1} \big( \symbf{O} ; \symbf{\eta}_0, \symbf{\beta} \big)
    = & \ 
    \underbrace{(\mathbb{P}_n  - \mathbb{P})
    \symbf{\phi}_j^{1} \big( \symbf{O} ; \symbf{\eta}_0, \symbf{\beta}\big)
    }_{term1} \ + \ \underbrace{(\mathbb{P}_n  - \mathbb{P}) \left\{\symbf{\phi}_j^{1} \big( \symbf{O} ; \widehat{\symbf{\eta}}, \symbf{\beta} \big) - \symbf{\phi}_j^{1} \big( \symbf{O} ; \symbf{\eta}_0, \symbf{\beta} \big)\right\}}_{term2}\\
    & \quad + \underbrace{\mathbb{P}
    \symbf{\phi}_j^{1} \big( \symbf{O} ; \widehat{\symbf{\eta}}, \symbf{\beta} \big) - \mathbb{P}
    \symbf{\phi}_j^{1} \big( \symbf{O} ; \symbf{\eta}_0, \symbf{\beta} \big)}_{term3} 
\end{align*}

We next check each term:

(1) By Lemma \ref{lemma:L2P}, $\frac{\partial \varphi_j(\symbf{O}; \symbf{\eta}_0, \symbf{\beta})}{\partial \symbf{\beta}} \in L^2(P)$, thus 
$\symbf{\phi}_j^{1} \big( \symbf{O} ; \symbf{\eta}_0, \symbf{\beta} \big) \in L^2(P)$. Since $\symbf{O}_1,\cdots, \symbf{O}_n$ are i.i.d., term1 is asymptotically normal by the central limit theorem:
\begin{align*}
    \sqrt{n}(\mathbb{P}_n - \mathbb{P}) \symbf{\phi}_j^{1} \big( \symbf{O} ; \symbf{\eta}_0, \symbf{\beta}\big)
    \overset{d}{\rightarrow} \mathcal{N}(\symbf{0}, \Sigma_j^1)
\end{align*}
where 
\begin{align*}
    \Sigma_j^1 = \mathbb{E} \left\{
    \symbf{\phi}_j^{1} \big( \symbf{O} ; \symbf{\eta}_0, \symbf{\beta}\big)
    \symbf{\phi}_j^{1} \big( \symbf{O} ; \symbf{\eta}_0, \symbf{\beta}\big)^\prime \right\},
\end{align*}
or we can write it as 
\begin{align*}
(\mathbb{P}_n - \mathbb{P}) \symbf{\phi}_j^{1} \big( \symbf{O} ; \symbf{\eta}_0, \symbf{\beta}\big) = n^{-1/2} \mathcal{N}(\symbf{0}, \Sigma_j^1) + o_P(n^{-1/2})
\end{align*}

(2)  By Lemma \ref{lemma:L2P_continuous}, $\mathbb{P} 
\Big|\Big|
\frac{\partial \varphi_j(\symbf{O}; \widehat{\symbf{\eta}}, \symbf{\beta})}{\partial \symbf{\beta}} - \frac{\partial \varphi_j(\symbf{O}; \symbf{\eta}_0, \symbf{\beta})}{\partial \symbf{\beta}} \Big|\Big|_2^2  = o_p(1)$ as $\widehat{\symbf{\eta}} \rightarrow \symbf{\eta}_0$. Thus, 
$\mathbb{P} 
\big|\big| \symbf{\phi}_j^{1} \big( \symbf{O} ; \widehat{\symbf{\eta}}, \symbf{\beta} \big) - \symbf{\phi}_j^{1} \big( \symbf{O} ; \symbf{\eta}_0, \symbf{\beta} \big)\big|\big|_2^2  = o_p(1)$ as $\widehat{\symbf{\eta}} \rightarrow \symbf{\eta}_0$.
Under Assumption \ref{assum:donsker2}, by Lemma 19.24 \citep{van2000asymptotic}, $\big|\big|(\mathbb{P}_n - \mathbb{P}) \left\{\symbf{\phi}_j^{1} \big( \symbf{O} ; \widehat{\symbf{\eta}}, \symbf{\beta} \big) - \symbf{\phi}_j^{1} \big( \symbf{O} ; \symbf{\eta}_0, \symbf{\beta} \big)\right\}\big|\big|_2$ is the empirical process remainder which is ensured to be negligible ($o_P(n^{-1/2})$) \citep{wellner2013weak}. Then each element of  
$$(\mathbb{P}_n  - \mathbb{P}) \left\{\symbf{\phi}_j^{1} \big( \symbf{O} ; \widehat{\symbf{\eta}}, \symbf{\beta} \big) - \symbf{\phi}_j^{1} \big( \symbf{O} ; \symbf{\eta}_0, \symbf{\beta} \big)\right\} $$ 
is also $o_P(n^{-1/2})$.

(3) By Lemma \ref{lemma:convg2}, we know that each element of $\mathbb{P}\symbf{\phi}_j^{1} \big( \symbf{O} ; \widehat{\symbf{\eta}}, \symbf{\beta} \big) - \mathbb{P}\symbf{\phi}_j^{1} \big( \symbf{O} ; \symbf{\eta}_0, \symbf{\beta} \big)$ is $o_P(n^{-1/2})$.

Thus, we know that $\mathbb{P}_n \symbf{\phi}_j^{1}\big(\symbf{O} ; \widehat{\symbf{\eta}}, \symbf{\beta}\big) - \mathbb{P}\symbf{\phi}_j^{1}\big(\symbf{O} ; \symbf{\eta}_0, \symbf{\beta}\big) = n^{-1/2} \mathcal{N}(\symbf{0}, \Sigma^1_j) + o_P(n^{-1/2})$, where $\Sigma^1_j = E \big\{\symbf{\phi}_j^1 \left( \symbf{O}; \symbf{\eta}_0, \symbf{\beta}\right)  \symbf{\phi}_j^1 \left( \symbf{O}; \symbf{\eta}_0, \symbf{\beta}\right) ^\prime \big\}$, 
  for $\symbf{\beta} \in \mathcal{B}^*$, $j=1,2$.
\end{proof}

\subsubsection{Proof of Theorem \ref{theorem:consist}}
\label{app:proof_consist}

\begin{proof}

We first prove the uniform convergence $\sup_{\symbf{\beta}\in \mathcal{B}^*} ||\mathbb{P}_n\varphi_j(\symbf{O}; \widehat{\symbf{\eta}}, \symbf{\beta}) - \mathbb{P}\varphi_j(\symbf{O}; \symbf{\eta}_0, \symbf{\beta}) ||_2 \overset{p}{\rightarrow} 
0$ and existence of $\widehat{\symbf{\beta}}_j$ in a neighborhood of $\symbf{\beta}_0$ for each influence function $\symbf{\varphi}_j$, $j=1,2$. Finally, we
apply the proof of Theorem 5.9 in  \cite{van2000asymptotic} to derive the consistency.

\textbf{Part 1: proof of }$\sup_{\symbf{\beta}\in \mathcal{B}^*} || \mathbb{P}_n\varphi_j(\symbf{O}; \widehat{\symbf{\eta}}, \symbf{\beta}) - \mathbb{P}\varphi_j(\symbf{O}; \symbf{\eta}_0, \symbf{\beta}) ||_2 \overset{p}{\rightarrow} 
0$, $j=1,2$.
\begin{align*}
    \mathbb{P}_n \symbf{\varphi}_j
    \big( \symbf{O} ; \widehat{\symbf{\eta}}, \symbf{\beta} \big) - \mathbb{P} \symbf{\varphi}_j \big( \symbf{O} ; \symbf{\eta}_0, \symbf{\beta} \big)
    = & \quad \underbrace{(\mathbb{P}_n  - \mathbb{P})\symbf{\varphi}_j \big( \symbf{O} ; \symbf{\eta}_0, \symbf{\beta}\big)}_{term1} \ \\
    & \quad + \ \underbrace{(\mathbb{P}_n  - \mathbb{P}) \left\{\symbf{\varphi}_j \big( \symbf{O} ; \widehat{\symbf{\eta}}, \symbf{\beta} \big) - \symbf{\varphi}_j \big( \symbf{O} ; \symbf{\eta}_0, \symbf{\beta} \big)\right\}}_{term2}\\
    & \quad + \underbrace{\mathbb{P}\symbf{\varphi}_j \big( \symbf{O} ; \widehat{\symbf{\eta}}, \symbf{\beta} \big) - \mathbb{P}\symbf{\varphi}_j \big( \symbf{O} ; \symbf{\eta}_0, \symbf{\beta} \big)}_{term3}.
\end{align*}

For term1, under Assumption \ref{assum:donsker1}, the class of functions $\left\{ \symbf{\varphi}_j \left(\symbf{O} ;  
   \symbf{\eta}_0, \symbf{\beta}\right):
  \symbf{\beta} \in \mathcal{B}^* \right\}$ 
  is also P-Donsker, and thus is P-Glivenko-Cantelli, then $\sup_{\symbf{\beta} \in \mathcal{B}^*}||(\mathbb{P}_n  - \mathbb{P})\symbf{\varphi}_j \big( \symbf{O} ; \symbf{\eta}_0, \symbf{\beta}\big)||_2 = o_P(1) $, $j=1,2$.

For term2, by Lemma \ref{lemma:L2P_continuous}, we know $\sup_{\symbf{\beta}\in \mathcal{B}^*}
   \mathbb{P} 
   ||\varphi_j(\symbf{O}; \widehat{\symbf{\eta}}, \symbf{\beta}) - \varphi_j(\symbf{O}; \symbf{\eta}_0, \symbf{\beta})||_2^2  = o_P(1)$, combined with Assumption \ref{assum:donsker1}, by Lemma 19.24 \citep{van2000asymptotic}, 
   $$\sup_{\symbf{\beta} \in \mathcal{B}^*} 
   ||(\mathbb{P}_n  - \mathbb{P}) \left\{\symbf{\varphi}_j \big( \symbf{O} ; \widehat{\symbf{\eta}}, \symbf{\beta} \big) - \symbf{\varphi}_j \big( \symbf{O} ; \symbf{\eta}_0, \symbf{\beta} \big)\right\}||_2 = o_P(n^{-1/2}).$$
   
   For term3, by Lemma \ref{lemma:L2P_continuous} and Jensen's inequality
   \begin{align*}
      \sup_{\symbf{\beta} \in \mathcal{B}^*} 
      || \mathbb{P}\symbf{\varphi}_j \big( \symbf{O} ; \widehat{\symbf{\eta}}, \symbf{\beta} \big) - \mathbb{P}\symbf{\varphi}_j \big( \symbf{O} ; \symbf{\eta}_0, \symbf{\beta} \big)||_2^2 
      \leq \sup_{\symbf{\beta}\in \mathcal{B}^*}
     \mathbb{P} 
      ||\varphi_j(\symbf{O}; \widehat{\symbf{\eta}}, \symbf{\beta}) - \varphi_j(\symbf{O}; \symbf{\eta}_0, \symbf{\beta})||_2^2 \overset{p}{\rightarrow} 0.
   \end{align*}

Thus, we have $\sup_{\symbf{\beta} \in \mathcal{B}^*} ||\mathbb{P}_n \symbf{\varphi}_j
\big( \symbf{O} ; \widehat{\symbf{\eta}}, \symbf{\beta} \big) - \mathbb{P} \symbf{\varphi}_j \big( \symbf{O} ; \symbf{\eta}_0, \symbf{\beta} \big)||_2^2  = o_P(1)$, $j=1,2$.

\textbf{Part 2: proof of the existence of} $\widehat{\symbf{\beta}}_j$ \textbf{ in a neighborhood of }$\symbf{\beta}_0$, $j=1,2$.

By Lemma \ref{lemma:unique}, $\symbf{\beta}_0 \in \mathcal{B}^*$ is the locally unique solution of the equation $\mathbb{P}\varphi_j(\symbf{O}; \symbf{\eta}_0, \symbf{\beta}) = 0$.  
Then there exists $\delta_j >0$ and $c_j>0$ such that 
\begin{align*}
    ||\mathbb{P}\varphi_j(\symbf{O}; \symbf{\eta}_0, \symbf{\beta})||_2  \geq c_j ||\symbf{\beta} - \symbf{\beta}_0||_2,\mbox{ for all }\symbf{\beta} \in B_{\delta_j}(\symbf{\beta}_0).
\end{align*}
where $B_{\delta_j}(\symbf{\beta}_0)$ is an open ball of $\symbf{\beta}_0$, $B_{\delta_j}(\symbf{\beta}_0) = \{\symbf{\beta}: ||\symbf{\beta} - \symbf{\beta}_0||_2 < \delta_j \} \subset \mathcal{B}^*$, $j=1,2$. 

By the uniform convergence of $\mathbb{P}_n\varphi_j(\symbf{O}; \widehat{\symbf{\eta}}, \symbf{\beta})$ to $\mathbb{P}\varphi_j(\symbf{O}; \symbf{\eta}_0, \symbf{\beta})$ on $B_{\delta_j}(\symbf{\beta}_0) \subset \mathcal{B}^*$, for any $\epsilon_j < c_j\delta_j/2$, 
\begin{align*}
    \mathbb{P}\big( \sup_{\symbf{\beta}\in \mathcal{B}_{\delta_j} (\symbf{\beta}_0)} ||\mathbb{P}_n\varphi_j(\symbf{O}; \widehat{\symbf{\eta}}, \symbf{\beta}) - \mathbb{P}\varphi_j(\symbf{O}; \symbf{\eta}_0, \symbf{\beta}) ||_2 < \epsilon_j \big) \rightarrow 1.
\end{align*}

On this event, for all $\symbf{\beta}$ with $||\symbf{\beta} - \symbf{\beta}_0||_2 = \delta_j$, 
\begin{align*}
    & ||\mathbb{P}_n\varphi_j(\symbf{O}; \widehat{\symbf{\eta}}, \symbf{\beta})||_2 \geq ||\mathbb{P}\varphi_j(\symbf{O}; \symbf{\eta}_0, \symbf{\beta})||_2 - ||\mathbb{P}_n\varphi_j(\symbf{O}; \widehat{\symbf{\eta}}, \symbf{\beta}) - \mathbb{P}\varphi_j(\symbf{O}; \symbf{\eta}_0, \symbf{\beta})||_2 \\
    & \geq c_j\delta_j - \epsilon_j > c_j\delta_j/2 > 0.
\end{align*}

At the same time, $\mathbb{P}_n\varphi_j(\symbf{O}; \widehat{\symbf{\eta}}, \symbf{\beta}_0) \rightarrow \mathbb{P}\varphi_j(\symbf{O}; \symbf{\eta}_0, \symbf{\beta}_0) =0 $, so for large $n$, we have $||\mathbb{P}_n\varphi_j(\symbf{O}; \widehat{\symbf{\eta}}, \symbf{\beta}_0)||_2$ arbitrarily small with high probability. Under Assumption \ref{chp4:assum:marginal}, $\mathbb{P}_n\varphi_j(\symbf{O}; \widehat{\symbf{\eta}}, \symbf{\beta})$ is continuous on $\mathcal{B}_{\delta_j}(\symbf{\beta}_0)$. Thus, for large $n$, $\mathbb{P}_n\varphi_j(\symbf{O}; \widehat{\symbf{\eta}}, \symbf{\beta})$ is close to zero at $\symbf{\beta}_0$, and is bounded away from zero on the boundary for all $\symbf{\beta}$ with $||\symbf{\beta} - \symbf{\beta}_0||_2 = \delta_j$. A multivariate intermediate value theorem implies there exists at least one $\widehat{\symbf{\beta}}_j \in \mathcal{B}_{\delta_j}(\symbf{\beta}_0)$ with $\mathbb{P}_n\varphi_j(\symbf{O}; \widehat{\symbf{\eta}}, \widehat{\symbf{\beta}}_j) = 0$, $j=1,2$.

\textbf{Part 3: proof of the consistency.}

Since $\widehat{\symbf{\beta}}_j \in \mathcal{B}_{\delta_j}(\symbf{\beta}_0)$ with $\mathbb{P}_n\varphi_j(\symbf{O}; \widehat{\symbf{\eta}}, \widehat{\symbf{\beta}}_j) = 0$, $j=1,2$, the quantitative inequality gives:
\begin{align*}
    c_j ||\widehat{\symbf{\beta}}_j - \symbf{\beta}_0||_2 & \ \leq ||\mathbb{P}\varphi_j(\symbf{O}; \symbf{\eta}_0, \widehat{\symbf{\beta}}_j )||_2 \\
    & \ = ||\mathbb{P}\varphi_j(\symbf{O}; \symbf{\eta}_0, \widehat{\symbf{\beta}}_j) - \mathbb{P}_n\varphi_j(\symbf{O}; \widehat{\symbf{\eta}}, \widehat{\symbf{\beta}}_j) ||_2 \\
    & \ = ||\mathbb{P}_n\varphi_j(\symbf{O}; \widehat{\symbf{\eta}}, \widehat{\symbf{\beta}}_j) - \mathbb{P}\varphi_j(\symbf{O}; \symbf{\eta}_0, \widehat{\symbf{\beta}}_j) ||_2 \\
    & \leq \sup_{\symbf{\beta}\in \mathcal{B}^*} || \mathbb{P}_n\varphi_j(\symbf{O}; \widehat{\symbf{\eta}}, \symbf{\beta}) - \mathbb{P}\varphi_j(\symbf{O}; \symbf{\eta}_0, \symbf{\beta}) ||_2 \overset{p}{\rightarrow} 0.
\end{align*}
Thus, we know that $||\widehat{\symbf{\beta}}_j - \symbf{\beta}_0||_2 \overset{p}{\rightarrow} 0$, which means $\widehat{\symbf{\beta}}_j$ is a consistent estimator for $\symbf{\beta}_0$.
\end{proof}

\subsubsection{Proof of Theorem \ref{theorem:asy_normal}}
\label{app:proof_asy}

\begin{proof}

\textbf{Part 1: Focus on the term }$\mathbb{P}_n \symbf{\varphi}_j \big( \symbf{O} ; \widehat{\symbf{\eta}}, \symbf{\beta}_0 \big)$.

Based on Lemma \ref{lemma:convg3}, we have
\begin{align*}
    \mathbb{P}_n \symbf{\varphi}_j \big( \symbf{O} ; \widehat{\symbf{\eta}}, \symbf{\beta}_0 \big) & = \mathbb{P}_n\symbf{\varphi}_j \big(\symbf{O} ; \widehat{\symbf{\eta}}, \symbf{\beta}_0\big) - \mathbb{P}\symbf{\varphi}_j \big(\symbf{O} ; \symbf{\eta}_0, \symbf{\beta}_0\big) \\
    & = n^{-1/2} \mathcal{N}(\symbf{0}, \Sigma_j) + o_P(n^{-1/2}),
\end{align*}
where $\Sigma_j = E \big\{\symbf{\varphi}_j \left( \symbf{O}; \symbf{\eta}_0, \symbf{\beta}_0 
\right)  \symbf{\varphi}_j \left( \symbf{O}; \symbf{\eta}_0, \symbf{\beta}_0
\right) ^\prime \big\}$, for $j=1,2$.

\textbf{Part 2: Application of the mean value theorem}

Under Assumption \ref{chp4:assum:marginal},
$\mathbb{P}_n \symbf{\varphi}_j \big( \symbf{O} ; \widehat{\symbf{\eta}}, \symbf{\beta} \big)$ is continuous in $\symbf{\beta}$, then we can apply the mean value theorem to $ \mathbb{P}_n \symbf{\varphi}_j \big( \symbf{O} ; \widehat{\symbf{\eta}}, \symbf{\beta} \big)$ for $\symbf{\beta} \in \mathcal{B}^*$,
\begin{align*}
    \sqrt{n}\left\{
    \mathbb{P}_n \symbf{\varphi}_j \big( \symbf{O} ; \widehat{\symbf{\eta}}, \widehat{\symbf{\beta}}_j \big) 
    - 
    \mathbb{P}_n \symbf{\varphi}_j \big( \symbf{O} ; \widehat{\symbf{\eta}}, \symbf{\beta}_0 \big)
    \right\} = \sqrt{n} \left\{
    \frac{\partial \mathbb{P}_n \symbf{\varphi}_j ( \symbf{O} ; \widehat{\symbf{\eta}}, \symbf{\beta} )}{\partial \symbf{\beta}} \Big|_{\symbf{\beta} = \widetilde{\symbf{\beta}}}\right\} \big(\widehat{\symbf{\beta}}_j - \symbf{\beta}_0\big), 
\end{align*}
and 
\begin{align*}
    \sqrt{n}\big(\widehat{\symbf{\beta}} - \symbf{\beta}_0\big) 
    = & \ 
    \left\{
    \frac{\partial \mathbb{P}_n \symbf{\varphi}_j ( \symbf{O} ; \widehat{\symbf{\eta}}, \symbf{\beta} )}{\partial \symbf{\beta}} \Big|_{\symbf{\beta} = \widetilde{\symbf{\beta}}}\right\}^{-1}
    \sqrt{n}
    \left\{
    \mathbb{P}_n \symbf{\varphi}_j \big( \symbf{O} ; \widehat{\symbf{\eta}}_j, \widehat{\symbf{\beta}} \big) 
    - 
    \mathbb{P}_n \symbf{\varphi}_j \big( \symbf{O} ; \widehat{\symbf{\eta}}, \symbf{\beta}_0 \big)
    \right\}\\
    = & \ 
    - 
     \left\{\frac{\partial \mathbb{P}_n \symbf{\varphi}_j ( \symbf{O} ; \widehat{\symbf{\eta}}, \symbf{\beta} )}{\partial \symbf{\beta}} \Big|_{\symbf{\beta} = \widetilde{\symbf{\beta}}}\right\}^{-1}
     \sqrt{n} 
     \mathbb{P}_n \symbf{\varphi}_j \big( \symbf{O} ; \widehat{\symbf{\eta}}, \symbf{\beta}_0 \big)
\end{align*}
where $\widetilde{\symbf{\beta}}$ takes some value between 
$\widehat{\symbf{\beta}}_j$ and $\symbf{\beta}_0$ and $\widetilde{\symbf{\beta}} \in \mathcal{B}^*$.

For the first part, based on Part 1, we have
\begin{align*}
     \sqrt{n}\mathbb{P}_n \symbf{\varphi}_j \big( \symbf{O} ; \widehat{\symbf{\eta}}, \symbf{\beta}_0 \big) 
    =
    \mathcal{N}(\symbf{0}, \symbf{\Sigma}_j) + o_P(1),
\end{align*}
where 
\begin{align*}
    \Sigma_j = \mathbb{E} \left\{\symbf{\varphi}_j \left( \symbf{O}; \symbf{\eta}_0, \symbf{\beta}_0\right)  \symbf{\varphi}_j \left( \symbf{O}; \symbf{\eta}_0, \symbf{\beta}_0\right) ^\prime \right\}
\end{align*}

For the second part, 

\begin{align*}
   & \ \frac{\partial \mathbb{P}_n \symbf{\varphi}_j \big( \symbf{O}; \widehat{\symbf{\eta}}, \symbf{\beta} \big)}{\partial \symbf{\beta}} \Big|_{\symbf{\beta} = \widetilde{\symbf{\beta}}} - 
    \mathbb{P} \frac{\partial \symbf{\varphi}_j \big( \symbf{O} ; \symbf{\eta}_0, \symbf{\beta}_0 \big)}{\partial \symbf{\beta}}\\
= & \ \mathbb{P}_n \frac{\partial \symbf{\varphi}_j \big( \symbf{O}; \widehat{\symbf{\eta}}, \symbf{\beta} \big)}{\partial \symbf{\beta}} \Big|_{\symbf{\beta} = \widetilde{\symbf{\beta}}} - 
    \mathbb{P} \frac{\partial \symbf{\varphi}_j \big( \symbf{O} ; \symbf{\eta}_0, \symbf{\beta}_0 \big)}{\partial \symbf{\beta}}\\
   = &\  \underbrace{
   \left\{ \mathbb{P}_n 
    \frac{\partial  \symbf{\varphi}_j \big( \symbf{O}; \widehat{\symbf{\eta}}, \symbf{\beta} \big)}{\partial \symbf{\beta}} \Big|_{\symbf{\beta} = \widetilde{\symbf{\beta}}} 
    - \mathbb{P} \frac{\partial \symbf{\varphi}_j \big( \symbf{O} ; \symbf{\eta}_0, \widetilde{\symbf{\beta}} \big)}{\partial \symbf{\beta}}\right\} 
     }_{term1} 
     - 
    \underbrace{
    \left\{
    \mathbb{P} \frac{\partial \symbf{\varphi}_j \big( \symbf{O} ; \symbf{\eta}_0, \widetilde{\symbf{\beta}} \big)}{\partial \symbf{\beta}} 
    -
    \mathbb{P} \frac{\partial \symbf{\varphi}_j \big( \symbf{O} ; \symbf{\eta}_0, \symbf{\beta}_0 \big)}{\partial \symbf{\beta}}
    \right\}
    }_{term2} 
\end{align*}
where each element of term1 is proved to be $o_P(1)$ by Lemma \ref{lemma:convg4}. Since $\widetilde{\symbf{\beta}}$ takes some value between 
$\widehat{\symbf{\beta}}_j$ and $\symbf{\beta}_0$ and $\widehat{\symbf{\beta}}_j \overset{p}{\rightarrow} \symbf{\beta}_0$, we can prove term2 is $o_P(1)$ by the continuous mapping theorem.

Thus, we have 
\begin{align*}
    & \mathbb{P}_n \frac{\partial  \symbf{\varphi}_j \big( \symbf{O} ; \widehat{\symbf{\eta}}, \symbf{\beta} \big)}{\partial \symbf{\beta}} \Big|_{\symbf{\beta} = \widetilde{\symbf{\beta}}} \overset{p}{\rightarrow} \mathbb{P} \frac{\partial \symbf{\varphi}_j \big( \symbf{O} ; \symbf{\eta}_0, \symbf{\beta}_0 \big)}{\partial \symbf{\beta}},\\
    \Rightarrow & \left\{ \mathbb{P}_n \frac{\partial  \symbf{\varphi}_j \big( \symbf{O} ; \widehat{\symbf{\eta}}, \symbf{\beta} \big)}{\partial \symbf{\beta}} \Big|_{\symbf{\beta} = \widetilde{\symbf{\beta}}}\right\}^{-1}
     \overset{p}{\rightarrow} 
     \left\{\mathbb{P} \frac{\partial \symbf{\varphi}_j \big( \symbf{O} ; \symbf{\eta}_0, \symbf{\beta}_0 \big)}{\partial \symbf{\beta}}\right\}^{-1},
\end{align*}
where  
\begin{align*}
    \frac{\partial \symbf{\varphi}_j \big( \symbf{O} ; \symbf{\eta}_0, \symbf{\beta}_0 \big)}{\partial \symbf{\beta}} = & \ 
    \int_{t=t_1}^{t_2} \frac{\partial  \symbf{W}_j(t;\symbf{\beta}_0)}{\partial \symbf{\beta}}\frac{
    \left[ Y(t) - \mathbb{E}\left\{Y(t)\mid \overline{\symbf{O}}(t)\right\}\right]}{\rho\left\{t \mid \overline{\symbf{O}}(t), Y(t)\right\}} dN(t) \\
    & \quad + \int_{t=t_1}^{t_2} \frac{\partial  \symbf{W}_j(t;\symbf{\beta}_0)}{\partial \symbf{\beta}}\left[ \mathbb{E}\left\{Y(t)\mid \overline{\symbf{O}}(t)\right\} -  s\left\{ \symbf{B}(t)^\prime \symbf{\beta}_0\right\}\right]dt \\
    & \quad -
    \int_{t=t_1}^{t_2} \symbf{W}_j(t;\symbf{\beta}_0)\frac{\partial s\left\{ \symbf{B}(t)^\prime \symbf{\beta}_0\right\}}{\partial \symbf{\beta}}dt.
\end{align*}

By Slutsky's theorem, we have
\begin{align*}
    \sqrt{n}\big(\widehat{\symbf{\beta}}_j - \symbf{\beta}_0\big)  
    \overset{d}{\rightarrow}
    \mathcal{N}(\symbf{0},  \Sigma_j^*),
\end{align*}
where 
\begin{align*}
     \Sigma_j^* = 
     \left[\mathbb{P} \frac{\partial \symbf{\varphi}_j \big( \symbf{O} ; \symbf{\eta}_0, \symbf{\beta}_0 \big)}{\partial \symbf{\beta}}\right]^{-1}
     \mathbb{E} \left\{\symbf{\varphi}_j \left( \symbf{O}; \symbf{\eta}_0, \symbf{\beta}_0\right)  \symbf{\varphi}_j \left( \symbf{O}; \symbf{\eta}_0, \symbf{\beta}_0\right) ^\prime \right\} \left[\mathbb{P} \frac{\partial \symbf{\varphi}_j \big( \symbf{O} ; \symbf{\eta}_0, \symbf{\beta}_0 \big)}{\partial \symbf{\beta}}\right]^{-1, \prime}, \ j=1,2.
\end{align*}
\end{proof}

\subsection{Proof of Theorem \ref{chp4:theorem:mono}}
\label{apxB:app:mono_proof}

\begin{proof}
    In the proof, we suppress dependence on the treatment arm $a$ for notational simplicity. 
    Define:
    \begin{align*}
        A(\alpha, t) & = \mathbb{E}\big[\exp\{\alpha Y(t)\} \mid \Delta N(t) = 1, \overline{\symbf{O}}(t) \big]\\
        B(\alpha, t) & = \mathbb{E}\big[Y(t)\exp\{\alpha Y(t)\} \mid \Delta N(t) = 1, \overline{\symbf{O}}(t) \big], \\
        C(\alpha, t) &= \mathbb{E}[Y^2(t)\exp\{\alpha Y(t)\}| \Delta N(t) = 1, \overline{\symbf{O}}(t)]
    \end{align*}
    Under Assumption \ref{chp4:assum:tilting}, given Proposition \ref{prop:mu_identity}, we have
    \begin{align*}
        \mathbb{E}\{Y(t);\alpha\} 
        = \mathbb{E} \left\{\frac{B(\alpha, t)}{A(\alpha, t)} \right\}.
    \end{align*}
    Under Assumption \ref{assum:bound_Y}, by the Leibniz integral rule, we have
    \begin{align*}
        \frac{d \mathbb{E}\{Y(t);\alpha \}}{ d \alpha} 
        & = \mathbb{E} \left[\frac{d \{B(\alpha, t) / A(\alpha, t)\} }{d \alpha}\right]\\
        & = \mathbb{E}\left\{\frac{A(\alpha,t) C(\alpha,t) - B^2(\alpha,t)}{A^2(\alpha,t)}\right\}
    \end{align*}
 By the Cauchy-Schwarz inequality, we have $ A(\alpha, t) C(\alpha, t) \geq B^2(\alpha, t)$, then $ \frac{d \mathbb{E}\{Y(t);\alpha \}}{ d \alpha} \geq 0$.
 Thus, we have proved that $\mathbb{E}\{Y(t);\alpha\}$ is non-decreasing in $\alpha$.
\end{proof}



\section{Supplementary Data Analysis}
\label{apxB:appendix:data_analysis}

This appendix presents additional data analysis results obtained using the second weight function, $\symbf{W}_2(t;\symbf{\beta})$. Overall, the results are qualitatively similar to those obtained using the first weight function, $\symbf{W}_1(t;\symbf{\beta})$.

For the ACQ score, Figure \ref{fig:mean_curves_log_weight2} shows  the estimated mean outcome curves,  $\mu^a(t)$, under the log-link model with $\symbf{W}_2(t;\symbf{\beta})$ for each treatment arm over the interval $t\in[60,460]$ and under selected values of the sensitivity parameters. Point estimates and corresponding confidence intervals for $\mu^a(t)$ at the target follow-up times are also presented. Figure \ref{fig:contour_log_weight2} displays the estimated treatment effect $\delta(t)$ and its confidence intervals at 6 and 12 months over a range of sensitivity parameters under the log-link model with $\symbf{W}_2(t;\symbf{\beta})$.

\begin{figure}[ht]
\centering
\begin{subfigure}[b]{0.43\textwidth}
    \centering
    \includegraphics[width=\textwidth]{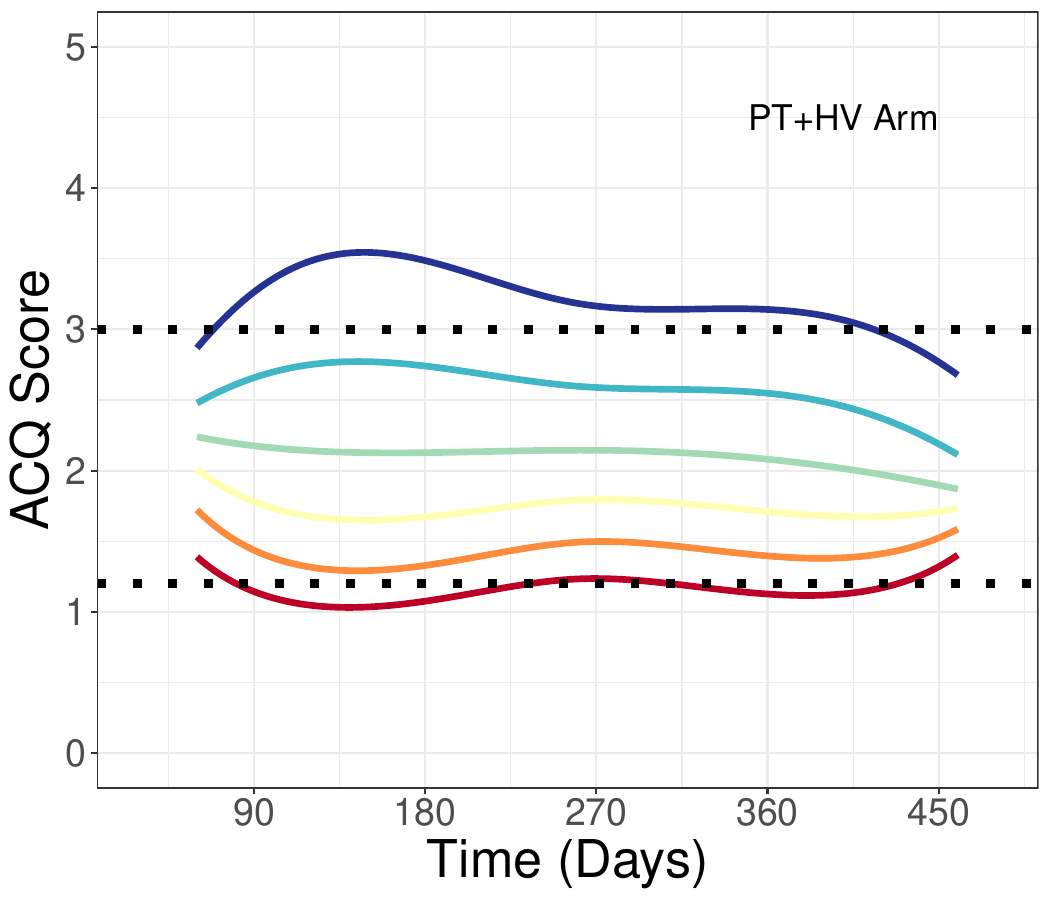}
\end{subfigure}
\hfill
\begin{subfigure}[b]{0.49\textwidth}  
    \centering 
    \includegraphics[width=\textwidth]{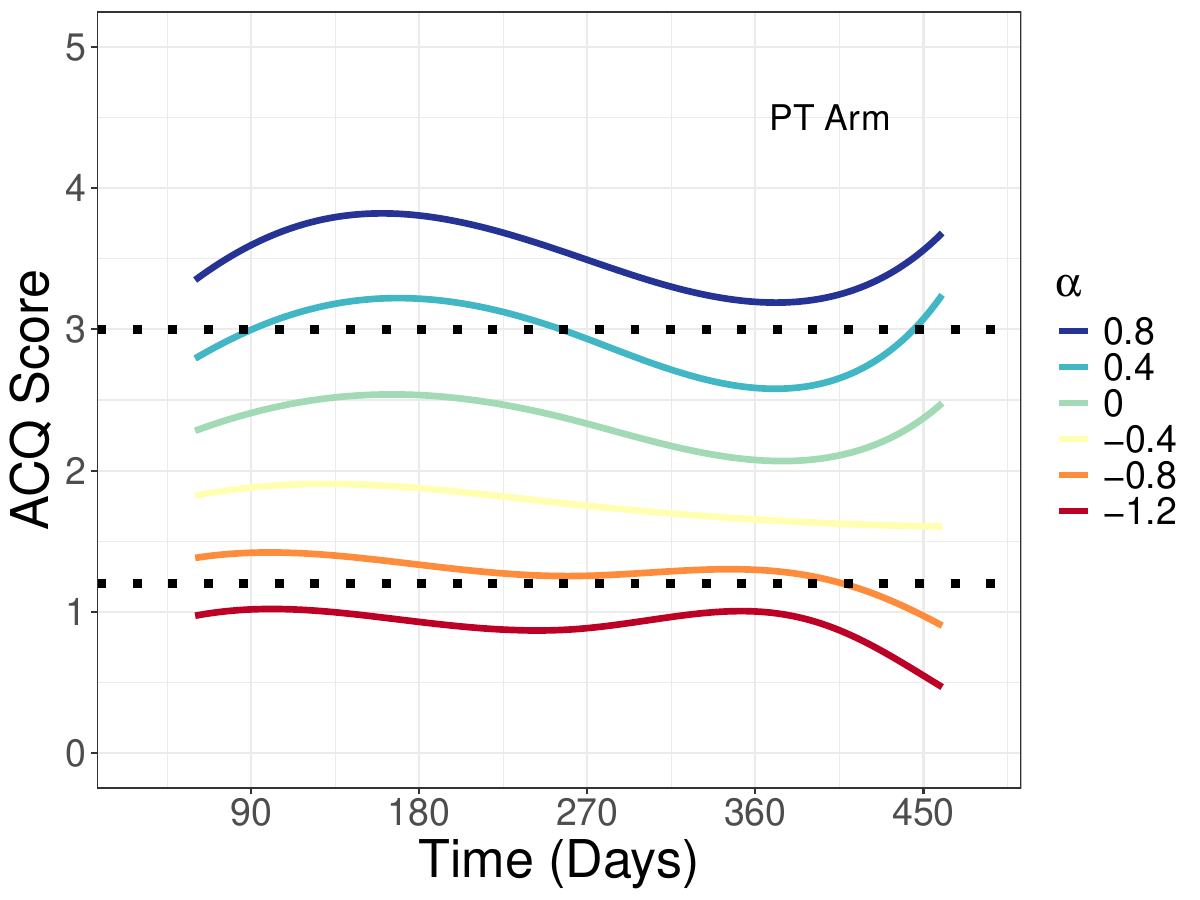}
\end{subfigure}

\medskip

\begin{subfigure}[b]{0.43\textwidth}   
    \centering 
    \includegraphics[width=\textwidth]{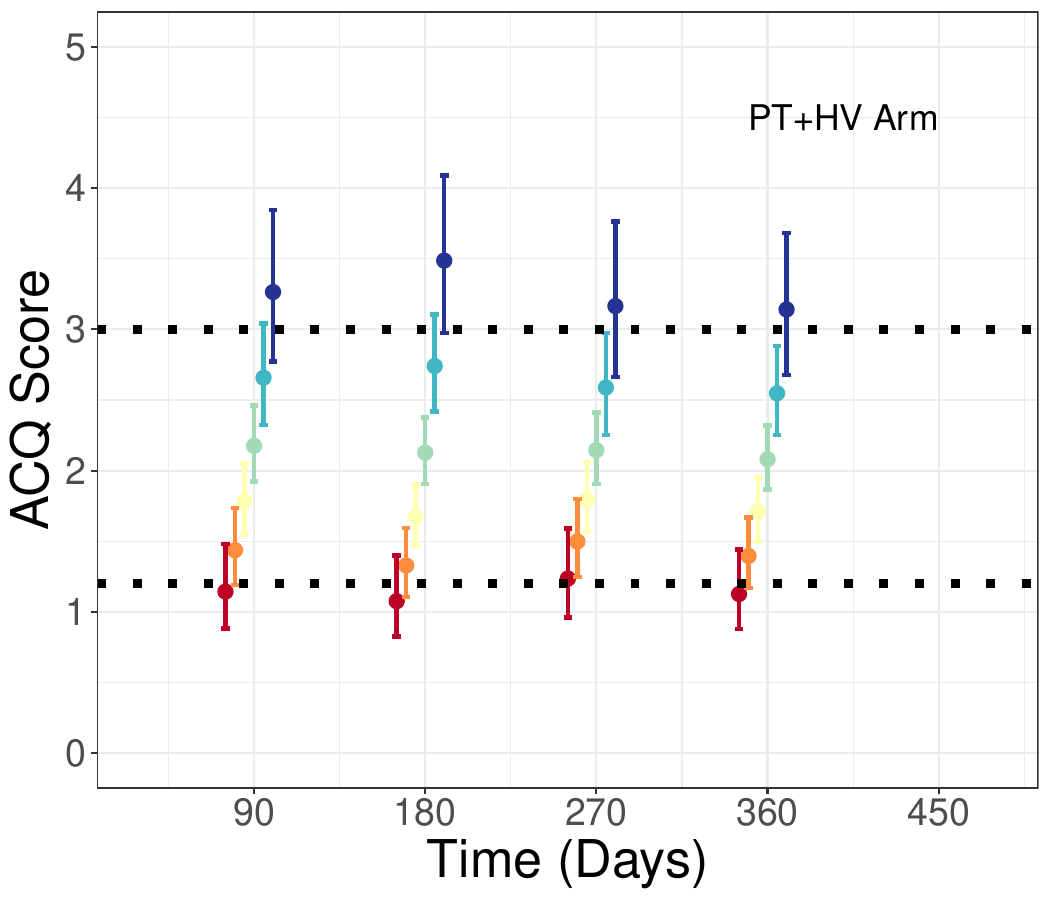}
   \end{subfigure}
\hfill
\begin{subfigure}[b]{0.49\textwidth}   
    \centering 
    \includegraphics[width=\textwidth]{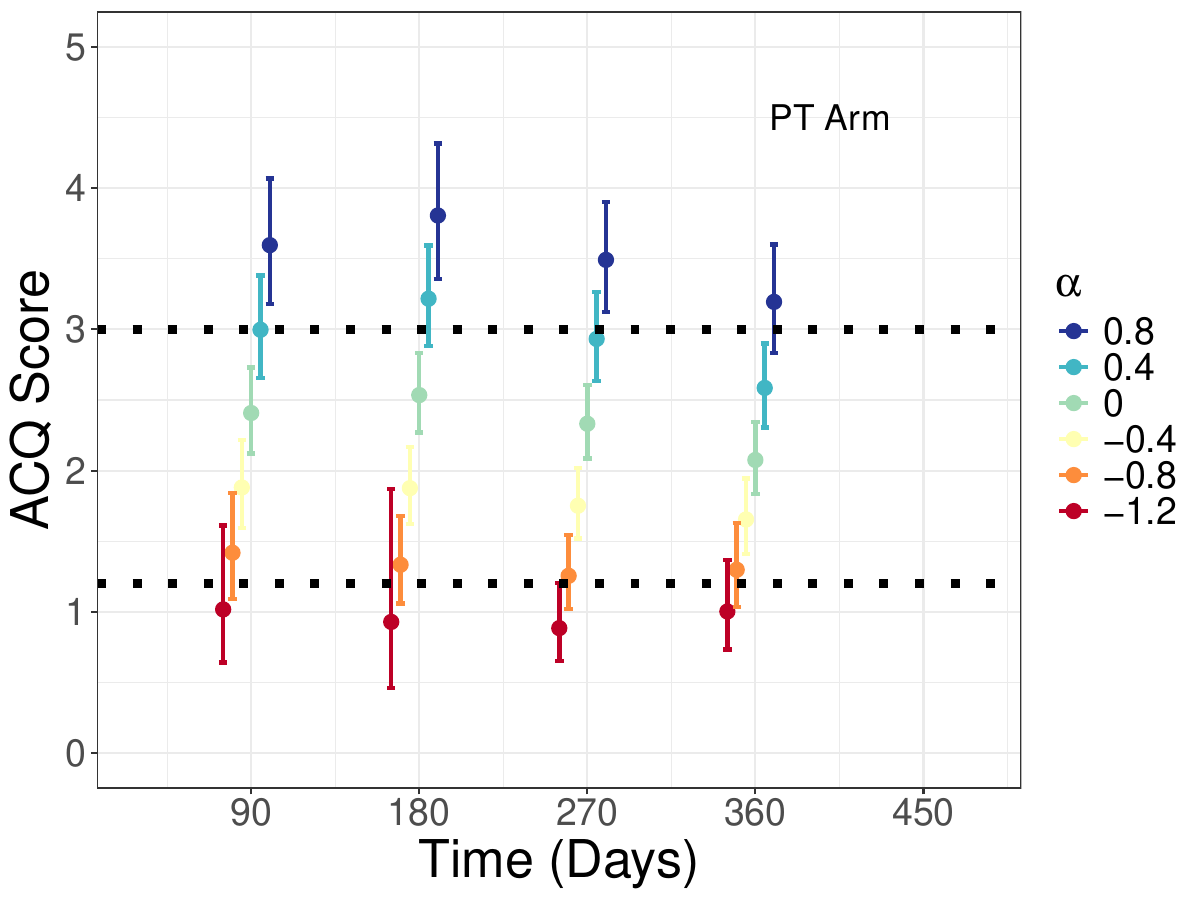}
\end{subfigure}

\caption[Population means of the Asthma Control Questionnaire (ACQ) score in the ARC trial under a range of sensitivity parameter values (log link with weight $\symbf{W}_2(t;\symbf{\beta})$).]{Population means of the Asthma Control Questionnaire (ACQ) score in the ARC trial under a range of sensitivity parameter values (log link with weight $\symbf{W}_2(t;\symbf{\beta})$). Here we show inference for the population mean ACQ score in the intervention (PT+HV) arm and the control (PT) arm.  Estimation is made under sensitivity parameter values of $\alpha = -1.2, -0.8, -0.4, 0, 0.4, 0.8$, where a positive (negative) value of $\alpha$ posits that unobserved outcome values at time $t$ tend to be higher (lower) than observed values, after controlling for variables observed before time $t$. Upper panels: point estimates on the interval 60 to 460 days after randomization.  Lower panels: point estimates and 95\% Wald confidence intervals using the parametric bootstrap variance estimate at each target time of 90, 180, 270, and 360 days. For each arm, only those values of $\alpha$ under which the estimated curves lie completely between the dotted lines at lower bound 1.2 and upper bound 3 are considered plausible based on subject-matter expertise.} 
\label{fig:mean_curves_log_weight2}
\end{figure}

\begin{figure}[ht]
\centering
\begin{subfigure}[b]{0.42\textwidth}
    \centering
    \includegraphics[width=\textwidth]{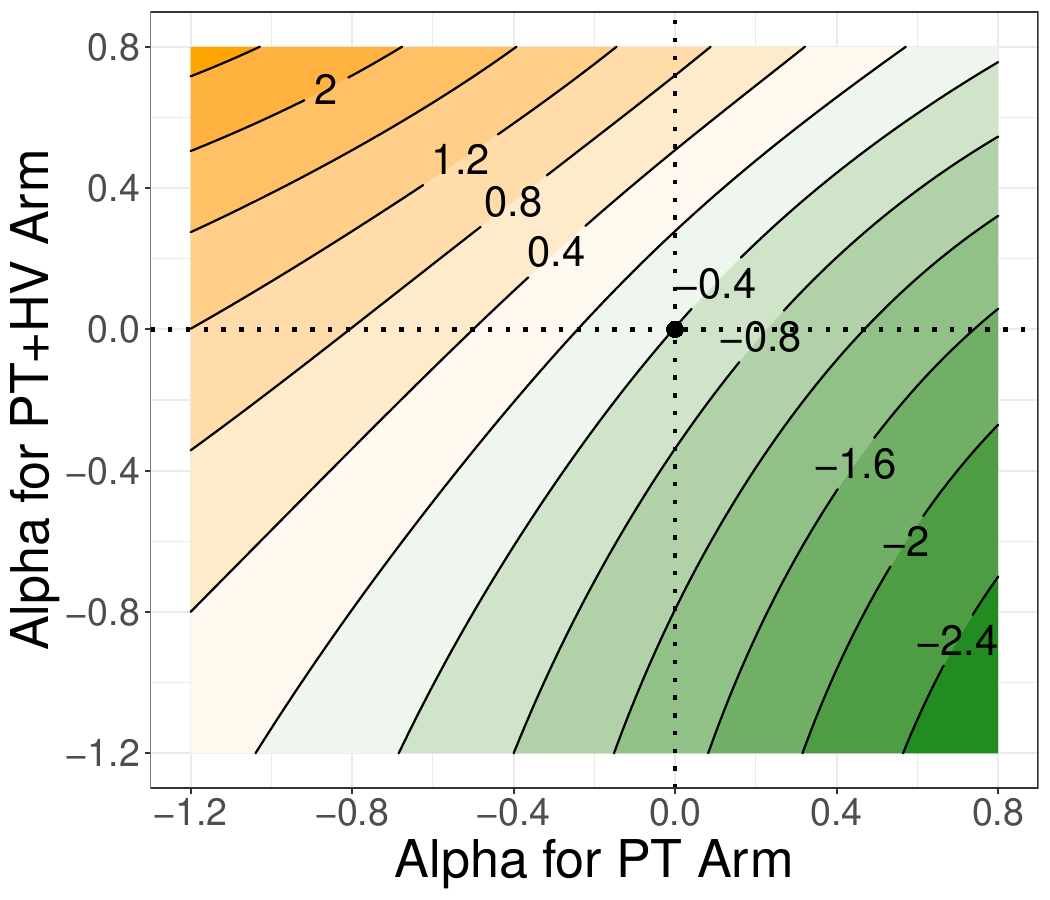}
\end{subfigure}
\hfill
\begin{subfigure}[b]{0.48\textwidth}  
    \centering 
    \includegraphics[width=\textwidth]{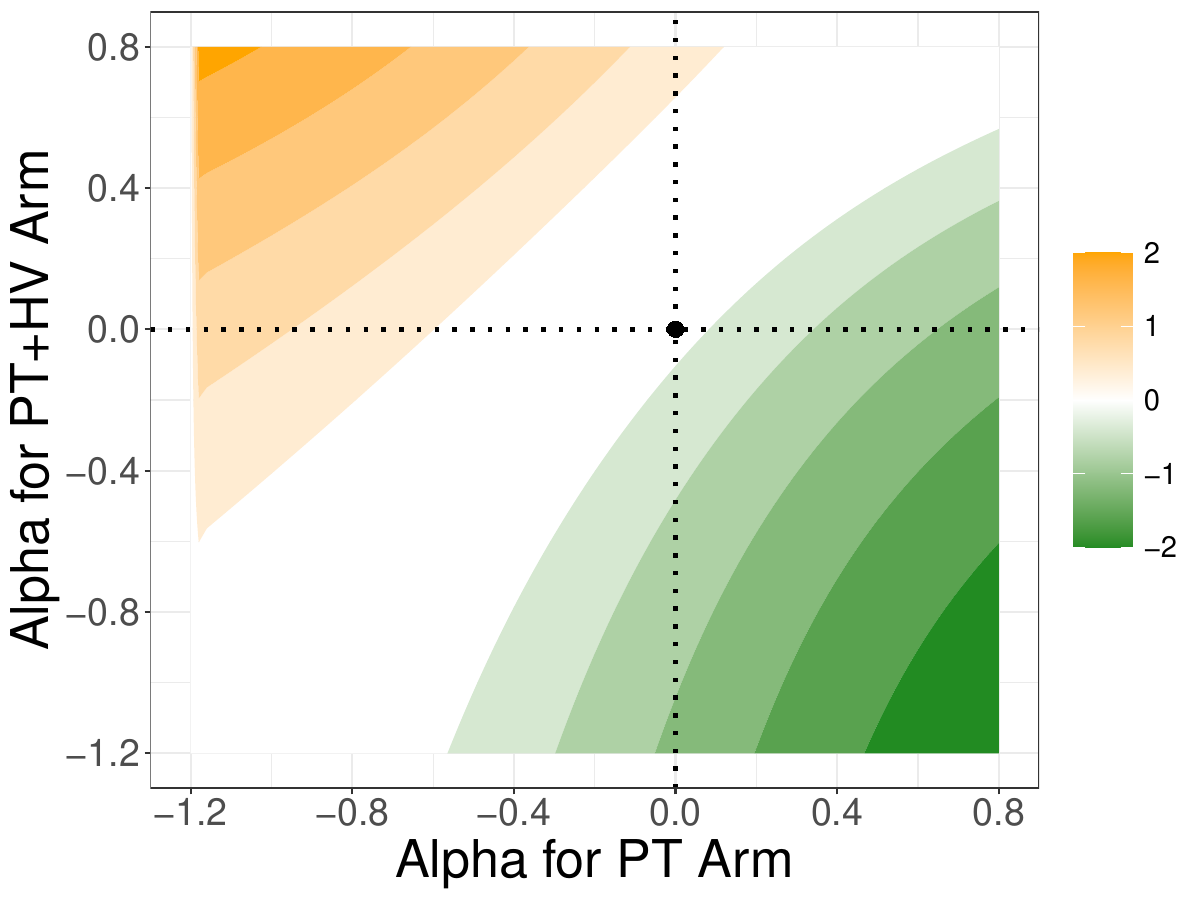}
\end{subfigure}
\medskip
\begin{subfigure}[b]{0.42\textwidth}   
    \centering 
    \includegraphics[width=\textwidth]{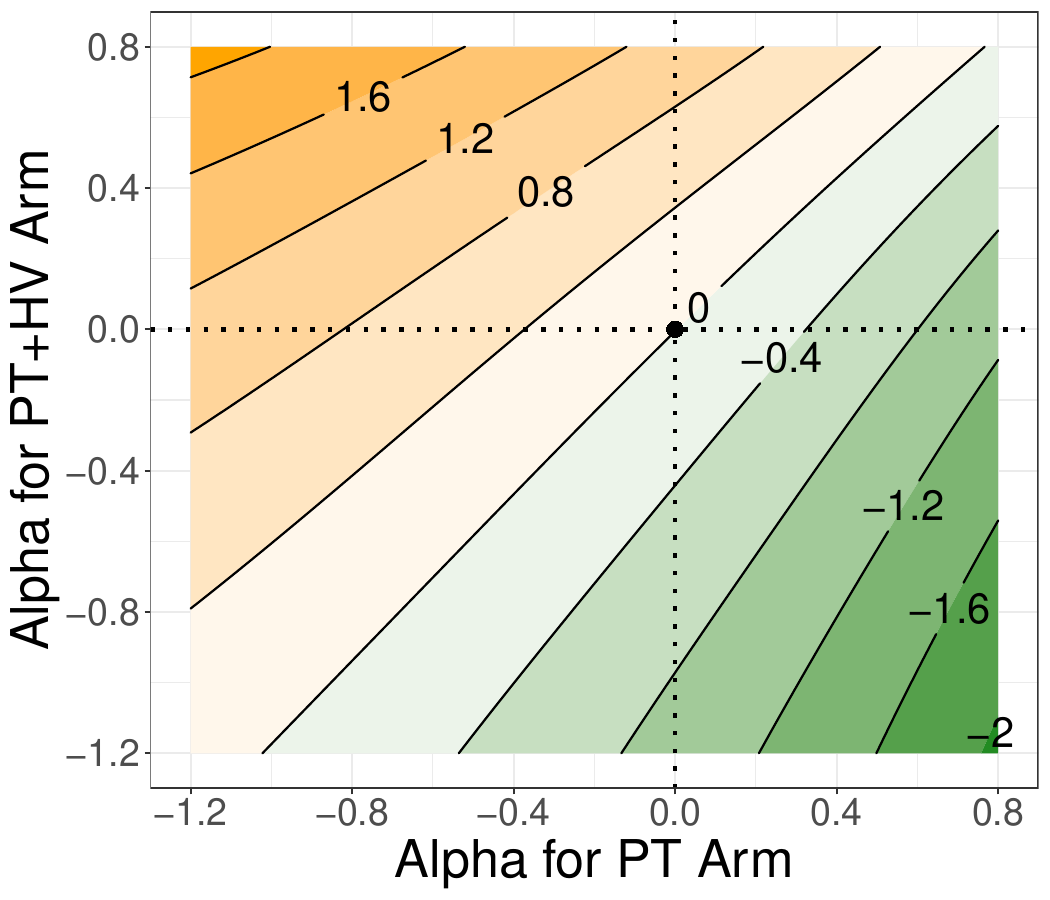}
   \end{subfigure}
\hfill
\begin{subfigure}[b]{0.48\textwidth}   
    \centering 
    \includegraphics[width=\textwidth]{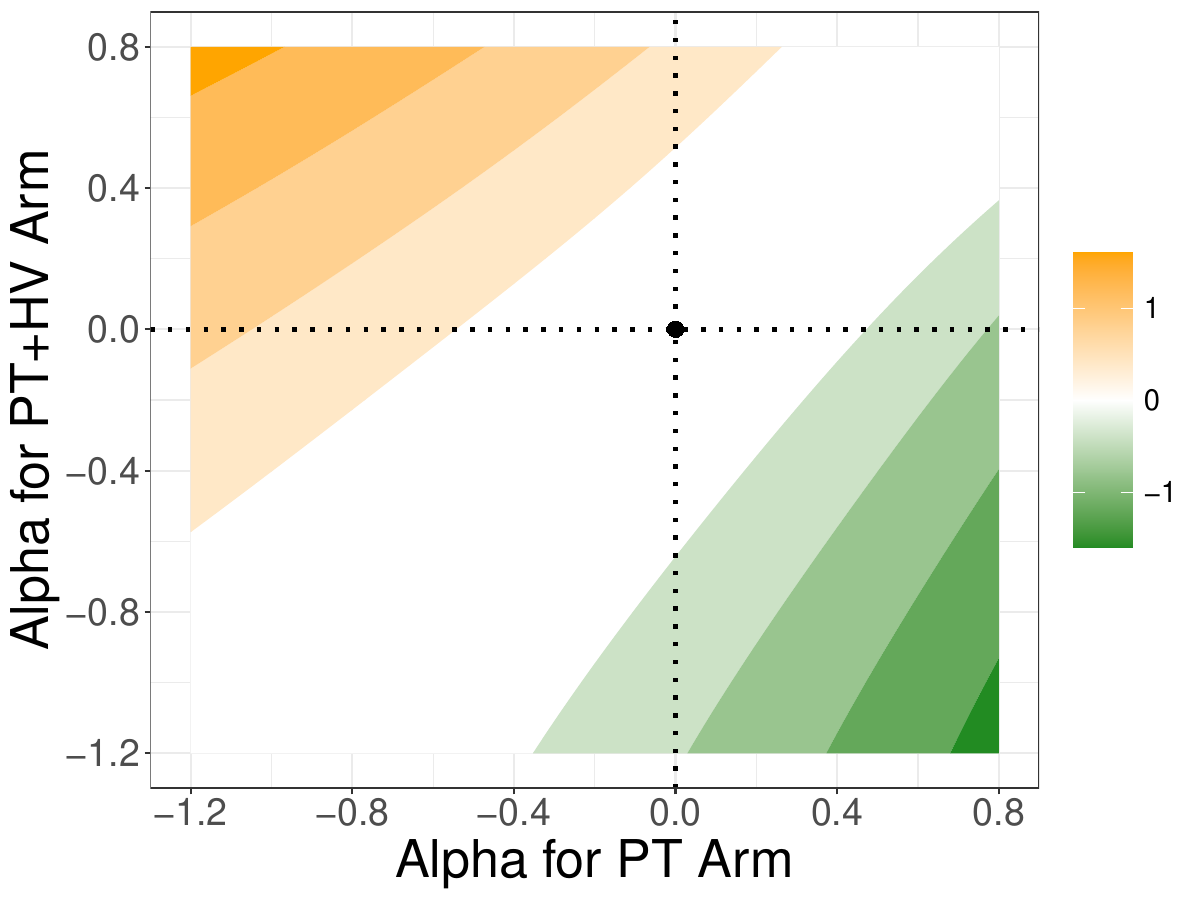}
\end{subfigure}

\caption[Sensitivity analysis for the Asthma Control Questionnaire (ACQ) score in the ARC trial (log link with weight $\symbf{W}_2(t;\symbf{\beta})$).]{Sensitivity analysis for the Asthma Control Questionnaire (ACQ) score in the ARC trial (log link with weight $\symbf{W}_2(t;\symbf{\beta})$). The left panels show the point estimate of the treatment effect $\delta(t)$ at 6 months (upper left) and 12 months (lower left) for each pair of sensitivity parameter values. The right panels summarize the 95\% confidence intervals for $\delta(t)$ at 6 months (upper right) and 12 months (lower right). The white region corresponds to sensitivity parameter values for which the confidence interval contains 0; the shaded region in the lower right (upper left) corresponds to values for which the confidence interval lies entirely below (above) 0. Contour heights give the value in the confidence interval closest to 0. Confidence intervals are Wald intervals based on the parametric bootstrap variance estimate.}
\label{fig:contour_log_weight2}
\end{figure}

For the asthma control indicator, Figure \ref{fig:mean_curves_logit_weight2}  displays the estimated mean response curves, $\mu^a(t)$, under the logit-link model with $\symbf{W}_2(t;\symbf{\beta})$ for each treatment arm over the interval $t\in[60,460]$ and under selected values of the sensitivity parameters. Point estimates and corresponding confidence intervals for $\mu^a(t)$ at the target follow-up times are also presented. Figure \ref{fig:contour_logit_weight2} displays the estimated treatment effect $\delta(t)$ and its confidence intervals at 6 and 12 months over a range of values of the sensitivity parameters under the logit-link model with $\symbf{W}_2(t;\symbf{\beta})$.

\begin{figure}[ht]
\centering
\begin{subfigure}[b]{0.43\textwidth}
    \centering
    \includegraphics[width=\textwidth]{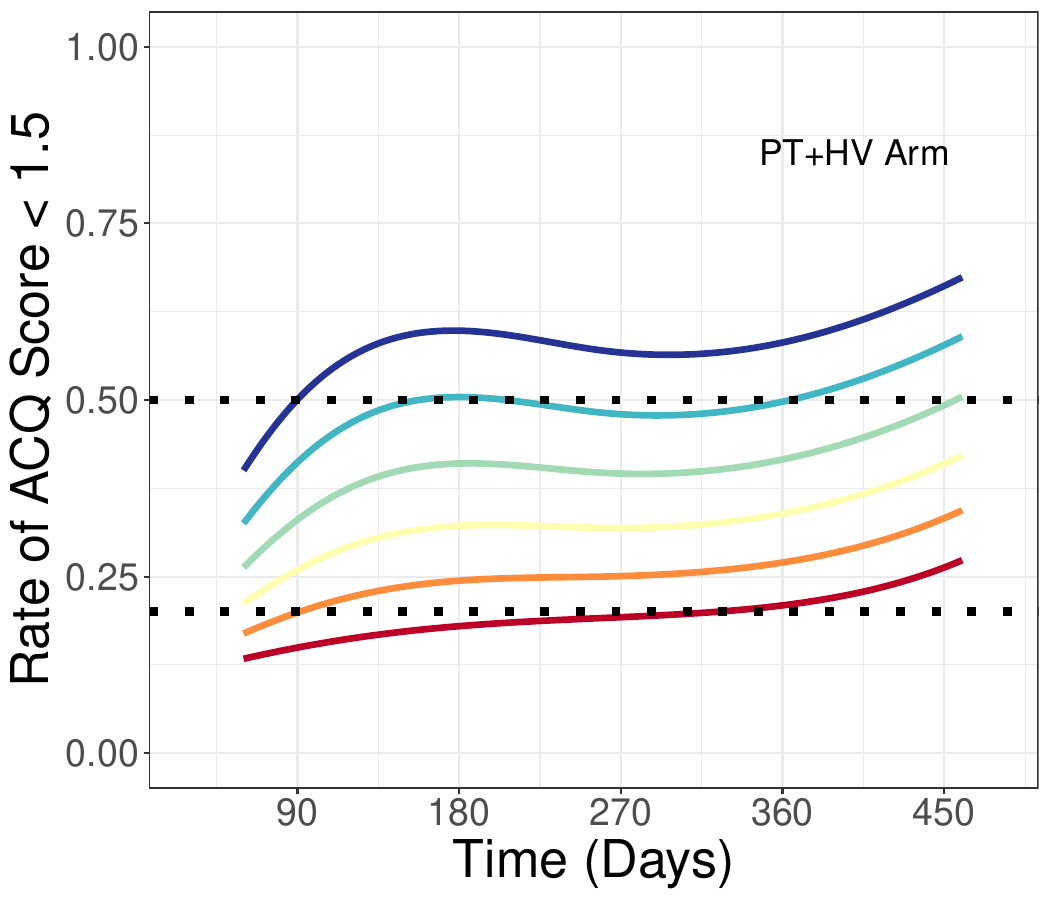}
\end{subfigure}
\hfill
\begin{subfigure}[b]{0.49\textwidth}  
    \centering 
    \includegraphics[width=\textwidth]{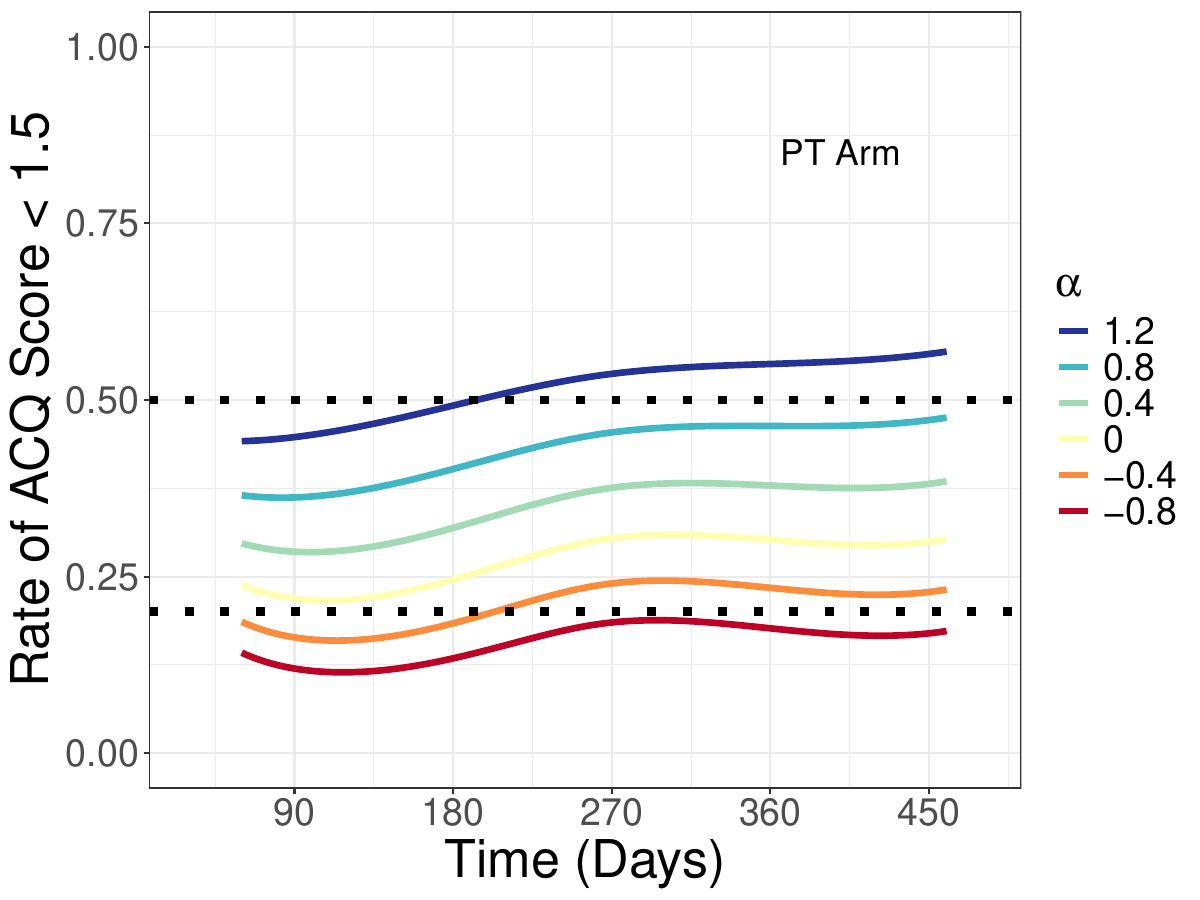}
\end{subfigure}

\medskip

\begin{subfigure}[b]{0.43\textwidth}   
    \centering 
    \includegraphics[width=\textwidth]{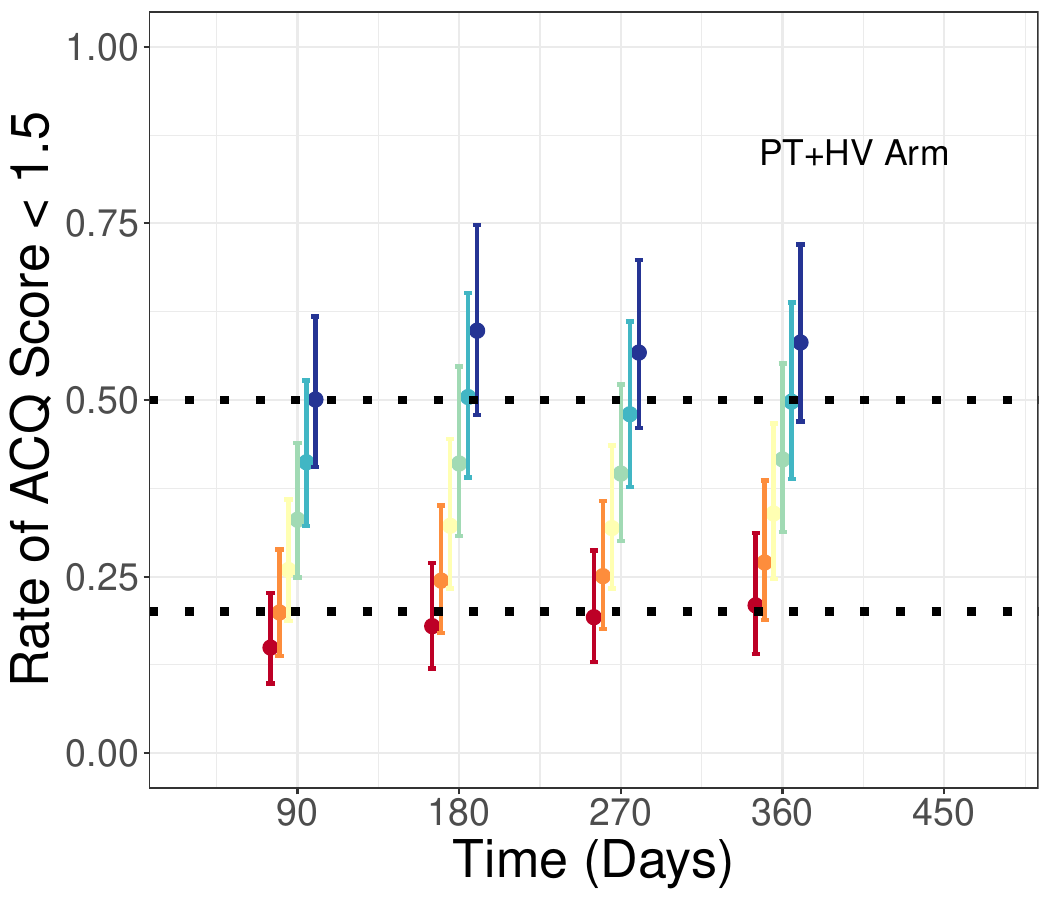}
   \end{subfigure}
\hfill
\begin{subfigure}[b]{0.49\textwidth}   
    \centering 
    \includegraphics[width=\textwidth]{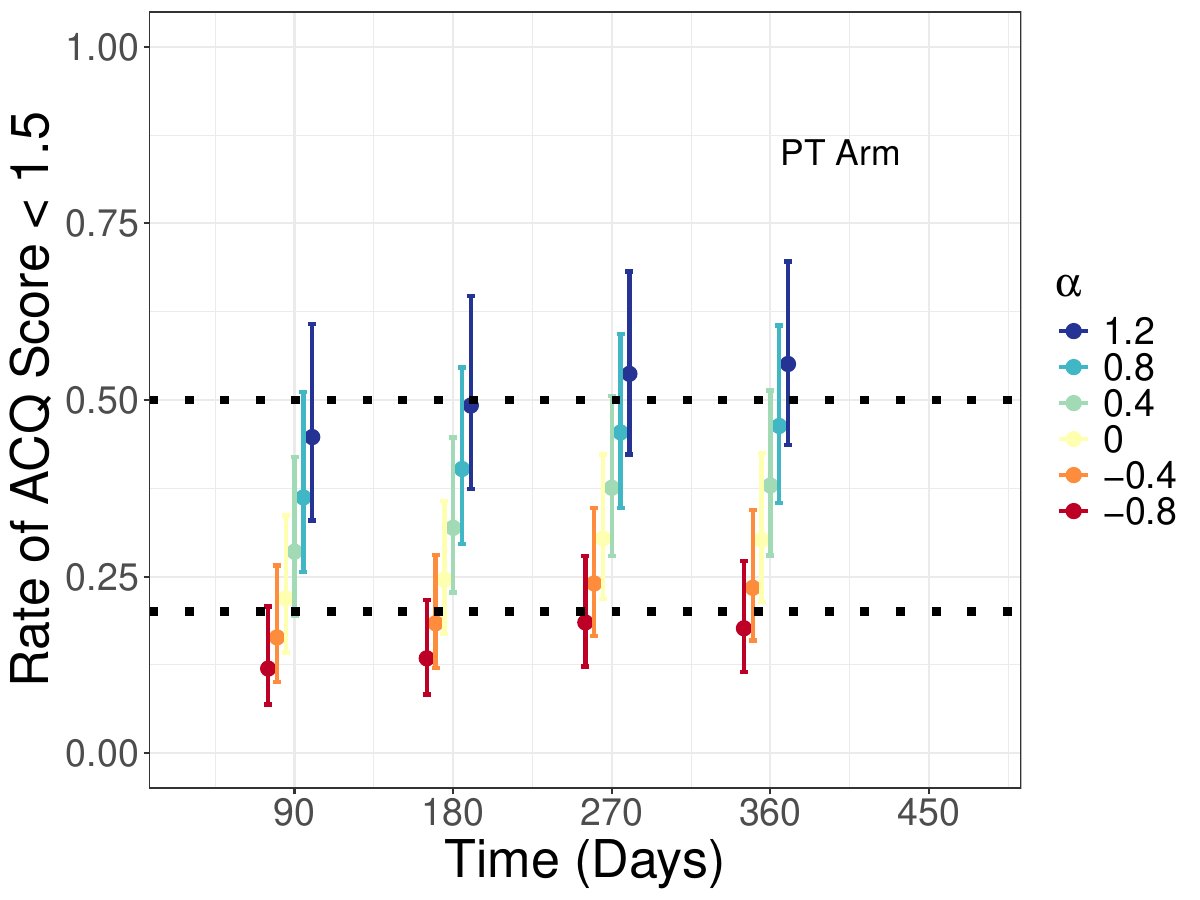}
\end{subfigure}

\caption[Population mean estimates of the proportion of patients with Asthma Control Questionnaire (ACQ) score less than 1.5 in the ARC trial under a range of sensitivity parameter values (logit link with weight $\symbf{W}_2(t;\symbf{\beta})$).]{Population mean estimates of the proportion of patients with Asthma Control Questionnaire (ACQ) score less than 1.5 in the ARC trial under a range of sensitivity parameter values (logit link with weight $\symbf{W}_2(t;\symbf{\beta})$). Here we show inference for the proportion of patients with ACQ score less than 1.5 in the intervention (PT+HV) arm and the control (PT) arm.  Estimation is made under sensitivity parameter values of $\alpha = -0.8, -0.4, 0, 0.4, 0.8, 1.2$, where a positive (negative) value of $\alpha$ posits that unobserved outcome values at time $t$ tend to be higher (lower) than observed values, after controlling for variables observed before time $t$. Upper panels: point estimates on the interval 60 to 460 days after randomization.  Lower panels: point estimates and 95\% Wald confidence intervals using the parametric bootstrap variance estimate at each target time of 90, 180, 270, and 360 days. For each arm, only those values of $\alpha$ under which the estimated curves lie completely between the dotted lines at lower bound 0.2 and upper bound 0.5 are considered plausible based on subject-matter expertise.} 
\label{fig:mean_curves_logit_weight2}
\end{figure}

\begin{figure}[ht]
\centering
\begin{subfigure}[b]{0.42\textwidth}
    \centering
    \includegraphics[width=\textwidth]{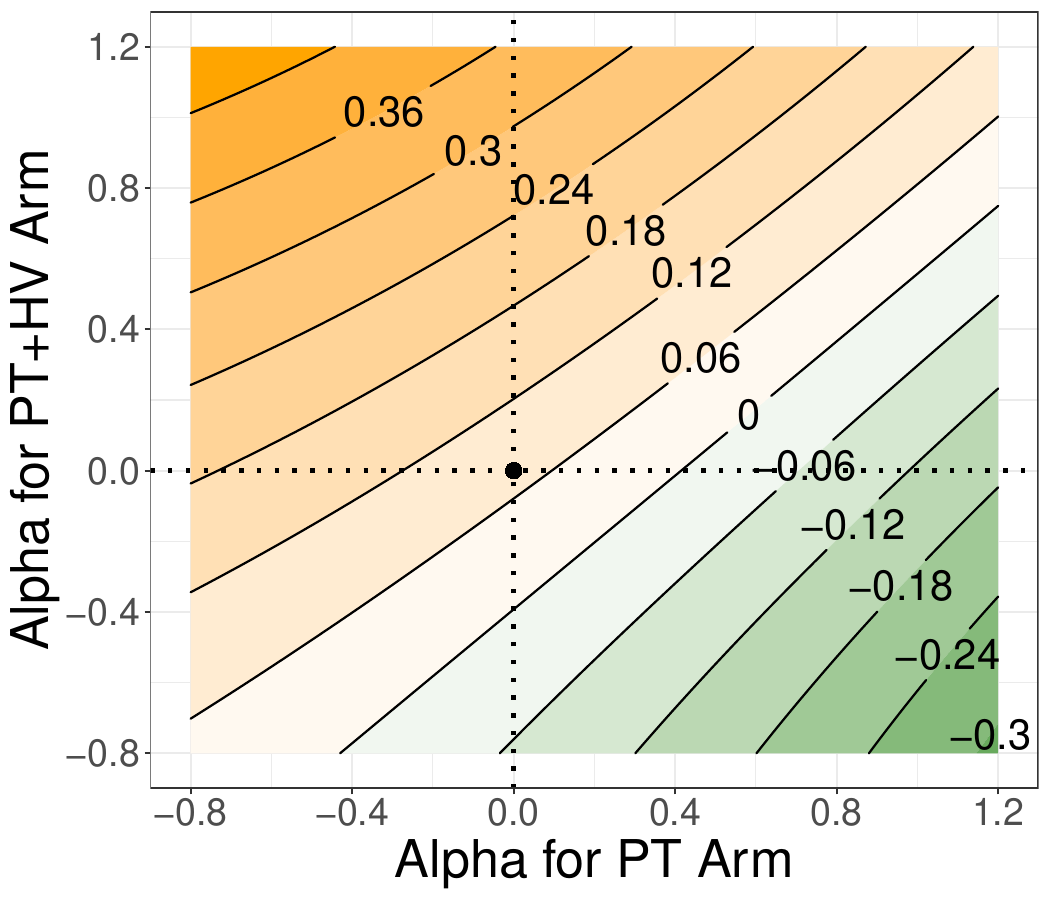}
\end{subfigure}
\hfill
\begin{subfigure}[b]{0.48\textwidth}  
    \centering 
    \includegraphics[width=\textwidth]{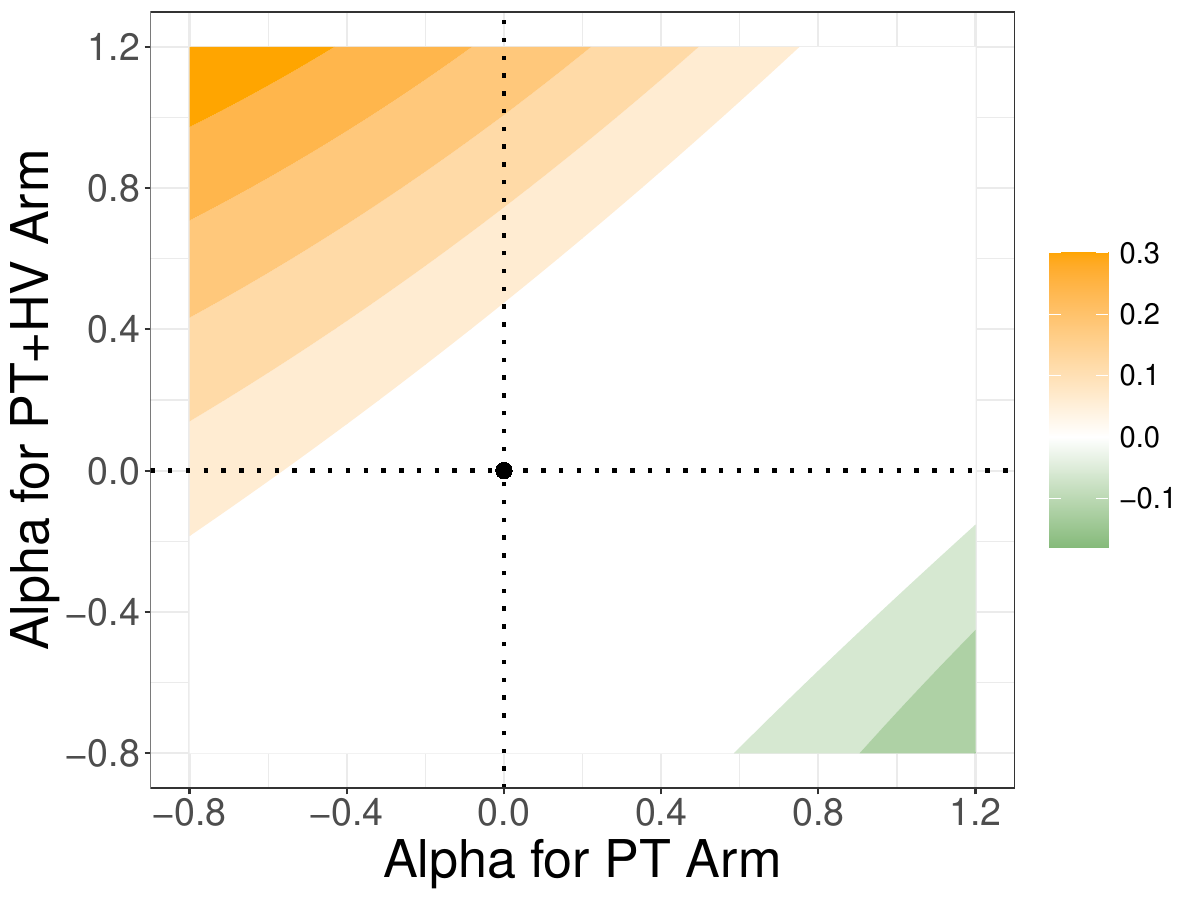}
\end{subfigure}

\medskip

\begin{subfigure}[b]{0.42\textwidth}   
    \centering 
    \includegraphics[width=\textwidth]{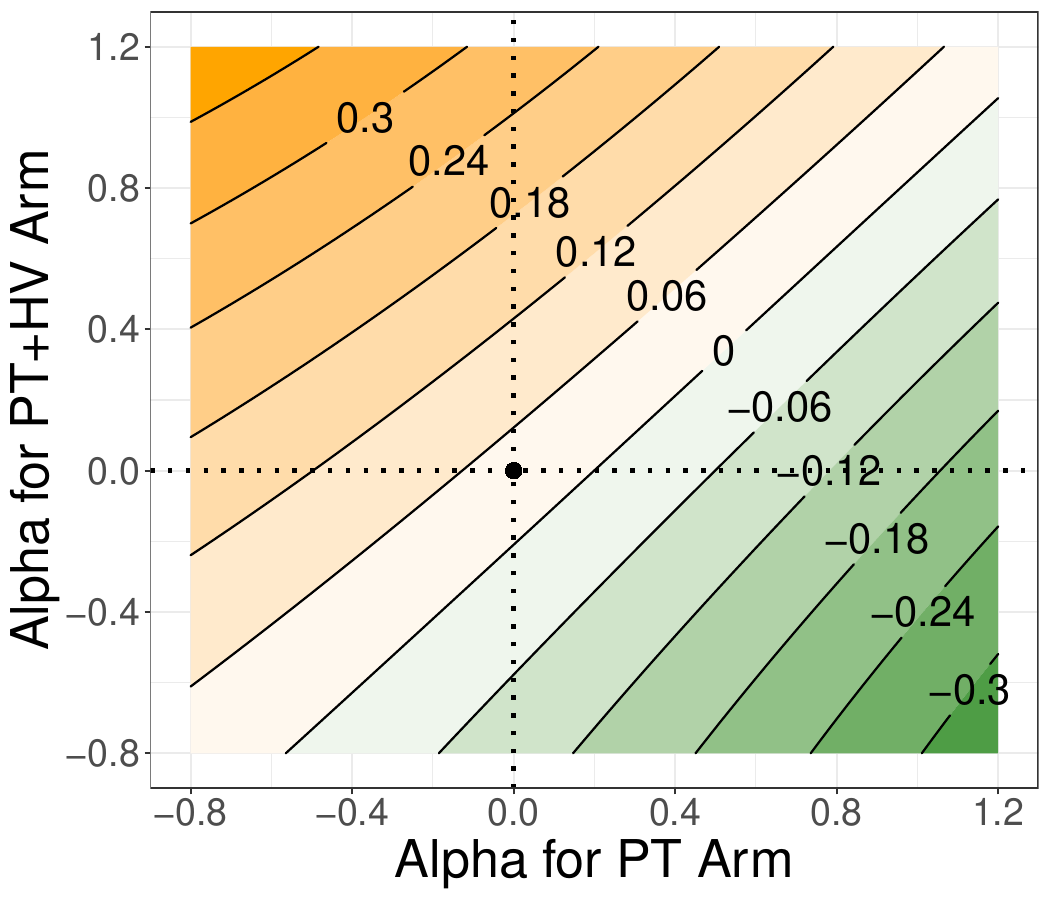}
   \end{subfigure}
\hfill
\begin{subfigure}[b]{0.48\textwidth}   
    \centering 
    \includegraphics[width=\textwidth]{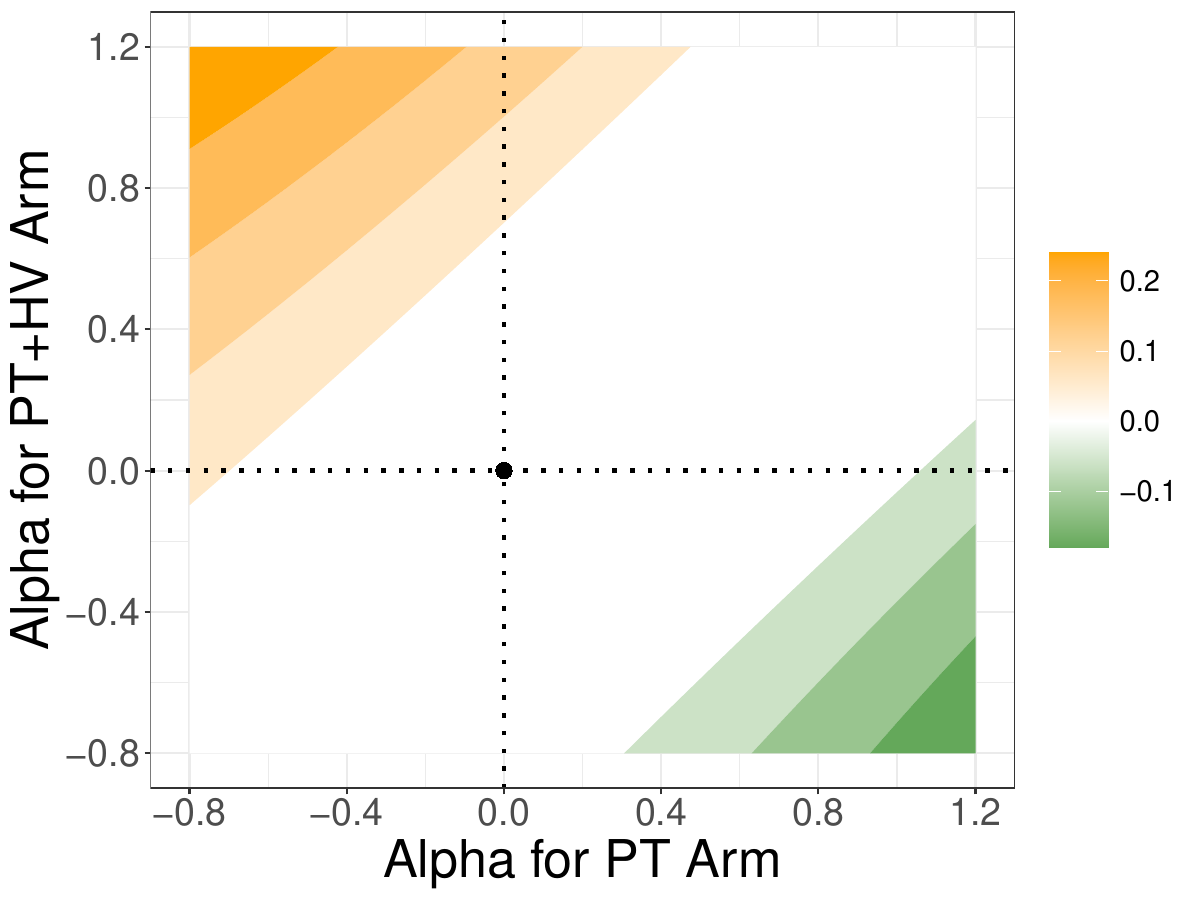}
\end{subfigure}
\caption[Sensitivity analysis for the proportion of patients with ACQ score less than 1.5 in the ARC trial (logit link with weight $\symbf{W}_2(t;\symbf{\beta})$).]{Sensitivity analysis for the proportion of patients with Asthma Control Questionnaire (ACQ) score less than 1.5 in the ARC trial (logit link with weight $\symbf{W}_2(t;\symbf{\beta})$). The left panels show the point estimate of the treatment effect $\delta(t)$ at 6 months (upper left) and 12 months (lower left) for each pair of sensitivity parameter values. The right panels summarize the 95\% confidence intervals for $\delta(t)$ at 6 months (upper right) and 12 months (lower right). The white region corresponds to sensitivity parameter values for which the confidence interval contains 0; the shaded region in the lower right (upper left) corresponds to values for which the confidence interval lies entirely below (above) 0. Contour heights give the value in the confidence interval closest to 0. Confidence intervals are Wald intervals based on the parametric bootstrap variance estimate.}
\label{fig:contour_logit_weight2}
\end{figure}

\section{Supplementary Simulation Study}
\label{apxB:appendix:simu}

\subsection{Simulation Mechanism and True Value}
\label{apxB:appendix:simu_mechanism}

This appendix describes the simulation procedure used to generate realistic datasets based on the ARC study and the corresponding true value calculation. The data-generating mechanism for the ACQ score is the same as that described in \cite{Smith2024}. For the asthma control indicator, however, we replaced the outcome model with a logistic regression model for $dF\{ y(t) \mid \Delta N(t)=1, \overline{\symbf{O}}(t)\}$.

Simulated datasets were generated separately for each treatment arm. For participant $i$, we first generated the baseline outcome $Y_i(0)$ by sampling from the empirical distribution of baseline outcomes in the corresponding treatment arm of the ARC data. We then alternated between generating subsequent assessment times and the corresponding outcomes. Specifically, given the $k$th assessment time $T_{ik}$ and the observed outcome $Y_i(T_{ik})$, we generated the $(k+1)$st assessment time, $T_{i,k+1}$, using Ogata's thinning algorithm \citep{ogata1981simulation} as follows:

\begin{enumerate}

\item Set $\lambda^*$ to be a value that is greater than or equal to $\sup \big\{ \widehat{\lambda}_{0,k}(t) \exp\{ \widehat{\gamma} \  Y_i(T_{ik}) \} : t \in [0,\tau] \big\}$, where $\widehat{\lambda}_{0,k}(t)$ and $\widehat{\gamma}$ are the estimates from the stratified Andersen-Gill model that we fit on the ARC data. 

\item Take the start time to be $T_{ik}$.  

\item Draw a potential gap time $s^* \sim Exp(\lambda^*)$, and set $t^* := $ start time $ + s^*$.  Accept the candidate assessment time $t^*$ with probability $ \widehat{\lambda}_{0,k}(t^*) \exp\{ \widehat{\gamma} \ Y_i(T_{ik}) \} /\lambda^*$.

\item  If $t^*$ is accepted, set $T_{i,k+1} := t^*$.  If $t^*$ is rejected, set the new start time to be $t^*$ and return to Step 3.  Iterate until a time $t^*$ is accepted, or until $t^* > \tau$, in which case $T_{ik}$ was participant $i$'s final assessment time. 
\end{enumerate}

Outcomes were generated from a fitted model for $dF\{y(t) \mid \Delta N(t)=1, \overline{\symbf{O}}(t)\}$ based on the ARC data. For the ACQ score, the outcome model was a negative binomial regression with predictors including time, the lag time since the previous assessment, and the outcome value at the previous assessment. For the asthma control indicator, the outcome model was a logistic regression with the same set of predictors. After generating the assessment time $T_{i,k+1}$, we obtained the estimated conditional distribution $d\widehat{F}\{ y(T_{i,k+1}) \mid \Delta N(T_{i,k+1})=1,\overline{\symbf{O}}_i(T_{i,k+1})\}$ 
using the fitted model for $dF\{y(t) \mid \Delta N(t)=1, \overline{\symbf{O}}(t)\}$. We then generated the outcome $Y_i(T_{i,k+1})$ by sampling from this estimated conditional distribution.

We computed the true value of $\symbf{\beta}$ using the identification formula provided in Proposition \ref{chp4:prop:identify}:
$$\symbf{\beta} = \left\{\int_{t=t_1}^{t_2} \symbf{B}(t) \symbf{B}(t)^\prime dt\right\}^{-1} \int_{t=t_1}^{t_2} \symbf{B}(t) s^{-1}\left\{\mu(t)\right\}dt,$$
under specific sensitivity parameter values. 
We simulated a large sample of $N=1,000,000$ participants.  For each time $t \in [t_1,t_2]$, we obtained the predicted value of 
$$
\mathbb{E}\{Y(t) \mid \overline{\symbf{O}}_i(t)\} = \frac{ \sum_{y(t)} y(t) \exp\{ \alpha y(t) \}dF\{ y(t) \mid \Delta N(t)=1,\overline{\symbf{O}}_i(t)\} }{ \sum_{y(t)} \exp\{ \alpha y(t) \}dF\{ y(t) \mid \Delta N(t)=1,\overline{\symbf{O}}_i(t)\} } 
$$
for each participant.  We then used the approximation $\mu(t) = \mathbb{E}[Y(t)] \approx \frac{1}{N} \sum_{i=1}^N \widehat{\mathbb{E}}[ Y(t) | \overline{\symbf{O}}_i(t) ]$.  We approximated the integral  $\int_{t=t_1}^{t_2} \symbf{B}(t) s^{-1}\{\mu(t) \}dt$ using intervals of one day.  The value of $\symbf{\beta}$  we obtain corresponds to the projection of the true curve of means onto spline curves of the form $\symbf{\beta}^\prime \symbf{B}(t)$. To obtain the true range of the sensitivity parameters, we first generated a large simulated dataset consisting of $N = 200,000$ participants and used the procedure described above to approximate the true value of $\mu(t)$. We then applied Algorithm \ref{alg:point_est} using a grid search method to identify the corresponding true upper and lower bounds of the sensitivity parameters.

\subsection{Supplementary Simulation Results for Section \ref{chp4:sec:simu}}
\label{apxB:appendix:simu_result}

This appendix presents additional simulation results obtained using the second weight function, $\symbf{W}_2(t;\symbf{\beta})$.
Overall, the results are qualitatively similar to those obtained using the first weight function, $\symbf{W}_1(t;\symbf{\beta})$.

Simulation results for the treatment effect in each treatment arm using the second weight function $\symbf{W}_2(t;\symbf{\beta})$ are shown in Tables \ref{table:coverage_log_weight2} and \ref{table:coverage_logit_weight2}. 
When the analyses were conducted using the true values of $\alpha_1$ and $\alpha_0$, the proposed estimators exhibited good finite-sample performance. Across 500 simulation replicates, the empirical bias was small, with absolute bias below 0.01 for the asthma control indicator and below 0.07 for the ACQ score. The empirical coverage probabilities of the nominal 0.95 confidence intervals were close to the target level, ranging from 0.944 to 0.976 for the asthma control indicator and from 0.926 to 0.970 for the ACQ score. In contrast, when the analyses assumed $\alpha_1 = \alpha_0 = 0$, the resulting estimators often exhibited substantial bias and the confidence interval coverage deviated markedly from the nominal 0.95 level.
This highlights the importance of considering a range of different assumptions through sensitivity analysis.


\begin{table}[!htbp]
\centering
\caption[Simulation results for the ACQ score with the log link ($\symbf{W}_2(t;\symbf{\beta})$).]{Simulation results for the ACQ score with the log link ($\symbf{W}_2(t;\symbf{\beta})$). Data were generated under our sensitivity analysis assumption using values of $\alpha_0,\alpha_1=-1.2, -0.8, -0.4, 0, 0.4, 0.8$. The treatment effects at 6 and 12 months were then estimated using our augmented inverse intensity-weighted estimator with the log link ($\symbf{W}_2(t;\symbf{\beta})$): (a) using the true values of $\alpha_0,\alpha_1$ (rows denoted ``S.A.''), and (b) under the explainable assessment assumption that $\alpha_0=\alpha_1=0$ (rows denoted ``Expl.''). Shown are the absolute value of the empirical bias and the confidence interval coverage across 500 simulations. Confidence intervals are Wald confidence intervals using the parametric bootstrap variance estimate.}
\label{table:coverage_log_weight2}
\footnotesize
\setlength{\tabcolsep}{3pt}
\begin{tabular}{ccc*{12}{c}}
\hline
 & & & \multicolumn{12}{c}{True $\alpha_0$} \\
\cline{4-15}
 & & & \multicolumn{2}{c}{$-1.2$} & \multicolumn{2}{c}{$-0.8$} & \multicolumn{2}{c}{$-0.4$}
     & \multicolumn{2}{c}{0} & \multicolumn{2}{c}{0.4} & \multicolumn{2}{c}{0.8} \\
\cline{4-15}
Month & True $\alpha_1$ & Method & Bias & Cov. & Bias & Cov. & Bias & Cov. & Bias & Cov. & Bias & Cov. & Bias & Cov. \\
\hline
\multirow{12}{*}{6}
 & \multirow{2}{*}{0.8}    & S.A.  & 0.045 & 0.958 & 0.041 & 0.960 & 0.038 & 0.964 & 0.030 & 0.960 & 0.019 & 0.958 & 0.004 & 0.958 \\
 &                         & Expl. & 2.155 & 0.000 & 1.916 & 0.000 & 1.573 & 0.000 & 1.073 & 0.000 & 0.408 & 0.304 & 0.356 & 0.412 \\ \cline{2-15}
 & \multirow{2}{*}{0.4}    & S.A.  & 0.024 & 0.966 & 0.020 & 0.968 & 0.017 & 0.964 & 0.010 & 0.968 & 0.002 & 0.962 & 0.024 & 0.954 \\
 &                         & Expl. & 1.571 & 0.000 & 1.333 & 0.000 & 0.989 & 0.000 & 0.490 & 0.156 & 0.176 & 0.834 & 0.940 & 0.000 \\ \cline{2-15}
 & \multirow{2}{*}{0}      & S.A.  & 0.013 & 0.964 & 0.009 & 0.952 & 0.006 & 0.948 & 0.002 & 0.952 & 0.013 & 0.958 & 0.036 & 0.950 \\
 &                         & Expl. & 1.080 & 0.000 & 0.842 & 0.002 & 0.498 & 0.148 & 0.002 & 0.952 & 0.667 & 0.032 & 1.431 & 0.000 \\ \cline{2-15}
 & \multirow{2}{*}{$-0.4$} & S.A.  & 0.006 & 0.960 & 0.002 & 0.950 & 0.002 & 0.944 & 0.009 & 0.942 & 0.020 & 0.950 & 0.043 & 0.938 \\
 &                         & Expl. & 0.697 & 0.010 & 0.459 & 0.214 & 0.115 & 0.884 & 0.384 & 0.346 & 1.050 & 0.000 & 1.814 & 0.000 \\ \cline{2-15}
 & \multirow{2}{*}{$-0.8$} & S.A.  & 0.002 & 0.960 & 0.002 & 0.962 & 0.006 & 0.946 & 0.013 & 0.940 & 0.024 & 0.942 & 0.047 & 0.940 \\
 &                         & Expl. & 0.412 & 0.294 & 0.174 & 0.830 & 0.170 & 0.838 & 0.670 & 0.028 & 1.335 & 0.000 & 2.099 & 0.000 \\ \cline{2-15}
 & \multirow{2}{*}{$-1.2$} & S.A.  & 0.003 & 0.960 & 0.007 & 0.960 & 0.010 & 0.946 & 0.018 & 0.936 & 0.029 & 0.946 & 0.052 & 0.944 \\
 &                         & Expl. & 0.197 & 0.804 & 0.041 & 0.938 & 0.385 & 0.346 & 0.884 & 0.000 & 1.550 & 0.000 & 2.314 & 0.000 \\
\hline
\multirow{12}{*}{12}
 & \multirow{2}{*}{0.8}    & S.A.  & 0.032 & 0.954 & 0.028 & 0.962 & 0.022 & 0.964 & 0.012 & 0.958 & 0.003 & 0.960 & 0.011 & 0.940 \\
 &                         & Expl. & 2.028 & 0.000 & 1.805 & 0.000 & 1.485 & 0.000 & 1.019 & 0.000 & 0.387 & 0.302 & 0.363 & 0.362 \\ \cline{2-15}
 & \multirow{2}{*}{0.4}    & S.A.  & 0.018 & 0.944 & 0.014 & 0.950 & 0.008 & 0.956 & 0.002 & 0.964 & 0.011 & 0.966 & 0.025 & 0.942 \\
 &                         & Expl. & 1.466 & 0.000 & 1.243 & 0.000 & 0.922 & 0.000 & 0.456 & 0.188 & 0.175 & 0.780 & 0.926 & 0.000 \\ \cline{2-15}
 & \multirow{2}{*}{0}      & S.A.  & 0.011 & 0.946 & 0.006 & 0.944 & 0.000 & 0.950 & 0.010 & 0.964 & 0.019 & 0.970 & 0.033 & 0.952 \\
 &                         & Expl. & 1.000 & 0.000 & 0.777 & 0.000 & 0.456 & 0.188 & 0.010 & 0.964 & 0.641 & 0.010 & 1.391 & 0.000 \\ \cline{2-15}
 & \multirow{2}{*}{$-0.4$} & S.A.  & 0.006 & 0.940 & 0.001 & 0.948 & 0.004 & 0.954 & 0.015 & 0.966 & 0.024 & 0.956 & 0.038 & 0.944 \\
 &                         & Expl. & 0.641 & 0.006 & 0.418 & 0.248 & 0.097 & 0.924 & 0.369 & 0.346 & 1.000 & 0.000 & 1.750 & 0.000 \\ \cline{2-15}
 & \multirow{2}{*}{$-0.8$} & S.A.  & 0.001 & 0.942 & 0.005 & 0.948 & 0.011 & 0.956 & 0.021 & 0.952 & 0.030 & 0.948 & 0.044 & 0.932 \\
 &                         & Expl. & 0.370 & 0.344 & 0.147 & 0.858 & 0.173 & 0.788 & 0.639 & 0.010 & 1.271 & 0.000 & 2.021 & 0.000 \\ \cline{2-15}
 & \multirow{2}{*}{$-1.2$} & S.A.  & 0.007 & 0.940 & 0.011 & 0.938 & 0.017 & 0.942 & 0.027 & 0.944 & 0.036 & 0.944 & 0.050 & 0.926 \\
 &                         & Expl. & 0.167 & 0.820 & 0.056 & 0.946 & 0.376 & 0.340 & 0.842 & 0.000 & 1.474 & 0.000 & 2.224 & 0.000 \\
\hline
\end{tabular}
\end{table}



\begin{table}[!htbp]
\centering
\caption[Simulation results for the AC indicator with the logit link ($\symbf{W}_2(t;\symbf{\beta})$).]{Simulation results for the AC indicator with the logit link ($\symbf{W}_2(t;\symbf{\beta})$). Data were generated under our sensitivity analysis assumption using values of $\alpha_0,\alpha_1 = -0.8, -0.4, 0, 0.4, 0.8, 1.2$. The treatment effects at 6 and 12 months were then estimated using our augmented inverse intensity-weighted estimator with the logit link ($\symbf{W}_2(t;\symbf{\beta})$): (a) using the true values of $\alpha_0,\alpha_1$ (rows denoted ``S.A.''), and (b) under the explainable assessment assumption that $\alpha_0=\alpha_1=0$ (rows denoted ``Expl.''). Shown are the absolute value of the empirical bias and the confidence interval coverage across 500 simulations. Confidence intervals are Wald confidence intervals using the parametric bootstrap variance estimate.}
\label{table:coverage_logit_weight2}
\footnotesize
\setlength{\tabcolsep}{3pt}
\begin{tabular}{ccc*{12}{c}}
\hline
 & & & \multicolumn{12}{c}{True $\alpha_0$} \\
\cline{4-15}
 & & & \multicolumn{2}{c}{$-0.8$} & \multicolumn{2}{c}{$-0.4$} & \multicolumn{2}{c}{0}
     & \multicolumn{2}{c}{0.4} & \multicolumn{2}{c}{0.8} & \multicolumn{2}{c}{1.2} \\
\cline{4-15}
Month & True $\alpha_1$ & Method & Bias & Cov. & Bias & Cov. & Bias & Cov. & Bias & Cov. & Bias & Cov. & Bias & Cov. \\
\hline
\multirow{12}{*}{6}
 & \multirow{2}{*}{1.2}    & S.A.  & 0.004 & 0.972 & 0.004 & 0.974 & 0.004 & 0.966 & 0.005 & 0.962 & 0.006 & 0.964 & 0.005 & 0.966 \\
 &                         & Expl. & 0.365 & 0.000 & 0.315 & 0.000 & 0.255 & 0.010 & 0.184 & 0.158 & 0.105 & 0.634 & 0.017 & 0.944 \\ \cline{2-15}
 & \multirow{2}{*}{0.8}    & S.A.  & 0.004 & 0.970 & 0.004 & 0.976 & 0.004 & 0.968 & 0.005 & 0.962 & 0.006 & 0.962 & 0.005 & 0.966 \\
 &                         & Expl. & 0.278 & 0.002 & 0.228 & 0.046 & 0.167 & 0.220 & 0.097 & 0.698 & 0.017 & 0.944 & 0.070 & 0.808 \\ \cline{2-15}
 & \multirow{2}{*}{0.4}    & S.A.  & 0.003 & 0.964 & 0.003 & 0.966 & 0.003 & 0.964 & 0.004 & 0.960 & 0.005 & 0.962 & 0.005 & 0.960 \\
 &                         & Expl. & 0.192 & 0.126 & 0.143 & 0.376 & 0.082 & 0.788 & 0.011 & 0.952 & 0.068 & 0.814 & 0.155 & 0.274 \\ \cline{2-15}
 & \multirow{2}{*}{0}      & S.A.  & 0.003 & 0.956 & 0.003 & 0.954 & 0.003 & 0.954 & 0.004 & 0.958 & 0.005 & 0.954 & 0.004 & 0.954 \\
 &                         & Expl. & 0.114 & 0.576 & 0.064 & 0.856 & 0.003 & 0.954 & 0.068 & 0.818 & 0.147 & 0.334 & 0.234 & 0.044 \\ \cline{2-15}
 & \multirow{2}{*}{$-0.4$} & S.A.  & 0.004 & 0.948 & 0.003 & 0.954 & 0.004 & 0.950 & 0.005 & 0.958 & 0.006 & 0.954 & 0.005 & 0.946 \\
 &                         & Expl. & 0.045 & 0.910 & 0.005 & 0.954 & 0.066 & 0.828 & 0.136 & 0.398 & 0.216 & 0.066 & 0.303 & 0.004 \\ \cline{2-15}
 & \multirow{2}{*}{$-0.8$} & S.A.  & 0.002 & 0.946 & 0.002 & 0.944 & 0.003 & 0.948 & 0.003 & 0.950 & 0.004 & 0.952 & 0.004 & 0.944 \\
 &                         & Expl. & 0.015 & 0.952 & 0.065 & 0.832 & 0.126 & 0.468 & 0.197 & 0.120 & 0.276 & 0.016 & 0.363 & 0.000 \\
\hline
\multirow{12}{*}{12}
 & \multirow{2}{*}{1.2}    & S.A.  & 0.001 & 0.958 & 0.002 & 0.954 & 0.002 & 0.952 & 0.002 & 0.954 & 0.003 & 0.950 & 0.003 & 0.946 \\
 &                         & Expl. & 0.366 & 0.000 & 0.311 & 0.000 & 0.246 & 0.012 & 0.171 & 0.204 & 0.089 & 0.692 & 0.002 & 0.954 \\ \cline{2-15}
 & \multirow{2}{*}{0.8}    & S.A.  & 0.002 & 0.956 & 0.002 & 0.950 & 0.002 & 0.954 & 0.002 & 0.954 & 0.003 & 0.956 & 0.003 & 0.954 \\
 &                         & Expl. & 0.281 & 0.000 & 0.226 & 0.026 & 0.160 & 0.272 & 0.086 & 0.708 & 0.003 & 0.952 & 0.083 & 0.698 \\ \cline{2-15}
 & \multirow{2}{*}{0.4}    & S.A.  & 0.003 & 0.958 & 0.003 & 0.952 & 0.003 & 0.954 & 0.003 & 0.956 & 0.004 & 0.954 & 0.004 & 0.950 \\
 &                         & Expl. & 0.196 & 0.096 & 0.141 & 0.372 & 0.076 & 0.756 & 0.001 & 0.958 & 0.081 & 0.708 & 0.168 & 0.230 \\ \cline{2-15}
 & \multirow{2}{*}{0}      & S.A.  & 0.003 & 0.952 & 0.003 & 0.952 & 0.003 & 0.954 & 0.003 & 0.952 & 0.004 & 0.954 & 0.004 & 0.954 \\
 &                         & Expl. & 0.117 & 0.530 & 0.062 & 0.812 & 0.003 & 0.952 & 0.078 & 0.730 & 0.161 & 0.256 & 0.247 & 0.024 \\ \cline{2-15}
 & \multirow{2}{*}{$-0.4$} & S.A.  & 0.003 & 0.954 & 0.003 & 0.956 & 0.004 & 0.952 & 0.004 & 0.948 & 0.004 & 0.950 & 0.004 & 0.952 \\
 &                         & Expl. & 0.046 & 0.896 & 0.009 & 0.950 & 0.075 & 0.746 & 0.149 & 0.318 & 0.232 & 0.032 & 0.318 & 0.000 \\ \cline{2-15}
 & \multirow{2}{*}{$-0.8$} & S.A.  & 0.004 & 0.958 & 0.004 & 0.952 & 0.004 & 0.956 & 0.004 & 0.950 & 0.005 & 0.950 & 0.005 & 0.950 \\
 &                         & Expl. & 0.017 & 0.940 & 0.072 & 0.768 & 0.137 & 0.382 & 0.211 & 0.060 & 0.294 & 0.000 & 0.381 & 0.000 \\
\hline
\end{tabular}
\end{table}



\end{document}